\documentclass[12pt,letterpaper]{article}

\usepackage{amssymb,amsthm,amsmath,multirow,amsfonts, subcaption,authblk,enumerate,setspace,pdflscape,lscape,comment,ulem,upgreek,graphicx}
\RequirePackage{paralist,microtype,dsfont,longtable,booktabs,subcaption,xcolor,ragged2e}
\usepackage[T1]{fontenc}
\usepackage{lmodern}
\usepackage[margin=1.25in]{geometry}
\usepackage[compact]{titlesec}
\usepackage{yfonts}
\usepackage{appendix} 
\usepackage{chngcntr} 
\usepackage[round,authoryear]{natbib}
\usepackage[colorlinks=false,linkcolor=red,urlcolor=blue,citecolor=blue,anchorcolor=blue,
pdftex,breaklinks,pdfencoding=auto,psdextra,
bookmarksopenlevel=1,bookmarksopen=true]{hyperref}

\hypersetup{
	breaklinks=true, 
	bookmarksopen=true,
	bookmarksnumbered=true,
	linkcolor=blue,
	urlcolor=blue,
	citecolor=blue,
	colorlinks=true}

\usepackage{booktabs}
\usepackage{graphicx,epstopdf,color}
\graphicspath{{graphics/}}
\usepackage{subcaption} 
\makeatletter
\DeclareRobustCommand\citepos
{\begingroup\def\NAT@nmfmt##1{{\NAT@up##1's}}%
	\NAT@swafalse\let\NAT@ctype\z@\NAT@partrue
	\@ifstar{\NAT@fulltrue\NAT@citetp}{\NAT@fullfalse\NAT@citetp}}
\makeatother

\makeatletter
\renewcommand*{\eqref}[1]{\hyperref[{#1}]{\textup{\tagform@{\ref*{#1}}}}}
\makeatother

\usepackage[inline,shortlabels]{enumitem}

\expandafter
\def \expandafter \normalsize \expandafter{\normalsize \setlength \abovedisplayskip{10pt plus 2pt minus 7pt}}
\expandafter
\def \expandafter \normalsize \expandafter{\normalsize \setlength \abovedisplayshortskip{0pt plus 2pt}}
\expandafter
\def \expandafter \normalsize \expandafter{\normalsize \setlength \belowdisplayskip{10pt plus 2pt minus 7pt}}
\expandafter
\def \expandafter \normalsize \expandafter{\normalsize \setlength \belowdisplayshortskip{5pt plus 2pt minus 3pt}}

\theoremstyle{plain}
\newtheorem{theorem}{Theorem}
\newtheorem{lemma}{Lemma}
\newtheorem{corollary}{Corollary}
\theoremstyle{definition}

\newtheorem{theoremA}{Theorem}

\newtheorem{theoremAA}{Theorem}

\newtheorem{rmk}{Remark}
\newenvironment{remark}{\begin{rmk}}{\hfill $\square$ \end{rmk}}
\newtheorem{example}{Example}

\newtheorem{assumFLS}{Assumption}

\newtheorem{assumIV}{Assumption}

\newtheorem{assumpA}{Assumption}

\newtheorem{assumpB}{Assumption}
\numberwithin{assumpB}{subsection}

\DeclareMathOperator{\sgn}{sgn}
\DeclareMathOperator{\IRF}{SIRF}
\DeclareMathOperator{\DDD}{\Delta}
\DeclareMathOperator{\KK}{\mathrm{K}_{\reg}}
\DeclareMathOperator{\KKK}{\mathrm{K}_{\reg}}
\newcommand{\lambdatw}{\bar{\lambda}}
\newcommand{\vtw}{\bar{v}}

\DeclareMathOperator{\op}{op}
\newcommand{\HS}{\text{HS}}
\newcommand{\CC}{\mathtt{C}}

\newcommand{\elltwo}{\widetilde{\mathcal H}}
\newcommand{\XX}{\Upsilon}
\newcommand{\ZZ}{\mathcal Z}
\newcommand{\llambda}{\nu}

\newcommand{\ww}{\varphi}
\newcommand{\wv}{\phi}

\newcommand{\CCC}{\Gamma}
\newcommand{\PP}{\mathcal{P}}
\renewcommand{\SS}{\mathcal{S}}
\newcommand{\SA}{\mathrm{S}}
\newcommand{\nutw}{\bar{\nu}}
\newcommand{\wtw}{\bar{\phi}}
\newcommand{\wwtw}{\bar{\varphi}}
\newcommand{\PPI}{\scalebox{1.3}{$\uppi$}}
\newcommand{\varph}{\vartheta}
\newcommand{\bdw}{\mathsf{b}}

\newcommand{\pcr}{P_{\mathbb{R}}}
\newcommand{\pcrr}{P_{\mathbb{R}}^\ast}
\newcommand{\pch}{P_{\mathcal{H}}}
\newcommand{\pchh}{P_{\mathcal{H}}^\ast}
\renewcommand{\tilde}{\widetilde}
\renewcommand{\hat}{\widehat}
\newcommand{\UU}{\mathcal U}
\newcommand{\UUU}{\mathfrak U}
\newcommand{\UUUU}{U}
\newcommand{\reg}{\tau}

\numberwithin{equation}{section}

\newtheorem{proposition}{Proposition}

\titleformat{\paragraph}
{\normalfont\normalsize\bfseries}{\theparagraph}{1em}{}
\titlespacing*{\paragraph}{0pt}{1.5ex plus 1.0ex minus .2ex}{0.1ex plus .05ex}

\makeatletter
\usepackage{xr}
\newcommand*{\addFileDependency}[1]{% argument=file name and extension
	\typeout{(#1)}
	\@addtofilelist{#1}
	\IfFileExists{#1}{}{\typeout{No file #1.}}
}
\makeatother

\usepackage{xr}
\titlespacing*{\subsubsection}{0pt}{2.5ex plus 1ex minus .2ex}{1.3ex plus .2ex}

\begin{document}
	\title{\Large Functional Linear Projection and Impulse Response Analysis\thanks{We are grateful to \`{O}scar Jord\`{a}, Bonsoo Koo, Benjamin Wong, and seminar participants at the 18th International Symposium on Econometric Theory and Applications (SETA2024) and 33rd Australian New Zealand Econometric Study Group Meeting for their invaluable comments.} }
	\author{	Won-Ki Seo\textsuperscript{a}   \quad\quad Dakyung Seong\textsuperscript{a}\thanks{Corresponding author. Address: School of Economics, University of Sydney, Camperdown, 2006, NSW, Australia. E-mail addresses: \texttt{dakyung.seong@sydney.edu.au (D.\ Seong), \texttt{won-ki.seo@sydney.edu.au (W.-K.\ Seo)}}}\\	\large{\textsuperscript{a} School of Economics, University of Sydney}
	}
	
	\maketitle
	\vspace{-1em}
	\begin{abstract}
		This paper proposes econometric methods for studying how economic variables  respond to function-valued shocks. Our methods are developed based on linear projection estimation of predictive regression models with a function-valued predictor and  other control variables. We show that the linear projection coefficient associated with the functional variable allows for the impulse response interpretation in a functional structural vector autoregressive model under a certain identification scheme, similar to well-known \citeauthor{Sims1972}' (\citeyear{Sims1972}) causal chain, but with nontrivial complications in our functional setup. A novel estimator based on an operator Schur complement is proposed and its asymptotic properties are studied. We illustrate its empirical applicability with two examples involving functional variables:  economy sentiment distributions and functional monetary policy shocks.  
	\end{abstract}

	\section{Introduction}\label{sec: intro} 
	
	After the initial development by \cite{Oscar2005}, the \textit{local projection} approach has become one of the foremost applications of time series analysis in the fields of macroeconomics and public policy for studying dynamic causal effects, known as structural impulse responses.  The work by \cite{PW2021}  provides more rigorous grounding for the approach as an important complement to   structural vector autoregressive (SVAR) models, by establishing the theoretical equivalence between  estimation results of the two. Practitioners have benefited from this result, as the local projection approach is easier to implement and interpret, computationally less burdensome, and less sensitive to model misspecification (see \citealp{Oscar2005} for a detailed discussion).

	This paper contributes to recent developments in this area by providing statistical inference methodologies to  study responses of target variables when  shocks are characterized by function-valued random variables, such as the functional monetary policy shocks in \cite{IR2021}. We consider a linear projection of a target variable onto the space spanned by the variables in the conditioning set, including function-valued  shocks. Similar to  \cite{Oscar2005}, this setup simplifies  impulse response analysis into the estimation of a regression model.  Our first result is the formal theoretical evidence for interpreting the parameters of interest in our regression model as structural impulse response functions of SVAR models involving a functional variable, under a certain identification condition. Considering the scarcity of the literature, our paper will be a valuable first step toward formally understanding functional structural shocks and their impact.  Furthermore,  while the equivalence for this special case has been formally established by \cite{Oscar2005} for finite dimensional VAR models, its extension to the case with a functional covariate has not yet been fully explored.

	Our benchmark model to be studied is closely related to the \textit{functional local projection} approach proposed by \cite{IR2021},  but a crucial distinction exists. Specifically, in contrast to them,  we do not require a parametric assumption on the structure of function-valued shocks and their coefficients. In  \cite{IR2021},  functional shocks are represented by a few known functional factors. This assumption crucially reduces the dimensionality of the variables to be analyzed, thereby making existing estimation and inference methodologies valid despite the infinite dimensionality of functional data. However,  such a parametric structure  is not always available in practice and could result in model misspecification; this point will be further demonstrated in Sections~\ref{sec: model} and~\ref{subsec: para} by using an example.    Considering  practical challenges  of exact identification of those factors in infinite dimensional function spaces, our approach could be an appealing alternative for practitioners.

	The current paper is closely related to recently growing studies on the VAR model involving a functional covariate, including \cite{IR2021},  \cite{BHCJ2023}, and \cite{Chang_Chen_Schorfheide_2024}. However, our focus on identification and the relationship between SVAR models and linear projection in the presence of a functional covariate distinguishes the current paper from theirs.  Moreover, the aforementioned studies typically approximate a functional covariate using a few factors or basis functions and then apply identification strategies developed in the finite dimensional SVAR literature (e.g.,\ \citealp{kilian2017structural}),  assuming that the approximation error is either zero or negligible as the sample size increases. Practitioners might prefer these approaches due to their ease of implementation. However, this simplified method does not always ensure the identification of infinite dimensional structural parameters on the entire space in which the parameters take values, thereby restricting our ability to fully comprehend the features of structural parameters. This point will be detailed in Section~\ref{subsec: para} by using an example.

	We propose an estimator  based on the operator Schur complement and establish its asymptotic properties, including  (local) asymptotic normality. As our model belongs to the so-called  scalar-on-function models studied in, e.g., \cite{Hall2007}, \cite{SHIN2009}, \cite{Florence2015},  \cite{imaizumi2018}, and \cite{Babii2022}, our estimator can be understood as a complement to the existing ones developed therein.  However,  existing estimators   (i) first reduce the dimension of functional predictors  and then use the dimension-reduced ones in the analysis (e.g., \citealp{AP2006,SHIN2009}), (ii) are designed without scalar-valued covariates (e.g., \citealp{Hall2007, Florence2015, imaizumi2018}), or (iii) do not account for time series dependence, which is crucial in our context (except for \citealp{Babii2022}, among the aforementioned articles). The approaches in (i) may not be preferred because they do not take into account covariation across variables in the dimension reduction procedure, while the latter two may limit practical applicability. Our estimator complements them by using a regularization method that allows us to account for covariation between scalar- and function-valued variables. This point is particularly relevant if the function-valued variable is given by a regressor without a structural interpretation, as in Section~\ref{sec: svar}. Later in Appendix~\ref{sec:est2}, we further extend our method to accommodate an endogenous functional predictor.  %functional instrumental variable (IV) estimation 

	As discussed in \cite{Hall2007}, \cite{Florence2015}, and references therein, estimating functional linear regression models intrinsically involves an inverse problem similar to those in nonparametric estimation. Although we focus on a model that is linear in functional variables, these function-valued random variables can be represented by an infinite number of basis functions. In this regard, our paper is also related to studies on nonlinear impulse responses in, e.g., \cite{Oscar2005}, \cite{kilian2017structural}, and \cite{KP2024}.  
	
	Although our main focus is on estimating impulse response coefficients when shocks are characterized by functions, our methodology could also be an empirically appealing alternative to prediction methods, such as that in \cite{Barbaglia2023}, in that it allows us to utilize richer information, such as distributional information or observations at different frequencies. In particular,  if a functional variable is given by a predictor without a structural interpretation, our model reduces to a predictive regression model with a functional predictor. As studied in \cite{Babii2022} and \cite{seong2021functional}, functional predictors could lead to better forecasting outcomes by exploiting information overlooked in conventional estimation approaches.

	In our empirical application, we first employ the quantile curve of the economic sentiment measure proposed by \cite{Barbaglia2023} and study the impact of perturbations given to the sentiment distribution on Total Nonfarm Payrolls in the US, without a structural interpretation. The sentiment measure is observed at a daily frequency and describes the presence of negative or positive tones in relation to specific terms, such as \textsl{economy} or \textsl{inflation}. The proposed estimator shows that the response depends not only on the magnitude of the perturbation, but also on its shape, which has not been observed previously. Moreover, we find that location shifts in the sentiment distribution have significant effects on the target variable in the short term. For instance, if the overall sentiment distribution shifts to the left, the average sentiment on the economy becomes more pessimistic, predicting significantly negative economic growth in the near future.

	We then revisit the work by \cite{IR2021} on the impact of monetary policy shocks, which are characterized by shifts in yield curves on monetary policy announcement dates. We find that during conventional periods, inflation responses are generally insignificant; however, their impact tends to vary depending on how the shock affects interest rates at short or long maturities. That is, the impact of functional monetary policy shocks depends on their shapes and directions.

	The rest of the paper is organized as follows. Section~\ref{sec: model} motivates our benchmark model and provides examples of  functional variables. In Section~\ref{sec: svar}, we consider the SVAR model involving a functional variable and discuss when the parameters in our model can be interpreted as functional structural impulse responses. Section~\ref{sec:est} proposes our estimator and establishes its asymptotic properties.  We apply our estimators to the empirical data in Section~\ref{sec:emp}. Section~\ref {sec:sim} summarizes simulation results.  Section~\ref{sec:con} concludes.  The appendix includes  proofs and an extension of the main theoretical results to instrumental variable (IV) estimation. In the Supplementary Material, we provide mathematical preliminaries and extension to functional SVAR models.

	\section{Model} \label{sec: model}
	We let $y_t$ denote a scalar-valued target (dependent) variable and $\mathbf{w}_t= ( {w}_{1,t},\ldots, {w}_{m,t})'$ be a vector of exogenous (scalar-valued) control variables possibly containing lagged $y_t$'s. The variable  $X_t = \{X_t (r) : r \in [0,1]\} $ denotes a \textit{function-valued} variable that takes values in some function space $\mathcal H$ to be detailed shortly.\footnote{For the ease of explanation, we suppose  that the domain of $X_t$ is $[0,1]$, but its extension to any arbitrary compact interval $[a,b]$ is straightforward.} Examples of $X_t$ include, but are not limited to, sentiment quantile curves in Example~\ref{example1} and functional monetary policy shocks in \cite{IR2021}. In this paper, we are particularly interested in the response of $y_{t+h}$ when an additional perturbation (or shock) $\zeta$ is introduced to $X_t$; that is, for any $x \in \mathcal H$, we are interested in \begin{equation}
		%	\mathbb E\left[y_{t+h} |X_t = x+ \zeta , \{ \mathbf{w}_{t-i} \}_{i=1} ^p\right] - \mathbb E\left[y_t | X_t = x, \{ \mathbf{w}_{t-i} \}_{i=1} ^p\right] = \int_{0}^{1} \beta_h(r)\zeta(r)dr = \langle \zeta, \beta_h\rangle_{\mathcal H}.  \label{def: local: irf} 
		\mathbb E\left[y_{t+h} |X_t = x+ \zeta , \mathbf{w}_{t} \right] - \mathbb E\left[y_{t+h} | X_t = x, \mathbf{w}_{t}\right]. \label{eq: irf: general}
	\end{equation}
	In previous studies on time series analysis, including \cite{Oscar2005} and \cite{PW2021}, $X_t$, the variable exposed to a perturbation, is characterized by a scalar or a vector, and the linear projection of $y_t$ onto the space spanned by $\{X_t, \mathbf w_t\}$ reduces  the impulse response analysis to inference of  projection coefficients. However, its extension to the case where $X_t$ is given by a function is not  straightforward. For example, in our setup, the perturbation $\zeta$ in \eqref{eq: irf: general} is specified as a function on $[0,1]$. Thus,  the response of $y_t$ in \eqref{eq: irf: general}  depends not only on the magnitude of $\zeta$, but also on its shape. This feature has not been observed in the conventional studies on SVAR models. %Consequently, there is no canonical choice of perturbation in this setting, such as a ``unit shock'' in the Euclidean space context. This also holds true in the functional impulse response analysis to be developed. 

	We follow \cite{Oscar2005} and consider the following linear projection of the scalar target variable $y_t$ onto the space spanned by the variables in the conditioning set, including the function-valued $X_t$: \begin{equation}
		y_{t+h} =  \int_{0}^{1} \beta_h(r)X_t(r) dr +   \mathbf{w}_{t} ' \alpha_h  % \mathbf{w}_{t}'\gamma_h 
		+  u_{h,t}. \label{eq: model: benchmark}
	\end{equation}
	The functional covariate $X_t$ and its coefficient $\beta_h$ are assumed to take values in $\mathcal H=L^2[0,1]$ (the Hilbert space of square-integrable functions), which allows us to conveniently write $\int_{0}^{1} \beta_h(r)X_t(r) dr$ as the inner product $\langle \beta_h,  X_t \rangle$ defined on $\mathcal H$.  
	
	Under the linear projection in \eqref{eq: model: benchmark}, we  identify   \eqref{eq: irf: general} as follows:\begin{equation} 
		\mathbb E\left[y_{t+h} |X_t = x+ \zeta , \mathbf{w}_{t} \right] - \mathbb E\left[y_t | X_t = x, \mathbf{w}_{t}\right] = \int_{0}^{1} \beta_h(r)\zeta(r)dr = \langle  \beta_h,\zeta\rangle.  \label{def: local: irf} 
	\end{equation}
	If $X_t$ is given by a functional structural shock (such as the functional monetary policy shock in \citealp{IR2021}) or if the true data generating process (DGP) follows a special structure  in Section~\ref{sec: svar}, the parameter $\beta_h$ reduces to  the \textit{structural impulse response function} (SIRF) of $y_t$ when $X_t$ experiences a function-valued shock $\zeta$. In this regard, we extend the seminal work of \cite{Oscar2005} to allow for a functional variable.

	\begin{figure}[h!]
		\centering 
		\caption{Quantile Curves of Economy Sentiment}
		\begin{subfigure}{.6\textwidth}\subcaption{Overall sentiment}
			\includegraphics[width = \textwidth, height = 0.5\textwidth, trim= {0 0 0 5cm},clip]{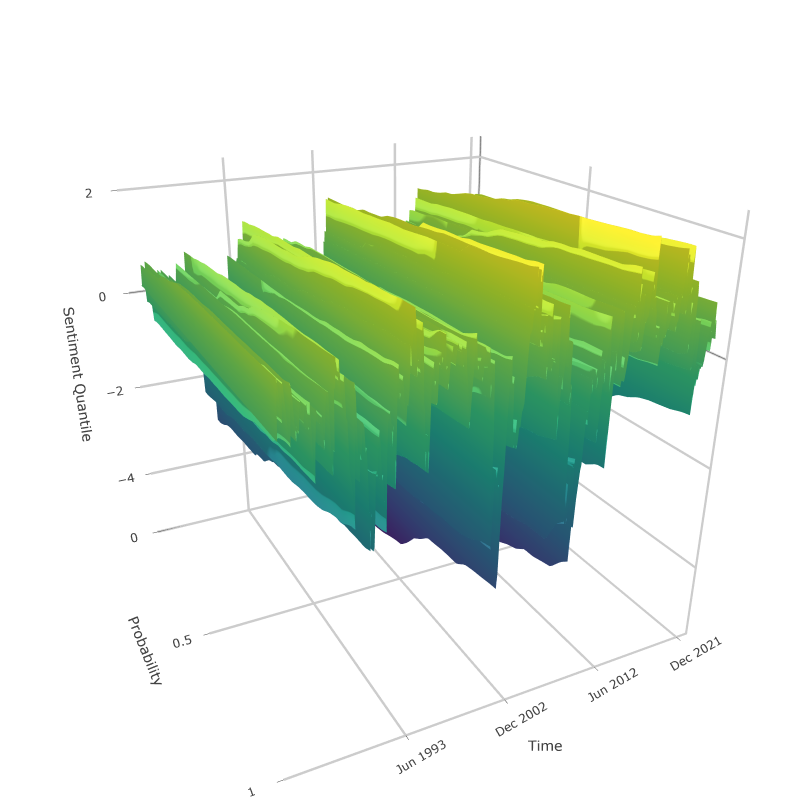}  
			\label{fig0a}
		\end{subfigure}\\
		\begin{subfigure}{.32\textwidth}\subcaption{Overall period}
			\includegraphics[width = \textwidth]{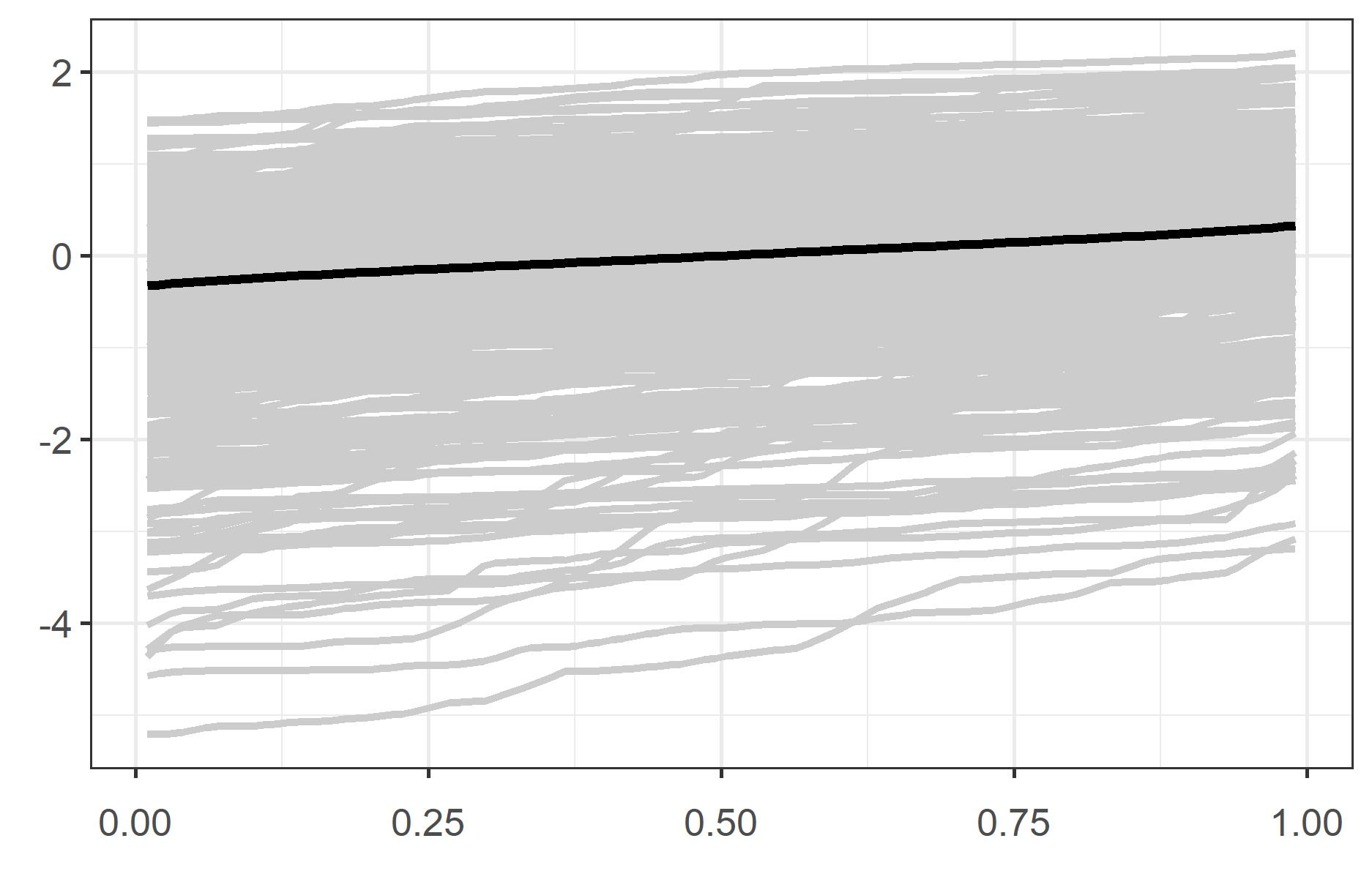}  
			\label{fig1a}
		\end{subfigure}
		\begin{subfigure}{.32\textwidth}\subcaption{Recessive period}
			\includegraphics[width = \textwidth]{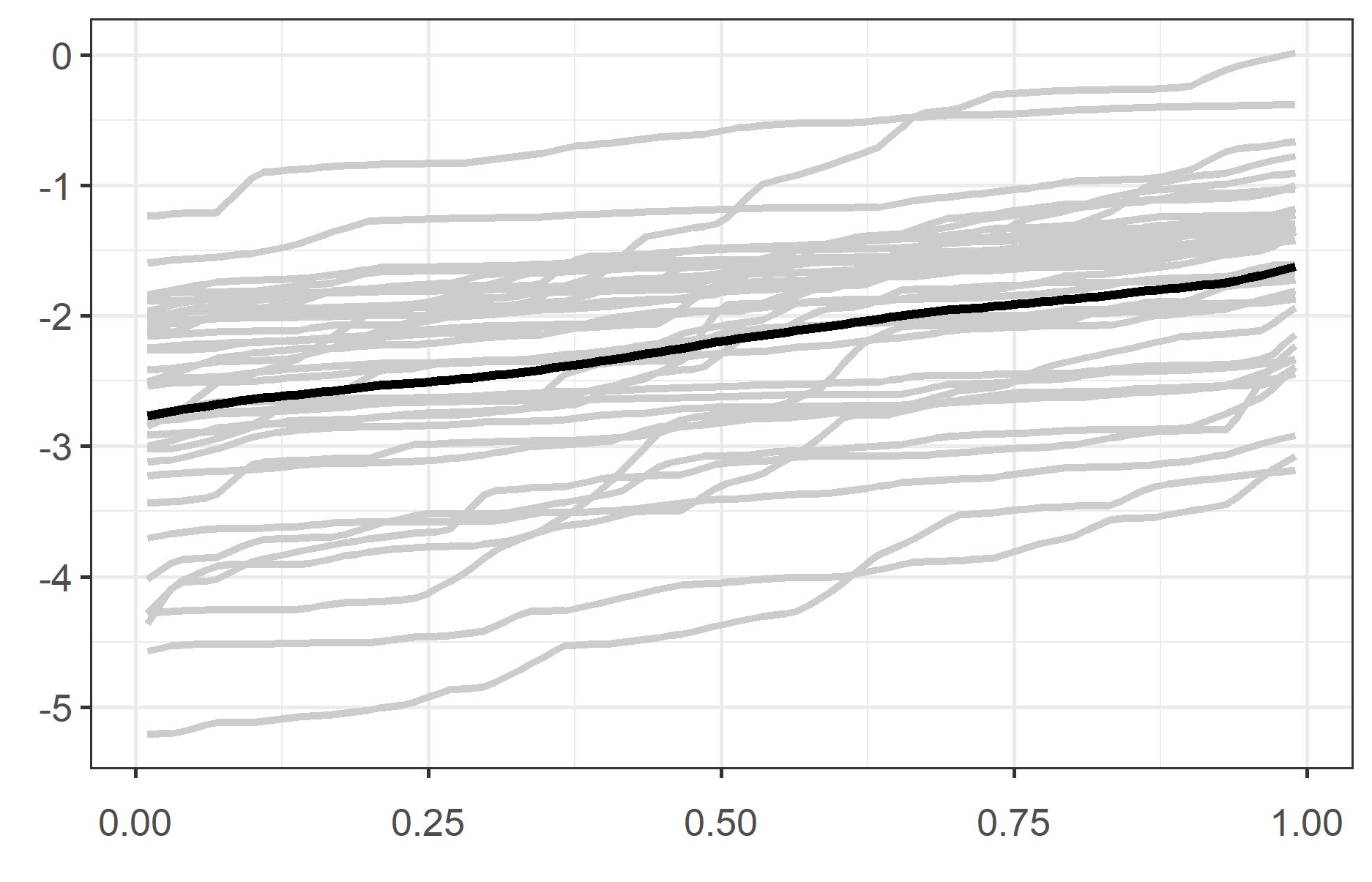}  
			\label{fig1b}
		\end{subfigure}
		\begin{subfigure}{.32\textwidth}\subcaption{Expansive period}
			\includegraphics[width = \textwidth]{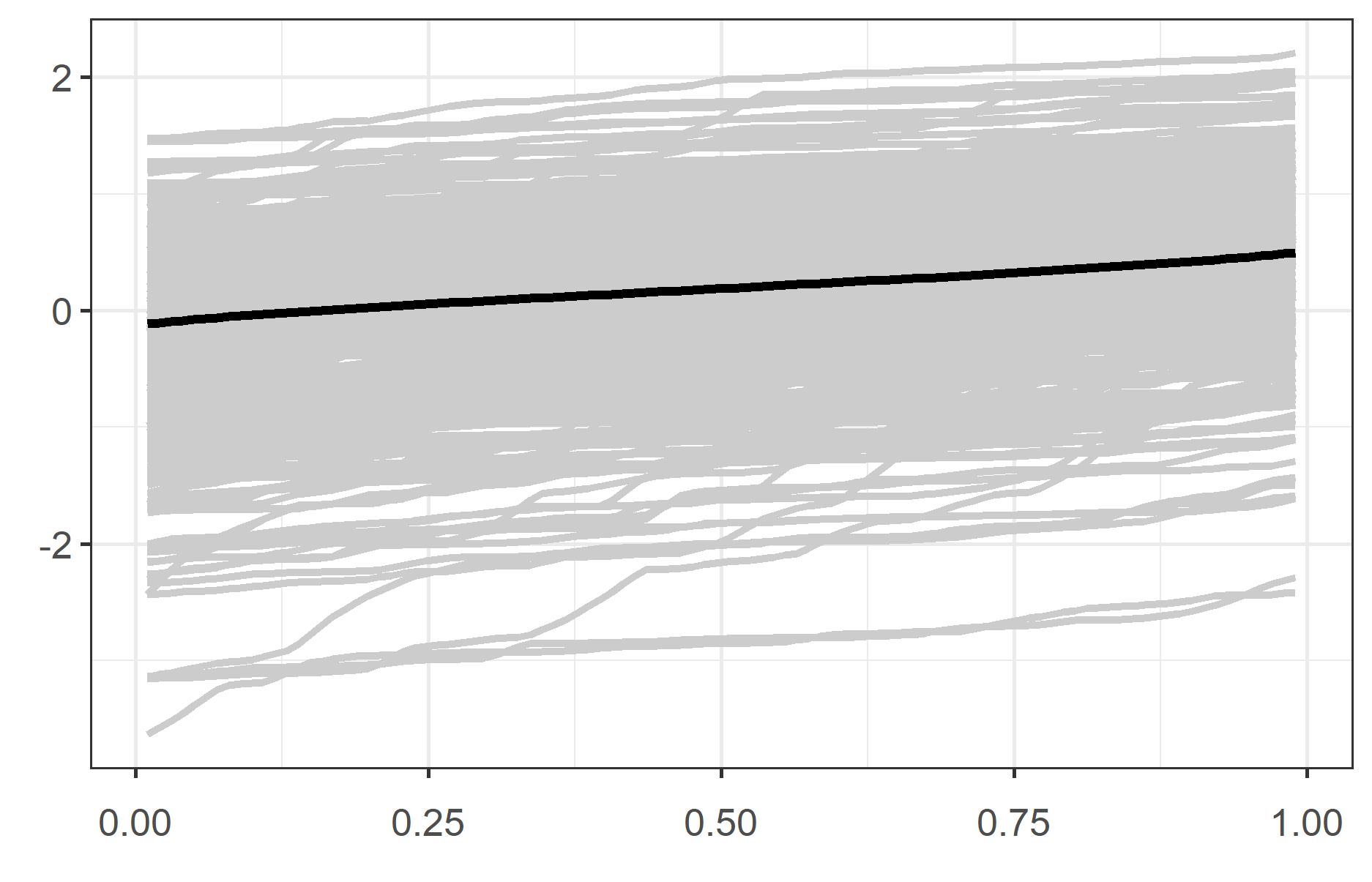}  
			\label{fig1c}
		\end{subfigure}\vspace{-1.5em}
		\flushleft{\scriptsize{Notes: The monthly quantile curves of economic sentiment are reported for overall, recessive, and expansive periods with each other their mean functions (black). Each curve is obtained by smoothing the daily economic sentiment measure proposed by \cite{Barbaglia2023} for each month, see Section~\ref{sec:emp1} for details.   }}
		\label{fig: 1}
	\end{figure}
	
	\begin{figure}[h!]  
		\caption{Examples of $\zeta$ in \eqref{def: local: irf} and   distributional changes}
		\begin{subfigure}{.5\textwidth} \flushleft{\subcaption{ positive location shift ($\zeta_1$)\label{fig2a}}}
			\includegraphics[width = 0.49\linewidth]{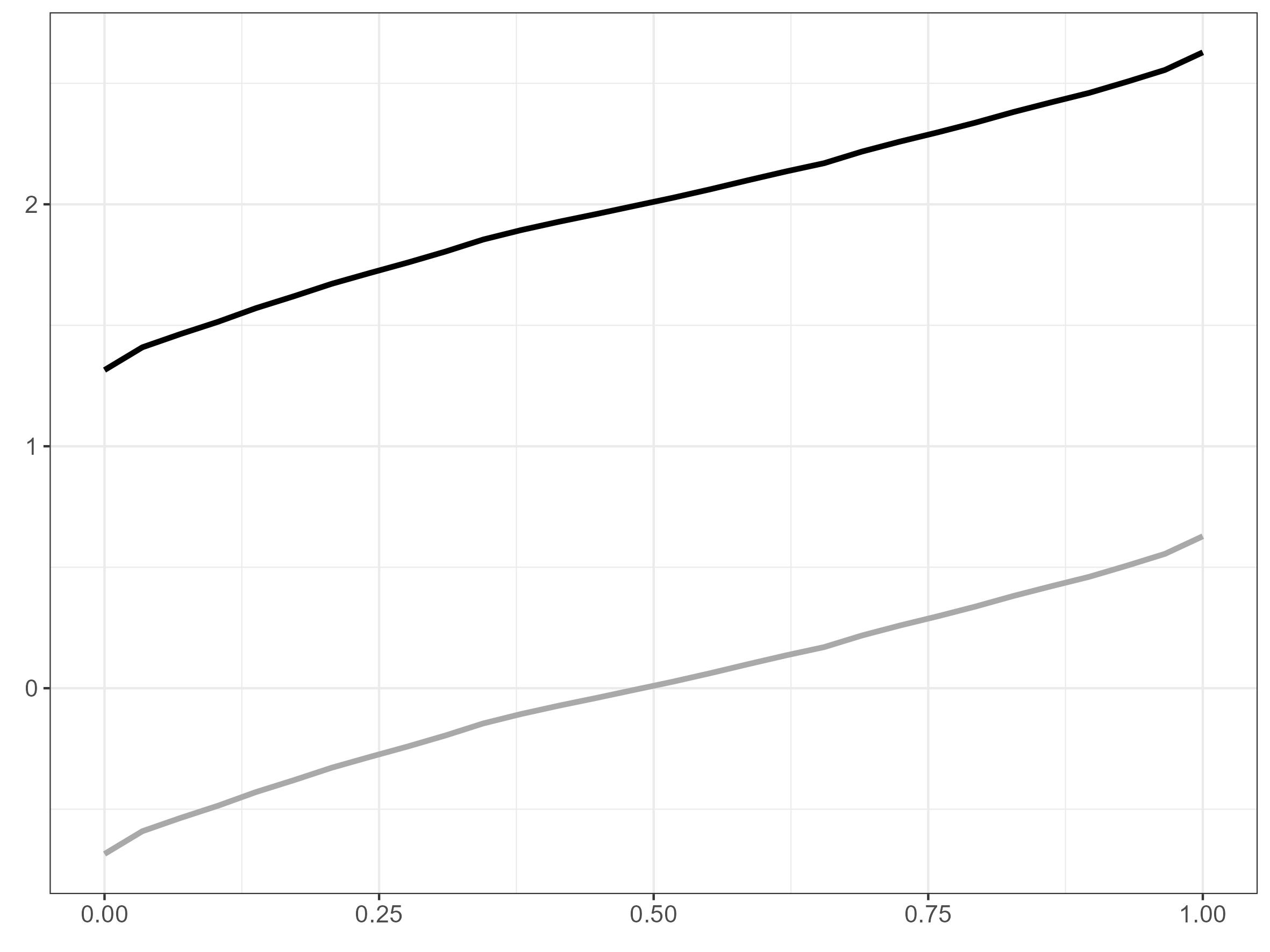}
			\includegraphics[width = 0.49\linewidth]{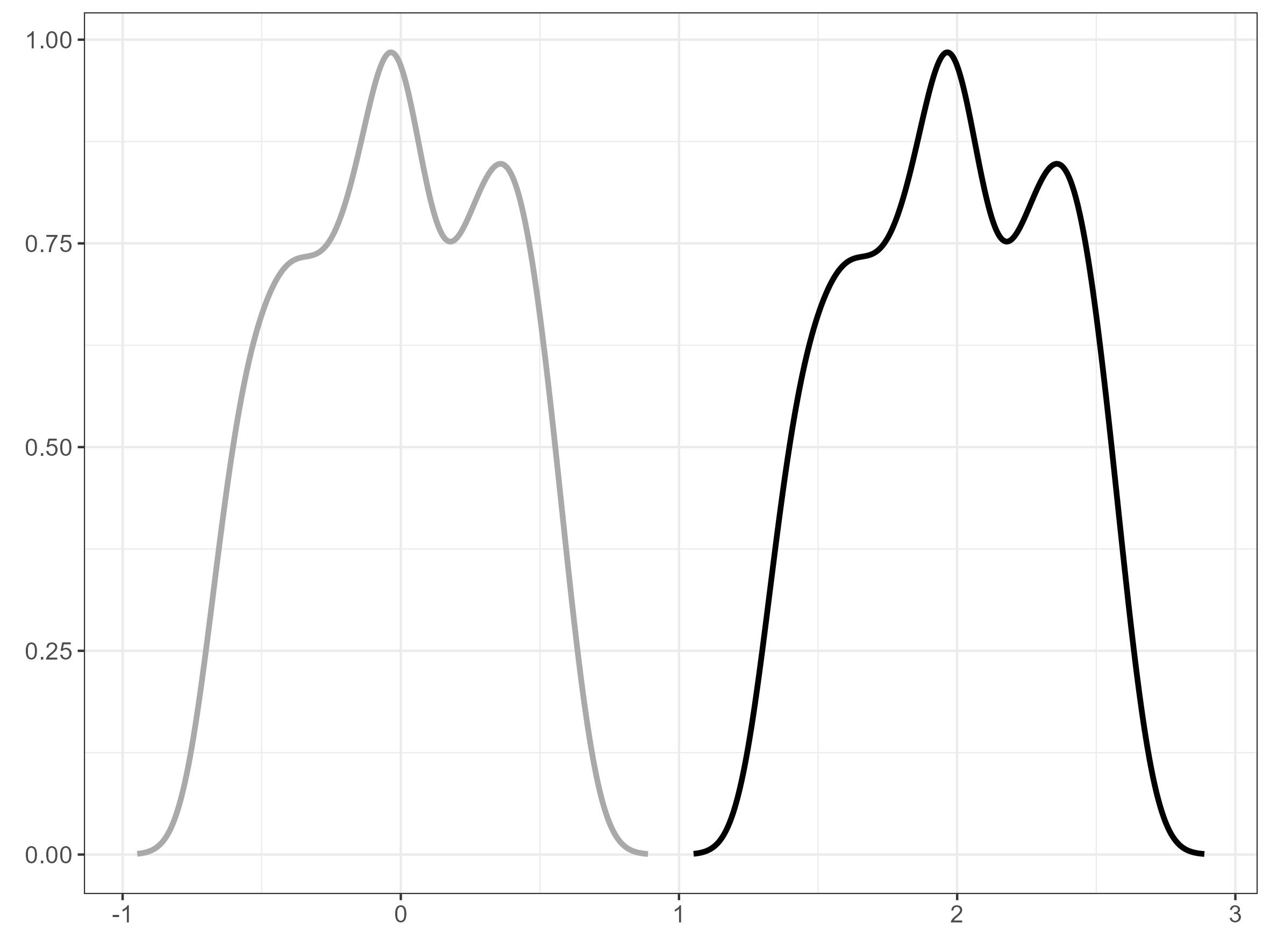}
		\end{subfigure}
		\begin{subfigure}{.5\textwidth} \subcaption{ negative location shift ($\zeta_2$)\label{fig2b}}
			\includegraphics[width = 0.49\linewidth]{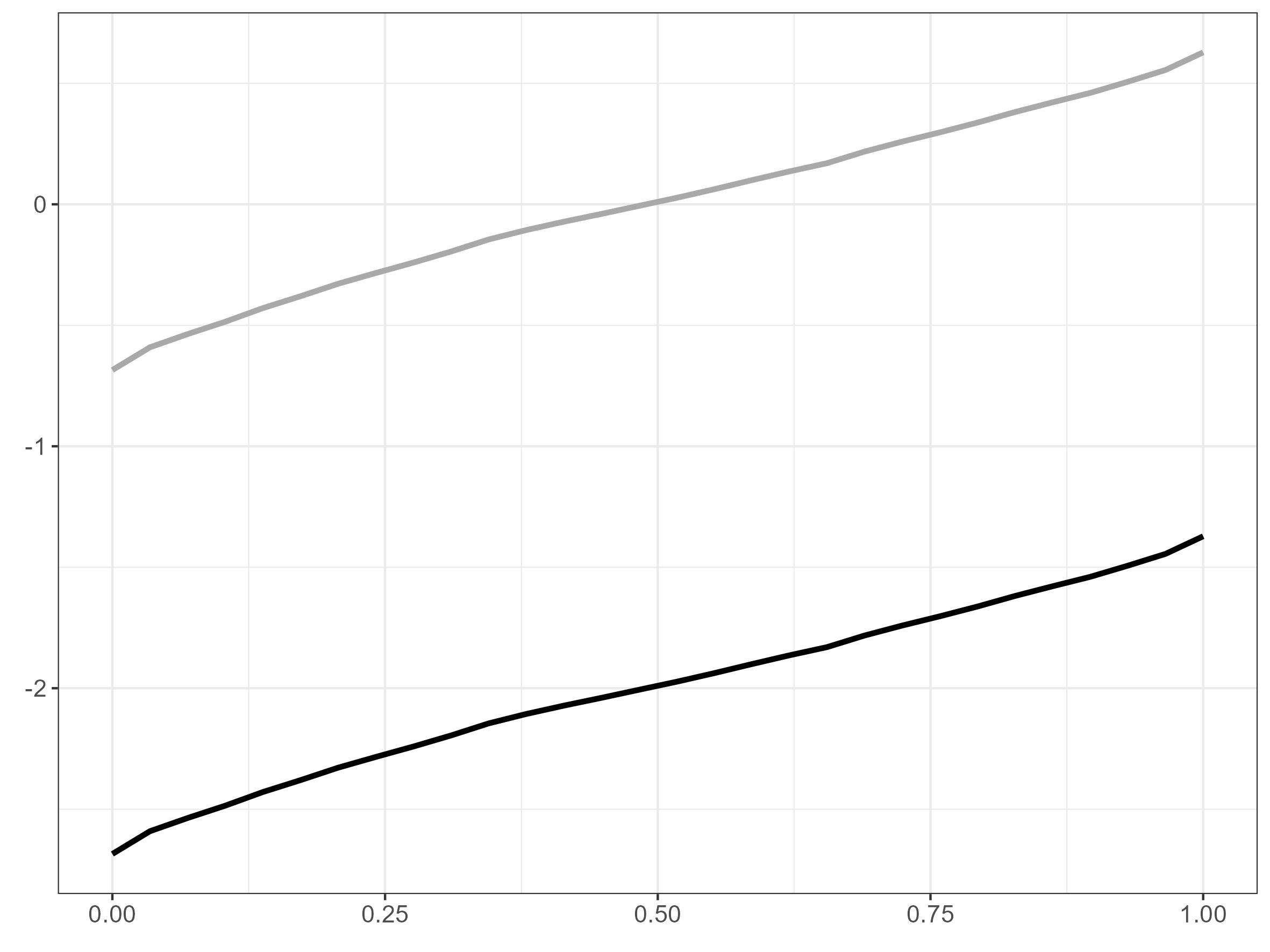}
			\includegraphics[width = 0.49\linewidth]{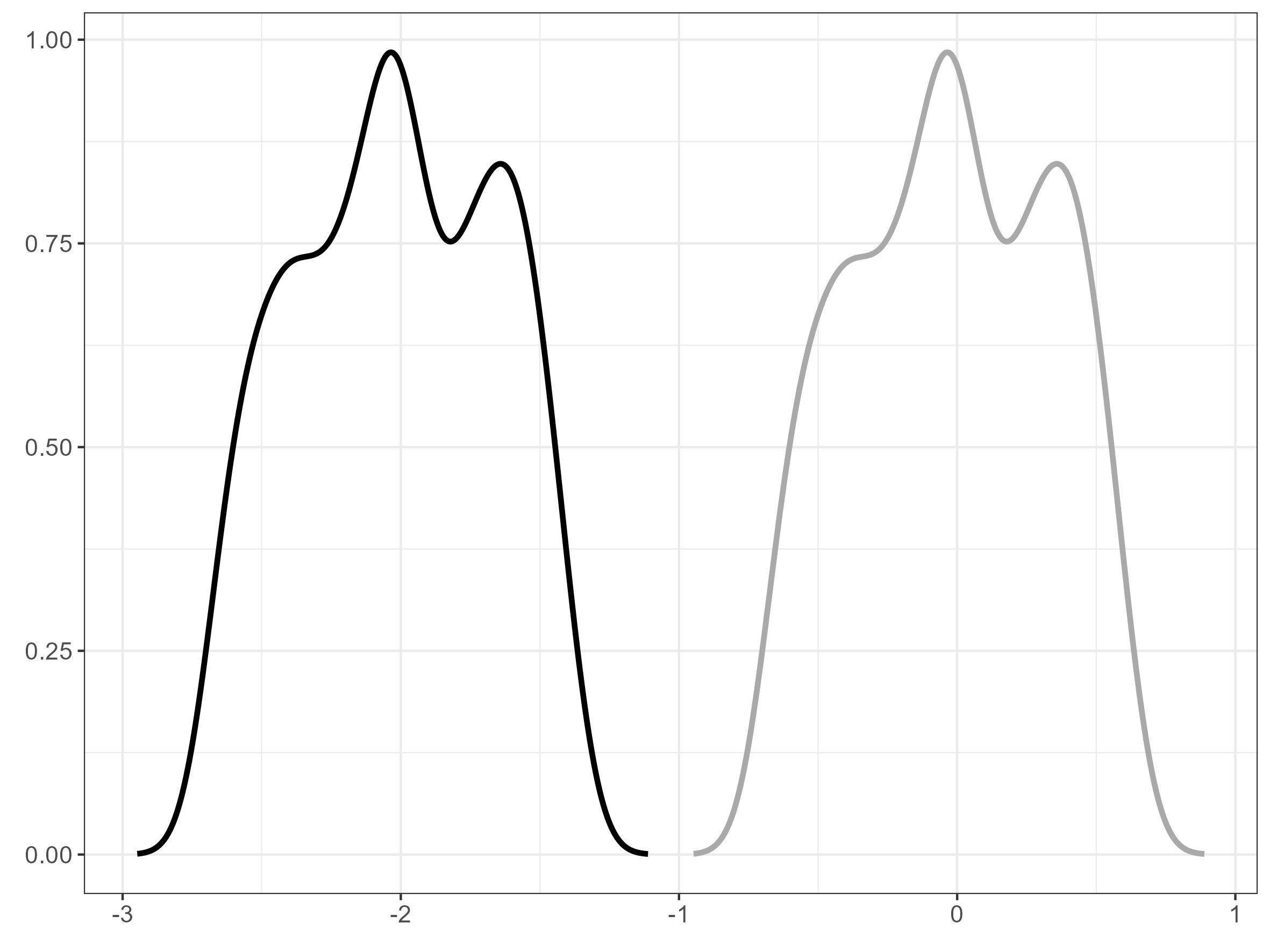}
		\end{subfigure}
		\begin{subfigure}{.5\textwidth} \subcaption{ greater dispersion ($\zeta_3$) \label{fig2c}}
			\includegraphics[width = 0.49\linewidth]{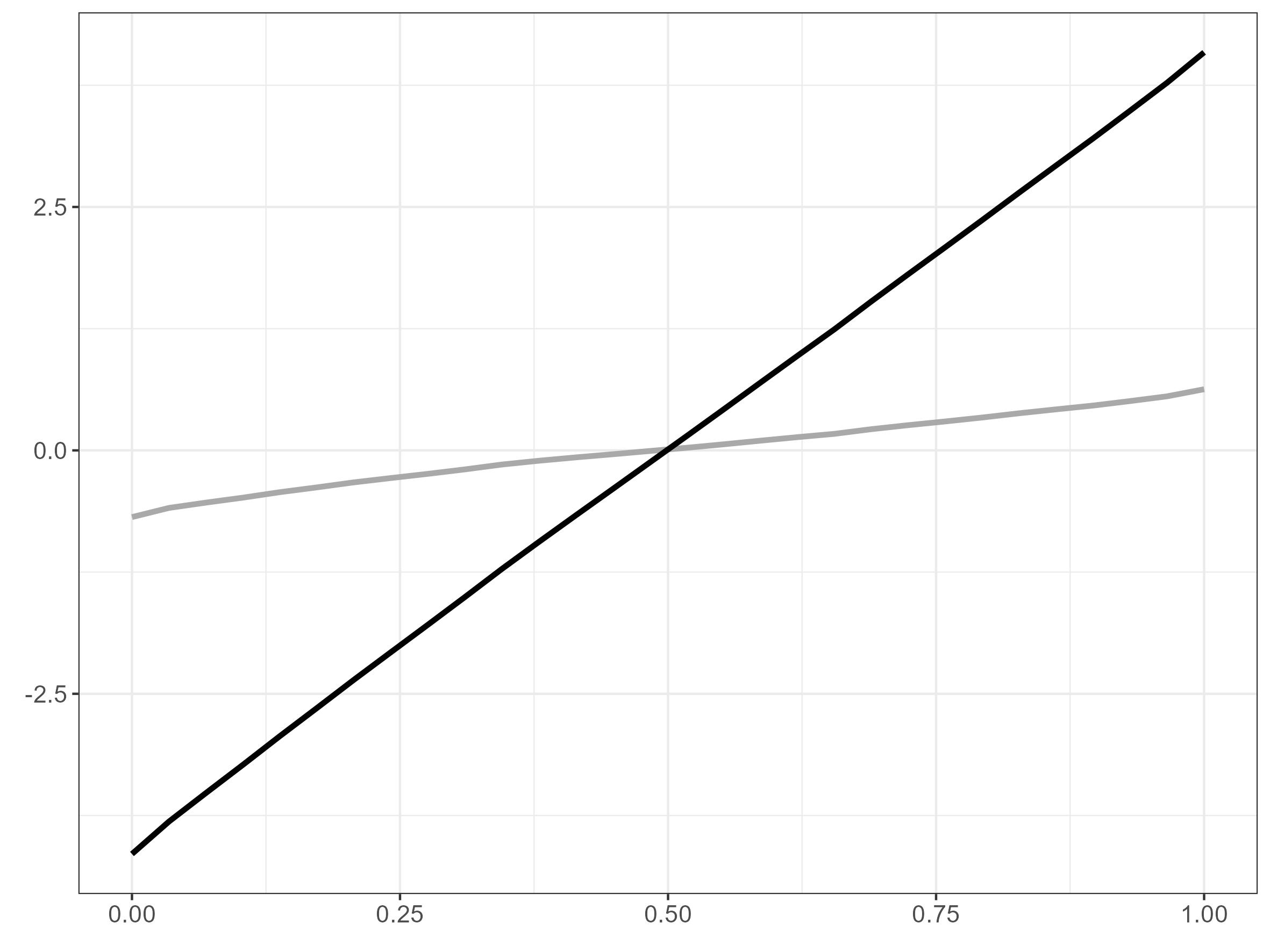}
			\includegraphics[width = 0.49\linewidth]{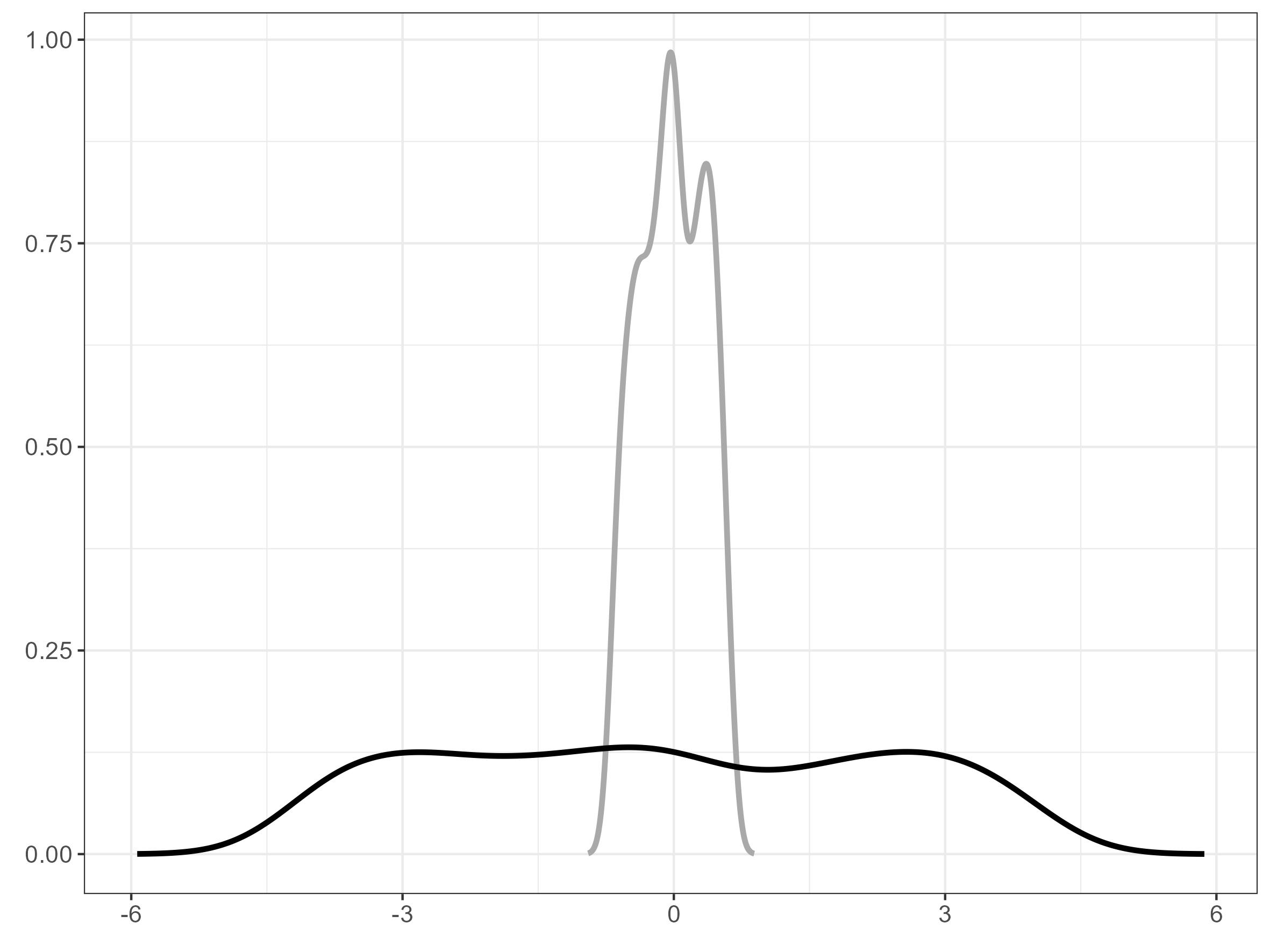}
		\end{subfigure}
		\begin{subfigure}{.5\textwidth} \subcaption{ greater dispersion, negative location shift ($\zeta_4$)\label{fig2d}}
			\includegraphics[width = 0.49\linewidth]{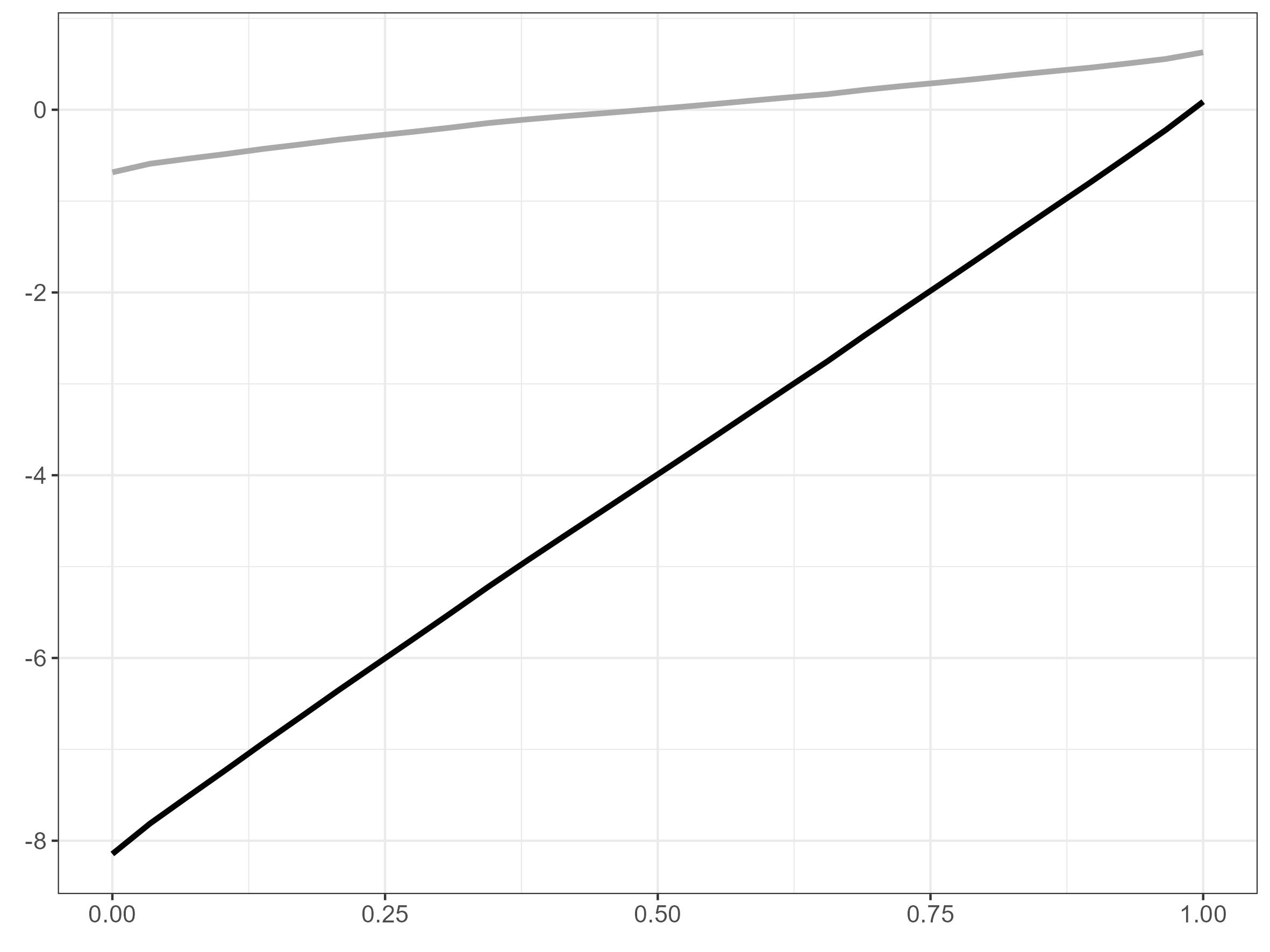}
			\includegraphics[width = 0.49\linewidth]{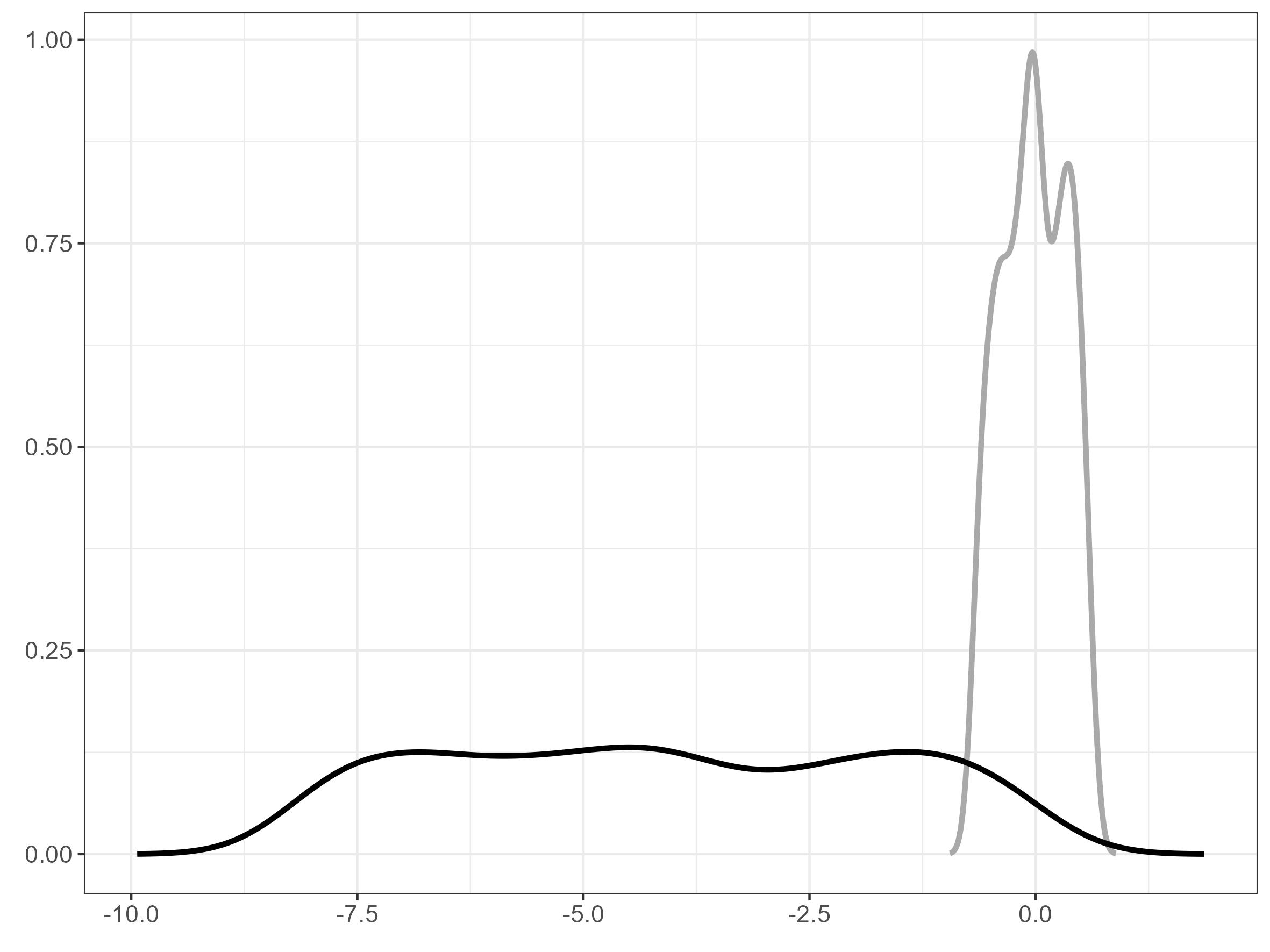}
		\end{subfigure}
		\label{fig: 2}\vspace{-.5em}
		\flushleft{\scriptsize{Notes: Each panel reports the impact of the perturbation $\zeta$ on the average sentiment quantile (left) and its associated distribution (right). The grey (resp.\ black) lines report the average quantile and its associated probability density function before (resp.\ after) the perturbation is given.  }} 
	\end{figure}

	\begin{example}[Economic Sentiment Quantile Functions]\label{example1}
		In Section~\ref{sec:emp1}, we consider a quantile function as an example of $X_t$, which represents the distribution of \citepos{Barbaglia2023} sentiment measure about the US economy. The distributional information is expected to provide meaningful insights into how economic variables respond when individuals' sentiment about the economy changes. %That is, $X_t(r) = \inf \{ q : F_t (q) \leq r\}$, where $F_t$ denotes the cumulative distribution function of the economy sentiment at time $t$. \cite{Barbaglia2023}  use the daily sentiment information to predict macroeconomic variables, such as GDP growth, by using the mixed-data sampling approach. However, there may be additional gains if the economy sentiment distribution is used. 
		To help with intuition, we report the sentiment curves in Figure~\ref{fig: 1}. Figure~\ref{fig: 1}.\ref{fig0a} reports the  sentiment quantile data over the sample period from January 1984 to December 2021.  Figures~\ref{fig: 1}.\ref{fig1a}, \ref{fig: 1}.\ref{fig1b}, and \ref{fig: 1}.\ref{fig1c} respectively report  the quantiles during the overall, recessive, and expansive periods with the probability $r\in[0,1]$ on the x-axis.  The black lines in the figures represent the mean functions for each period.\footnote{The recessive and expansive periods are defined following the NBER Business Cycle dates: \url{https://www.nber.org/research/business-cycle-dating}.}  The first interesting observation is  the different scales of the vertical axis; during recessive periods, individuals tend to hold negative sentiments on the economy and its mean function  exhibits a larger slope ($\approx$ -1.96), compared to overall ($\approx$ -0.002) and expansive ($\approx$ 0.2) periods. This suggests that	the economy, on average, tends to hold a higher level of heterogeneous beliefs, which may be explained by greater uncertainty during that time. Such distributional information will be lost if we aggregate $X_t$ into a single scalar-valued random variable. 
		
		When a random shock is given to the sentiment distribution, its impact on economic variables may depend on whether the shock increases or decreases overall economic sentiment and/or changes the level of disagreement about the economic sentiment. This can be studied by, for example, setting $\zeta (\equiv \zeta(\cdot))$ in \eqref{def: local: irf} to $\zeta_1(s) = 2  $, $\zeta_2(s) = -2 $, $\zeta_3(s) = 2 s$, and $\zeta_4(s) = -2  + 2s$, for $s \in[0,1]$. %These are particularly chosen because of our consideration of quantile functions as the functional predictor; the shock $\zeta$ should preserve the monotonicity of $X_t$. 
		Figures~\ref{fig2a}--\ref{fig2d} present the  average sentiment quantiles and distributions before (grey) and after (black) these perturbations are introduced. As shown in Figures~\ref{fig2a} and \ref{fig2b}, the mean function, tends to shift up or down in its level when $\zeta_1$ and $\zeta_2$ are introduced, which corresponds to the location shifts of the sentiment distribution. Meanwhile,  $\zeta_3$ increases the slope of the quantile without changing its average level. Thus, the sentiment distribution exposed to the perturbation will have a greater dispersion, meaning a higher level of disagreements about economy status.  The last shock $\zeta_4$ increases the level of disagreement and negatively shifts the overall sentiment distribution, so that it gets similar  to the average quantile observed in recessive periods.  
	\end{example}

	\begin{figure}[h!]
		\centering 
		\caption{  Functional Monetary Policy Shocks in \cite{IR2021} }
		\begin{subfigure}{.32\textwidth}\subcaption{Overall period}
			\includegraphics[width = \textwidth]{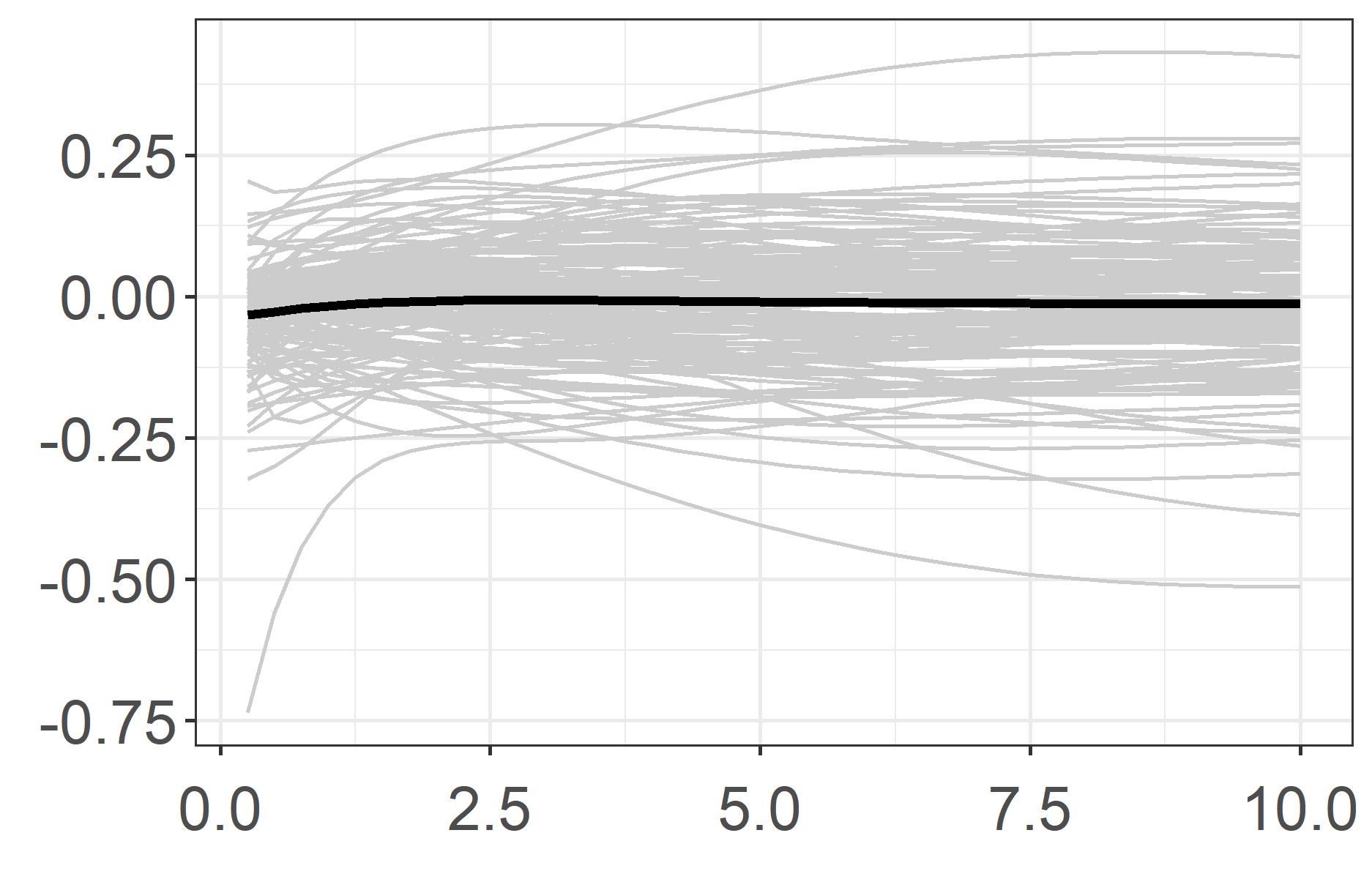}  	\label{fig1a:ir}
		\end{subfigure}
		\begin{subfigure}{.32\textwidth}\subcaption{Conventional period}
			\includegraphics[width = \textwidth]{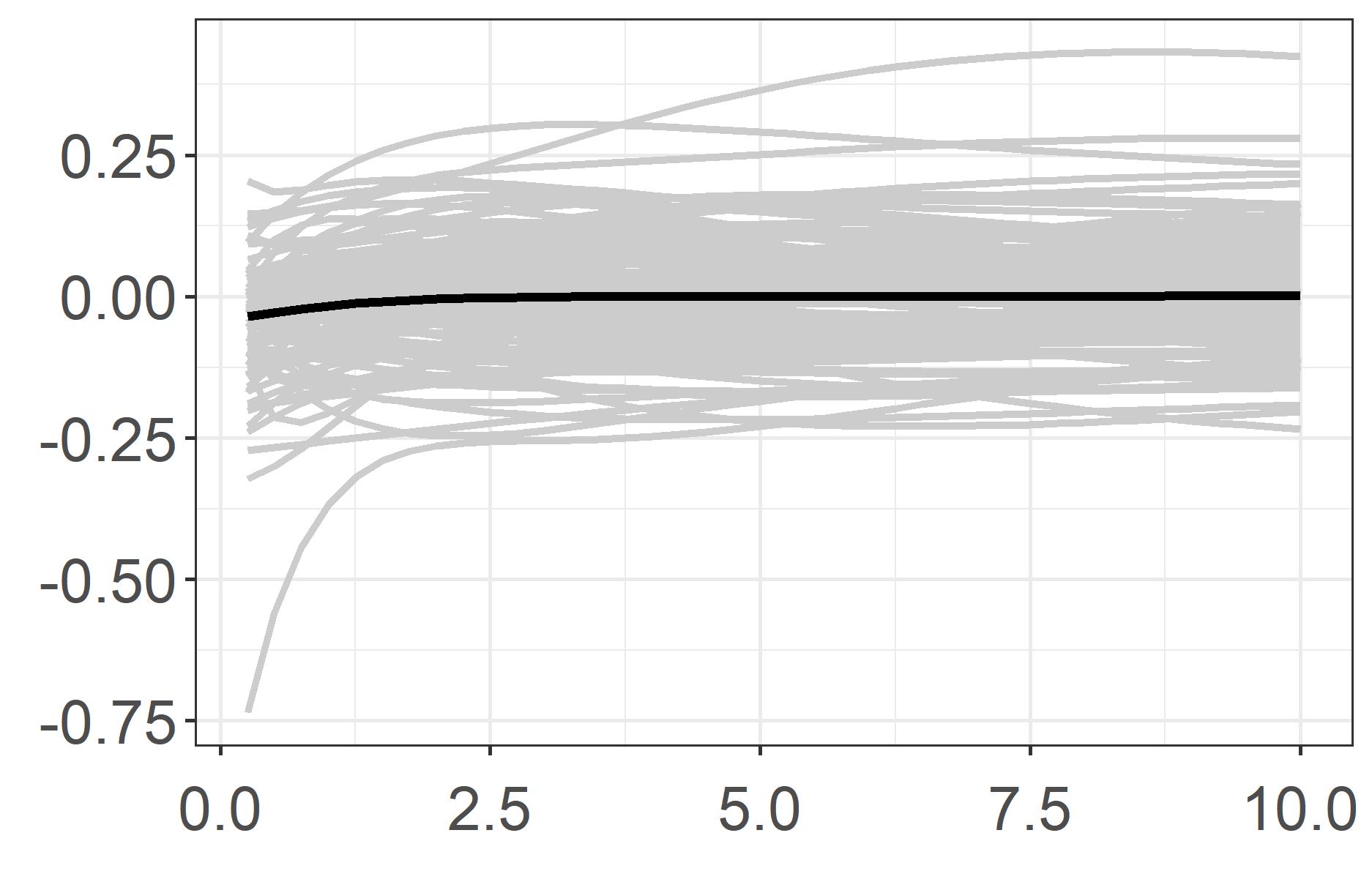} 	\label{fig1b:ir}
		\end{subfigure}
		\begin{subfigure}{.32\textwidth}\subcaption{Unconventional period}
			\includegraphics[width = \textwidth]{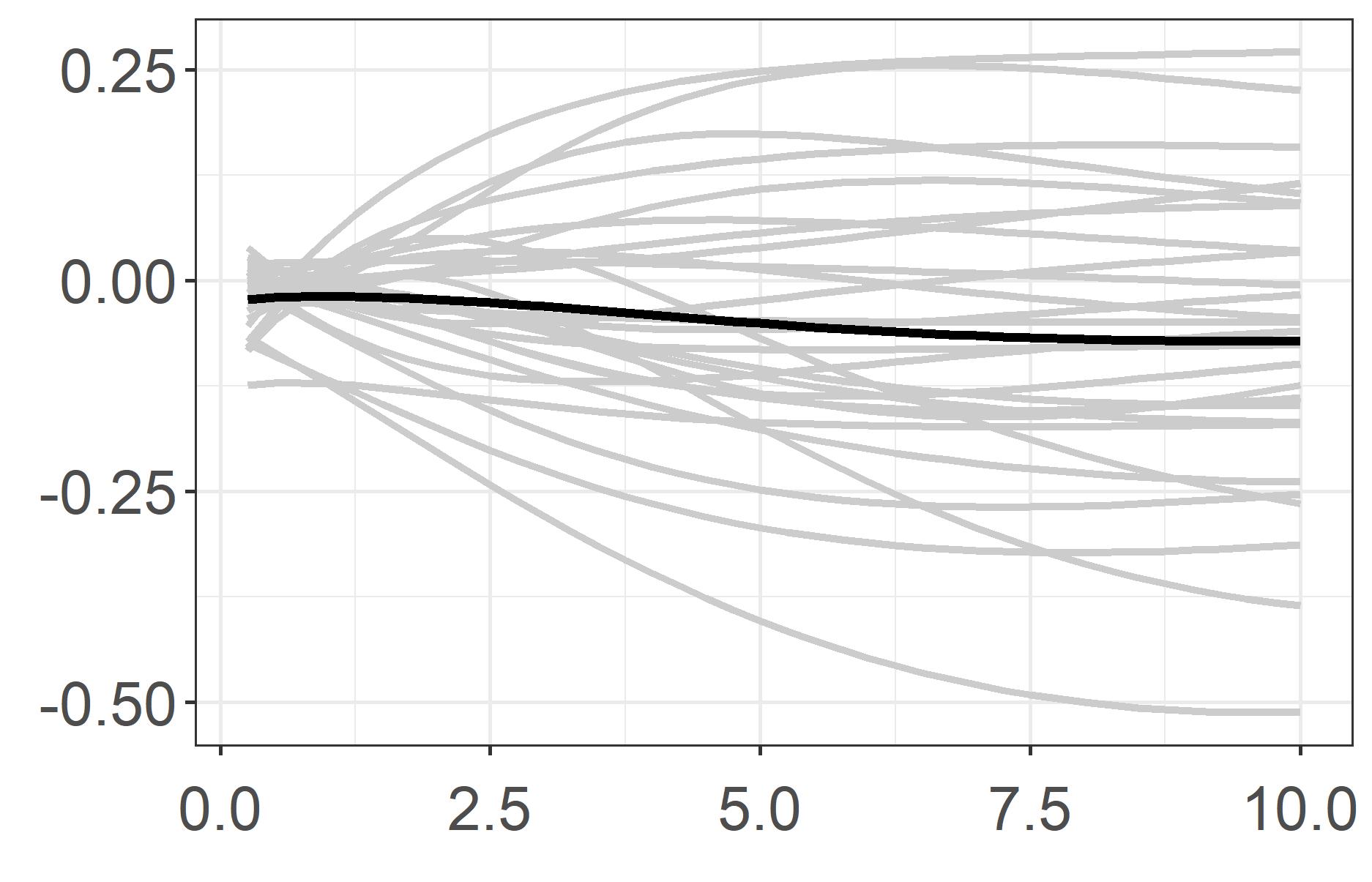}	\label{fig1c:ir}
		\end{subfigure}\vspace{-1.5em}
		\flushleft{\scriptsize{Notes: The monthly functional monetary policy shocks proposed by \cite{IR2021} are reported for overall, conventional, and unconventional periods for the period from January 1995 to June 2016. The black lines represent the mean functions for each of these periods. The shocks are characterized by the shifts in yield curves on monetary policy announcement dates, see \cite{IR2021} for details. }}
		\label{fig: ir: 1}   
	\end{figure}

	\begin{example}[Functional Monetary Policy Shocks]\label{example2} Another important example of $X_t$ is functional monetary policy shocks proposed by \cite{IR2021}. %In Section~\ref{sec:emp2}, we revisit their study on the response of inflation growth to functional monetary policy shocks with our estimator for \eqref{eq: sirf} to be developed. % \cite{IR2021} summarize the information on their function-valued shock into level, slope, and curvature components in \citepos{Nelson1987} model and it is assumed that the same components expand the coefficient function, which makes it valid to estimate responses to functional shocks ($\beta_h$ in our notation) by classical estimators in the local projection literature (e.g., \citealp{Oscar2005}; \citealp{PW2021}) or in the textbook SVAR literature (e.g., \citealp{hamilton2020time}). However, such a parametric assumption is not always supported in practice, because some important features of data generating processes (DGP) involving function-valued variables cannot be preserved under the parametric representation (see, e.g., \citealp{Nielsen2023}). Thus, our approach could be an appealing alternative to \citepos{IR2021} approach to study the impact of functional shocks. 
		The functional monetary policy shocks are reported in Figure~\ref{fig: ir: 1} for the overall (Fig.\ \ref{fig1a:ir}), conventional (Fig.\ \ref{fig1b:ir}), and unconventional (Fig.\ \ref{fig1c:ir}) periods, along with their mean functions (black). The functional monetary policy shocks in conventional periods appear to have regular shapes. Specifically, during  conventional periods, the shocks tend to have a greater negative value at short maturities, compared to those at long maturities. Such a regular shape is not observed during unconventional periods in Figure~\ref{fig1c:ir}.  The figures suggest that monetary policy shocks realized by shifts in yield curves may allow us to utilize richer information related to their shape and direction at different maturities. This example was formerly studied by \cite{IR2021} and will be further explored in  Section~\ref{sec:emp2}.
	\end{example}

	\begin{figure}[h!]
		\centering 
		\caption{Financial Risk Shocks  }
		\begin{subfigure}{.32\textwidth}\subcaption{Overall period}
			\includegraphics[width = \textwidth]{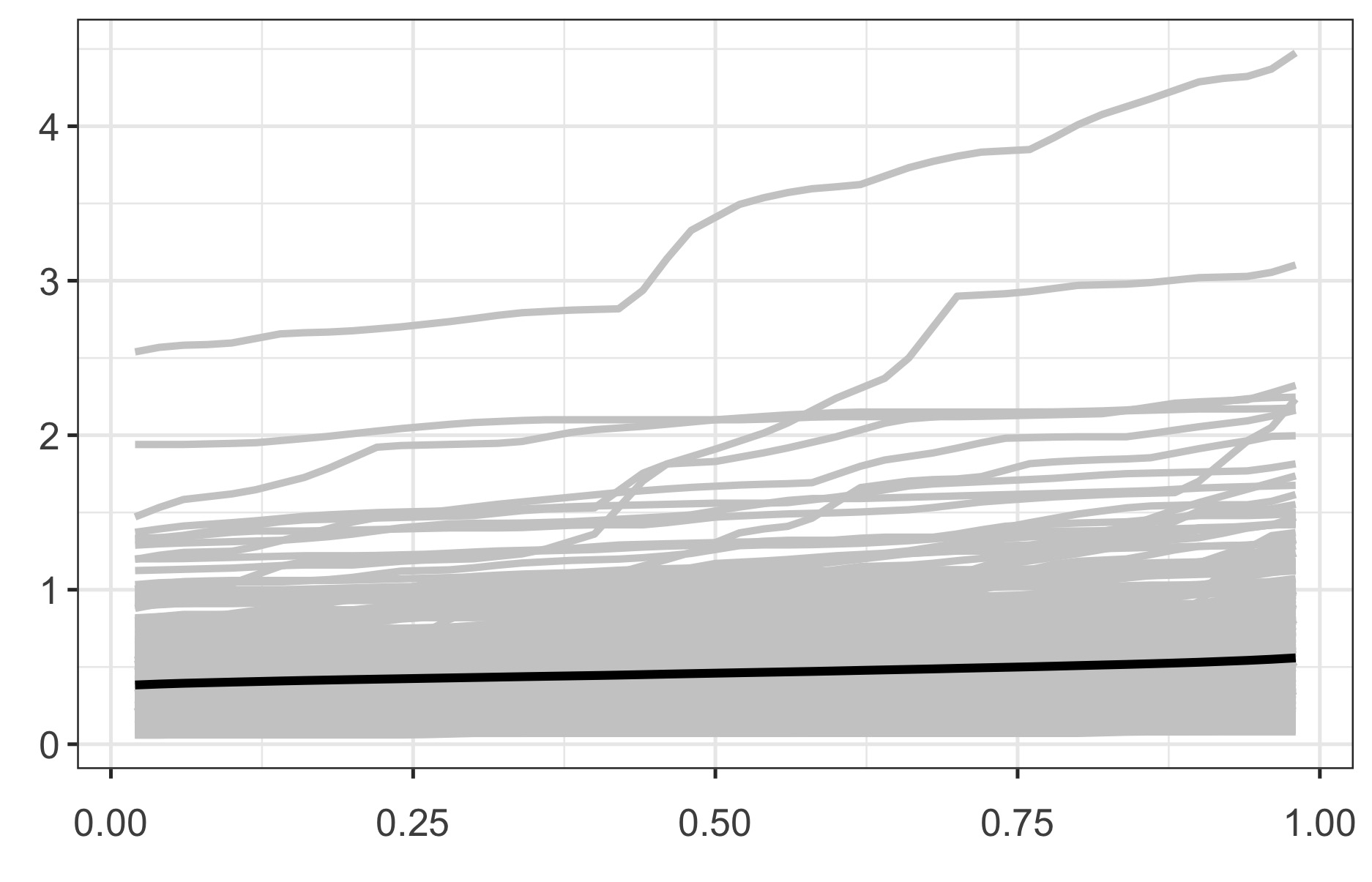}  
			\label{fig1a:ted}
		\end{subfigure}
		\begin{subfigure}{.32\textwidth}\subcaption{Recessive period}
			\includegraphics[width = \textwidth]{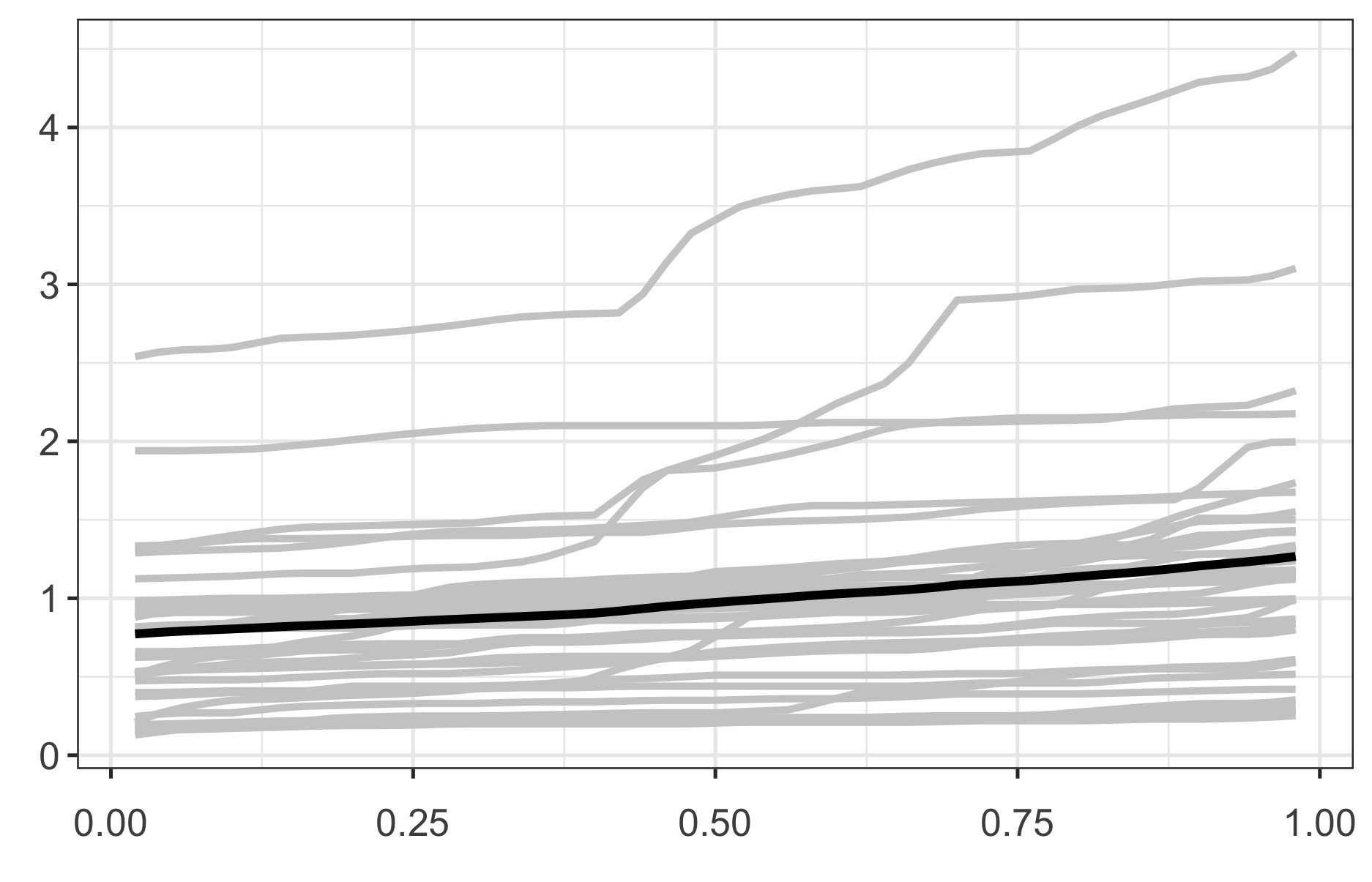} 
			\label{fig1b:ted}
		\end{subfigure}
		\begin{subfigure}{.32\textwidth}\subcaption{Expansive period}
			\includegraphics[width = \textwidth]{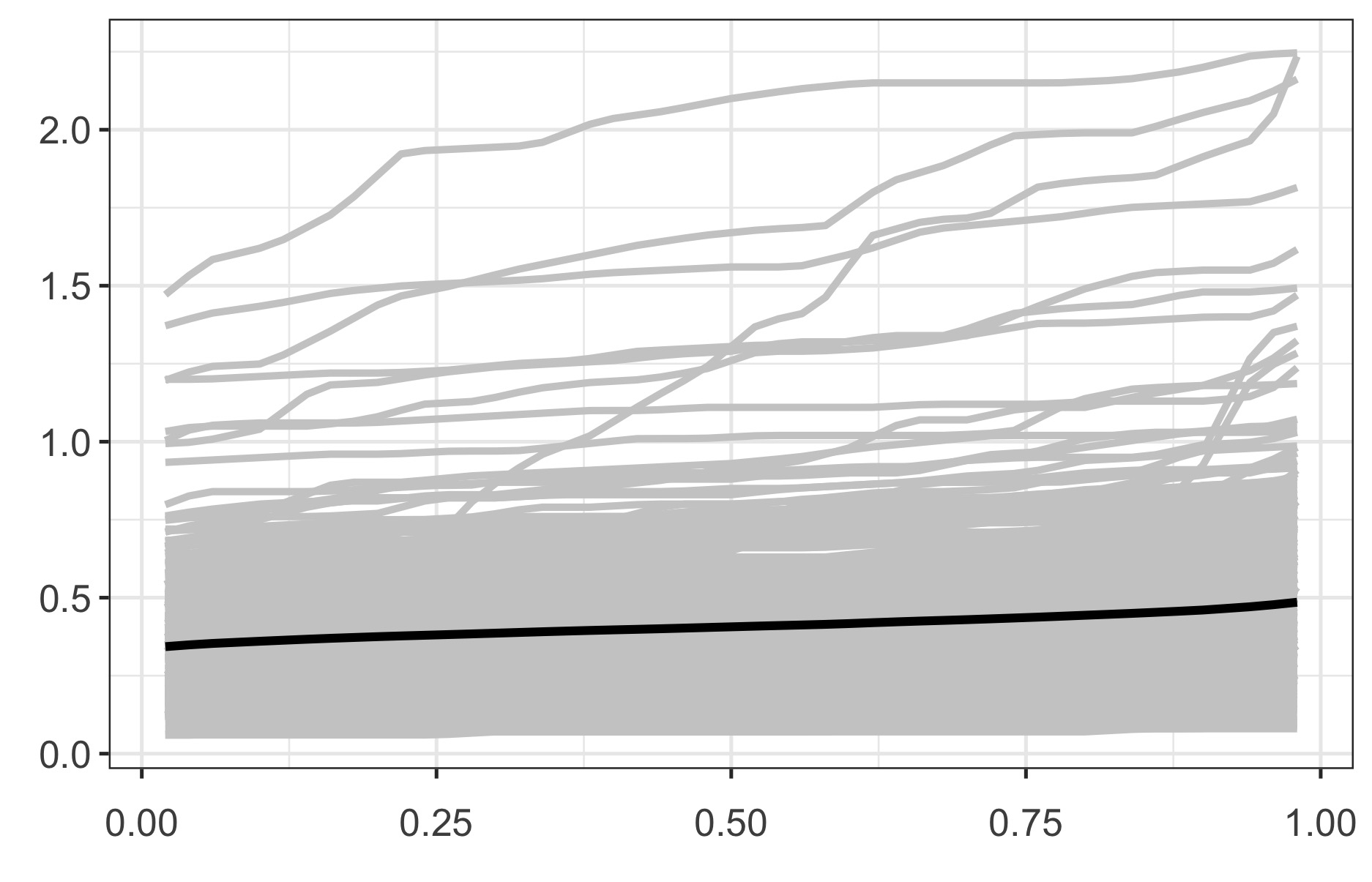}
			\label{fig1c:ted}
		\end{subfigure}		\label{fig: ted}  \vspace{-1.5em}
		\flushleft{\scriptsize{Notes: Each figure reports the monthly quantiles of TED spread for the period from February 1990 to December 2022. The black lines report the mean function in the overall, recessive, and expansive periods. }}
	\end{figure}

	\begin{example}[Proxy of Financial Risk Shocks] 
		Another example is the monthly quantiles of the TED spread that captures the likelihood of unpredictable default risk in financial markets. The TED spread is popularly employed in the literature as a proxy or an IV to measure the structural shock associated with financial markets (see \citealt[p.\ 110]{SW2012}). This quantile function could be used in place of $X_t$ in \eqref{eq: irf: general} as a proxy for the financial or liquidity shock, or as a function-valued instrument for the IV estimator developed in Appendix~\ref{sec:est2}.
	\end{example}
	
	%To be specific, inspired by the dynamic macroeconomic model characterizing the term structure in the yield curve, \cite{IR2021} suggest studying the functional shock that heterogeneously affects interest rates for each maturity via the \textsl{functional local projection} estimation. Their functional structural shocks are defined by shifts in yield curves when unconventional monetary policy announcements were made. For example, a shock given to interest rates at short maturity would differently affect the economy compared to a shock given to interest rates at long maturity. 
	%In fact, from the perspective that we do not require any parametric assumptions on the functional variables and their coefficients, our paper can be understood as a generalization of  \citepos{IR2021} work to allow for more flexible dynamics in functional structural shocks and their impact. In this paper, the impact of functional shocks is represented by an infinite number of basis functions, which necessitates estimation and inference strategies different from those in literature concerning the finite dimensional time series.

	\section{Relation to SVAR with a functional covariate} \label{sec: svar}
	If $X_t$ is identified as a functional structural shock,  as  in \cite{IR2021}, the parameter  $\beta_h$ in \eqref{eq: model: benchmark} can be certainly interpreted as the SIRF of $y_t$ to that shock. However, its general 
	link to the SIRF similar to those in \cite{Oscar2005} or \cite{PW2021} has not yet formally established in the functional setup, although  \eqref{eq: irf: general} and \eqref{eq: model: benchmark} partly provide an intuitive guidance on its interpretation as a direct estimator of the reduced-form impulse response. 
	
	In this section, we show that the parameter $\beta_h$ in \eqref{def: local: irf} provides a crucial and interesting interpretation as a SIRF when the true DGP follows an autoregressive structure. This suggests the importance of our benchmark model \eqref{eq: model: benchmark}, particularly as an extension of \cite{Oscar2005}.

	\subsection{Formulation of the functional SVAR model}
	For the ease of exposition, we exclude $\mathbf w_t$ in this section. The results given in this section can be easily extended to include other covariates as long as they are finite dimensional. %\footnote{To illustrate this point, suppose that $\mathbf{P}_{\mathbf{w}}$ denote the projection matrix defined by $ \mathbf{W} (\mathbf{W} ' \mathbf{W} )^{-1} \mathbf{W}'$ where $\mathbf{W}$ is the matrix stacks $\{\mathbf{w}_t\}_{t=1} ^T$. Let $\mathcal T_n : \mathcal H \to \mathbb R ^n$ and $\mathcal T_n ^\ast :\mathbb R^n \to \mathcal H$ such that $\mathcal T_n \zeta = (\langle X_1, \zeta \rangle, \ldots , \langle X_n, \zeta\rangle )'$ and $\mathcal T_n ^\ast a = n^{-1}\sum_t X_i a_i$.  Then, by stacking the observations, the model can be represented by $\mathbf{y} = \mathcal T_n \beta_h +\mathbf{W}\gamma + \mathbf{u}$ and, by multiplying $\mathbf{I} -\mathbf{P}_{\mathbf{w}}$ to both sides, we have $y_t ^\perp =  \langle X_t ^\perp , \beta_h \rangle +u_t ^\perp $ where $y_t ^\perp$ and $u_t ^\perp$ are the $t$-th row of $(\mathbf{I} -\mathbf{P}_{\mathbf{w}})\mathbf{y}$ and $(\mathbf{I} -\mathbf{P}_{\mathbf{w}})\mathbf{u}$. In addition, $X_t ^\perp$ is the Riesz representation of the $t$th row of $(\mathbf{I} -\mathbf{P}_{\mathbf{w}})\mathcal T_n$. } 
	Thus we do not lose any generality by this simplification.
	
	As mathematical preliminaries, Section~\ref{sec_prelim} of the Supplementary Material provides a brief review of definitions and properties of \(\mathcal{H}\)-valued random variable \(X_t\), linear operators on \(\mathcal{H}\), and the product Hilbert space \(\mathbb{R} \times \mathcal{H}\), where the tuple \(\left[\begin{smallmatrix} y_t\\ X_t\end{smallmatrix}\right]\) takes values.\footnote{We treat the tuple \( \left[\begin{smallmatrix} y\\ x\end{smallmatrix} \right] \), consisting of \(\mathbb{R}\)-valued and \(\mathcal{H}\)-valued random variables, as a “column vector” in the  multivariate setting. This notation introduces no confusion or mathematical imprecision in the subsequent discussion. Similar notation will be used for any other tuples appearing later.}  
	As detailed in the appendix, a linear operator \(\mathcal{D}\) on \(\mathbb{R} \times \mathcal{H}\) can be represented as an operator matrix, say $\mathcal{D} = \left[\begin{smallmatrix} d_{11} & d_{12} \\ d_{21} & d_{22} \end{smallmatrix}\right].$
	This operator maps $ \left[\begin{smallmatrix} y\\ x\end{smallmatrix} \right] \in \mathbb{R} \times \mathcal{H} $ to 
	$\left[ \begin{smallmatrix}
		d_{11}y + d_{12}x\\ d_{21}y + d_{22}x
	\end{smallmatrix}\right] \in \mathbb{R} \times \mathcal{H}$,
	analogous to transforming a \((2 \times 1)\) vector using a \((2 \times 2)\) matrix. For this reason, with a slight abuse of notation, we henceforth treat the \((\mathbb{R} \times \mathcal{H})\)-valued random element \( \left[\begin{smallmatrix} y\\ x\end{smallmatrix} \right]\) as if it were a \((2 \times 1)\) random vector, and any linear operator on \(\mathbb{R} \times \mathcal{H}\) as if it were a \((2 \times 2)\) matrix. This simplification is valid in the considered setup without any loss of technical rigor (see Section~\ref{sec_prelim2} of the Supplementary Material).
	
	%  matrix transformation of $(2\times 1)$ vectors via $(2\times 2)$ matrix. Thus, viewing $\mathcal D$ (resp.\ $(y,x)'$) as $(2\times 2)$ matrix (resp.\ $(2\times 1)$ vector), $\mathcal D(y,x)'$ can simply be written as  understood as if the matrix multiplication $\left[\begin{smallmatrix}     d_{11} & d_{12} \\ d_{21}& d_{22}	\end{smallmatrix} \right] \left[\begin{smallmatrix}	    y  \\ x\end{smallmatrix}\right]$.    For more details on our setup and technical details, see Section \ref{sec_prelim}.  }

In this section, we suppose that the tuple $\left[\begin{smallmatrix} y_t\\ X_t\end{smallmatrix}\right] \in  \mathbb R \times \mathcal H$ follows the model below:
\begin{equation}
	\underbrace{\begin{bmatrix} I_1& \beta_{12} \\ \beta_{21}&  I_2\end{bmatrix} }_{:=\mathcal B }\begin{bmatrix} y_t 
		\\ X_t \end{bmatrix} =\underbrace{ \begin{bmatrix} \alpha_{11} & \alpha_{12} \\ \alpha_{21} & a_{22} \end{bmatrix}}_{:=\mathcal A}\begin{bmatrix} y_{t-1} 
		\\ X_{t-1} \end{bmatrix} + \begin{bmatrix} u_{1t} 
		\\ U_{2t}\end{bmatrix}.\label{eq: model: svar}
\end{equation}
%In the above, and write linear operators on $\mathbb{R}\times\mathcal H$ as matrices of operator elements (a.k.a. operator matrices) accordingly. This is allowed in the considered setup and there is no loss of technical accuracy (see Section \ref{sec_prelim2}).    
%The above is simply a compact expression of the two equations:  (a) \(y_t + \beta_{12} X_t = \alpha_{11} y_{t-1} + \alpha_{12} X_{t-1} + u_{1t}\)  and (b) \(X_t + \beta_{21} y_t = \alpha_{21} y_{t-1} + \alpha_{22} X_{t-1} + U_{2t}\),  treating \((y_t, X_t)'\) as a \(2 \times 1\) vector and using appropriately defined operator matrices \(\mathcal{A}\) and \(\mathcal{B}\) (see Section \ref{sec_prelim} for more details).
%In the above, $(y_t, X_t)'$ is a random variable in the (product) Hilbert space $\mathbb{R}\times \mathcal{H}$, the operator matrices $\mathcal A$ is a representation of the linear operator mapping $(x,y) \mapsto (\alpha_{11}x + \alpha_{12}y, \alpha_{21}x + \alpha_{22}y )'$; standard matrix algebra Section \ref{sec_prelim} of the Appendix briefly reviews these notions
The above model is identical to the standard bivariate SVAR model other than that $X_t$ is a $\mathcal H$-valued random variable, and thus the coefficients $\mathcal A$ and $\mathcal B$ are in fact operator matrices. %appropriately defined linear maps; for example, $\beta_{12}$ is given by a linear map from $\mathcal H$ to $\mathbb{R}$.   	
Specifically, $\alpha_{12}$ and $ \beta_{12}$ are given by linear maps from $ \mathcal H$ to $ \mathbb R$, while  $\alpha_{21}$ and $\beta_{21}$ are maps from $\mathbb R$ to $\mathcal H$. The $(1,1)$-th (resp.\ $(2,2)$-th) elements, $I_1$ and $a_{11}$ (resp.\ $I_2$ and $a_{22}$), are linear operators acting on $\mathbb{R}$ (resp.\ $\mathcal H)$, where $I_1$ and $I_2$ denote the identity maps in the relevant spaces. The tuple of structural shocks $\left[\begin{smallmatrix} u_{1t}\\U_{2t}\end{smallmatrix}\right]$ is a $(\mathbb R \times \mathcal H)$-valued random element. Its covariance operator (see \eqref{eqcovop} in the Supplementary Material) is given as follows: 
\begin{equation*}
	\mathbf \Sigma  %= Cov \left( \begin{matrix} u_{1t} \\  U_{2t} \end{matrix} \right) 
	=  \begin{bmatrix} \mathbb{E}[u_{1t} \otimes u_{1t}] &  \mathbb{E}[U_{2t} \otimes u_{1t}] \\  \mathbb{E}[u_{1t} \otimes U_{2t}] &  \mathbb{E}[U_{2t}\otimes U_{2t}]\end{bmatrix} = \begin{bmatrix} \sigma_{11}   &  0 \\  0 & \mathbf{\Sigma}_{22} \end{bmatrix},
\end{equation*} 
where $\otimes$ denotes the  tensor product, which generalizes the outer product in the Euclidean space, and $\sigma_{11}$ (resp.\ $\mathbf{\Sigma}_{22}$) is a linear operator acting on $\mathbb{R}$ (resp.\ $\mathcal H$).\footnote{Since  $\mathbf{\Sigma}$ is an operator matrix whose $(i,j)$-th entry is a map from $\mathcal{H}_i$ to $\mathcal{H}_j$, with $\mathcal{H}_1 = \mathbb R$ and  $\mathcal{H}_2=\mathcal{H}$, $\sigma_{11}$ is given by a linear map on $\mathbb{R}$. However, any linear map on $\mathbb{R}$ is nothing but a scalar multiplication, given by $c I_1$ for some $c \in \mathbb{R}$. Thus, there is little risk of confusion even if we understand $\sigma_{11}$ as a real-valued constant. \label{foot1}
} See Section~\ref{sec_prelim}  of the Supplementary Material for details. % it may be equivalently understood as  a constant since $\sigma_{11}$  it will be more precise to write the $(1,1)$-th element of $\mathbf{\Sigma}$ as the scalar multiplication map $\sigma_{11}I_1$, where $I_1$ is the identity map on $\mathbb{R}$. However, there seems a little gain from doing so, with bearing notational complexity to be accompanied. Thus with a slight abuse of notation, we hereafter understand $\sigma_{11}$ (or any other scalar) as the corresponding scalar multiplication map whenever it does not cause confusion.}  

We first establish the fully functional identification of structural parameters and highlight its differences from existing approaches. Then, we link  the SIRF implied by the SVAR structure to the coefficients in our benchmark model.

\subsection{Fully Functional Identification of the SVAR with a functional covariate} \label{sec_svar}
It can be shown that the operator $\mathcal B$ in \eqref{eq: model: svar}  is invertible  and $\Gamma:=\mathcal B^{-1}\mathcal A$ can be well defined under the condition in Proposition \ref{prop: svar: identification} that will appear shortly. Thus \eqref{eq: model: svar} can be written into the following reduced-form VAR (RFVAR) model:
\begin{equation} 
\begin{bmatrix} y_t 
	\\ X_t \end{bmatrix} =\underbrace{\begin{bmatrix} \gamma_{11} & \gamma_{12} \\ \gamma_{21} & \gamma_{22} \end{bmatrix}}_{:=\Gamma}\begin{bmatrix} y_{t-1} 
	\\ X_{t-1} \end{bmatrix} + \begin{bmatrix} \varepsilon_{1t} 
	\\ \mathcal E_{2t}\end{bmatrix}, \quad\text{ where }\quad  \Sigma_{\varepsilon} = \begin{bmatrix} \sigma_{\varepsilon,11} & \Sigma_{\varepsilon,12} \\ \Sigma_{\varepsilon,21} & \Sigma_{\varepsilon,22} \end{bmatrix}.\label{eqreduced1}
	\end{equation}
	$\Sigma_{\varepsilon}$ denotes the covariance of $\left[\begin{smallmatrix}\varepsilon_{1t}\\\mathcal E_{2t}\end{smallmatrix}\right]$.
	%for $\sigma_{\varepsilon,11} \in \mathbb{R}$, $\Sigma_{\varepsilon,12}:\mathcal H \mapsto \mathbb{R}$, $\Sigma_{\varepsilon,21}:\mathbb{R} \mapsto \mathcal{H}$, and $\Sigma_{\varepsilon,22}:\mathcal {H} \mapsto \mathcal{H}$. 
	In this section, we discuss conditions under which the structural parameters in \eqref{eq: model: svar} can be identified from the reduced-form parameters in \eqref{eqreduced1}. %this allows us to infer the structural model from the reduced-form model.
	%Practitioners may often be interested in the structural parameters in \eqref{eq: model: svar} can be identified from the parameters in the reduced form model, which will be directly estimated from data, as follows:
	Our assumption for achieving such  identification is based on the standard causal ordering of the (contemporaneous) variables in the system, as in \cite{Sims1972,Sims1980}. Later, we will show that the coefficients of our benchmark model reduces to the SIRFs under this identification scheme.  %Throughout this section, we implicitly assume that the vector of $y_t$ and $X_t$ follows the SVAR model in \eqref{eq: model: svar}. 
	
	\begin{proposition}\label{prop: svar: identification}\normalfont
Suppose that $ \left[\begin{smallmatrix} y_t\\ X_t\end{smallmatrix} \right]$ satisfies  \eqref{eq: model: svar}. Then the structural parameters in \eqref{eq: model: svar} are identified %from the reduced form model  \eqref{eq: model: svar2} 
under either 
\begin{enumerate*}[(i)]
	\item\label{prop: svar: identification1} \textsl{$\beta_{12} = 0$} or 
	\item\label{prop: svar: identification2} \textsl{$\beta_{21}=0$ and $\mathbf\Sigma_{22}$ is injective.} 
\end{enumerate*}
\end{proposition}

Although the conditions \ref{prop: svar: identification1} and \ref{prop: svar: identification2} in Proposition \ref{prop: svar: identification} are seemingly similar to each other, they suggest  different levels of complexity in estimating the outcome equations. Specifically, if there is no contemporaneous impact of function-valued $X_t$ to scalar-valued $y_t$ (i.e., $\beta_{12} = 0$), the outcome equation of $y_t$ reduces to $y_t = \alpha_{11} y_{t-1} + \alpha_{12}X_{t-1} + u_{1t}$, and thus it includes   only one infinite dimensional parameter $\alpha_{12}$. On the other hand, under the second identification condition (i.e., $\beta_{21} = 0$), the outcome equation reduces to $y_t = \alpha_{11} y_{t-1} -\beta_{12}X_t +  \alpha_{12}X_{t-1} + u_{1t}$, in which the number of infinite dimensional parameters doubles. Estimating these parameters involves solving an inverse problem that necessitates the use of a regularization scheme. Considering this, the resulting estimator from the second identification scheme is likely to suffer from a larger regularization bias.

%(WK REVIEW STOPPED HERE)

Another interesting implication of Proposition~\ref{prop: svar: identification} is that the standard rank-based identification strategy cannot be straightforwardly translated into our setup involving a functional covariate. Instead, depending on the direction of the contemporaneous impact, we may additionally need an injectivity condition %in addition to the restriction on structural coefficients. This is required
to identify structural parameters from the reduced-form parameters. % specifically, under the identification condition (ii) in Proposition~\ref{prop: svar: identification}, we identify $\beta_{12}$   from the following moment condition: \begin{equation*}     \mathbb E[v_{1t}  V_{2t} ] =  -\mathbf{\Sigma}_{22}   \beta_{12} ^\ast  . \end{equation*}  
For example, under the condition \ref{prop: svar: identification2} in Proposition~\ref{prop: svar: identification},  $\Sigma_{\varepsilon,22} = \Sigma_{22}$ and  $\beta_{12}$ satisfies
\begin{equation} \label{eqinject}
\sigma_{e,12} =\beta_{12}\Sigma_{22}; %\quad   \text{and} \quad	\sigma_{\varepsilon,11} =  	\sigma_{11}  + \beta_{12}\Sigma_{22}\beta_{12}^\ast;
\end{equation} 
see our proof of Proposition~\ref{prop: svar: identification} in Appendix~\ref{app_proof}. In this case, the injectivity of $\Sigma_{22}$ is required to uniquely identify  $\beta_{12}$  from \eqref{eqinject} (see \citealp{carrasco2007linear,Seo2024}).  
%This observation and the remark to be given indicate that existing identification schemes developed in a finite dimensional setup cannot be easily translated into ours.%Remark \ref{remchol} illustrates the challenges in extending these conventional schemes to our functional setup, by discussing the identification scheme based on the Cholesky factorization of $\Sigma$.}    

%Thus  On the other hand, if $\mathbf{\Sigma}_{22}$ is injective, then the former equation of \eqref{eqinject} guarantees the unique existence of the linear map linear map $\beta_{12}$ necessary condition guarantees the unique linear map $\beta_{12}$ satisfying the above equation, and thus it can be understood as a condition corresponding to the rank condition on the covariance of structural errors in classical SVAR models. \commWK{(Last sentence seems to be vague.... recommend deleting it.)} 

The two observations mentioned above contrast with identification conditions and their implications   in the existing literature, such as \cite{Sims1972}, where a different ordering of variables primarily affects their economic interpretation and the direction of contemporaneous shocks, rather than the level of asymptotic bias or computational burden.  Thus, in the SVAR model with functional covariates,  identification conditions need to be cautiously chosen with taking into account the above.

\begin{remark}\label{remchol}
In the existing literature on SVAR models concerning $k$-dimensional vector-valued time series,  it is commonly assumed that the $k\times k$ matrix $\Sigma_\varepsilon$ allows the Cholesky factorization such that $\Sigma_\varepsilon = LL'$ for a lower triangular matrix $L
$. Then $\mathcal B $ in \eqref{eq: model: svar} simply reduces to $ L^{-1}$. This is one of the most popular identification strategies and similar to the causal ordering in our Proposition~\ref{prop: svar: identification}. However, in the considered setup, where $X_t$ is a function-valued random element, the covariance $\Sigma_\varepsilon$ is not invertible on $\mathbb{R}\times\mathcal H$ (see \citealp{Mas2007}, Section 2.2), which makes it infeasible to apply the existing identification strategy.
\end{remark}

\subsection{SIRF and its relationship with the linear projection}  \label{sec_IRF}

%For the remaining discussion, we use $\Gamma$ and $(v_{1t} \ V_{2t})'$ to denote the coefficient matrix and the error term of the reduced form of the VAR model associated with \eqref{eq: model: svar} such that \begin{equation}
%  \Gamma : = \mathcal B^{-1} \mathcal A \quad\quad\text{ and }\quad\quad(
%      v_{1t}  , V_{2t})':= \mathcal B^{-1}(
%     u_{1t} , U_{2t})' , \label{eq: def: gamma}
%\end{equation}
% where $\mathcal A$ and $\mathcal B$ are matrices of operators defined in \eqref{eq: model: svar}. Then, the proof of  Proposition \ref{prop: svar: identification} is based on that $\mathcal B^{-1}$ is well defined under the identification stated therein and the fact that the covariance kernel of $(
%       v_{1t}  , V_{2t})'$ is given by $\mathcal B^{-1}\mathbf{\Sigma} \mathcal B^{-1 \ast}$ with $\mathcal B^{-1\ast}$ is the adjoint of $\mathcal B^{-1}$.

To understand the relationship between the coefficient in our benchmark model \eqref{eq: model: benchmark} and the SIRF implied by \eqref{eq: model: svar}, it is convenient to consider the  MA($\infty$) representation.  %the process $\{ y_t, X_t \}_{t\geq 1}$ is stationary and  
Under the invertibility of  $\mathcal B$, the SVAR model in \eqref{eq: model: svar} allows the reduced-form representation in \eqref{eqreduced1} and furthermore, we have
\begin{equation} \label{eqrfvar}
\begin{bmatrix}y_t\\ X_t
\end{bmatrix} = \sum_{j=0}^\infty \Gamma^{j} \mathcal B^{-1} {\mathbf{u}}_{t-j},  %=  \sum_{j=0}^\infty \Gamma^{j} \mathcal B^{-1} \begin{bmatrix}u_{1,t-j} 	\\ U_{2,t-j} \end{bmatrix}  ,
\end{equation} 
where    $\mathbf{u}_t = \left[\begin{smallmatrix}{u}_{1,t}\\{U}_{2,t}\end{smallmatrix}\right] \in\mathbb R \times \mathcal H$ and $\Gamma ^0  = I$ which is the identity operator on $\mathbb{R}\times \mathcal H$.  \eqref{eq: irf: general} and \eqref{eqrfvar} suggest that the SIRF of $ \left[\begin{smallmatrix} y_t\\ X_t\end{smallmatrix} \right]$ at horizon~$h$, when the perturbation $\widetilde\zeta \equiv \left[\begin{smallmatrix}1 \\ \zeta\end{smallmatrix}\right] \in \mathbb{R}\times \mathcal H$ is introduced to $\mathbf{u}_t$, is given as follows:
\begin{equation}\label{eq: sirf}
\IRF_h(\widetilde{\zeta}) :=\mathbb{E}\left[\left.\begin{bmatrix} y_{t+h} \\ X_{t+h}   \end{bmatrix} \right| {\mathbf u}_t+\widetilde\zeta,{\mathbf u}_{t-1},\ldots\right] - \mathbb{E}\left[ \left.\begin{bmatrix} y_{t+h} \\ X_{t+h}   \end{bmatrix} \right| \mathbf u_t,\mathbf u_{t-1},\ldots\right]   = \Gamma^h\mathcal B^{-1} \widetilde{\zeta}.
\end{equation} 
Both $\Gamma ^h$ and $\mathcal B^{-1}$ are operator matrices mapping from $\mathbb R \times \mathcal H$ to $\mathbb R \times \mathcal H$, and the same holds for $\Gamma ^h\mathcal B^{-1}$. Therefore, $\Gamma ^h \mathcal B^{-1}$ allows the following representation:
%	\begin{equation} \nonumber
%		\Gamma^{h} \mathcal B^{-1} =  \begin{bmatrix} 1_{\mathbb{R}}'\Gamma^{h} \mathcal B^{-1}1_{\mathbb{R}}  &  1_{\mathbb{R}}' \Gamma^{h} \mathcal B^{-1} 1_{\mathcal{H}} \\  1_{\mathcal{H}}'\Gamma^{h} \mathcal B^{-1}1_{\mathbb{R}}   &  1_{\mathcal{H}}'\Gamma^{h} \mathcal B^{-1}1_{\mathcal{H}}     \end{bmatrix} = \begin{bmatrix} \IRF_{11,h} &  \IRF_{12,h} \\  \IRF_{21,h} &  \IRF_{22,h}   \end{bmatrix},
% =  \begin{bmatrix} \IRF_{11,j} &  \IRF_{12,j} \\  \IRF_{21,j} & \IRF_{22,j}   \end{bmatrix}
%	\end{equation}
\begin{equation} \nonumber
\Gamma^{h} \mathcal B^{-1} =  \begin{bmatrix} \pcr\Gamma^{h} \mathcal B^{-1}\pcrr & \pcr \Gamma^{h} \mathcal B^{-1}\pchh \\  \pch\Gamma^{h} \mathcal B^{-1}\pcrr   &  \pch\Gamma^{h} \mathcal B^{-1}\pchh   \end{bmatrix} = \begin{bmatrix} \IRF_{11,h} &  \IRF_{12,h} \\  \IRF_{21,h} &  \IRF_{22,h}   \end{bmatrix},
% =  \begin{bmatrix} \IRF_{11,j} &  \IRF_{12,j} \\  \IRF_{21,j} & \IRF_{22,j}   \end{bmatrix}
\end{equation}
where $\pcr$ denotes the projection map given by $\pcr\left[\begin{smallmatrix} x_1\\ x_2\end{smallmatrix}\right] = x_1$ and  $\pcrr$ is its adjoint satisfying $\pcrr(x_1) = \left[\begin{smallmatrix} x_1\\ 0\end{smallmatrix}\right]$. Similarly, $\pch$ (resp.\ $\pchh$) is a projection map satisfying $\pch \left[\begin{smallmatrix}
x_1\\x_2
\end{smallmatrix}\right] = x_2$ (resp.\  $\pch^\ast (x_2) =\left[\begin{smallmatrix}
0\\x_2
\end{smallmatrix}\right] $) for $x_2 \in \mathcal H$. Given these projection maps, each element of $\Gamma ^h \mathcal B^{-1}$ is characterized by different linear maps\footnote{$\IRF_{11,h} : \mathbb R \to \mathbb R$ (i.e., a scalar multiplication), $\IRF_{12,h}: \mathcal H \to \mathbb R$, $\IRF_{21,h}: \mathbb R \to \mathcal H$ and $\IRF_{22,h}:\mathcal H \to \mathcal H$.} and measures an effect of an additional structural shock on $y_{t+h}$ or $X_{t+h}$. For example, the response of $y_{t+h}$ to a function-valued shock  $\zeta\in \mathcal H$ is given by a scalar such that 	\begin{equation}\label{eqpartialirf1}
\IRF_{12,h}(\zeta) = \pcr \Gamma^{h} \mathcal B^{-1}\pchh (\zeta) = \delta_{12,h} \in \mathbb R.
\end{equation}
The response of $X_{t+h}$ to the unit shock on $u_{1,t}$ is characterized by a function such that
\begin{equation}\label{eqpartialirf2}
\IRF_{21,h}(1)=\pch \Gamma^{h} \mathcal B^{-1} \pcrr (1) = \delta_{21,h} \in \mathcal H. 
\end{equation}
These are often of interest to practitioners. %In particular, we note that the interpretation of the coefficient in \eqref{eqpartialirf1} is essentially different from that in the standard SVAR literature. To clarify intuitively, consider an example when $X_t$ is a $d_x$-dimensional random vector, where each element of a $(d_x + 1) \times (d_x +1)$ matrix $\Gamma ^h\mathcal B^{-1}$ is interpreted as a SIRF. The SIRF in \eqref{eqpartialirf1} is similar not to this element-wise definition of the IRFs, but to a response when a shock is jointly given to the structural errors associated with the elements of $X_t$. If a perturbation is introduced jointly to $d_x$ shocks associated with $X_t$, the response will be characterized by a $d_x$-dimensional row vector. Similarly, when $X_t$ is a function-valued predictor, its structural error will respond  to a new perturbation $\zeta$ as follows: 
The SIRFs in \eqref{eqpartialirf1} and \eqref{eqpartialirf2} can be interpreted similarly to those of standard SVAR models. However, because $X_t$ and $U_{2,t}$ are given by $\mathcal H$-valued functions of infinite dimension, the implementation of $\zeta$ to the structural error can be represented with an infinite number of basis functions. That is, for   $\{ \xi_j\}_{j \geq 1}$,   a set of orthonormal basis functions that span $\mathcal H$, we have \begin{equation*}
\langle  U_{2,t},\zeta \rangle  =\langle   \sum_{j=1} ^\infty \langle U_{2,t}, \xi_j \rangle \xi_j,\zeta \rangle  = \sum_{j=1} ^\infty  \langle U_{2,t}, \xi_j \rangle \langle    \xi_j ,\zeta\rangle.
\end{equation*}
In this regard, the SIRF in \eqref{eqpartialirf1} will be interpreted as the response of $y_t$ when all the basis functions move \textit{jointly} to the direction of $\zeta$. A similar observation was previously made by \cite{IR2021} under a parametric assumption, and it remains valid even in our setting without that assumption.

%	Despite this crucial difference, some interesting results developed for standard SVAR models can be extended to our setup. This extension provides wide applicability for our model. In particular, we present the following proposition that establishes the equivalence between \eqref{def: local: irf} and \eqref{eq: sirf} when the true data generating process follows \eqref{eq: model: svar}.

\begin{proposition}\label{prop: svar: identification: a} \normalfont Suppose that $\left[\begin{smallmatrix} y_t\\ X_t\end{smallmatrix} \right]$ satisfies \eqref{eq: model: svar} and \eqref{eqrfvar}.  Then the following hold: below, $v_{1,t}$ (resp.\ $V_{2,t}$) is some $\mathbb{R}$-valued (resp.\ $\mathcal H$-valued) MA($h-1$) process, and  %such that Then, there exists $(\mathbb{R}\times \mathcal H)$-valued MA($h-1$) process $\varepsilon_t$  such that $\varepsilon_{1,t}=1_{\mathbb{R}}'\varepsilon_{t}$, $\varepsilon_{2,t}=1_{\mathcal{H}}'\varepsilon_t$, and 
$v_{1,t}$ (resp.\ $V_{2,t}$) is not correlated with $\{y_{t-\ell},X_{t-\ell}\}$ for every $\ell \geq 0$.
\begin{enumerate}[(i)]
	\item \label{eq: model: svar1}Under the conditions in Proposition \ref{prop: svar: identification}\ref{prop: svar: identification1}, $\IRF_{12,h}$ is equivalent to the coefficient map $\langle \beta_h^y, \cdot \rangle:\mathcal H \to \mathbb{R}$  in \begin{equation}
		y_{t+h}=  \alpha_h ^y y_t + \langle \beta_{h} ^y, X_{t}\rangle + v_{1,t}, \label{prop2: eq1}
	\end{equation}
	where $\alpha_h ^y \in \mathbb{R}$ and $\beta_{h}^y \in \mathcal H$. Moreover, $\IRF_{21,h}$ is equivalent to the coefficient $a_h ^X : \mathbb{R} \to \mathcal H$ in
	\begin{equation}
		X_{t+h} = \alpha_{h} ^X y_t + \beta_{h}^X y_{t-1} + \gamma_h^X X_{t-1} + V_{2,t},\label{prop2: eq2}
	\end{equation}
	where $\beta_h^X : \mathbb{R}\to\mathcal H$ and $\gamma_h^X : \mathcal H\to \mathcal H$.
	\item Under the conditions in Proposition \ref{prop: svar: identification}\ref{prop: svar: identification2},  $\IRF_{12,h}$ is equivalent to the coefficient map $\langle \alpha_{h} ^y,\cdot \rangle:\mathcal H \to \mathbb{R}$ in
	\begin{equation}
		y_{t+h} =  \langle \alpha_{h} ^y, X_t \rangle  + \langle {\beta}_{h} ^y, X_{t-1}\rangle  + \gamma_h ^y y_{t-1} + v_{1,t},\label{prop2: eq4}
	\end{equation}
	where $\alpha_h^y \in \mathcal H$, $\beta_h^y \in \mathcal H$ and $\gamma_h^y \in \mathbb{R}$. Moreover, $\IRF_{21,h}$ is equivalent to the coefficient $ \beta_h ^X:  \mathbb{R} \to \mathcal H$ in 
	\begin{equation}
		X_{t+h} =  {\alpha}_h ^X X_t +  {\beta}_{h} ^X   y_{t} + V_{2,t}, \label{prop2: eq3}
	\end{equation}
	where $\alpha_h^X : \mathcal H\to \mathcal H$.  
	\end{enumerate}\end{proposition}
	Proposition~\ref{prop: svar: identification: a}\ref{eq: model: svar1} provides a theoretical foundation for interpreting the map $\langle \beta_h, \cdot \rangle$ in the benchmark model \eqref{eq: model: benchmark} as $\IRF_{12,h}$. Specifically, if the identification condition $\beta_{12} = 0$ holds and the true DGP follows \eqref{eq: model: svar}, then the coefficient map $\langle \beta_h, \cdot \rangle$ %, where $\beta_h$ in our benchmark model \eqref{eq: model: benchmark} with $\mathbf w_t$ set to $y_t$, 
	represents the SIRF at horizon $h$ when the functional structural error associated with $X_t$ experiences a shock $\zeta$. It may be possible to estimate the SIRFs directly from the RFVAR model by applying an appropriate regularization scheme required for an infinite dimensional setup. We provide a brief outline of this approach in  Section~\ref{sec: svar: est} of the Supplementary Material. However, as detailed in the appendix, this approach makes statistical inference on the SIRFs much more challenging, %, as it requires estimating a function-on-function regression model, where the coefficients and their variances are characterized by linear operators. 
	while our benchmark model simplifies the inference procedure by exclusively focusing on the outcome equation of interest. 
	
	Proposition~\ref{prop: svar: identification: a}\ref{eq: model: svar1} implies that $\IRF_{21,h}$ can be characterized by a certain coefficient map in a function-on-function regression model with  scalar control variables. We  study how statistical inference can be implemented for this quantity in Section~\ref{sec_extension} of the Supplementary Material.

	\subsection{Caveats of finite dimensional approximation\label{subsec: para}}
	In empirical studies involving functional random variables, it is a common practice to first approximate the functional variable as a finite dimensional vector and then apply existing estimation or inference methods. However, as \cite{Nielsen2023} pointed out in the context of cointegration tests for functional time series, this finite dimensional approximation can lead to misspecification errors, causing  misleading estimation and interpretation.  To illustrate this, suppose
	%	\begin{equation}
%	\begin{bmatrix} y_t 
	%			\\ X_t \end{bmatrix} ={ \begin{bmatrix} \alpha_{11} & \alpha_{12} \\ \alpha_{21} & \mathcal A_{22} \end{bmatrix}}\begin{bmatrix} y_{t-1} 
	%			\\ X_{t-1} \end{bmatrix} + \begin{bmatrix} u_{1t} 
	%			\\ U_{2t}\end{bmatrix}.  \label{eq: model: svar2}
%	\end{equation}
\begin{equation}
\widetilde	\Upsilon_t = A \widetilde\Upsilon_{t-1} + \mathbf{u}_t, \label{eq: model: svar3}
\end{equation} where $ \widetilde	\Upsilon_t = \left[\begin{smallmatrix} y_t\\ X_t\end{smallmatrix} \right] \in \mathbb{R}\times \mathcal H$. To study the above model, a popular approach is to transform $X_t$ in $\widetilde\Upsilon_t$ into a finite dimensional vector, say $\mathbf{x}_t$, in advance and then estimate a VAR model similar to \eqref{eq: model: svar3} with $(y_t,\mathbf{x}_t')'$. The transformation from $X_t$ to $\mathbf{x}_t$ may be expressed as a projection of $X_t$ onto a finite dimensional space such that $P_0{\widetilde\Upsilon}_t = (y_t, \langle X_t, \xi_1 \rangle,\ldots, \langle X_t, \xi_K \rangle)'$ with some finite $K$. $\xi_j$ may be either a pre-specified parametric function (\citealp{IR2021}) or a principal component (\citealp{BHCJ2023}). However, as shown by \cite{Nielsen2023}, this projection operation does not preserve the VAR structure. Specifically, from \eqref{eq: model: svar3}, we have
\begin{equation}
P_0 \widetilde \Upsilon_t = P_0 A P_0\widetilde\Upsilon_{t-1} + \mathbf{v}_t, \qquad \mathbf{v}_t = P_0 A (I-P_0) \widetilde\Upsilon_{t-1} + P_0 \mathbf{u}_t, \nonumber
\end{equation} If $P_0 \neq I$, $P_0\widetilde\Upsilon_{t}$ does not follow the VAR(1) structure in \eqref{eq: model: svar3} since it is generally correlated with $\mathbf{v}_t$, which potentially invalidates many existing estimation and inference methodologies developed for VAR models; for instance, \cite{Nielsen2023} show that the cointegration rank test of \cite{Johansen1991,Johansen1995} is generally misleading in this setup. This finite dimensional approximation can be justified only when $P_0$ is close enough to $I$ (unless we consider the special case where $X_t$ can be fully expressed by a finite number of known basis functions, which is essentially equivalent to a finite dimensional setup). This implies that \( K \) must be sufficiently large. In functional linear models, it is well known that an increase in \( K \) can significantly raise the variance of coefficient estimators while reducing the regularization bias. Thus, choosing \( K \) in advance without accordance with the desired inferential methods may introduce an unbalanced bias-variance trade-off. 
%	However, it is  well known in the literature on functional linear models that even a small increase in $K$ can significantly increase the variance of the coefficient estimators. Therefore, $K$ is typically kept significantly smaller than the sample size $T$, though it may increase as $T$ grows, in accordance with estimator-specific theoretical justification. %This introduces a conflict between two opposite goals in the finite-dimensional approach. 
%In fact, our approach to estimating the $\IRF$ fully address the regularization error associated  does not involve similar misspecification issues arising from approximation.

\section{Estimation}\label{sec:est}
\subsection{Representation of the proposed model} 

To facilitate our discussion, it is useful to consider an alternative representation of our benchmark model in \eqref{eq: model: benchmark}. We  let $\elltwo$ denote the product Hilbert space, whose inner product is given by the sum of the inner products in $\mathbb{R}^m$ and $\mathcal H$ (see Section~\ref{sec_prelim2} of the Supplementary Material). With a slight abuse of notation, we let $\langle \cdot, \cdot \rangle$ denote the inner product defined on any of $\mathbb{R}^m$, $\mathcal H$, and $\elltwo$ (e.g., if $h_1, h_2 \in \mathbb{R}^m$ then $\langle h_1,h_2 \rangle=h_1'h_2$). Because the inner product is inherently defined for elements in the same space, there is little risk of confusion following this simplification. For any elements $h_j \in \mathcal H_j$ and $h_k\in \mathcal H_k$, where $\mathcal H_j$ and $\mathcal H_k$ can be any of $\mathbb{R}^m$, $\mathcal H$, and $\elltwo$, we let $\otimes$ denote the tensor product, defined by $h_j\otimes h_k (\cdot) = \langle h_j,\cdot \rangle h_k$ (see Section~\ref{sec_prelim} of the Supplementary Material). In particular, if $h_1 = \left[\begin{smallmatrix}h_{11}\\ h_{12}\end{smallmatrix}\right] \in \elltwo$ and $h_{2} = \left[\begin{smallmatrix}h_{21}\\ h_{22}\end{smallmatrix}\right] \in \elltwo$, then $h_1\otimes h_2$ can be understood as an operator matrix given by $\left[\begin{smallmatrix} h_{11} \otimes h_{21} & h_{12} \otimes h_{21} \\   h_{11} \otimes h_{22} &   h_{12} \otimes h_{22}\end{smallmatrix}\right]$.   We define the following $\elltwo$-valued random element  $\Upsilon_t$ and its coefficient $\theta_h$:
\begin{equation} \label{eqdefY}
\Upsilon_t =  \begin{bmatrix}\mathbf{w}_t\\ X_t\end{bmatrix} \quad\text{ and }\quad \theta_h = \begin{bmatrix}
	\alpha_h \\\beta_h
\end{bmatrix}.
\end{equation}
%That is, $\Upsilon_t$ is  a random tuple consisting of a numeric vector and a function. 
Under this representation, $ \langle\theta_h, \Upsilon_t\rangle = \langle  \alpha_h,\mathbf{w}_t \rangle + \langle \beta_h, X_t \rangle$, and  \eqref{eq: model: benchmark} is simplified into the following regression model:
\begin{equation} \label{eq: model: benchmark: reduced2}
y_{t+h} =   \langle\theta_h, \Upsilon_t\rangle  + u_{h,t}.
\end{equation}  %The above model involves $(\mathbb{R}^m \times \mathcal H)$-valued predictor $\Upsilon_t$, given by a random tuple of a numeric vector and a function. 

%As an alternative form of the model \eqref{eq: model: benchmark: reduced2}, we may represent $X_t$ as a random sequence $X_t^{(\ell)} = (\langle X_t,f_1 \rangle, \langle X_t,f_2 \rangle,\ldots)'$ for any arbitrary orthonormal basis $f_j$ of $\mathcal H$, and thus may view the model \eqref{eq: model: benchmark: reduced2} as the functional linear model involving the predictor taking values in the Hilbert space of square-summable sequences. 

%=   \alpha_h' \mathbf{w}_t + \langle \beta_h, X_t \rangle$.

% \{\beta_{h,j}^{(\ell)}\}$ with $\beta_{h,j}^{(\ell)} = \langle \beta_h,f_1 \rangle,\langle \beta_h,f_2 \rangle,\ldots)'$, and hence $\beta_h = \sum_{j=1}^\infty \beta_{h,j}^{(\ell)} $	
%\begin{equation} \label{eqtheta1}
%	\theta_{\ha}^{(\ell)} = (\theta_{h,1}^{(\ell)},\ldots,\theta_{h,m}^{(\ell)})' \in \mathbb{R}^m, \qquad \theta_{\hb} = (\theta_{h,m+1}^{(\ell)}, \theta_{h,m+2}^{(\ell)},\ldots)' \in \mathcal H^{(\ell)}. \nonumber
%\end{equation} Then the estimation of $\alpha_h$ (resp.\ $\beta_h$) reduces to the estimation of $\theta_{1,h}$ (resp.\ $\theta_{2,h}$), since	\begin{equation}\label{eqtheta2}
%	\alpha_h =  (\theta_{h,1}^{(\ell)},   \ldots, \theta_{h,m}^{(\ell)})' \in \mathbb{R}^m, \qquad   \beta_h = \sum_{j=1}^\infty \theta_{h,m+j}^{(\ell)} f_j \in \mathcal H.  \nonumber
%\end{equation}

\subsection{Estimator}\label{sec:est1}

We  hereafter assume that the variables $y_{t+h}$, $X_t$ and $\mathbf w_t$ have zero means. Extending to the case where these means are unknown is straightforward by considering their demeaned values, under the assumptions detailed shortly. Furthermore, we assume that $X_t$ and $\mathbf{w}_t$ are exogenous with respect to $u_{h,t}$. Then, by its definition  in \eqref{eqdefY}, $\Upsilon_t$  satisfies the exogeneity condition such that  $C_{\Upsilon u} = \mathbb{E}[u_{h,t}\Upsilon_t] = 0$.  This in turn implies the following moment condition: %Our estimator is associated with  \eqref{eq: model: benchmark: reduced2} and   the following moment condition is implied by the exogeneity condition: 
\begin{equation}
C_{\Upsilon y}\equiv \mathbb{E}[y_{t+h}\Upsilon_t] = \mathbb E[ \langle\Upsilon_t ,\theta_h \rangle \Upsilon_t ]   = \mathbb E[\Upsilon_t \otimes \Upsilon_t] \theta_h \equiv  C_{\Upsilon\Upsilon}\theta_h. \label{eqpopmoment} 
\end{equation} %where $C_{\Upsilon y} = \mathbb{E}[y_{t+h}\Upsilon_t] = \mathbb E[ \langle\Upsilon_t ,\theta_h \rangle \Upsilon_t ]$ and   $C_{\Upsilon \Upsilon} =  \mathbb{E}[\Upsilon_t\otimes \Upsilon_t]$. %The above moment condition follows from  $ C_{\Upsilon y}\equiv \mathbb{E}[y_{t+h}   \Upsilon_t] = C_{\Upsilon\Upsilon} \theta_h$. 
Under the identification condition to be detailed shortly, $\theta_h$ can be estimated from the sample counterpart of \eqref{eqpopmoment} given by
\begin{equation}
\widehat{C}_{\Upsilon y} = \widehat{C}_{\Upsilon\Upsilon}\bar{\theta}_h, \label{eqsammoment} 
\end{equation}where $\widehat{C}_{\Upsilon y}$ and $\widehat{C}_{\Upsilon\Upsilon}$  are given by $
\widehat{C}_{\Upsilon y} = T^{-1}\sum_{t=1}^T  y_{t+h}  \Upsilon_t$ and $  \widehat{C}_{\Upsilon \Upsilon} = T^{-1} \sum_{t=1}^T  \Upsilon_t \otimes \Upsilon_t$.   However, in our setup, $\bar{\theta}_h$ is not generally obtainable from \eqref{eqsammoment}, since  $  \widehat{C}_{\Upsilon\Upsilon}$, an operator acting on $\widetilde{\mathcal H}$, is not invertible. We circumvent this issue by considering an estimator constructed using a regularized inverse of  $\widehat{C}_{\Upsilon\Upsilon}$ based on its operator Schur complement. Specifically, we  note that whenever convenient, ${C}_{\Upsilon\Upsilon}$ and $\widehat{C}_{\Upsilon\Upsilon}$ can be understood as operator matrices on $\elltwo$ such that 
\begin{equation} \label{covblock}
{C}_{\Upsilon\Upsilon} = \begin{bmatrix}
	\CCC_{11} &\CCC_{12} \\ \CCC_{21}& \CCC_{22}
\end{bmatrix} \quad\text{and}\quad
\widehat{C}_{\Upsilon\Upsilon} = \begin{bmatrix}
	\widehat{\CCC}_{11} & \widehat{\CCC}_{12} \\ \widehat{\CCC}_{21} & \widehat{\CCC}_{22}
\end{bmatrix},
\end{equation}
where  $\CCC_{ij} = \mathbb{E}[g_{jt}\otimes g_{it}]$,  $\widehat{\CCC}_{ij} = T^{-1}\sum_{t=1}^T g_{jt}\otimes g_{it}$,  $g_{1t} = \mathbf{w}_t$, and $g_{2t} = X_t$. That is, for any $\mathbf{a} \in \mathbb R^m$, $\CCC_{11}\mathbf{a} = \mathbb E[ \langle \mathbf w_t, \mathbf a \rangle \mathbf{w}_t]$ and $\CCC_{21}\mathbf{a} = \mathbb E[ \langle \mathbf w_t, \mathbf a \rangle X_t]$. Similarly, for any $ h\in \mathcal{H}$, $\CCC_{12} h =  \mathbb E[\langle X_t, h\rangle \mathbf w_t]$ and $\CCC_{22} h = \mathbb E[\langle X_t, h \rangle X_t]$. Note that $\CCC_{11}$ is the covariance of  $ \mathbf{w}_t$ which is assumed to be invertible throughout the paper (see Assumption \ref{assum2}).  
%$\widehat\CCC_{11}: \mathbb{R}^k \mapsto \mathbb{R}^k$,  $\widehat\CCC_{12}:\mathcal H\mapsto \mathbb{R}^k$,  $\widehat\CCC_{21}: \mathbb{R}^k \mapsto \mathcal H$ and $\widehat\CCC_{22}: \mathcal H \mapsto \mathcal H$.
We then  define the following operator Schur complements $\SA$ (of ${C}_{\Upsilon\Upsilon}$) and $\widehat{\SA}$ (of $\widehat{C}_{\Upsilon\Upsilon}$) (see \citealp[Section 2.2]{Bart2007}):
\begin{equation} \label{eqopschur}
\SA ={\CCC}_{22}-{\CCC}_{21}{\CCC}_{11}^{-1}{\CCC}_{12}\quad\text{and}\quad \widehat{\SA}=\widehat{\CCC}_{22}-\widehat{\CCC}_{21}\widehat{\CCC}_{11}^{-1}\widehat{\CCC}_{12}.
\end{equation} 
{The operator Schur complement $\SA$ can  be understood as the covariance of the residuals obtained by regressing $X_t$   on $\mathbf{w}_t$, and $\widehat{\SA}$  is its sample counterpart (see  \citealp{fukumizu2004dimensionality}).} Lastly, to introduce our regularization scheme, we further represent $\SA$ (resp.\  $\widehat{\SA}$) with respect to its eigenvalues and eigenvectors $\{ \lambda_j,\nu_j \}_{j \geq 1}$ (resp.\ $\{ \widehat\lambda_j, \widehat \nu_j \}_{j \geq 1}$) as follows:
\begin{equation}\label{eqshur0}
{\SA}= \sum_{j=1}^\infty {\lambda}_j {v}_j \otimes {v}_j\quad\text{and}\quad \widehat{\SA}= \sum_{j=1}^\infty \hat{\lambda}_j \hat{v}_j \otimes \hat{v}_j,
\end{equation}
where $\lambda_1\geq\lambda_2\geq\ldots \geq 0$ and $\hat{\lambda}_1\geq\hat{\lambda}_2\geq\ldots\geq 0$. %Note that $\widehat{\SA}$ is a self-adjoint, nonnegative and compact. 
This representation is possible as $\SA$ and $\widehat{\SA}$ are self-adjoint, nonnegative, and compact  (see \citealp{Bosq2000}, p.\ 34). The empirical eigenelements $\{\hat{\lambda}_j , \hat{v}_j\}$ can be obtained using the functional principal component analysis (FPCA). From \eqref{eqshur0},  noninvertibility of $\widehat{\SA}$ is evident as its partial inverse $\widehat{\SA}_K^{-1} = \sum_{j=1}^K \hat{\lambda}_j ^{-1}\hat{v}_j\otimes \hat{v}_j$ grows without bound in the operator norm as $K$ increases. This consequently leads to noninvertibility of $\widehat{C}_{\Upsilon\Upsilon}$ (see \citealp{Bart2007}, p.\ 29). Therefore, to construct a regularized inverse of $\widehat{C}_{\Upsilon\Upsilon}$, we first construct regularized inverses of the operator Schur complements $\SA$ and $\widehat{\SA}$ as follows:
\begin{equation} \label{eqshur1}
\SA_{\KK}^{-1} = \sum_{j=1}^{\KK} \lambda_j^{-1}v_j\otimes v_j \quad\text{and}\quad \widehat{\SA}_{\KK}^{-1}= \sum_{j=1}^{\KK} \hat{\lambda}_j^{-1} \hat{v}_j \otimes \hat{v}_j,\quad \text{ where}\quad \KK = \max\{j: \hat{\lambda}_j^2 \geq \reg\} ,
\end{equation}
for the regularization parameter $\reg$, decaying to zero as $T$ increases. Then,   we define the following regularized inverse $\widehat{C}_{\Upsilon\Upsilon,\KK}^{-1}$ of $\widehat{C}_{\Upsilon\Upsilon}$:

\begin{equation} \label{eqreginv}
\widehat{C}_{\Upsilon\Upsilon,\KK}^{-1}  = \begin{bmatrix} 
	\widehat{\CCC}_{11}^{-1} + \widehat{\CCC}_{11}^{-1}\widehat{\CCC}_{12} \widehat{\SA}^{-1}_{\KK} \widehat{\CCC}_{21}\widehat{\CCC}_{11}^{-1} & -\widehat{\CCC}_{11}^{-1} \widehat{\CCC}_{12} \widehat{\SA}^{-1}_{\KK} \\ -\widehat{\SA}^{-1}_{\KK} \widehat{\CCC}_{21}\widehat{\CCC}_{11}^{-1} &
	\widehat{\SA}^{-1}_{\KK}
\end{bmatrix}.
\end{equation}
By using the regularized inverse \eqref{eqreginv} and the moment condition \eqref{eqsammoment}, we define our estimator $\hat\theta_h$  as follows:
\begin{equation}
\hat{\theta}_h = \widehat{C}_{\Upsilon\Upsilon,\KK}^{-1}  \widehat{C}_{\Upsilon y}. \nonumber
\end{equation} 
\begin{remark}
Similar to the literature in Section~\ref{subsec: para},  \cite{AP2006} and \cite{SHIN2009} resolve the ill-posed inverse problem similar to \eqref{eqsammoment} by approximating    $X_t$ with  a finite number of eigenvectors of $\widehat\CCC_{22}$. %A least-square-type estimator is then defined using the variables $\mathbf{w}_t$ and $\mathbf{x}_t$. 
%Due to the finite dimensionality of $(\mathbf{x}_t', \mathbf{w}_t')'$, its inverse is well defined in finite samples, which is used as a regularized inverse to define a least-squares-type estimator. 
On the other hand, we in this paper utilize the fact that the ill-posedness is directly associated by the Schur complement $\widehat{\SA}$, so we  regularize it to solve the inverse problem. The operator to be regularized  $\widehat{\SA}$ is the the covariance of the residuals obtained by regressing $X_t$   on $\mathbf{w}_t$, which is differentiated from $\widehat{\Gamma}_{22}$.  Moreover, in contrast to the aforementioned literature whose focus lies on consistency, we provide a formal statistical inference procedure for $\theta_h$ based on the asymptotic properties of our proposed estimators, which is another novel contribution of this paper.
		\end{remark}

		\subsection{Identification and consistency} \label{sec:est1a} 
		Let $\ker A$ denote the kernel of $A$. To uniquely identify $\theta_h$ from \eqref{eqpopmoment}, we assume the following: 
		\begin{assumFLS} \label{assum1} $\langle x,\theta_h \rangle=0$ for all $x\in \ker C_{\Upsilon\Upsilon}$. % and for any orthonormal basis $\{g_j\}_{j\geq 1}$ of $\mathcal {H}$, $\sum_{j=1}^\infty \langle \theta_h, g_j \rangle^2 < \infty$. %, i.e., $\langle \theta_h, v_j \rangle = 0$ for all $j$ such that $\lambda_j = 0$. 
		\end{assumFLS}
		Note that $\Upsilon_t$ involves the $\mathcal H$-valued random variable $X_t$. In the literature on functional data analysis, it is commonly assumed that the covariance  of such a functional random variable allows infinitely many nonzero eigenvalues. Therefore, as deduced from \citet[Proposition 2.1]{Mas2007} and the Riesz representation theorem (\citealp{Conway1994}, p.\ 13), the parameter of interest in \eqref{eqpopmoment}, $\theta_h$, is not uniquely identified as an element of $\elltwo$ if $\ker C_{\Upsilon \Upsilon} \neq \{0\}$. The failure of identification occurs because, for any $\psi \in \ker C_{\Upsilon\Upsilon}$, $C_{\Upsilon\Upsilon} ( \theta_h+\psi ) =  C_{\Upsilon\Upsilon}\theta_h$. Thus, if $\theta_h$ satisfies \eqref{eqpopmoment}, $\theta_h+\psi$ also satisfies it. Assumption~\ref{assum1} prevents this failure of identification.

		Our next assumption is related to asymptotic properties of $\widehat\theta_h$. Below, $\widehat{C}_{\Upsilon u} = T^{-1}\sum_{i=1} ^T u_{h,t}  \Upsilon_t$  and  $\Vert \cdot \Vert _{\op}$ is the operator norm (see  Section~\ref{sec_prelim} of the Supplementary Material for its formal definition). %Also note that, due to the equation $u_{t+h}\otimes \Upsilon_t (x) = \langle u_{t+h}\Upsilon_t,x \rangle$ for every $x$, our conditions on the $\elltwo$-valued time series $\{u_{t+h}\Upsilon_{t}\}$ have implications on the properties of ${C}_{\Upsilon u}$ and $\widehat{C}_{\Upsilon u}$.}
\begin{assumFLS} \label{assum2}
\begin{enumerate*}[(i)]
\item \label{assum2a}  \eqref{eq: model: benchmark: reduced2}  holds with $\mathbb{E}[u_{h,t}\Upsilon_t]=0$;
\item\label{assum2b}   $\{\Upsilon_t\}$, $\{u_{h,t}\}$ and $\{u_{h,t}\Upsilon_t\}$ are stationary and  $L^4$-$m$-approximable; %  $\{\Upsilon_t - \mathbb{E}[\Upsilon_t]\}$,  $\{u_{h,t}\}$ and  are  in the relevant spaces (see Section \ref{Section_AFTS}); and 
% \item For some iid sequence $\{\varepsilon\}_{t\in \mathbb{Z}}$ satisfying $\mathbb{E}[\varepsilon_t]=0$ and $\mathbb{E}[\|\varepsilon_t\|^{2+\delta}] < \infty$, 
% $$X_t = \sum_{j=0}^\infty \Psi_j \varepsilon_{t-j},$$
% where $\{\Psi_j\}_{j\geq 0}$ is a sequence of linear operators satisfying $\sum_{j=1}^\infty j\|\Psi_j\| < \infty$.  
\item \label{assum2c}  $\|\widehat{C}_{\Upsilon u}\|_{\op} = O_p(T^{-1/2})$, $\|\widehat{C}_{\Upsilon\Upsilon} -C_{\Upsilon\Upsilon}\|_{\op} = O_p(T^{-1/2})$, and $\|\widehat{\CCC}_{11}^{-1} - \CCC_{11}^{-1}\|_{\op} = O_p(T^{-1/2})$.
\end{enumerate*} 
\end{assumFLS}
The $L^4$-$m$-approximability in Assumption \ref{assum2}\ref{assum2b}, whose formal definition is provided in  Section~\ref{Section_AFTS} of the Supplementary Material, is employed to use existing limit theorems. The assumption is not only widely adopted in the literature on stationary functional time series but also inclusive of many practical and interesting examples, such as the SVAR model in Section \ref{sec: svar}. Assumption~\ref{assum2}\ref{assum2c} contains high-level conditions on the limiting behavior of $\widehat{C}_{\Upsilon \Upsilon}$ and $\widehat{C}_{\Upsilon u}$, and  is not restrictive given the stationarity of $\{\Upsilon_t\}$ and $\{ u_{h,t}\Upsilon_t\}$. Some primitive sufficient conditions can be found in  \citet[Chapter  2]{Bosq2000}. 

We  impose the following conditions to characterize the rate of convergence of $\hat\theta_h$.
\begin{assumFLS} \label{assum3} For a generic constant $\CC>0$, the following holds: \begin{enumerate*}[(i)] 
\item\label{assum3a}  
$\lambda_j^2 \leq \CC j^{-\rho}$ and $\lambda_j^2-\lambda_{j+1}^2 \geq \CC j^{-\rho-1}$ for $\rho>2$;
\item\label{assum3b}   $|\langle \beta_h, v_j \rangle| \leq \CC j^{-\varsigma}$ for   $\varsigma > 1/2$. 
\end{enumerate*}
\end{assumFLS}
% The polynomially decaying upper envelop  of the eigenvalues of $C_{\Upsilon \Upsilon}$ be well-separated, and this condition  is    
When $\SA$ is given by the standard covariance of the functional explanatory variable, Assumption~\ref{assum3}\ref{assum3a} reduces to a standard assumption commonly used in the literature on functional linear models.
Given that $\sum_{j=1}^\infty \lambda_j < \infty$ must hold
(see Lemma \ref{lem1}), the first condition $\lambda_j^2 \leq \CC j^{-\rho}$ for some $\rho>2$ is natural (obviously, $\rho\leq 2$ may not result in the summability of $\{\lambda_j\}_{j\geq 1}$). By the latter condition, we require that the eigenvalues of $\SA$ are well separated. As documented in a similar context concerning  functional linear models (see e.g., \citealp{Hall2007,imaizumi2018,seong2021functional}), this separation is crucial for achieving sufficient accuracy in the estimation of eigenelements. Assumption~\ref{assum3}\ref{assum3b}  can be understood as a smoothness condition on $\beta_h$  with respect to the eigenvectors $\{v_j\}_{j \geq 1} $. In contrast to $\beta_h$, a similar restriction is not necessary for $\alpha_h$ associated with the vector-valued random variable $\mathbf{w}_t$. Considering that $\beta_{h}$ is an element in a Hilbert space satisfying $\sum_{j=1}^\infty \langle \beta_h, v_j \rangle^2 <\infty$,  Assumption \ref{assum3}\ref{assum3b} does not seem too restrictive. 

%(WK REVIEW STOPPED HERE)

The following theorem states consistency and the rate of convergence of $\hat\theta_h$.
\begin{theorem} \label{thm1} Suppose Assumptions \ref{assum1}--\ref{assum3} hold and $T\reg^{1+4/\rho} \to \infty$. Then, \begin{equation}	\|\hat{\theta}_h - \theta_h\| = O_p(T^{-1/2}\reg^{-1/2-2/\rho} + \reg^{(2\varsigma-1)/2\rho}).\nonumber\end{equation}
\end{theorem}
In Theorem \ref{thm1}, the choice of the regularization parameter $\reg$ to ensure the consistency of the proposed estimator depends on $\rho$, allowing $\reg$ to decay at a faster rate as $\rho$ increases. Assuming that the same regularization parameter is employed, an increase in $\rho$ implies a faster convergence of $\hat{\theta}_h$ to $\theta_h$.
It is worth noting that the consistency is established if $\reg$ satisfies $T\reg^{3} \to \infty$  (such as e.g.,  $\reg = T^{-1/3 + \epsilon}$ or $T^{-1/3} \log^{\epsilon} T$ for small $\epsilon>0$) as long as $\rho >2$ as assumed in Assumption \ref{assum3}; of course, $\reg$ decaying at a slower rate can also be considered in practice, without affecting consistency. Nevertheless, this naive selection of $\tau$ may be particularly advantageous for practitioners with little knowledge of these eigenvalues.

\subsection{Statistical inference based on local asymptotic normality}\label{sec:est1b}
Given that $\langle \beta_h, \zeta\rangle$ can naturally be interpreted as the response to a perturbation applied to $X_t$, as in \eqref{def: local: irf}, and in some special cases, such as in Proposition \ref{prop: svar: identification: a}, it can further be interpreted as $\IRF_{12,h}(\zeta)$, it is of our interest to conduct statistical inference on this quantity. More generally,  we consider inference on \( \langle \theta_h, \zeta \rangle \) for any \( \zeta \in \elltwo \). To clarify the perturbation \( \zeta \) in the subsequent discussion and avoid potential confusion regarding the space in which \( \zeta \) takes values, we let  
\begin{equation}  \label{eqzetadecom}
\zeta =\left[\begin{matrix}\zeta_{1}\\ \zeta_{2}\end{matrix}\right] \in \widetilde{\mathcal H},  \quad\text{where}\quad \zeta_1 \in \mathbb{R}^m \quad\text{and}\quad \zeta_2 \in \mathcal H .
\end{equation}  
By setting \( \zeta_1 = 0 \), inference on \( \langle \theta_h, \zeta \rangle \) reduces to inference on \( \langle \beta_h, \zeta_2 \rangle \) for \( \zeta_2 \in \mathcal H \). Meanwhile, if $\zeta_2 = 0$, it reduces to  inference on $\langle \alpha_h, \zeta_1\rangle \equiv \alpha_h ' \zeta_1$.
Let $\hat{u}_{h,t}= y_{t+h}- \langle \Upsilon_t, \hat{\theta}_h\rangle$ and $\mathrm{k}(\cdot)$ be a standard weight function to be specified shortly. Then, under Assumption \ref{assum2}, % the $L^2$-$m$-approximability of $		\{u_{h,t} \Upsilon_t\}$ (see  and Section \ref{Section_AFTS}), 
we define the following long-run covariance operator and its sample counterpart (see  \citealt[Theorem 2]{berkes2013weak}):   below, ${\UU}_{h,t} = {u}_{h,t}\Upsilon_{t}$ and $\hat{\UU}_{h,t} = \hat{u}_{h,t}\Upsilon_{t}$.
\begin{equation*} 
\Lambda_{\UU} =\sum_{s=-\infty}^\infty \mathbb{E}[\UU_{h,t} \otimes \UU_{h,t-s}]   \quad	\text{and}\quad		\widehat{\Lambda}_{\UU} =  \frac{1}{T} \sum_{s=-\bdw}^{\bdw}\mathrm{k}\left(\frac{s}{\bdw}\right)\left( \sum_{1 \leq t, t-s \leq T} \hat{\UU}_{h,t}  \otimes \hat{\UU}_{h,t-s} \right).
\end{equation*} 
%Since  $\Lambda_{u\Upsilon}$ is nonnegative and self-adjoint, it allows the following spectral representation:\begin{equation} \label{eqspectrallambda}
%    \Lambda_{u\Upsilon} = \sum_{j=1}^{\infty} \mu_{j} \varpi_j\otimes \varpi_j, \quad \mu_1\geq\mu_2\geq\ldots\geq 0.
%\end{equation} 
%	Our asymptotic normality  to be developed is based on restriction on $\Lambda_{\UU}$ and the following quantity:
We also define the quantity $c_{m,j}$ as
\begin{equation} \label{eqcm}
	c_{m,j}(\zeta) =  \lambdatw_j^{-1}\langle \zeta,\vtw_j \rangle\bigg/\sqrt{\sum_{j=1}^{m} \lambdatw_j^{-2}\langle \zeta,\vtw_j \rangle^2} ,
\end{equation}
where   $\{  \lambdatw_j \}$ and  $ \{ \vtw_j\}$ denote the eigenvalues and eigenvectors of $C_{\Upsilon\Upsilon}$, which are generally different from the eigenelements of $\SA$. The quantity satisfies $\sum_{j=1}^{m} c_{m,j}^2(\zeta) = 1$ for every $m$ unless $\langle \zeta,\vtw_j \rangle \neq 0$ for at least one $j$. 

To establish local asymptotic normality of $\hat\theta_h$, we employ the following assumption:
\begin{assumFLS} \label{assum4} 
	\begin{enumerate*}[(i)]
		\item \label{assum4a}  $\bdw=\bdw(T) \to \infty$, $\bdw(T)/T \to 0$ and  $\mathrm{k}(\cdot)$ is an even function with $\mathrm{k}(0)=1$, $\mathrm{k}(s) = 0$ if $|s| > c$ for some $c>0$ and $\mathrm{k}(\cdot)$ is continuous on $[-c,c]$;
		%\item \label{assum4aa}  $\Lambda_{u\Upsilon} v_j$
		%\item\label{assum4a} For some small enough $\CC$, there exists a finite integer $m$ such that  
		%\begin{equation}
		%\mathbb{P}\left\{\sum_{k=1}^\infty \mu_k \left( \frac{\left\|\sum_{j=1}^{\KK} \lambda_j^{-1} \langle \zeta,v_j \rangle\langle v_j,\varpi_k \rangle \right\|^2}{\sum_{j=1}^{\KK} \lambda_j^{-2}\langle \zeta,v_j \rangle^2}\right) > \CC \right\} \to_p 1.\end{equation}
	%\begin{equation}
	%\mathbb{P}\left\{\sum_{k=1}^\infty \mu_k \left( \left\|\sum_{j=1}^{\KK} \lambda_j^{-1} \langle \zeta,v_j \rangle\langle v_j,\varpi_k \rangle \right\|^2 \bigg/ \sum_{j=1}^{\KK} \lambda_j^{-2}\langle \zeta,v_j \rangle^2\right) > \CC \right\} \to_p 1,\end{equation}
%where $\mu_k$ and $\varpi_k$ are the eigenpair of the self-adjoint, nonnegative compact operator $\Lambda_{u\Gamma}$, i.e., $\Lambda_{u\Upsilon} = \sum_{k=1}^{\infty} \mu_{k} \varpi_j\otimes \varpi_k$ for $\mu_1\geq\mu_2\geq\ldots\geq 0.$
%\begin{equation} \label{eqspectrallambda}
%    \Lambda_{u\Upsilon} = \sum_{k=1}^{\infty} \mu_{k} \varpi_j\otimes \varpi_k, \quad \mu_1\geq\mu_2\geq\ldots\geq 0.
%\end{equation} 
\item \label{assum4aa}  $\zeta \notin \ker C_{\Upsilon\Upsilon}$ and $\Lambda_{\UU} \vtw_j \neq 0$ for $\vtw_j$ corresponding to $\lambdatw_j >0$;
\item  \label{assum4bb} $\sum_{j=1}^{m} \sum_{\ell=1}^{m} c_{m,j}(\zeta) c_{m,\ell}(\zeta) \langle \Lambda_{\UU}\vtw_j,\vtw_{\ell} \rangle \to \CC>0$ as $m\to \infty$; \item  \label{assum4cc}  $\sup_{1\leq t\leq T} \|\Upsilon_t\| = O_p(1)$.
\end{enumerate*}
\end{assumFLS}
Assumption \ref{assum4}\ref{assum4a} is adopted from \cite{horvath2013estimation} and imposed for convenience in our proof of consistency of $\widehat{\Lambda}_{\UU}$. Assumption \ref{assum4}\ref{assum4aa} is imposed to ensure nondegenerate convergence rate in our asymptotic normality result; in the special case where $\Lambda_{\UU} = \sigma_u^2 C_{\Upsilon\Upsilon}$ for some $\sigma_u^2>0$,\footnote{This can happen when $\{u_t\}$ is a homogeneous martingale difference with respect to the  filtration $\mathfrak F_t=\sigma(\{u_{s}\}_{s\leq t-1},\{\Upsilon_{s}\}_{s\leq t})$ as in the case considered by \cite{seong2021functional}}  the former condition $\zeta \notin \ker C_{\Upsilon\Upsilon}$ implies the latter, and hence Assumption \ref{assum4}\ref{assum4aa} can be simplified. Assumptions~\ref{assum4}\ref{assum4bb} and \ref{assum4}\ref{assum4cc} are technical requirements facilitating our asymptotic analysis, and are not quite restrictive. In particular, as $\Lambda_{\UU}$ is a covariance operator (\citealt[Theorem 1.7]{Bosq2000}),  the quantity in Assumption~\ref{assum4}\ref{assum4bb} converges to a nonnegative constant for every $\zeta \in \elltwo$. Thus, what we require of $\zeta$ is solely to ensure the positivity of the limit. Assumption \ref{assum4}\ref{assum4cc} is similar to standard assumptions found in the literature (see  \citealt[Theorems 2.12-2.14]{Bosq2000}).
%\footnote{We expect that $\zeta$, satisfying Assumption \ref{assum4}\ref{assum4aa}, tends to easily satisfy Assumption \ref{assum4}\ref{assum4bb} as well. However, the relationship between these two conditions are not within the scope and interest of the present paper, and thus it is not further pursued.} 

%Likewise, a similar (but stronger) requirement for Assumption~\ref{assum4}\ref{assum4cc} was earlier employed in some well-known limit theorems for function-valued random sequences .   

We  employ the following assumption: below, we let $X_{w,t} = X_t-\CCC_{21}\CCC_{11}^{-1}\mathbf{w}_t$, $r_t(j,\ell) = \langle X_t,v_j \rangle \langle X_{w,t},v_{\ell} \rangle - \mathbb E[\langle X_t,v_j \rangle \langle X_{w,t},v_{\ell} \rangle]$ for $j,\ell \geq 1$, and  $\CC$ be a generic positive constant.
\begin{assumFLS} \label{assum5}
\begin{enumerate*}[(i)]
\item\label{assum5a} $\mathbb{E}[\langle X_t,v_j \rangle^4] \leq \CC \lambda_j^2$, $\mathbb{E}[\langle X_{w,t},v_j \rangle^4] \leq \CC \lambda_j^2$, and  for some $\tilde{\CC} > 1$ and $s\geq 1$, $\mathbb{E}[r_t(j,\ell)r_{t-s}(j,\ell)]\leq \CC s^{-\tilde{\CC}}\mathbb{E}[r_t^2(j,\ell)]$;
%T\item $\lambda_j^2 \leq C j^{-\rho}$  $\lambda_j^2 - \lambda_j^2 \geq C j^{-\rho-1}$. 
%\item \label{assum5c}	$ \|T^{-1/2}\sum_{t=1}^T (\langle X_t,v_j \rangle \mathbf{w}_t - \mathbb{E}[ \langle X_t,v_j \rangle \mathbf{w}_t])\|^2 = O_p (\lambda_{j})$;
\item\label{assum5b} for some $\delta>1/2$,  $|\langle \zeta_2,v_j \rangle| \leq \CC j^{-\delta}$ and $|\langle \zeta_1, {\CCC}_{11}^{-1}{\CCC}_{12} v_j \rangle| \leq \CC j^{-\delta}$.
%\item\label{assum4d} $h=o(T^{1/2})$.
\end{enumerate*}
\end{assumFLS}
% Assumption \ref{assum4}\ref{assum4a} is a technical requirement for facilitating mathematical proofs of our results to appear. It can be shown from a bit of algebra that the condition is satisfied if there exists a finite integer $m$ such that $\mu_m >0$ and 
% \begin{equation}
%     \frac{\sum_{j=1}^{\KK} \lambda_j^{-2} \langle \zeta,v_j \rangle^2 (1-\langle v_j,\varpi_m \rangle^2)- \sum_{j\neq i, i=1}^{\KK}\lambda_j^{-1}\lambda_i^{-1} \langle \zeta,v_j \rangle\langle v_j,\varpi_m\rangle\langle \zeta,v_i \rangle\langle v_i,\varpi_m \rangle}{\sum_{j=1}^{\KK} \lambda_j^{-2}\langle \zeta,v_j \rangle^2} \not\to 1,
% \end{equation}
% as $\KK \to \infty$.
% Even if there may be some extreme cases violating the above condition, such as that $\varpi_k$ is orthogonal to $v_1,\ldots,v_{\KK}$, for all $k=1,2,\ldots,m$,  we can expect this assumption to be satisfied. 

A similar but slightly different assumption can be found in  \cite{seong2021functional} and references therein.  In particular, Assumption~\ref{assum5}\ref{assum5b} pertains to the smoothness of $\zeta$; similar to Assumption~\ref{assum3}\ref{assum3b},  it is natural to consider $\delta >1/2$ as $\zeta_2\in\mathcal H$ and $ {\CCC}_{11}^{-1}{\CCC}_{12}$ is a bounded linear functional.\footnote{There exists unique element $h \in \mathcal H$ such that ${\CCC}_{11}^{-1}{\CCC}_{12}v_j = \langle h, v_j \rangle$ by the Riesz representation theorem (see  \citealp{Conway1994}, p.\ 13), and $\sum_{j=1}^\infty\|{\CCC}_{11}^{-1}{\CCC}_{12}v_j\|^2 < \infty$ as $\{v_j\}$ is an orthonormal basis in $\mathcal H$.} %Assumption \ref{assum5}\ref{assum5c} is a high-level condition that facilitate our mathematical proof of the local asymptotic normality result, and it does not seem restrictive.\footnote{For example, if the vector valued time series of $\langle X_t,v_j \rangle \mathbf{w}_t$ satisfies the CLT and its long-run covariance is $O(\lambda_j)$ (this is reasonable given that $\mathbb{E}[\langle X_t,v_j \rangle^4] \leq \lambda_j^2$), then this condition is satisfied.} 
% $\mathbb{E}[\langle X_t,v_j \rangle] \leq $  will converge to multivariate normal random vector, l, whose covariance  $T^{-1/2}\sum_{t=1}^T (\langle X_t,v_j \rangle w_t - \mathbb{E}[ \langle X_t,v_j \rangle w_t])$ converges in }
%Assumption \ref{assum4}\ref{assumgaussian} is not necessary for the results to be developed. This is for a more detailed illustration of our results on the limiting behavior of the estimator when $X_t$ is of finite-type. 

Lastly, we let $\widehat{P}_{\KK}=\widehat{C}_{\Upsilon\Upsilon,\KK}^{-1} \widehat{C}_{\Upsilon\Upsilon} $ and $P_{\KK} =  {C}_{\Upsilon\Upsilon,\KK}^{-1}  {C}_{\Upsilon\Upsilon}$, where  ${C}_{\Upsilon\Upsilon,\KK}^{-1}$ is defined by replacing $\hat\CCC_{ij}$ and  $\hat\SA_{\KK}^{-1}$ in \eqref{eqreginv} with $\CCC_{ij}$ and $\SA_{\KK}^{-1}$ respectively. Then, $\widehat{P}_{\KK}  $ and $P_{\KK}  $ allow  the following representations:
\begin{equation}\label{eqprojections}
\widehat{P}_{\KK} 
= \begin{bmatrix}
I_1 & \widehat{\CCC}_{11}^{-1}\widehat{\CCC}_{12}(I_2-\widehat{\Pi}_{\KK}) \\ 0 & \widehat{\Pi}_{\KK}
\end{bmatrix}\quad\text{ and }\quad   {P}_{\KK}   
= \begin{bmatrix}
I_1 & {\CCC}_{11}^{-1}{\CCC}_{12}(I_2-{\Pi}_{\KK}) \\ 0 & {\Pi}_{\KK}
\end{bmatrix},
\end{equation} 
where $I_1$ (resp.\ $I_2$) is the identity map on $\mathbb{R}^{m}$ (resp.\ $\mathcal H$) and
\begin{equation}
\widehat{\Pi}_{\KK}=\sum_{j=1}^{\KK} \hat{v}_j \otimes \hat{v}_j\quad\text{and}\quad  \Pi_{\KK}=\sum_{j=1}^{\KK} v_j \otimes v_j. \nonumber
\end{equation}
That is, $\Pi_{\KK}$ is the projection onto the span of $\{v_j\}_{j=1}^{\KK}$ and $\widehat{\Pi}_{\KK}$ is its sample counterpart. %Note that $\widehat{P}_{\KK}$ and $P_{\KK}$ are non-orthogonal projections on $\ell^2$. 
Using $\widehat{P}_{\KK}$ and ${P}_{\KK}$, we decompose $\langle \hat{\theta}_h-\theta_h,\zeta\rangle$ into $ \widehat{\Theta}_1+\widehat{\Theta}_{2A}+\widehat{\Theta}_{2B}$ such that
\begin{equation} \label{eqdecomtheta}
\widehat{\Theta}_1 =  \langle \hat{\theta}_h-\widehat{P}_{\KK}\theta_h, \zeta \rangle,\quad \widehat{\Theta}_{2A}=  \langle\widehat{P}_{\KK}\theta_h - P_{\KK}{\theta}_h,\zeta\rangle, \text{ and } \widehat{\Theta}_{2B}=  \langle {P}_{\KK}\theta_h - {\theta}_h,\zeta\rangle. 
\end{equation}
%We then have $\widehat{\Theta} = \widehat{\Theta}_1+\widehat{\Theta}_2+\widehat{\Theta}_3$.  

\begin{theorem} \label{thm2}  Suppose that Assumptions \ref{assum1}-\ref{assum4} hold and  $T\reg^{2+4/\rho} \to \infty$.    Then the following holds:
\begin{enumerate}[(i)]
\item\label{thm2i0} $\sqrt{{T}/{\psi_{\KK}(\zeta)}}\widehat{\Theta}_1 \to_d N(0,1)$ for $	{\psi}_{\KK}(\zeta) = \langle {\Lambda}_{\UU}{C}_{\Upsilon\Upsilon,\KK}^{-1}\zeta,{C}_{\Upsilon\Upsilon,\KK}^{-1}\zeta \rangle$.	 %; if $\langle \theta_{h}, \zeta\rangle = \langle \theta_{1,h}, \zeta\rangle$, then $\widehat{\Theta}_2 = \widehat{\Theta}_3 = 0$ and thus $
%\sqrt{{T}/{\psi_{\KK}(\zeta)}}\langle \hat{\theta}_h-\theta_h,\zeta\rangle \to_d N(0,1).$
%\item\label{thm2i0} \begin{equation} 
%\sqrt{{T}/{\psi_{\KK}(\zeta)}}\widehat{\Theta}_1 \to_d N(0,1). %\label{eqthm1}
%\end{equation}
\end{enumerate} 

%Suppose that $\langle \theta_{h}, \zeta\rangle \neq  \langle \theta_{1,h}, \zeta\rangle$, 
If Assumption \ref{assum5} is additionally satisfied with $\rho/2 + 2 < \varsigma+ \delta$, %, regardless of if $\langle \theta_{h}, \zeta\rangle = \langle \theta_{1,h}, \zeta\rangle$, 
the following hold:
\begin{enumerate}[(i)]\addtocounter{enumi}{1}
\item\label{thm2i1}  If $\psi_{\KK}(\zeta) \to_p \infty$, $
\sqrt{{T}/{\psi_{\KK}(\zeta)}}\widehat{\Theta}_{2A} \to_p 0$.
\item \label{thm2i2} If $\psi_{\KK}(\zeta) \to_p \infty$ and $T^{1/2}\reg^{(\delta+\varsigma-1)/\rho} \to 0$, 
$\sqrt{{T}/{\psi_{\KK}(\zeta)}}\widehat{\Theta}_{2B} \to_p 0.$
%and thus, if $\psi_{\KK}(\zeta) \to_p \infty$, 
%$$\sqrt{{T}/{\psi_{\KK}(\zeta)}}\widehat{\Theta} \to_d N(0,1).$$

%   \item If  HH holds and $\delta + \varsigma > \rho/2 + 2$
%\item\label{thm2i3} 
\end{enumerate}
The results  in \ref{thm2i0}-\ref{thm2i2} hold when $\psi_{\KK}(\zeta)$ is replaced by $\widehat{\psi}_{\KK}(\zeta) = \langle  \widehat{\Lambda}_{\UU} \widehat{C}_{\Upsilon\Upsilon,\KK}^{-1}\zeta, \widehat{C}_{\Upsilon\Upsilon,\KK}^{-1}\zeta \rangle$. %\Vert   \widehat{\Lambda}_{\UU} ^{1/2}\widehat{C}_{\Upsilon\Upsilon,\KK}^{-1}\zeta\Vert ^2.$ % under the assumptions that each of  \ref{thm2i0}-\ref{thm2i2} requires. 
\end{theorem}

The quantity $\psi_{\KK}(\zeta)$ in Theorem~\ref{thm2} plays an important role as a normalizing factor in our asymptotic normality result. 
With an obvious adaptation of the discussion given by \citet[Remarks 4 and 11]{seong2021functional}, we know that it is convergent only on a strict subspace of $\elltwo$, and ${\psi}_{\KK}(\zeta)$ is likely to diverge (at a rate slower than $T$; see Remark \ref{remadd2}) for $\zeta$ arbitrarily chosen by practitioners.

Similar to Theorem \ref{thm1}, the decay rate of the regularization parameter in Theorem~\ref{thm2} depends on $\rho$, and a larger value of $\rho$ leads to a faster convergence. Note that $\reg$ satisfying $T\reg^{4} \to \infty$ (e.g.,  $\reg = T^{-1/4 + \epsilon}$ or $T^{-1/4} \log^{\epsilon} T$ for small $\epsilon>0$) meets the requirement for Theorem~\ref{thm2}\ref{thm2i0}, provided that $\rho >2$ as assumed in Assumption~\ref{assum3}. Such a choice of $\reg$ may be preferred for practitioners. Of course, as in Theorem~\ref{thm1}, we may also consider $\reg$ decaying to zero at a slower rate, but this is not recommended due to the condition $T^{1/2}\reg^{(\delta+\varsigma-1)/\rho}  \to 0$ for the result given in Theorem~\ref{thm2}\ref{thm2i2}. This condition can be understood as requiring sufficiently large $\varsigma$ and $\delta$ for a given $\reg$; moreover, as the decay rate of $\reg$ decreases, stricter requirements are imposed on $\varsigma$ and $\delta$.   %the Obviously, practitioners need to avoid $\alpha$ decaying to zero at a slower rate than what we require for \ref{thm2i2} to avoid strong requirements on $\varsigma$ and $\delta$. 

\begin{remark} \label{remadd2} 
The normalizing factor $\sqrt{{T}/{\psi_{\KK}(\zeta)}}$ can diverge at a different rate depending on the choice of $\zeta$, and the factor is not stochastically bounded for any choice of $\zeta$. From \eqref{addeqrem} and Assumption \ref{assum3}, it is straightforward to see that  ${\psi_{\KK}(\zeta)}$ is bounded above by $O_p(\KK^{\rho})$, which is in turn $O_p(\reg^{-1})$ as shown in our proof of Theorem \ref{thm1}. Since $T\tau \to \infty$, this implies that ${\psi_{\KK}(\zeta)}/T$ must decay to zero (i.e., $\sqrt{{T}/{\psi_{\KK}(\zeta)}}$ is divergent) regardless of the choice of $\zeta$.   
\end{remark}

\begin{remark}\normalfont
If $\reg = T^{-1/4} \log^{\epsilon} T$ for  $\epsilon>0$, then $T^{1/2}\reg^{(\delta+\varsigma-1)/\rho}
=T^{(2\rho+1-\delta-\varsigma)/(4\rho)}\times$ $(\log T)^{\epsilon(\delta+\varsigma-1)/\rho}$. Since $\varsigma>1/2$ and $\delta>1/2$ under the employed assumptions, either (i) $\delta+\varsigma > 2\rho+1$ or (ii) $\varsigma \geq 2\rho + 1/2$ ensures the condition that $T^{1/2}\reg^{(\delta+\varsigma-1)/\rho} =o(1)$. In case (ii), Theorems \ref{thm2}\ref{thm2i1}-\ref{thm2i2} hold for any $\delta > 1/2$.  More generally, if $\reg = T^{-\rho/(2\rho+4)} \log^{\epsilon}T$, it can be shown that the aforementioned condition is satisfied under either of (i)$'$ $\delta+\varsigma>\rho+3$  or (ii)$'$ $\varsigma \geq \rho + 5/2$. 
Given that $\rho>2$, (i)$'$ (resp.\ (ii)$'$) is weaker than (i) (resp.\ (ii)), but they are nearly identical if $\rho$ is close to $2$. 
\end{remark}

%\begin{remark} \label{remothertests}
%	In the functional linear model,  there are available asymptotic normality results for the discrepancy $\langle \hat{\beta}-\beta, \zeta \rangle$, providing statistical inference on $\langle \beta, \zeta \rangle$; see e.g., \cite{cardot2007}, \cite{yeon2023bootstrap}, and \cite{seong2021functional}. The former two require an iid setting and $\zeta$ is restricted to some specific random variable (such as $X_{T+1}$ or another independent of the data).  \cite{seong2021functional} establish a similar result in the endogenous functional linear model, which corresponds to a special case of our Theorem \ref{thm4} to appear as an extension of Theorem \ref{thm2}. Nevertheless, their result not only necessitates a more complicated and subtle choice of $\reg$ but also requires the error term ($u_{h,t}$ in our setup) to be uncorrelated. 
%\end{remark}

If $\zeta$ is restricted to have a nonzero element only in $\mathbb{R}^m$ (i.e., $\zeta_2 = 0 \in \mathcal{H}$) and $\psi_{\KK}(\zeta) \to_p  C_{\zeta} <\infty$, it can be shown that $\langle \hat{\theta}_h - \theta_h, \zeta \rangle$ converges in distribution to a normal random variable at the rate of $\sqrt{T}$. However, we here emphasize that the former condition, $\zeta_2 = 0$, does not guarantee the latter condition, $\psi_{\KK}(\zeta) \to_p  C_{\zeta}$. Thus, it is common to observe a nonparametric rate of convergence even for the parameters associated with the finite dimensional random variables in this setup (see Remark~\ref{remmas} for more details).

\begin{corollary} \label{coradd}
Suppose that Assumptions \ref{assum1}-\ref{assum5} hold, $\psi_{\KK}(\zeta) \to_p  C_{\zeta} < \infty$, and $\zeta = \left[\begin{smallmatrix}\zeta_1\\0\end{smallmatrix}\right]$ where $\zeta_1 \in \mathbb{R}^m$ and $0\in \mathcal H$. Then,
\begin{equation}
\sqrt{T}\langle \hat{\theta}_h-\theta_h,\zeta\rangle =  \sqrt{T} \langle \hat \alpha_h - \alpha_h, \zeta_1 \rangle  \to_d N(0, C_{\zeta}).      
\end{equation}
\end{corollary}

\begin{remark} \label{remmas}
As deduced from our results in Theorem \ref{thm2} and the discussions provided by \cite{Mas2007}, the convergence rate of the proposed estimator depends on the choice of $\zeta$ and is generally slower than $\sqrt{T}$. Moreover,  even when we concern the coefficients associated with finite dimensional elements, the parametric $\sqrt{T}$-convergence rate of $\widehat\alpha_h$ requires $\lim_{\KK \to \infty} \psi_{\KK}(\zeta) < \infty$, which, in turn, necessitates  sufficient smoothness of  $\zeta$   with respect to the eigenvectors ${\vtw_j}$ of ${C}_{\Upsilon\Upsilon}$.
%Observe that, for any $\theta_{j,h}$ with $j=1, \ldots, m$, there exists a choice of vector $\zeta$ satisfying $\langle \theta_h, \zeta \rangle = \theta_{j,h}$. From this, it follows that the estimator of $\theta_{j,h}$ associated with a scalar-valued control variable $w_{j,t}$ does not generally achieve $\sqrt{T}$-convergence, as there is no guarantee that the choice of vector $\zeta$ is smooth enough with respect to ${\vtw_j}$.
%As shown by \cite{SHIN2009} for the partially linear functional regression model, achieving the $\sqrt{T}$-convergence rate for parameters associated with scalar-valued covariates requires nontrivial conditions (see Theorem 3.1 and Eq.\ (20) therein), which are not required in a finite dimensional setting.
A similar condition can be found in the literature  on partially linear functional regression models, including   \cite{AP2006} and \citet[Theorem 3.1 and eqn.\ (20)]{SHIN2009}. Given its nontrivial nature, this condition reflects the cost of implementing inference in a model involving functional variables.
%It is not difficult to see that if $\mathbf{w}_t$ and $X_t$ are contemporaneously orthogonal (and hence $\Gamma_{12} = 0$ and $\Gamma_{21} = 0$ in \eqref{covblock}), then all such choice vectors for $\theta_{h,j}$ can be expressed as $\zeta = \sum_{j=1}^m a_j \vtw_j$ for some $a_1, \ldots, a_m$, and the $\sqrt{T}$-convergence rate is achieved. However, this is an overly restrictive scenario in the considered setup.
\end{remark}

%\begin{remark}
%\commWK{(To be added Discussion when $\theta_{2,h}=0$.)}
%\end{remark}

\section{Empirical Application \label{sec:emp}}

\subsection{US Economic Sentiment Quantile Curves\label{sec:emp1}}
In this section, we study how economic variables respond to random perturbations introduced to economic sentiment distribution in the US. To measure sentiment, we adopt \citepos{Barbaglia2023} data, which quantify sentiment on specific economic subjects from sentences published in major US news articles. The sentiment measure is defined by each token (i.e., word) and its neighboring sentences, and includes both intensity and tone for each subject. We refer interested readers to \cite{Barbaglia2023} for a detailed definition of the sentiment measure.

Among others, we particularly focus on the daily sentiment measure on the \textsl{economy}, and use it to construct monthly quantile sentiment curves.  The quantile curves are represented by 31 Fourier basis functions, which become our functional predictor, and are reported in Figure~\ref{fig: 1}. As mentioned earlier in Section~\ref{sec: model} and also by \cite{Barbaglia2023}, the sentiment  quantile and its associated distribution seem to be closely related to business cycle fluctuations: during recessive periods, the overall quantiles tend to shift downward and exhibit a steeper slope coefficient, indicating  a larger dispersion in the sentiment distribution. This could be interpreted as a higher level of disagreement about economic sentiment during recessive periods.

The dependent variable $y_t$ is specified to monthly Total Nonfarm Payroll (PAYEMS) in the US. The data are provided by the Federal Reserve Bank.   We follow \cite{Barbaglia2023} and let    $\mathbf{w}_t$ in \eqref{eq: model: benchmark} be the vector consisting of the Chicago Fed National Activity Index (CFNAI), the National Financial Conditions Index (NFCI), and \citepos*{ADS2009} measure of economic activity (ADS). Because NFCI and ADS are observed at a higher frequency, we average them over months to match the frequency of the target variable.  The sample runs from January 1984 to December 2021, with a total size of 456 observations.

The regularization parameter of our estimator is chosen by roughly considering $\rho$ satisfying Assumption~\ref{assum3}.   Specifically, if $\mathtt{C}=1$, Assumption~\ref{assum3}\ref{assum3b} tells us that  $\rho \geq \rho^\ast  \equiv -(\log (\lambda_j ^2  - \lambda_{j+1} ^2) /\log j)-1$. Thus, we set $\widetilde \rho$ to $\lceil 100 \rho ^\ast \rceil/100$, where $\lceil \cdot \rceil$ is the ceiling function. Then the regularization parameter is set to $0.01\Vert  \widehat C_{\Upsilon\Upsilon}\Vert_{\HS}  T^{-\widetilde \rho/(\widetilde \rho+2)}$, where $\Vert \cdot \Vert_{\HS}$ denotes the Hilbert-Schmidt norm. This approach may provide  practical insights into the magnitude of $\rho$, thereby allowing us to choose $\reg$ without violating the assumption. We investigate the finite-sample performance of our estimator computed with this regularization parameter in Section~\ref{sec:sim}.

\begin{figure}[h!]
\caption{ Responses of Growth in Total Nonfarm Payroll to Sentiment Shocks} \label{fig: emp-1} 
\includegraphics[width = \textwidth, height= 0.2\textheight]{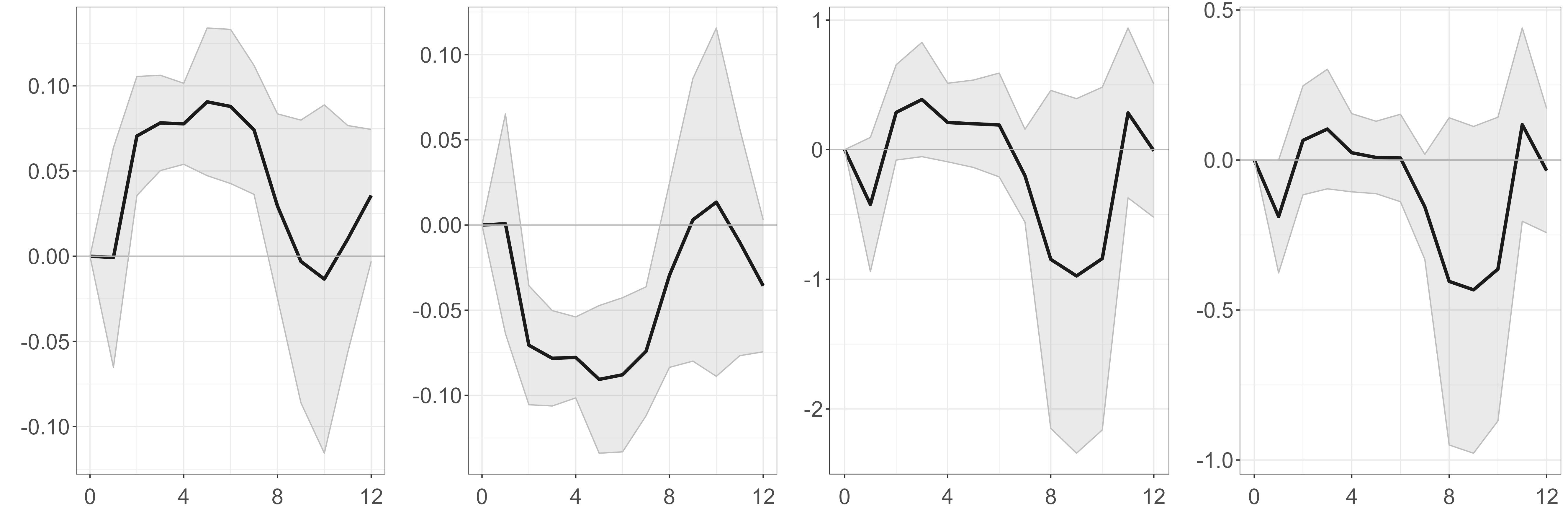}  
\flushleft{\scriptsize{Notes: The figure reports the impulse response estimates when the shocks are specified to $\zeta_1, \zeta_2$, $\zeta_3$, and $\zeta_4$ in Figure~\ref{fig: 2} respectively. The shocks are normalized to a unit norm before analysis. The solid lines are reporting the pointwise estimates and grey areas report their 90\% confidence interval.}}
\end{figure}

We first estimate the impulse responses of PAYEMS to shocks given to the sentiment distribution using our benchmark model in \eqref{eq: model: benchmark} for $ h \in\{1,2,\ldots,12\}$, without a structural interpretation.  We let $\zeta$ be the shocks considered in Figure \ref{fig: 2}, with their norms normalized to one. The normalization is taken to ensure reasonable comparison between their effects. The estimation results are reported in Figure~\ref{fig: emp-1}, where solid lines indicate $\{ \langle \beta_h, \zeta \rangle \}_{h=1}^{12}$ in \eqref{eq: sirf} and the shaded area represents their 90\% confidence intervals computed using the local asymptotic normality result in Theorem~\ref{thm2}. % The first two columns report estimated effects when $\zeta = 2$ and $\zeta = -2$, which are associated with positive and negative shifts in average sentiments on the economy. The third and fourth columns in Figure~\ref{fig: emp-1} relate to cases when $\zeta = 2\xi_1$ and $\zeta = -2+2\xi_1$, where $\xi_1$ is the linear Legendre polynomial with unit norm. Considering the shape of quantile functions (e.g., Figure~\ref{fig: 1}), the coefficient of $\xi_1$ can be understood as the  disagreement measure on sentiment about the economy. 
As  mentioned earlier, the first two shocks are associated with positive and negative shifts in the sentiment distribution. Meanwhile, in the last two columns, we consider cases where sentiment about the economy exhibits greater dispersion, with and without a negative shift, respectively. The responses in Figure~\ref{fig: emp-1} have different vertical scales, which may require caution in their interpretation.

Figure~\ref{fig: emp-1} suggests that the effects of the distributional shocks tend to reach their peak after 6 to 8 months, and their magnitude tends to decrease as the forecasting horizon increases. % However, these effects are not statistically significant, which may be attributed to the small sample size for the quarterly data. On the other hand, when the monthly data of PAYEMS are used, 
In particular, when the average  sentiment about the economy becomes more optimistic (resp.\ pessimistic), payroll growth appears to significantly increase (resp.\ decrease), especially during the first six months. The change in disagreement levels in sentiment appears to have a smaller effect on PAYEMS growth compared to locational shifts in the distribution.

Lastly, we study if our estimation approach produces any improvement compared to the existing estimators, by using the empirical median absolute prediction error (MAPE). The MAPEs are computed based on a rolling window with two different test sets, selected to include approximately 60\% and 70\% of the observations from the total sample. In addition to ours, denoted by SCInv, we consider two alternative estimators for comparison: PCA-FR and PCA-SVAR. The PCA-FR is obtained by applying \citepos{SHIN2009} estimation approach to our benchmark model. The PCA-SVAR is computed using the standard estimation methods developed for the recursive SVAR model in a finite dimensional setting, after pre-applying the FPCA-based dimension reduction to $X_t$.

\begin{table}[tbp]
\caption{Estimated Median Absolute Prediction Errors \label{tab: emp-1}} 
\vskip -8pt
\small \centering
\begin{tabular*}{\linewidth}{@{\extracolsep{\fill}}lcccccccccc}
\toprule 
& \multicolumn{5}{c}{2006 Q4$\sim$}  & \multicolumn{5}{c}{2010 Q3$\sim$}  \\ \cmidrule{2-6}\cmidrule{7-11}
h & 1 & 2 & 3 & 4 & 5 & 1 & 2 & 3 & 4 & 5 \\ \midrule  
SCInv & 0.055 & 0.060 & 0.067 & 0.067 & 0.063 & 0.054 & 0.058 & 0.054 & 0.059 & 0.053 \\ 
PCA-FR & 0.055 & 0.060 & 0.067 & 0.066 & 0.062 & 0.054 & 0.057 & 0.054 & 0.059 & 0.053 \\ 
PCA-SVAR &  0.060 & 0.065 & 0.061 & 0.069 & 0.072 & 0.056 & 0.064 & 0.059 & 0.058 & 0.056 \\ 
\bottomrule
\end{tabular*}
\flushleft{\scriptsize{Notes: The table reports the empirical median absolute prediction errors computed using our estimator (SCInv) and two FPCA-based estimators: PCA-FR and PCA-SVAR. PCA-FR is obtained by applying \citepos{SHIN2009} estimator to our benchmark model. PCA-SVAR is derived by applying the SVAR model to the vector consisting of $\mathbf w_t$		and $X_t$, with its dimension reduced using the FPCA approach. The regularization parameters are set to choose the first two score functions for all the estimators. }}
\end{table}
Table~\ref{tab: emp-1} summarizes MAPEs. To mitigate a potential effect of different choices of regularization parameters on  estimation results, we fix $\tau$ so that all the estimators select the first two score functions, regardless of the forecasting horizon. This choice is based on empirical data, which suggests that approximately 99\% of the variations are  explained by the first two scores. Overall, it seems that our estimator outperforms the PCA-SVAR in terms of smaller MAPEs, and the superior performance is more significantly observed as the forecasting horizon increases. An interesting observation is that the PCA-FR approach also outperforms the PCA-SVAR, although both estimators are computed with the same empirical eigenelements, and thus the PCA-FR can be understood as a single-equation estimation of SVAR models. This may be attributed to the fact that, in the PCA-SVAR estimation, the sentiment quantile curve serves as both the target (i.e., dependent) and prediction (i.e., explanatory) variables, whereas it is only used as a predictor in the PCA-FR approach. Therefore, the PCA-SVAR estimator is likely to be subject to a larger bias associated with \textsl{double truncation}. Meanwhile, in this particular example, the two functional approaches, the SCInv and the PCA-FR, produce similar MAPEs regardless of forecasting horizons.

\subsection{Impact of Functional Monetary Policy Shocks \label{sec:emp2}}
%	Considering its importance in empirical research, it is our main motivation to study the benchmark model \eqref{eq: model: benchmark} that  $\beta_h$ provides an interesting interpretation as a response of a target variable to certain functional shocks. A similar model has been studied by \cite{IR2021} who identify the monetary policy shock as shifts in yield curves and estimate its impact on output and inflation. 
In this section, we employ \citepos{IR2021} functional monetary policy shocks and study the monetary policy impact on the inflation growth by using the proposed estimator. The data span the period from January 1995 to June 2016 and the functional shock is reported in Figure~\ref{fig: ir: 1}. Under Assumption~I in \cite{IR2021}, the slope coefficient of the shock ($\beta_h$ in our notation) can be  understood as the impulse response to monetary policy shocks, with the effect represented by a function. In this section, we overall follow \cite{IR2021}. The control vector $\mathbf w_t$ is given by the set of the first two lagged dependent variables.  $y_t$ is given by the inflation growth rate. The impulse response $\{ \langle \beta_h,\zeta\rangle \}_{h=1}^{20}$ is estimated with $\zeta$ corresponding to shocks observed on three specific dates of interest identified by the authors:   9/1998,   2/1999, and 1/2007. Note that the shocks and the changes in yield curves induced by them in the current study may differ slightly from those in \cite{IR2021} because  we do not represent them into level and curvature components.%, where the shocks can be found in Figure 4 of \cite{IR2021}. %The estimates are smoothed over the forecasting horizon.

\begin{figure}[h!]
\centering
\caption{Inflation Response to Functional Shocks in \cite{IR2021}}\label{fig: emp: ir}
\begin{subfigure}{.32\textwidth}\caption{Change in Sep 1998} 
\includegraphics[width =  \textwidth ]{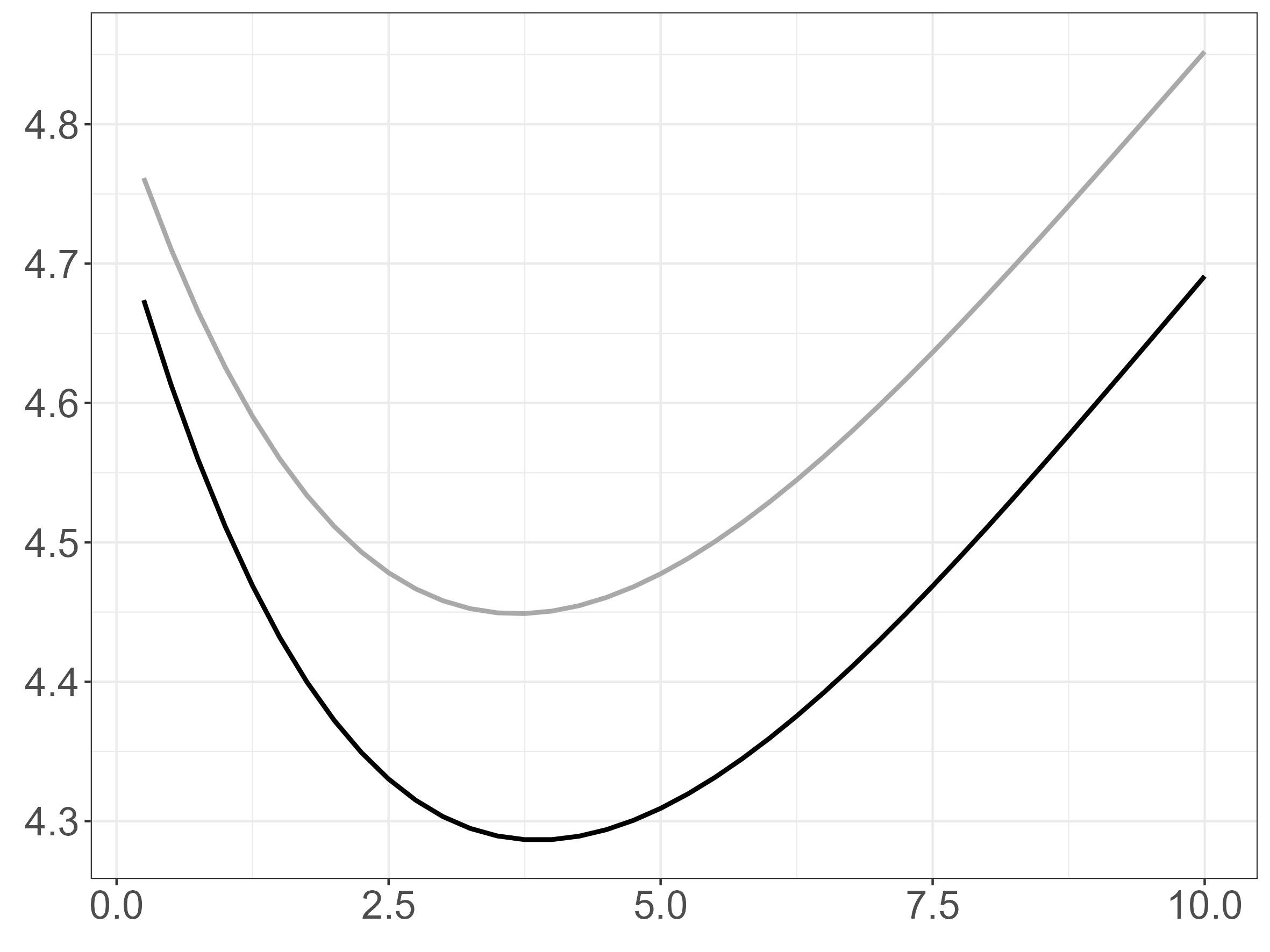} 
\end{subfigure}    
\begin{subfigure}{.32\textwidth}\caption{Change in Feb 1999}
\includegraphics[width = \textwidth ]{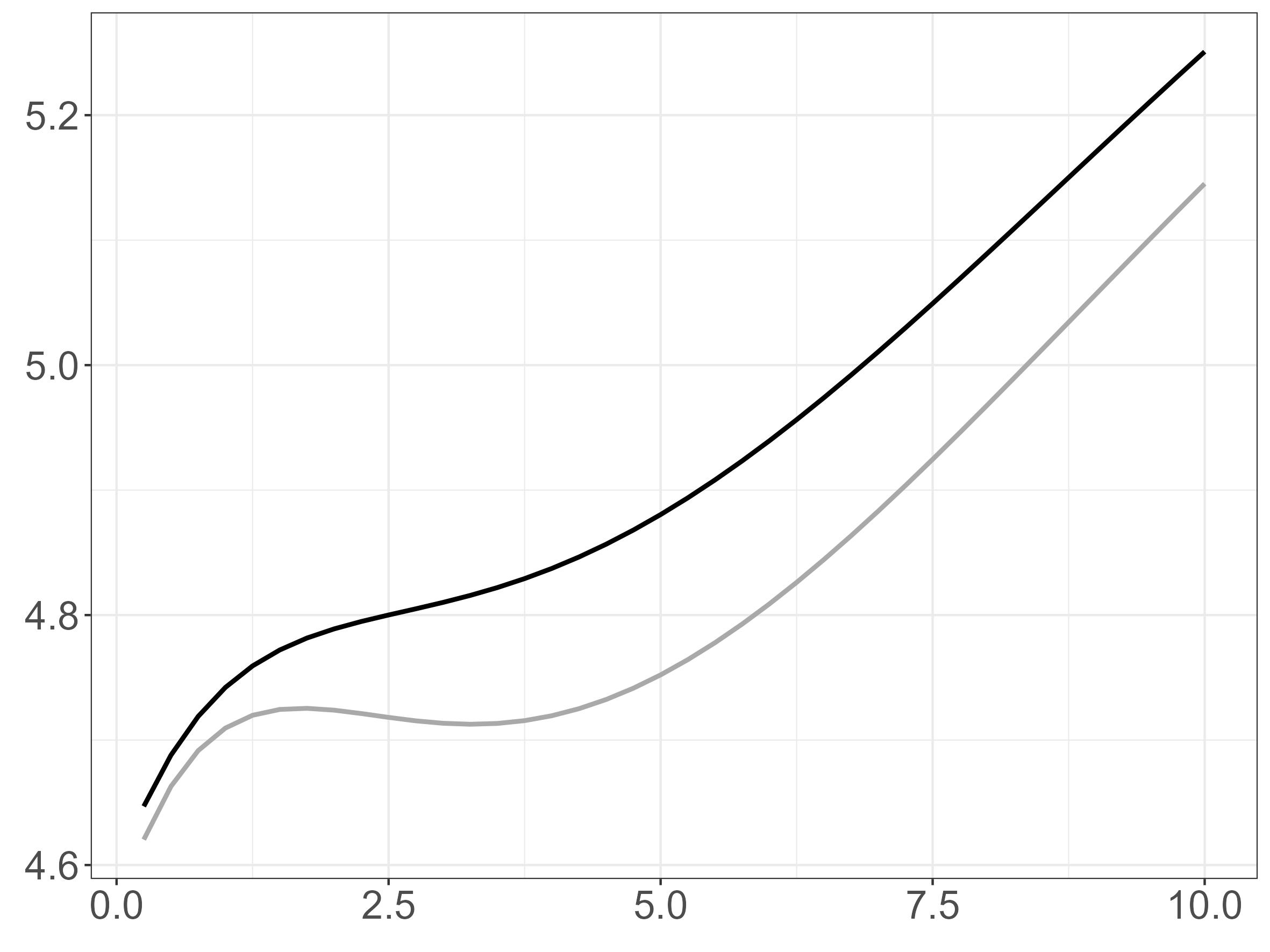} 
\end{subfigure}    
\begin{subfigure}{.32\textwidth}\caption{Change in Jan 2007}
\includegraphics[width = \textwidth ]{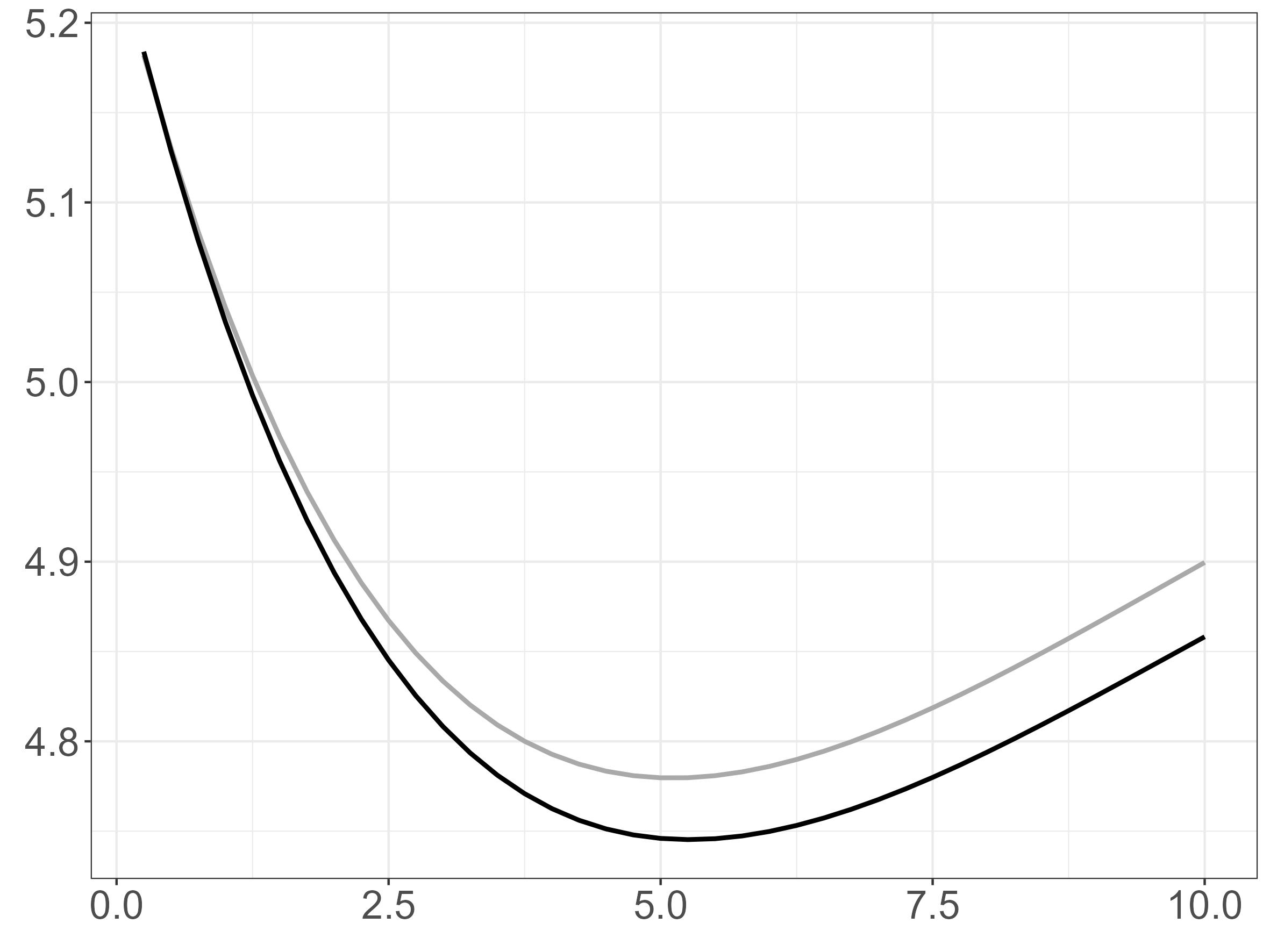} 
\end{subfigure}    
\begin{subfigure}{.32\textwidth}\caption{Inflation IRF (Sep 1998)}\label{figIRd}
\includegraphics[width = \textwidth ]{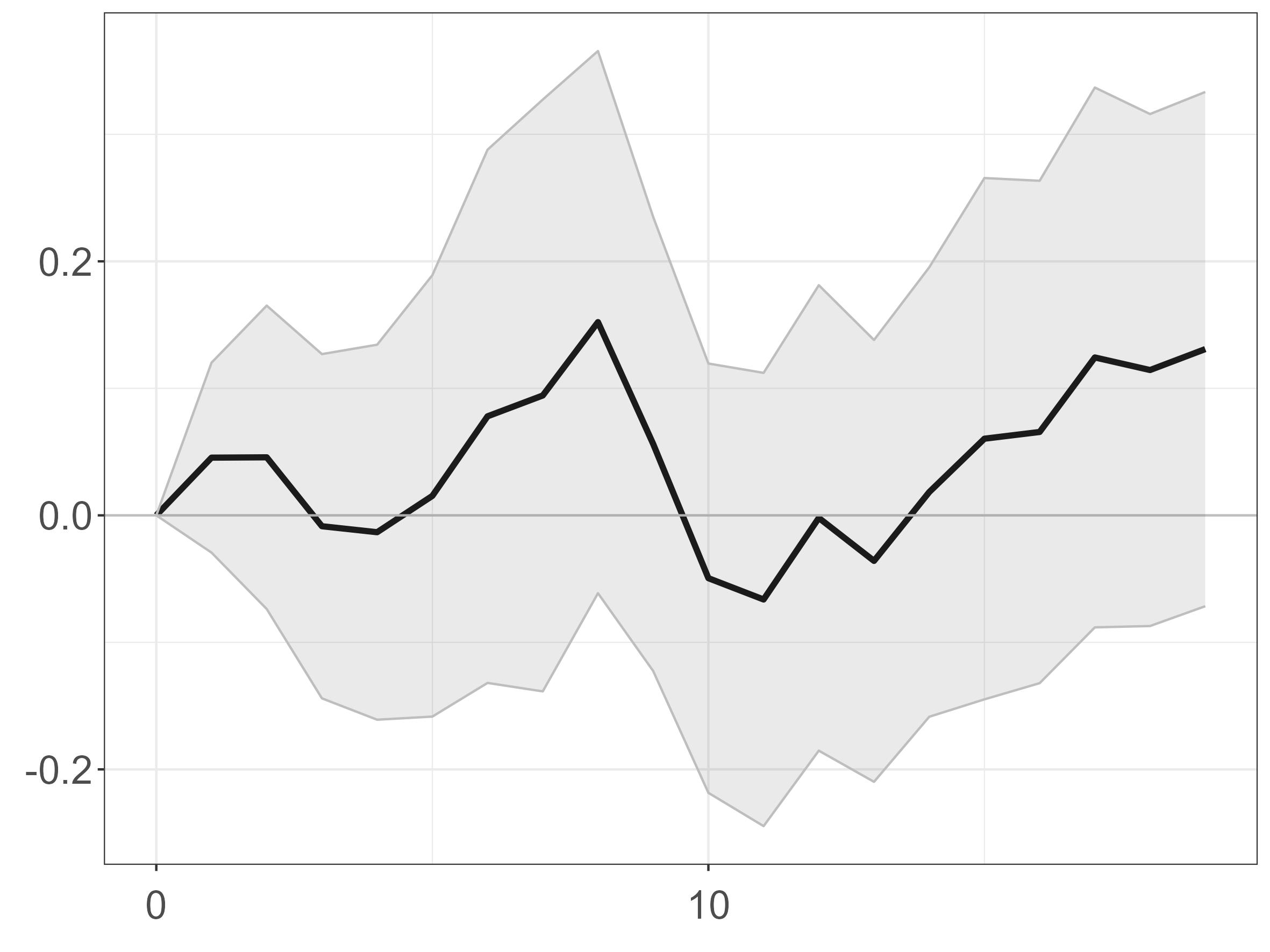} 
\end{subfigure}   
\begin{subfigure}{.32\textwidth}\caption{Inflation IRF (Feb 1999)}\label{figIRe}
\includegraphics[width = \textwidth ]{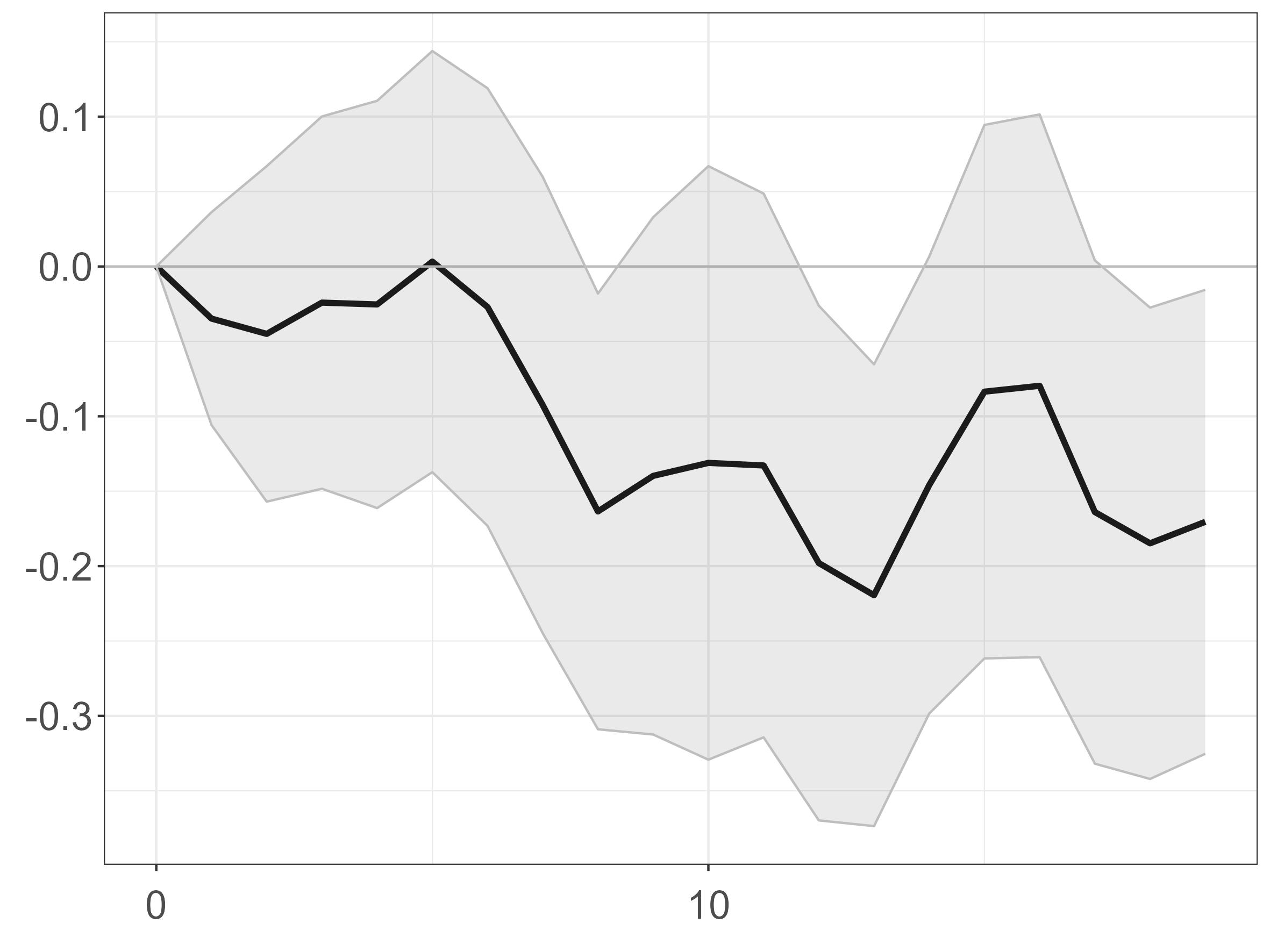} 
\end{subfigure} 
\begin{subfigure}{.32\textwidth}\caption{Inflation IRF (Jan 2007)}\label{figIRf}
\includegraphics[width = \textwidth ]{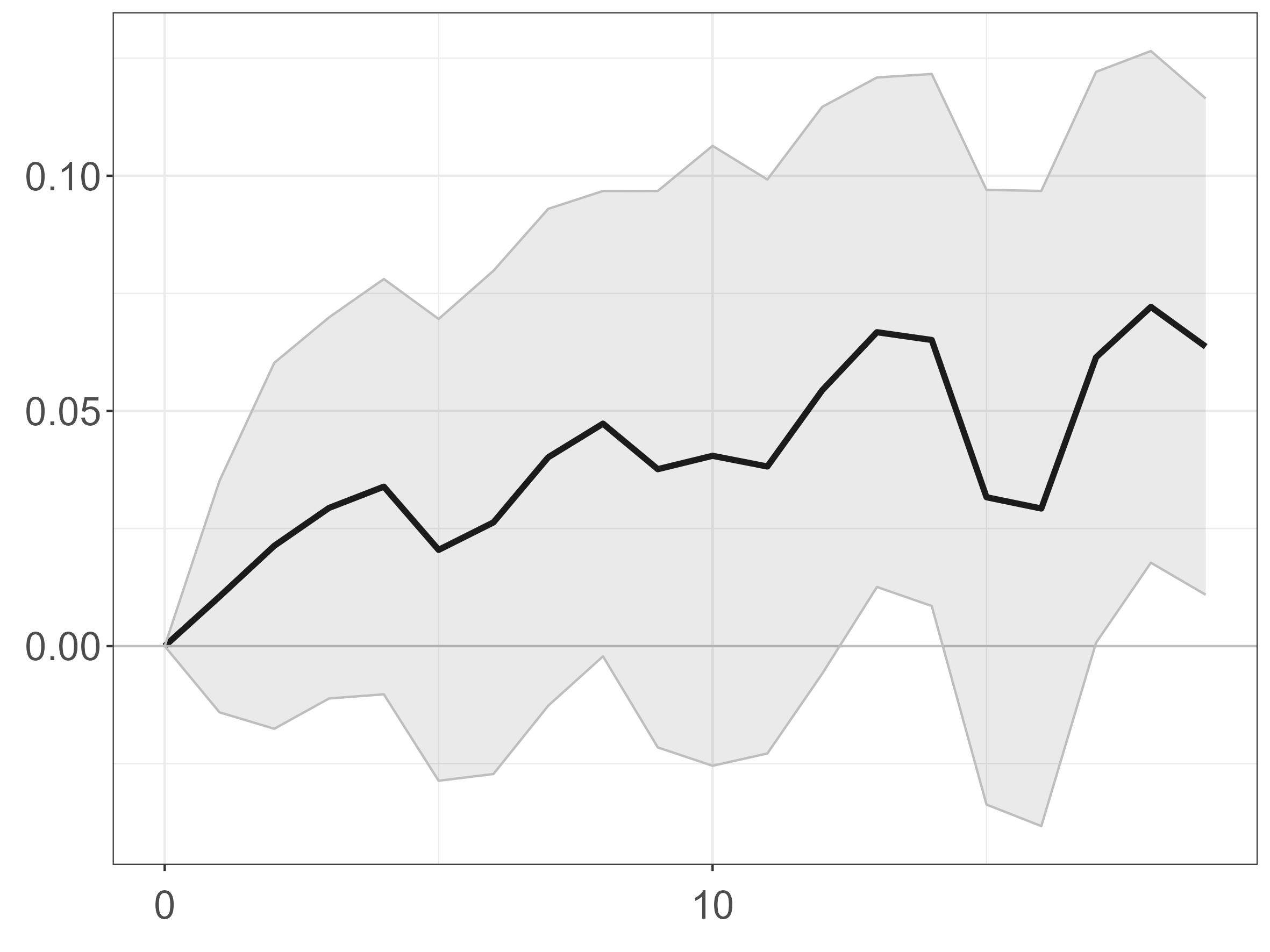} 
\end{subfigure}    
\flushleft{\scriptsize{Notes: The figure reports the change in the yield curves induced by shocks in 9/1998, 2/1999, and 1/2007, specified in \cite{IR2021}, and inflation response estimates. The top panel shows the yield curves before (grey) and after (black) the shock is introduced. In the bottom panel, the solid lines represent the pointwise estimates, while the dotted gray lines indicate their 90\% confidence interval.}}
\label{fig:enter-label}
\end{figure}

The estimated impulse responses and shifts in the yield curve caused by the shocks are reported in Figure~\ref{fig: emp: ir}. %Specifically, in contrast to \citepos{IR2021} impulse response estimator estimated with shocks summarized by changes in a few parametric factors, our monetary policy shocks utilize richer information on changes in the entire functions. Moreover, we impose no restriction on the parameter to be estimated, whereas their estimator requires the parameter of interest to be represented by the same factors in the functional shocks. Therefore, our estimator should be understood as a fully functional complement to the estimator considered by the authors. 
The first observation is that the response seems to depend on the shape of the shock. For example, in September 1998 and January 2007, the yield curves experience decreases in overall maturities. Their effects reported in Figures~\ref{figIRd} and \ref{figIRf} align with economic theory, which suggests that decreases in interest rates for all maturities lead to higher inflation growth. On the other hand, if the shock raises  interest rates in general (e.g., February 1999), inflation tends to respond negatively, with the effect becoming more pronounced over a longer horizon.  Although such a similar observation can be found in \cite{IR2021},  it is worth noting that  the results in Figure~\ref{fig: emp: ir} are not directly comparable with those reported in \cite{IR2021}. This is because of the nonparametric nature of our estimation approach, which does not require a parametric restriction on either the parameter of interest or the functional variables. Therefore, our estimator should be understood as a complement to the estimator considered by the authors.

\section{Simulation\label{sec:sim}}
In this section, we study  finite sample performance of estimators studied in Section~\ref{sec:emp} using Monte Carlo experiments. Throughout the section, we consider the case with $T=250$ and $T=500$, and %we discard the first 300 burn-in observations and consider two sample sizes: 250 and 500. 
the total number of replications is set to 1,000.

\subsection{Experiment A:  Linear projection  DGP\label{sec:simA}}
In our simulation, the variable $X_t $ is designed to mimic the economic sentiment quantile functions  in Section~\ref{sec:emp1}. To this end,  let $X_t =  \sum_{j=1} ^{31} x_{j,t}\xi_j$ where $\xi_j$ denotes the $j$-th eigenvector of $X_t$'s variance and the $j$-th coordinate process $x_{j,t} $ is given by $ \langle X_t, \xi_j\rangle$. Assume that $x_{j,t}$ follows AR(1) such that \begin{equation*}
x_{j,t} = {\alpha}_{j} ^x x_{j,t-1} + c_{e,j} { \sigma}_{j} e_{j,t}, 
\end{equation*} for $j=1, \ldots, 31$, where $e_{j,t}\sim_{iid} N(0,1)$. The parameter values $\alpha_j ^x$ and $\sigma_j$ are set to the estimates obtained from the economic sentiment quantile curves,  with $c_{e,j}$ set to one for all $j$. Later in the section, we keep or exaggerate the variance of each coordinate process by setting $c_{e,j} = c_{e,-1} 1\{ j\neq 1\} + 1\{j=1\}$ and considering two different values of $c_{e,-1}$: (i) $c_{e,-1} = 1$ and (ii) $c_{e,-1} = 4$. In the latter case, the ordered eigenvalues (from largest to smallest) are scaled up, except for the first one, without altering their order. Hence, $\lambda_j$ exhibits a significantly slower diminishing tendency for $j$ that is not large, compared to the former case.

% The larger value of $c_{e,j}$ allows us to have a slower convergence of $\lambda_j$ which is an essential requirement for our asymptotic results. 
We assume $\{ y_t,  X_t \}_{t=1} ^T$ follows the DGP such that  
\begin{equation*}
y_{t+h} =  \alpha_{h} y_{t} +    \int_0 ^1  X_t(s) \beta_h (s) ds  + c_u \sigma_{u,h} u_t, \quad \text{for }h=1,\ldots, H,
\end{equation*}
where $u_{t}\sim_{iid} N(0,1)$ and $\beta_h(\cdot) = \sum_{j=1} ^{J} \beta_{j,h} \xi_j(\cdot)$. In our empirical data, the first two coordinate processes ($x_{1,t}$ and $x_{2,t}$) explain more than 99\% of the variations. Thus, for each $h$, the parameters $\{\alpha_h, \beta_{1,h},  \beta_{2,h}, \sigma_{u,h}\}$ are replaced by the estimates from the empirical data when $c_u=1$ and $J=2$. In the simulation, we increase $J$ to 31 and let the remaining coefficients $\{\beta_{j+2,h}\}_{j=1} ^{29}$ be given by $ \min \{ |\beta_{1,h}|, |\beta_{2,h}| \} \times 0.7^j$. This simulation design allows us to avoid the inverse problem in obtaining the  coefficient estimates while keeping most variations in $X_t$, so that realizations from the DGP can mimic the empirical data. In the simulation, the constant $c_u$ is set to 0.5. We note that the coefficients in this section should not be interpreted as the coefficient estimates in Section~\ref{sec:emp1}.

We consider three estimators: SCInv, SCInv$_1$ and PCA-FR. %The former two are our estimators based on regularization of the operator Schur complement, and the other is the FPCA-based estimator proposed by \cite{SHIN2009}. 
SCInv and SCInv$_1$ are differentiated in the choice of the regularization parameter. SCInv is computed as in Section~\ref{sec:emp}. The second estimator SCInv$_1$ and also \citepos{SHIN2009} PCA-FR are computed by using the AIC criterion detailed in \cite{SHIN2009}. 

\begin{table}[tbp]
\caption{ Relative Bias and Variance Estimates (Experiment 1) } \label{tab: simA-1}
\vskip -8pt
\small
\begin{tabular*}{\linewidth}{@{\extracolsep{\fill}}lllccccccccc}
\toprule       &  &  & \multicolumn{3}{c}{$h= 1$} &     \multicolumn{3}{c}{ $h=3$} &    \multicolumn{3}{c}{$h= 5$}      \\\cmidrule{4-6}\cmidrule{7-9}\cmidrule{10-12}  
T & $c_{e,-1}$  &  & SCInv &SCInv$_1$ & PCA-FR & SCInv & SCInv$_1$ & PCA-FR & SCInv & SCInv$_1$ & PCA-FR \\ \midrule
\multirow{4}{*}{250} & \multirow{2}{*}{4} & Bias & 1.00 & 1.07 & 1.07 & 1.00 & 2.03 & 2.03 & 1.00 & 1.53 & 1.53 \\ 
&  & Var & 1.00 & 1.80 & 1.79 & 1.00 & 2.15 & 2.15 & 1.00 & 0.30 & 0.30 \\ \cmidrule{2-12}
& \multirow{2}{*}{1} & Bias & 1.00 & 3.13 & 3.13 & 1.00 & 4.44 & 4.44 & 1.00 & 1.48 & 1.48 \\ 
&  & Var & 1.00 & 1.83 & 1.83 & 1.00 & 1.24 & 1.24 & 1.00 & 0.23 & 0.23 \\ \midrule
\multirow{4}{*}{500} & \multirow{2}{*}{4} & Bias & 1.00 & 1.11 & 1.11 & 1.00 & 1.20 & 1.20 & 1.00 & 1.64 & 1.64 \\ 
&  & Var & 1.00 & 0.90 & 0.90 & 1.00 & 1.68 & 1.68 & 1.00 & 0.27 & 0.27 \\ \cmidrule{2-12}
& \multirow{2}{*}{1} & Bias & 1.00 & 1.82 & 1.82 & 1.00 & 4.27 & 4.27 & 1.00 & 1.49 & 1.49 \\ 
&  & Var & 1.00 & 2.57 & 2.57 & 1.00 & 2.11 & 2.11 & 1.00 & 0.15 & 0.15 \\ 
\bottomrule
\end{tabular*}
\flushleft{\scriptsize{Notes: Based on 1,000 replications. The table reports the $L_2$ norm of the bias (Bias) and the variance (Var) relative to those of SCInv. As $h$ increases, the estimators exhibit smaller bias and variance, which is expected, as the function values of $\beta_h$, computed from empirical data, tend to diminish in scale as $t$ increases.}}
\end{table}

\begin{table}[h!]
\caption{Coverage Probability (Experiment 1) } \label{tab: simA-2}
\vskip -8pt
\small
\begin{tabular*}{\linewidth}{@{\extracolsep{\fill}}lcccccc}
\toprule      &      \multicolumn{3}{c}{$T= 250$} &     \multicolumn{3}{c}{ $T=500$}        \\\cmidrule{2-4}\cmidrule{5-7}  
$c_{e,-1} \backslash h $  & 1&3&5&1&3&5   \\ \midrule  
4 & 0.95 & 0.93 & 0.93 & 0.95 & 0.94 & 0.94 \\  
1 & 0.94 & 0.93 & 0.93 & 0.95 & 0.94 & 0.94 \\  
\bottomrule
\end{tabular*}
\flushleft{\scriptsize{Notes: Based on 1,000 replications. The 95\% coverage probabilities are computed using SCInv and Theorem~\ref{thm2}. }}
\end{table}

The estimation results are summarized in Table~\ref{tab: simA-1} for three forecasting horizons: $h=1,3,5$. We report the bias (Bias) and variance (Var) of each estimator relative to those computed using SCInv. Overall, the proposed estimator based on the naive choice of $\rho$ (SCInv) tends to produce the smallest bias and variance, particularly when $c_{e,-1}$ is large, which is related to the relative decay rate of $\lambda_j$. %\commWK{As discussed in Section~\ref{sec:est}, our estimator is expected to perform well especially when $\lambda_j$ decays slowly, as in the case of $c_{e,-1}$ is large.} 
Meanwhile, the two estimators based on the same regularization parameter, SCInv$_{1}$ and PCA-FR, perform similarly. %\commWK{This may not be surprising as in this setup where $\mathbf{w}_t$ is simply given by the lagged dependent variable, the eigenfunctions of $\SS$ agree with those of $\CCC_{22}$.} Therefore, the two estimators produce similar performance.

% We also note that SCInv$_{1}$ and FPCA tend to produce smaller bias and variance as the forecasting horizon increases. This is related to the simulation design. As $h$ increases, the signal from other coordinate processes decreases significantly, which allows us to focus only the first one or two coordinate processes, which is the main determinants of the parameter estimator. 

We then study Theorem~\ref{thm2} by the means of local asymptotic normality to study the impulse responses $\{\langle \beta_h, \zeta \rangle; h=1,3,5 \}$, where $\zeta$ is set to a constant function for simplicity. Table~\ref{tab: simA-2} reports 95\% coverage probabilities computed with SCInv for each forecasting horizon. Overall, the coverage probability is very close to the nominal level in both sample sizes. As discussed in Section~\ref{sec:est}, to the best of the authors' knowledge, the assumptions on $\rho$, $\varsigma$ and $\delta$ in Theorem~\ref{thm2} are not directly testable in practice, and thus the asymptotic bias terms $\widehat \Theta_{2A}$ and $\widehat \Theta_{2B}$ may not be asymptotically negligible. Nevertheless, the simulation results reported in Table~\ref{tab: simA-2} suggest that those asymptotic bias would be small and thus do not distort testing results based on the asymptotic approximation in Theorem~\ref{thm2}. 

Figure~\ref{fig: simA-1} reports the pointwise estimates of the functional coefficient at $h=1$ and the impulse response estimates for $h=1,\ldots, 6$ when $c_{e,-1}=1$. The simulation results for $c_{e,-1}=4$ are similar; thus, we omit the figures to save space. The reported coefficient estimates are obtained by averaging the estimates across 1,000 simulations.  Consistent with the observations in the tables, our estimator produces both functional coefficient estimates and the pointwise impulse response estimates close to the true values, demonstrating the practical applicability of our approach.

\begin{figure}
\centering\caption{   $\beta_1$ and $\{\langle \zeta, \beta_h\rangle \}_{h=1} ^{6}$ Estimates (Experiment 1; $T=250$) \label{fig: simA-1}}
\begin{subfigure}{.45\linewidth}\subcaption{Functional coefficient estimates at $h=1$\label{fig: simA-1a}}
\includegraphics[width = \textwidth]{figure/Estcurve\_Lsigfac100lnfac2Ctmp60Cexp70ModelEx.jpeg}
\end{subfigure}
\begin{subfigure}{.45\linewidth} \subcaption{Impulse response estimates for $h=1,\ldots , 6$.\label{fig: simA-1b}}
\includegraphics[width = \textwidth]{figure/IRF\_Lsigfac100lnfac2Ctmp60Cexp70ModelSC.jpeg}
\end{subfigure}
\flushleft{\scriptsize{Notes: Averaged across 1,000 replications. Figure~\ref{fig: simA-1a} reports the functional coefficient estimate for $h=1$ computed with SCInv, SCInv1 and PCA-FR. Figure~\ref{fig: simA-1b} reports $\langle \zeta, \beta_h \rangle $ (blue), its pointwise estimates based on the SCInv (dashed) and its 90\% confidence interval (shaded area) for $h=1,\ldots, 6$, when $\zeta$ is specified to the constant function. $c_{e,-1} = 1$.  }}
\end{figure}

\subsection{Experiment B: SVAR-based DGP\label{sec:simB}}
Lastly, we examine the performance of our estimator as an estimator of the SIRF suggested in Propositions~\ref{prop: svar: identification} and \ref{prop: svar: identification: a}, using a small-scale Monte Carlo simulation based on the following SVAR(1) process:\begin{align*}
y_t &= \alpha_{11} y_{t-1} + \alpha_{12} x_{1,t-1} +  \alpha_{13} x_{2,t-1} +    u_{1,t},\\
x_{1,t} & = \beta_{1} y_{t} + \alpha_{21} y_{t-1} + \alpha_{22} x_{1,t-1} +  \alpha_{23} x_{2,t-1} +  u_{2,t},\\
x_{2,t} & = \beta_{2} y_{t} + \alpha_{31} y_{t-1} + \alpha_{32} x_{1,t-1} +  \alpha_{33} x_{2,t-1} + u_{3,t}, 
\end{align*} where $(u_{1,t}, u_{2,t}, u_{3,t} )' \sim_{iid} \mathcal N(0,   \text{diag}(\sigma_1^2, \sigma_2^2, \sigma_3^2))$ and $x_{j,t}$ denotes the $j$-th coordinate process of economic sentiment quantiles. The functional predictor $X_t$ is assumed to be generated as follows. \begin{equation}
X_t =  \sum_{j=1} ^{2} \mathbf{x}_{t, 2 }\xi_j  + \sum_{j= 3} ^{31}  u_{j+1, t} \xi_j.
\end{equation} 
where $  u_{j+1, t}  \sim_{iid}  N( {0},   \sigma_j ^2  ) $ across $j$ and $t$, and $\sigma_j ^2 = \sigma_3 0.8^j$ for $j = 3, \ldots, 31$, and $\xi_j$ are the eigenvectors of $X_t$'s covariance. The parameters and  $\{\xi_j\}_{j=1} ^{31}$  are estimated from empirical data as in Section \ref{sec:simA}. As before, this setup is designed to keep the structure of $X_t$ whose largest variations are determined by the first two coordinate processes. Then, in the simulation, we replace the variance structure of $\mathbf {u}_t=  (u_{1,t},u_{2,t}, \ldots , u_{32,t})'$ with $c_{\ast}\text{diag}(c_1 \sigma_1^2,   \sigma_2^2,\sigma_3^2, \ldots, \sigma_{32}^2 )$, where $c_{\ast}$ is the normalizing constant that allows us to keep the norm of the variance of $\mathbf {u}_t$ to be equal to one. The other constant $c_1$ determines the relative magnitude of the structural error associated with $y_t$, and thus it is inversely related to the signal-to-noise ratio. We consider three values of $c_1$: 1, 0.5, and 0.2. Along with our  estimators, the SCInv and the SCInv$_{1}$,  we consider  the PCA-SVAR for comparison. As the PCA-SVAR produces the same estimation results with the PCA-FR in this setup, we omit the latter.
%the SVAR estimator that is obtained by applying the SVAR model after reducing the dimension of $X_t$ via the FPCA. We focus on one time ahead shock; that is $h=1$.

\begin{table}[tbp]
\caption{Relative Bias and Variance Estimates (Experiment 2) \label{tab: simB-1}} 
\vskip -8pt
\small
\begin{tabular*}{\linewidth}{@{\extracolsep{\fill}}llccccccccc}
\toprule 
& $c_{1}$  & \multicolumn{3}{c}{ 1}   & \multicolumn{3}{c}{ 0.5} &   \multicolumn{3}{c}{ 0.2}   \\\cmidrule{3-5}\cmidrule{6-8}\cmidrule{9-11}
T &  &   SCInv & SCInv$_{1}$ & SVAR & SCInv & SCInv$_{1}$ & SVAR & SCInv & SCInv$_{1}$ & SVAR \\\midrule  
\multirow{2}{*}{250} & Bias &  1.00 &  8.98 &  8.98 &  1.00 & 10.27 & 10.25 &  1.00 &  9.35 &  9.32 \\  
& Var & 1.00 &  0.46 &  0.46 &  1.00 &  0.74 &  0.74 &  1.00 &  1.24 &  1.24 \\  \midrule
\multirow{2}{*}{500}& Bias & 1.00 & 29.57 & 29.53 &  1.00 & 25.11 & 25.06 &  1.00 & 13.61 & 13.56 \\  
& Var & 1.00 &  0.65 &  0.65 &  1.00 &  1.00 &  1.00 &  1.00 &  1.17 &  1.17 \\  \bottomrule
\end{tabular*}
\flushleft{\scriptsize{Notes: Based on 1,000 replications. The table reports the $L_2$ norm of the bias (Bias) and the variance (Var) of each estimator relative to those of SCInv. }}
\end{table}

\begin{figure}[h!]
\caption{Estimated Functional Coefficient (Experiment 2; $T=250$) \label{fig: simB-1}}
\includegraphics[width = \textwidth, height= .2\textheight]{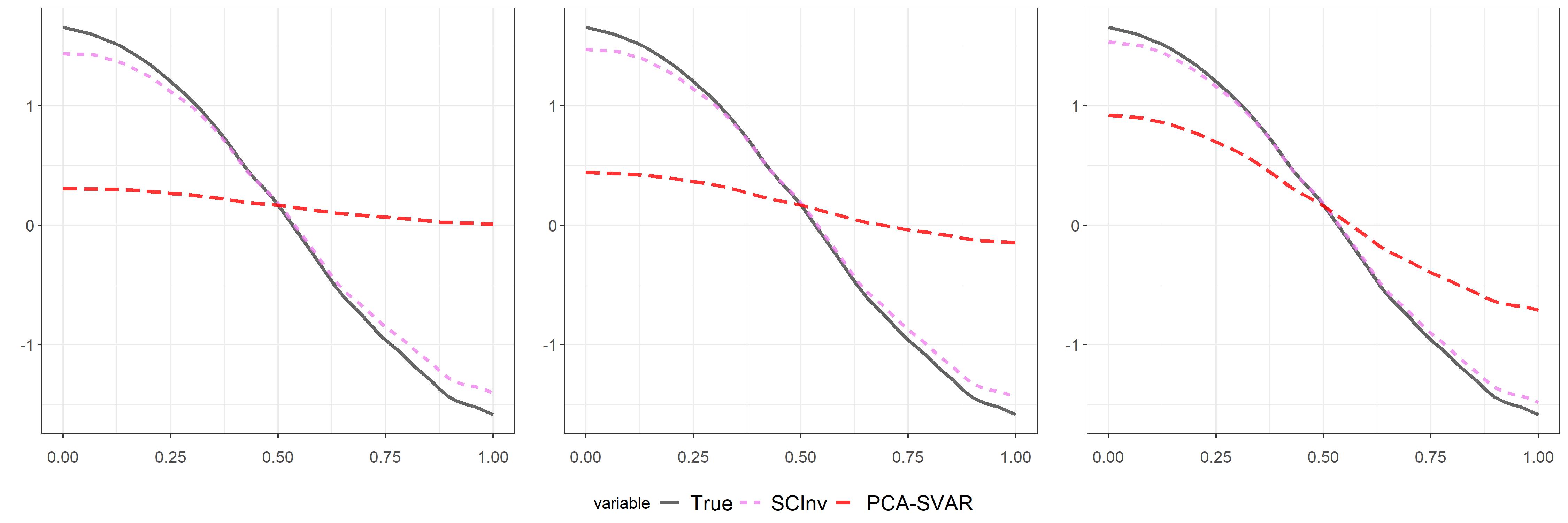}
\flushleft{\scriptsize{Notes: Averaged across 1,000 replications. The figure reports the true functional coefficient (black) and its estimates based on SCInv (dotted) and SVAR (dashed) when $c_1$ is 1 (left), 0.5 (middle), and 0.2 (right). $h=1$.  }}
\end{figure}

Table~\ref{tab: simB-1} summarizes  simulation results. Overall, the SCInv produces the smallest bias at the cost of a relatively large variance. As observed previously, the other two estimators, the SCInv$_1$ and the PCA-SVAR, report similar estimation results with each other, regardless of the value of $c_1$ and the sample size, as the difference between the two regularization schemes is not significant in this simulation design. %All the estimators appear to have a smaller bias and variance as $c_1$ increases. This is also observed in Figure~\ref{fig: simB-1}, in which the true functional coefficient and its estimates (averaged across 1,000 replications) are reported for each value of $c_1$. In the figure, it is noticeable that the PCA-SVAR estimator tends to get closer to the true functional coefficient as $c_1$ decreases. On the other hand, our estimator, the SCInv (dotted), is very robust and close to the true regardless of the value of $c_1$.
Figure~\ref{fig: simB-1} reports the true functional coefficient and its estimates (averaged across 1,000 replications) for each value of $c_1$ when $h=1$. In the figure, it is noticeable that the PCA-SVAR estimator tends to get closer to the true functional coefficient as $c_1$ decreases. On the other hand, our estimator, the SCInv (dotted), produces very robust estimation results close to the true regardless of the value of $c_1$.

\section{Conclusion \label{sec:con}}
In this paper, we study  impulse response analysis with functional predictors and other scalar-valued covariates and propose new estimation and inference methodologies. We show that the proposed estimator allows for an interesting interpretation as a structural impulse response in some special cases. In our empirical application, we study how economic variables respond to a certain distributional shock on sentiment. The results are consistent with existing observations and show negative (resp.\ positive) responses of economic growth variables when sentiment distributions shift to the left (resp.\ right). Monte Carlo simulation results also confirm our theoretical findings.

%	\section{Simulation Study}
%	\section{Application to prediction with sentiment density}

%\makeatletter %remove leading `A \quad` from section headers
%\def\@seccntformat#1{\@ifundefined{#1@cntformat}
%	{\csname the#1\endcsname\quad}
%	{\csname #1@cntformat\endcsname}}
%\newcommand\section@cntformat{}
%\makeatother
%\addcontentsline{toc}{section}{Appendix}
%\renewcommand{\thesubsection}{\Alph{subsection}}

\makeatletter
% Adjust section numbering to display the number with a letter
\def\@seccntformat#1{%
\csname the#1\endcsname.\quad % Keep the section numbering (A, B, etc.)
}
\makeatother

% Customize subsection numbering with letters (A, B, C...)
\renewcommand{\thesubsection}{\thesection.\arabic{subsection}}

\newpage 
\appendix
\section*{Appendix} \label{sec:app}

\section{Estimation with endogenous predictors}\label{sec:est2}
There are various circumstances where endogenous regressors are observed. This issue is particularly relevant when a proxy for a structural shock, rather than the shock itself, is available, as considered in the macroeconomics literature. Moreover, considering the practical issue that functional observations are almost always constructed from their discrete realizations (\citealp{Chen_et_al_2020,seong2021functional}), any regression model with functional regressors is likely to inherently involve endogeneity. In this section, we propose an IV estimation method to resolve the endogeneity problem. 

%To begin with, suppose  the following model holds:
%	\begin{equation}
%		y_{t+h} =  \mathbf{v}_t' \alpha_h + \langle  \beta_h, {X}_t^\circ  \rangle +  \uu_{h,t}= \langle \Upsilon_t,\theta_h \rangle + \uu_{h,t}. \label{eq: model: benchmark: reduced_add}
%	\end{equation}
%The variable available in the analysis is not regressor ${X}_t^\circ$   but  an observable predictor, denoted $\XX_t$, containing an additive measurement error $e_t$, i.e., $\XX_t = \Upsilon_t + e_t$. By replacing $\Upsilon_t$ with $\XX_t$, we have
%	\begin{equation} 
%		y_{t,h} = \langle \XX_t,\theta_h \rangle + u_{h,t}, \quad u_{h,t} = \uu_{h,t}-\langle e_t, \theta_h \rangle,  \label{equu}
%	\end{equation} 
%	where $\XX_t$ and $u_{h,t}$ are generally correlated (i.e., $\mathbb{E}[u_{h,t}\XX_t] \neq 0$) and thus the regressor $\XX_t$ is endogenous. In this case, it is not difficult to see that Assumption \ref{assum2}\ref{assum2a} is invalid and thus the proposed estimator in Section \ref{sec:est} becomes inconsistent; see also \citet[Sec.\ 2.1]{seong2021functional}. 

Suppose that $\XX_t$ is endogenous but there exists $\ZZ_t$ satisfying $\mathbb{E}[u_{h,t} \ZZ_t] = 0$ and $\mathbb{E}[ \XX_{t}\otimes\ZZ_t ] \neq 0$. If $\ZZ_t$ additionally satisfies the lead-lag exogeneity in  \cite{SW2018}, it could further allow us to identify the structural impulse response.  %In the considered scenario where the measurement error $e_t$ is generated by constructing a functional observation (contained in $\Upsilon_t$) with its discrete realizations, it may not be unreasonable to assume that $e_t$ is not correlated with $e_{t-\ell}$ (i.e., $\mathbb{E}[e_t\otimes e_{t-\ell}]=0$ for some large enough $\ell>0$). In this case, $\XX_{t-\kappa}$  can be a candidate for any $\kappa \geq \ell$.  
Let \begin{equation}
C_{\ZZ \XX } = \mathbb{E}[\XX_t\otimes\ZZ_t] \quad\text{and}\quad \widehat{C}_{\ZZ \XX} = \frac{1}{T}\sum_{t=1}^T \XX_t\otimes \ZZ_t. \nonumber
\end{equation}
%From \eqref{equu}, we find that $C_{ y\ZZ} = C_{\XX\ZZ}\theta_h$ , and thus $\ddot \theta$ satisfying $\widehat C_{ y\ZZ} = \widehat C_{\XX\ZZ} \ddot \theta_h$ may be a reasonable estimator. However,  
%As in Section \ref{sec:est}, 
%Our inferential method requires estimation of the eigenelements of a non-self-adjoint operator $C_{\XX\ZZ}$ whose (operator Schur complement is non-self-adjoint as well). The estimation is implemented by estimating the eigenelements of self-adjoint operators $C_{\XX\ZZ}^\ast C_{\XX\ZZ}$ and $C_{\XX\ZZ}C_{\XX\ZZ}^\ast$ (see \citealp{Bosq2000}, pp.\ 117-118), and thus their asymptotic properties depend on those of $C_{\XX\ZZ}^\ast C_{\XX\ZZ}$ and $C_{\XX\ZZ}C_{\XX\ZZ}^\ast$. Therefore, it would be theoretically convenient to take a slightly different approach, constructing our estimator based on regularizing $C_{\XX\ZZ}^\ast C_{\XX\ZZ}$, these operators rather than focusing on the operator Schur complement of $C_{\XX\ZZ}$, which is also non-self-adjoint. 
Similarly, we define $C_{\ZZ y}$ and $C_{\ZZ u }$ by replacing $\Upsilon_t$ with $\ZZ_t$ in our definitions of $C_{\Upsilon y}$ and $C_{\Upsilon u}$, and let $\widehat{C}_{\ZZ y}$ and $\widehat{C}_{\ZZ u}$ denote their sample counterparts. As in \eqref{covblock} and \eqref{eqopschur}, we write 
%\begin{equation}\nonumber
%	{C}_{\ZZ \XX} = \begin{bmatrix}		\DDD_{11} &\DDD_{12} \\ \DDD_{21}& \DDD_{22}	\end{bmatrix},   \,\, 
%	\widehat{C}_{\ZZ \XX} = \begin{bmatrix}		\widehat{\DDD}_{11} & \widehat{\DDD}_{12} \\ \widehat{\DDD}_{21} &\widehat{\DDD}_{22}	\end{bmatrix},  \,\, \SS ={\DDD}_{22}-{\DDD}_{21}{\DDD}_{11}^{-1}{\DDD}_{12},  \,\, \widehat{\SS}=\widehat{\DDD}_{22}-\widehat{\DDD}_{21}\widehat{\DDD}_{11}^{-1}\widehat{\DDD}_{12},\end{equation}
\begin{equation}\nonumber
{C}_{\ZZ \XX} = \begin{bmatrix}		\DDD_{11} &\DDD_{12} \\ \DDD_{21}& \DDD_{22}	\end{bmatrix}\quad\text{and}\quad   \SS ={\DDD}_{22}-{\DDD}_{21}{\DDD}_{11}^{-1}{\DDD}_{12},\end{equation}
where $\DDD_{ij} = \mathbb{E}[g_{j,t} \otimes h_{i,t}]$, $g_{1,t} = \mathbf{w}_t$, $g_{2,t}=X_t$,  $h_{1,t}=\mathbf{z}_{1,t}$ and $h_{2,t}=Z_{2,t}$  when
\( \ZZ_t = \left[\begin{smallmatrix}\mathbf{z}_{1,t}\\ Z_{2,t}\end{smallmatrix}\right] \) with $\mathbf{z}_{1,t}$ (resp.\   \( Z_{2,t} \)) taking values in $\mathbb R^{m}$ (resp.\ \( \mathcal H \)). We also define their sample counterparts by replacing $\Delta_{ij}$ with $\hat{\Delta}_{ij}= T^{-1}\sum_{t=1}^T g_{j,t} \otimes h_{i,t}$ as follows:
\begin{equation}\nonumber
\widehat{C}_{\ZZ \XX} = \begin{bmatrix}		\widehat{\DDD}_{11} & \widehat{\DDD}_{12} \\ \widehat{\DDD}_{21} &\widehat{\DDD}_{22}	\end{bmatrix}\quad\text{and}\quad \widehat{\SS}=\widehat{\DDD}_{22}-\widehat{\DDD}_{21}\widehat{\DDD}_{11}^{-1}\widehat{\DDD}_{12}. \end{equation}	 
Observing that $\Delta_{11}$ is the cross-covariance of $\mathbf{w}_t$ and $\mathbf{z}_{1,t}$ (corresponding to the first $m$ elements of $\ZZ_t$), one can easily see that the invertibility of $\DDD_{11}$ is analogous to the invertibility assumption of the cross-covariance between the vectors of endogenous variables and IVs in a finite dimensional just-identified endogenous linear model. Here, $\DDD_{11}$ and $\DDD_{22}$ (resp.\ $\hat{\DDD}_{11}$ and $\hat{\DDD}_{22}$) are not self-adjoint, and neither are ${C}_{\ZZ \XX}$ nor ${\SS}$ (resp.\ $\hat{C}_{\ZZ \XX}$ nor $\hat{\SS}$). 

We represent ${C}_{\ZZ\XX}$, $\SS$ and $\widehat{\SS}$ with respect to their eigenelements as follows:
\begin{equation}\nonumber
{C}_{\ZZ\XX} =  \sum_{j=1}^\infty {\nutw}_j  {\wtw}_j \otimes {\wwtw}_j, \qquad {\SS}= \sum_{j=1}^\infty {\nu}_j {\wv}_j \otimes {\ww}_j, \quad\text{and}\quad \widehat{\SS}= \sum_{j=1}^\infty \hat{\llambda}_j \hat{\wv}_j \otimes \hat{\ww}_j.
\end{equation}
Subsequently, we will need $\hat{\nu}_j, \hat{\wv}_j$, and  $\hat{\ww}_j$ to construct a regularized inverse of $\widehat{\SS}$. These can easily be computed, as the eigenelements of $\widehat{\SA}$ in Section \ref{sec:est}, from the fact that $\hat{\nu}_j^2$ and $\hat{\wv}_j$ (resp.\ $\hat{\nu}_j^2$ and $\hat{\ww}_j$) are the eigenelements of $\widehat{\SS}^\ast \widehat{\SS}$ (resp.\ $\widehat{\SS} \widehat{\SS}^\ast$), which is self-adjoint, nonnegative, and compact (see \citealp{Bosq2000}, pp.\ 117--118). 

%	Therefore, the spectral  decompositions of ${C}_{\ZZ \XX}$, $\SS$ and $\widehat{\SS}$ will be written as follows:
%	\begin{equation}\label{eqspectral1}
%		\qquad {\SS}= \sum_{j=1}^\infty {\nu}_j {w}_j \otimes {\ww}_j, \qquad \widehat{\SS}= \sum_{j=1}^\infty \hat{\llambda}_j \hat{w}_j \otimes \hat{\ww}_j.
%		\end{equation}
%		$\hat{\nu}_j, \hat{w}_j$, and  $\hat{\ww}_j$ can  be computed  from the fact that $\{ \hat {\nu}_j, \hat{w}_j\}_{j \geq 1}$ (resp.\ $\{\hat{\nu}_j,\hat{\ww}_j\}_{j \geq 1}$) are the eigenelements of $\widehat{\SS}^\ast \widehat{\SS}$ (resp.\ $\widehat{\SS} \widehat{\SS}^\ast$), which is self-adjoint, nonnegative and compact (\citealp{Bosq2000}, pp.\ 117--118).  %Unless otherwise stated, we always assume that $\hat{\nu}_j, \hat{w}_j$, and  $\hat{\ww}_j$ are computed in this way.The fact that the eigenelements are computed from $\widehat{\SS}^\ast \widehat{\SS}$  or $\widehat{\SS}\widehat{\SS}^\ast$ (rather than directly from $\widehat{\SS}$) affects the asymptotic results to be presented, particularly regarding the requirements on the regularization parameter $\reg$.
%Since $\widehat C_{\XX\ZZ}$ is a non-self-adjoint operator whose operator Schur complement is also non-self-adjoint, the estimation results given in Section  \ref{sec:est} are not directly extended. We thus propose a slightly different estimator, which can be understood as an extension of the functional IV estimator (FIVE) proposed by \cite{seong2021functional} to a more general setting.   

To introduce our estimator, we observe that $C_{ \ZZ y} =  C_{\ZZ \XX}\theta_h$ holds if \eqref{eq: model: benchmark: reduced2} is satisfied and  $\mathbb{E}[u_{h,t} \ZZ_t] = 0$. Our estimator is obtained by replacing the population covariance operators with their sample counterparts and then computing $\theta_h$ using an appropriate regularized inverse. Specifically, we consider 
\begin{equation}
\tilde{\theta}_h =  \widehat{C}_{\ZZ\XX,\KKK}^{-1} \widehat{C}_{\ZZ y},
\end{equation} 
where  $\KKK = \max\{j: \hat{\llambda}_j^2 \geq \reg\}$ for some regularization parameter $\reg$ decaying to zero as $T\to \infty$,
\begin{equation}  \label{eqreginv2}
\widehat{C}_{\ZZ\XX,\KKK}^{-1}  = \begin{bmatrix} 
\widehat{\DDD}_{11}^{-1} + \widehat{\DDD}_{11}^{-1}\widehat{\DDD}_{12} \widehat{\SS}^{-1}_{\KKK} \widehat{\DDD}_{21}\widehat{\DDD}_{11}^{-1} & -\widehat{\DDD}_{11}^{-1} \widehat{\DDD}_{12} \widehat{\SS}^{-1}_{\KKK} \\ -\widehat{\SS}^{-1}_{\KKK} \widehat{\DDD}_{21}\widehat{\DDD}_{11}^{-1} &
\widehat{\SS}^{-1}_{\KKK}
\end{bmatrix},
\end{equation}
and  $\widehat{\SS}^{-1}_{\KKK} = \sum_{j=1}^{\KKK} \hat{\nu}_j^{-1} \hat{\ww}_j \otimes \hat{\wv}_j$.
%and,, $\KKK$ is determined as 
%\begin{equation}
%\KKK = \max\{j: \hat{\llambda}_j^2 \geq \reg\}. \label{eqchoiceK2}
%\end{equation}
As in \eqref{eqprojections}, we let $\widehat{\PPI}_{\KKK} = \sum_{j=1}^{\KKK} \hat{\wv}_j \otimes \hat{\wv}_j$, ${\PPI}_{\KKK} = \sum_{j=1}^{\KKK} {\wv}_j \otimes {\wv}_j$,  
\begin{equation} \label{eqwkadd} %\nonumber 
\widehat{\PP}_{\KKK}  = \begin{bmatrix} 
I_1 & \widehat{\DDD}_{11}^{-1} \widehat{\DDD}_{12} (I_2-\widehat{\PPI}_{\KKK})  \\ 0 &  \widehat{\PPI}_{\KKK}  
\end{bmatrix}\quad \text{ and }\quad {\PP}_{\KKK} = \begin{bmatrix}
I_1 & {\DDD}_{11}^{-1}{\DDD}_{12}(I_2-{\PPI}_{\KK}) \\ 0 & {\PPI}_{\KKK}
\end{bmatrix}.
\end{equation} 
Similar to $\widehat P_{\KK}$ in Section~\ref{sec:est1b}, $\widehat{\PP}_{\KKK}=\widehat{C}_{\ZZ\XX,\KKK}^{-1}\widehat{C}_{\ZZ\XX}$ and ${\PP}_{\KKK} =  {C}_{\ZZ\XX,\KKK}^{-1}{C}_{\ZZ\XX}$, where ${C}_{\ZZ\XX,\KKK}^{-1}$ is defined by replacing $\hat\DDD_{ij}$ with $\DDD_{ij}$ and $\hat\SS_{\KK}^{-1}$ with $\SS_{\KK}^{-1} = \sum_{j=1}^{\KKK} \llambda_j^{-1}\ww_j \otimes \wv_j$ in \eqref{eqreginv2}. We  assume the following throughout this section: below, $\CC$ denotes a generic positive constant.

\begin{assumIV} \label{assumpIV2}
\begin{enumerate*}[(i)]
\item \label{assumpIV1} $\langle \theta_h, x \rangle = 0$ for all $x\in \ker C_{\ZZ\XX}$ and $ C_{\ZZ\XX} \neq 0$;
\item \label{assumpIV2a} $y_{h,t} = \langle \XX_t,\theta_h \rangle + u_{h,t}$ holds with $\mathbb{E}[u_{h,t}\XX_t]\neq 0$  and $\mathbb{E}[u_{h,t}\ZZ_t]=0$; 
\item\label{assumpIV2b} $\{\XX_t \}$,  $\{\ZZ_t\}$,  $\{u_{h,t}\}$, and $\{u_{h,t}\ZZ_t\}$  are stationary and  $L^4$-$m$-approximable;  %$\{\XX_t - \mathbb{E}[\XX_t]\}$,  $\{\ZZ_t - \mathbb{E}[\ZZ_t]\}$, $\{u_{h,t}\}$ and $\{u_{h,t}\ZZ_t\}$ are $L^2$-$m$-approximable in the relevant spaces; 
% \item For some iid sequence $\{\varepsilon\}_{t\in \mathbb{Z}}$ satisfying $\mathbb{E}[\varepsilon_t]=0$ and $\mathbb{E}[\|\varepsilon_t\|^{2+\delta}] < \infty$, 
% $$X_t = \sum_{j=0}^\infty \Psi_j \varepsilon_{t-j},$$
% where $\{\Psi_j\}_{j\geq 0}$ is a sequence of linear operators satisfying $\sum_{j=1}^\infty j\|\Psi_j\| < \infty$.  
\item \label{assumpIV2c} $\|\widehat{C}_{\ZZ\XX} -C_{\ZZ\XX}\|_{\op} = O_p(T^{-1/2})$, $\|\widehat{C}_{\ZZ u}\|_{\op} = O_p(T^{-1/2})$  and $\|\hat{\Delta}_{11}^{-1} - \Delta_{11}^{-1}\|_{\op} = O_p(T^{-1/2})$;
\item\label{assumpIV3a} for $\rho>2$, 
$\nu_j^2 \leq \CC j^{-\rho}$, $\nu_j^2-\nu_{j+1}^2 \geq \CC j^{-\rho-1}$;
\item\label{assumpIV3b} for $\varsigma > 1/2$,  $|\langle \beta_h, \wv_j \rangle| \leq \CC j^{-\varsigma}$.
%   \item\label{assumpIV2d} for some  $\rho>2$ and $\CC>0$, $\llambda_j^2 \leq \CC j^{-\rho}$ and $\llambda_j^2-\llambda_{j+1}^2 \geq \CC j^{-\rho-1}$. 
\end{enumerate*} 
\end{assumIV}
As in Assumption \ref{assum1}, Assumption \ref{assumpIV2}\ref{assumpIV1} is imposed for the unique identification of $\theta_h$ from the equation $C_{ \ZZ y} =  C_{\ZZ\XX}\theta_h$.   Assumption \ref{assumpIV2}\ref{assumpIV2a} allows for the endogeneity of $\XX_t$ and requires $\ZZ_t$ to be a valid IV (see e.g., \citealp{Florence2015, Benatia2017, seong2021functional}). Assumptions~\ref{assumpIV2}\ref{assumpIV2b} and~\ref{assumpIV2}\ref{assumpIV2c} do not appear restrictive, given that the considered time series are stationary. %($\{\XX_t \otimes \ZZ_t - C_{\ZZ \XX}\}$ is known to be a operator-valued stationary sequence in the space of Hilbert-Schmidt operators, which is also a Hilbert space; see \citealp{Bosq2000}, p.\ 34).   
Assumption \ref{assumpIV2}\ref{assumpIV3a} serves as a relevance condition in the conventional IV literature (see \citealp{seong2021functional}). Assumption \ref{assumpIV2}\ref{assumpIV3b} is natural as $\beta_h \in \mathcal H$. 

The consistency of $\tilde{\theta}_h$ is established as follows: 
\begin{theoremAA} \label{thm3} Let Assumption \ref{assumpIV2} hold and $T\reg^{2+4/\rho} \to \infty$. Then, 
\begin{equation}
\|\tilde{\theta}_h - \theta_h\| = O_p(T^{-1/2}\reg^{-1-2/\rho} + \reg^{(2\varsigma-1)/2\rho}).\nonumber
\end{equation} 
\end{theoremAA}
The requirement on $\reg$ for the consistency of $\tilde{\theta}_h$ in Theorem \ref{thm3} is slightly stronger than that in Theorem \ref{thm1}; for example,  the consistency  is ensured if   $T\reg^{4} \to \infty$, which can be compared with the requirement in Section~\ref{sec:est}  that $T\reg^{3} \to \infty$. %\commWK{As can be deduced from our proofs of Theorems \ref{thm1} and \ref{thm3}, this difference arises mainly because the eigenvector $\hat{\wv}_j$ is computed not from $\widehat{\SS}$ but from $\widehat{\SS}^\ast \widehat{\SS}$ for practical implementation.} % in contrast to existing consistency results for similar functional IV estimators that necessitate choosing regularization parameters based on certain population eigenvalues (see \citealp{seong2021functional}), may be particularly advantageous for practitioners who need to select the regularization parameter with little prior knowledge of the eigenvalues.

We next discuss on the distributional properties of $\langle \tilde{\theta}_h, \zeta\rangle$ for any $\zeta \in \widetilde{\mathcal H}$ as in Theorem \ref{thm2}. Consider the decomposition $\langle \tilde{\theta}_h-\theta_h,\zeta \rangle=\widetilde{\Theta}_1+\widetilde{\Theta}_{2A}+\widetilde{\Theta}_{2B}$, where 
\begin{equation}\nonumber
\widetilde{\Theta}_1 =  \langle \widetilde{\theta}_h-\widehat{\PP}_{\KKK}\theta_h, \zeta \rangle,\quad \widetilde{\Theta}_{2A}=  \langle\widehat{\PP}_{\KKK}\theta_h - {\PP}_{\KKK}{\theta}_h,\zeta\rangle, \quad\text{and}\quad \widetilde{\Theta}_{2B}=  \langle{\PP}_{\KKK}\theta_h - {\theta}_h,\zeta\rangle.
\end{equation}
We define  $\Lambda_{\UUU}$ and $  \widehat{\Lambda}_{\UUU}$ as follows: below $\UUU_t = u_{h,t}\ZZ_t$ and  $\hat{\UUU}_t = \tilde{u}_{h,t}\ZZ_t$ with $\tilde{u}_{h,t}= y_t - \langle \tilde{\theta}_h, \Upsilon_t \rangle$.
\begin{equation} \nonumber
\Lambda_{\UUU} =\sum_{s=-\infty}^\infty \mathbb{E}[\UUU_t \otimes \UUU_{t-s}]  \quad \text{and}\quad %= \sum_{j=1}^\infty \tilde{\mu}_j \tilde{\varpi}_j \otimes \tilde{\varpi}_j, \quad \tilde{\mu}_1 \geq \tilde{\mu}_2 \geq \ldots\geq 0 
\widehat{\Lambda}_{\UUU} =  \frac{1}{T} \sum_{s=-h}^h \mathrm{k}\left(\frac{s}{\bdw}\right) \sum_{1 \leq t, t-s \leq T}  \widehat{\UUU}_t \otimes \hat{\UUU}_{t-s}.
\end{equation}
%where the second equality follows from the spectral decomposition of $ \Lambda_{u\ZZ}$ as in \eqref{eqspectrallambda}. 
We introduce the required assumptions for the subsequent discussion. To this end, %we let $\ww_j$ be the eigenvector of $C_{\XX\ZZ}C_{\XX\ZZ}^\ast$ (i.e., $C_{\XX\ZZ}C_{\XX\ZZ}^\ast = \sum_{j=1}^\infty \llambda_j^2 \ww_j \otimes \ww_j$) and 
we define 
\begin{equation}\nonumber
d_{m,j}(\zeta)= \nutw_j^{-1} \langle  \wtw_j,\zeta \rangle \bigg/ \sqrt{\sum_{j=1}^{m} \nutw_j^{-2} \langle  \wtw_j,\zeta\rangle^2} .
\end{equation}
Below, we let \( Z_{w,t} = Z_{2,t} - \DDD_{21}\DDD_{11}^{-1}\mathbf{z}_{1,t} \), which takes values in $\mathcal H$. Additionally, we define  
$\tilde{r}_t(j,\ell) = \langle X_t, \wv_j \rangle \langle Z_{w,t}, \ww_{\ell} \rangle - \mathbb{E}[\langle X_t, \wv_j \rangle \langle Z_{w,t}, \ww_{\ell} \rangle]$
for \( j, \ell \geq 1 \), and let \( \CC \) denote a generic positive constant.

\begin{assumIV} \label{assumpIV4} 
\begin{enumerate*}[(i)] 
\item \label{assumpIV4a} Assumption \ref{assum4}\ref{assum4a} holds;
\item \label{assumpIV4aa}  $\zeta \notin \ker C_{\ZZ\XX}$ and $\Lambda_{\UUU} \wwtw_j \neq 0$  for all $\wwtw_j $ corresponding to $\nutw_j >0$;
\item  \label{assumpIV4bb} $\sum_{j=1}^{m} \sum_{\ell=1}^{m} d_{m,j}(\zeta) d_{m,\ell}(\zeta) \langle \Lambda_{\UUU}\wwtw_j,\wwtw_{\ell} \rangle \to \CC$ as $m\to \infty$; \item \label{assumpIV4cc} $\sup_{1\leq t\leq T} \|\ZZ_t\| = O_p(1)$.
\end{enumerate*}
\end{assumIV}
\begin{assumIV} \label{assumpIV5} 
\begin{enumerate*}[(i)] 
\item \label{assumpIV5b} $\mathbb{E}[\langle X_t,\wv_j \rangle^4] \leq \CC \llambda_j^2$, $\mathbb{E}[\langle Z_{2,t},\ww_j \rangle^4] \leq \CC \llambda_j^2$, $\mathbb{E}[\langle Z_{w,t},\ww_j \rangle^4] \leq \CC \llambda_j^2$,   and for some $\tilde{\CC}>1$ and $s\geq 1$, $\mathbb{E}[\tilde{r}_t(j,\ell)\tilde{r}_{t-s}(j,\ell)]\leq \CC s^{-\tilde{\CC}}\mathbb{E}[\tilde{r}_t^2(j,\ell)]$;
%\item \label{assumpIV5d} $\|T^{-1/2} \sum_{t=1}^T (\langle X_t,\ww_j \rangle  \mathbf{z}_{1,t}-\mathbb{E}[\langle X_t,\ww_j \rangle \mathbf{z}_{1,t}])\|^2 = O_p(\llambda_j)$. 
\item\label{assumpIV5c} for some $\delta>1/2$,  $|\langle \zeta_2,\wv_j \rangle| \leq \CC j^{-\delta}$ and $|\langle \zeta_1, \DDD_{11}^{-1}\DDD_{12}\wv_j \rangle|\leq \CC j^{-\delta}$.
\end{enumerate*}
\end{assumIV}
%Assumption \ref{assumpIV4} (resp.\ \ref{assumpIV5}) is similar to Assumption \ref{assum4}  (resp.\ \ref{assum5}). 
Assumption \ref{assumpIV4}\ref{assumpIV4a} is employed to facilitate the mathematical proofs of the subsequent results. Assumptions \ref{assumpIV4}\ref{assumpIV4aa} and \ref{assumpIV4}\ref{assumpIV4bb} are adapted from Assumptions \ref{assum4}\ref{assum4aa} and \ref{assum4}\ref{assum4bb}, respectively, and they play similar roles in our asymptotic analysis. Assumption \ref{assumpIV5} is adapted from Assumption \ref{assum5} to extend the discussion to the case of endogeneity.
%A similar assumption can be found in  \cite{seong2021functional}. 

We next present asymptotic normality results similar to those of Theorem~\ref{thm2}:
\begin{theoremAA} \label{thm4}   Suppose that Assumptions \ref{assumpIV2}-\ref{assumpIV4}  hold  and $T\reg^{3+4/\rho} \to \infty$. Let   ${C}_{\XX \ZZ,\KKK}^{-1}$ be the adjoint of $ {C}_{\ZZ\XX,\KKK}^{-1}$ (this is defined in a straightforward manner; see Remark \ref{remaddapp01} in Section \ref{sec_prelim0}).  Then the following holds:
%\begin{equation} 
%\sqrt{{T}/{\omega_{\KKK}(\zeta)}}\widetilde{\Theta}_1 \to_d N(0,1). %\label{eqthm4}
%\end{equation}
\begin{enumerate}[(i)]
\item\label{thm4i0} $\sqrt{{T}/{\omega_{\KKK}(\zeta)}}\widetilde{\Theta}_1 \to_d N(0,1)$ for {$\omega_{\KKK}(\zeta)= \langle {C}_{\ZZ\XX,\KKK}^{-1} {\Lambda}_{\UUU} {C}_{\XX\ZZ,\KKK}^{-1}\zeta, \zeta\rangle$}.%;  if $\langle \theta_{h}, \zeta\rangle = \langle \theta_{1,h}, \zeta\rangle$, then $\widetilde{\Theta}_2 = \widetilde{\Theta}_3 = 0$ and thus $
%\sqrt{{T}/{\omega_{\KKK}(\zeta)}}\langle \tilde{\theta}_h-\theta_h,\zeta\rangle \to_d N(0,1).$
\end{enumerate}
If Assumptions \ref{assumpIV5} is additionally satisfied with $\rho/2 + 2 < \varsigma+ \delta$, %, regardless of if $\langle \theta_{h}, \zeta\rangle = \langle \theta_{1,h}, \zeta\rangle$, 
the following hold:
\begin{enumerate}[(i)]\addtocounter{enumi}{1}
\item\label{thm4i1}  If $\omega_{\KKK}(\zeta) \to_p \infty$, $\sqrt{{T}/{\omega_{\KKK}(\zeta)}}\widetilde{\Theta}_{2A} \to_p 0$.
% \begin{equation} 
%\sqrt{{T}/{\psi_{\KKK}(\zeta)}}\widetilde{\Theta}_{2A} \to_p 0 \label{eqthm5}
%\end{equation}
\item \label{thm4i2} If $\omega_{\KKK}(\zeta) \to_p \infty$ and $T^{1/2}\reg^{(\delta+\varsigma-1)/\rho} \to 0$,  $\sqrt{{T}/{\omega_{\KKK}(\zeta)}}\widetilde{\Theta}_{2B} \to_p 0.$
% \begin{equation} 
	%\sqrt{{T}/{\psi_{\KKK}(\zeta)}}\widetilde{\Theta}_{2B} \to_p 0. \label{eqthm6}
	%\end{equation}
	%and thus, if $\psi_{\KK}(\zeta) \to_p \infty$, 
	%$$\sqrt{{T}/{\psi_{\KK}(\zeta)}}\widehat{\Theta} \to_d N(0,1).$$
	
	%   \item If  HH holds and $\delta + \varsigma > \rho/2 + 2$
\end{enumerate}
The results in \ref{thm4i0}-\ref{thm4i2} hold 
when $\omega_{\KKK}(\zeta)$ is replaced by $\hat\omega_{\KKK}(\zeta)= \langle \hat{C}_{\ZZ\XX,\KKK}^{-1} \hat{\Lambda}_{\UUU} \hat{C}_{\XX\ZZ,\KKK}^{-1} \zeta,\zeta\rangle$,
where $\hat{C}_{\XX\ZZ,\KKK}^{-1}$ is the adjoint of $\hat{C}_{\ZZ\XX,\KKK}^{-1}$. 
%Each of the results Let $\widehat{\omega}_{\KKK}(\zeta) = \langle \left(\widehat{C}_{\XX\ZZ}^\ast \widehat{C}_{\XX\ZZ}\right)_{\KKK}^{-1} \widehat{C}_{\XX\ZZ}^\ast\widehat{\Lambda}_{u \ZZ}\widehat{C}_{\XX\ZZ}\left(\widehat{C}_{\XX\ZZ}^\ast \widehat{C}_{\XX\ZZ}\right)_{\KKK}^{-1}\zeta, \zeta\rangle$.  \eqref{eqthm1}-\eqref{eqthm3} hold if .
\end{theoremAA}
Obviously, the decay rate of the regularization parameter $\reg$, required in Theorem~\ref{thm4}, is slightly stronger than that in Theorem~\ref{thm2}. However, still, $\reg$ satisfying $T\reg^{5} \to \infty$, such as e.g., $\reg= T^{-1/5+\epsilon}$ or  $T^{-1/5}\log^{\epsilon} T$ for some small $\epsilon>0$, meets the requirement for any $\rho>2$, and such a choice is recommended.  As shown by Theorems~\ref{thm3} and~\ref{thm4}, using the functional IV approach, we can consistently estimate the parameter of interest and also implement a statistical inference on $\theta_h$ as in Section \ref{sec:est1}.

\section{Proofs}\label{sec:app:pf}
We first discuss some useful lemmas, which will be repeatedly used in the subsequent sections, and then provide proofs of the theoretical results. 	With a slight abuse of notation, we will use $\langle \cdot, \cdot \rangle$ (resp. $\|\cdot\|$) to denote the inner product (resp. norm) regardless of the underlying Hilbert space. This may cause little confusion while significantly reducing notational burden. As another way to simplify notation, we let $I_1$ denote the identity map on $\mathbb{R}^m$ for various $m$ depending on the context. This is used together with $I_2$, denoting the identity map on $\mathcal H$, and $I$, denoting the identity map on $\mathbb{R}^m \times \mathcal H$.  
\subsection{Preliminary results and useful lemmas} \label{sec_prelim0} 
\begin{lemma} \label{lem1}
The operator Schur complement $\SA$ % and $\widehat{\SA}$ 
defined in \eqref{eqopschur} is self-adjoint, nonnegative and trace-class (i.e., 	$\sum_{j=1}^\infty \langle \SA h_j, h_j \rangle < \infty$ for some orthonormal basis $\{h_j\}_{j\geq 1}$).
\end{lemma}
\begin{proof}[Proof of Lemma \ref{lem1}]
From Theorem 1 of \cite{baker1973joint}, we know that $\CCC_{ij}=\CCC_{ii}^{1/2}V_{ij}\CCC_{jj}^{1/2}$ for $i,j \in \{1,2\}$ and for some $V_{ij}$ satisfying $\|V_{ij}\|_{\op} \leq 1$. Thus, $\SA = \CCC_{22} - \CCC_{22}^{1/2}V_{21}V_{12} \CCC_{22}^{1/2}$, from which we find that $\SA$ is self-adjoint and compact. Let $\{h_j\}_{j\geq 1}$ be an orthonormal basis, then we find that 
\begin{align}
	\sum_{j=1}^\infty \langle \SA h_j, h_j \rangle   &= \sum_{j=1}^\infty \langle \CCC_{22} h_j, h_j \rangle - \sum_{j=1}^\infty \langle \CCC_{22}^{1/2}V_{21}V_{12} \CCC_{22}^{1/2}  h_j, h_j \rangle \notag  \\
	&= \sum_{j=1}^\infty \langle \CCC_{22} h_j, h_j \rangle - \sum_{j=1}^\infty \langle V_{21}V_{12} \CCC_{22}^{1/2}  h_j, \CCC_{22}^{1/2} h_j \rangle.  \label{eqlem01}
\end{align}
Note that $\sum_{j=1}^\infty \langle V_{21}V_{12} \CCC_{22}^{1/2}  h_j, \CCC_{22}^{1/2} h_j \rangle  \leq \sum_{j=1}^\infty \|V_{21}V_{12}\|_{\op}\langle  \CCC_{22}^{1/2}  h_j, \CCC_{22}^{1/2} h_j \rangle \leq \sum_{j=1}^\infty \langle\CCC_{22} h_j,  h_j \rangle < \infty$, which follows from the fact that $\|V_{21}V_{12}\|_{\op}\leq \|V_{21}\|_{\op}\|V_{12}\|_{\op} \leq 1$ and $\CCC_{22}$ is trace-class. Combining this result with \eqref{eqlem01}, we find that $\SA$ is a trace-class operator. Similarly, using the properties of the inner product, $\CCC_{22}$ and $V_{ij}$, it can also be shown that $\langle \SA h,h \rangle = \langle \CCC_{22}h,h \rangle - \langle \CCC_{22}^{1/2}V_{21}V_{12}\CCC_{22}^{1/2}h,h \rangle  \geq 0$; see the proof of Proposition 2.1 of \cite{Seo2024}, in particular (A.2).
%	For the desired results for $\widehat{\SA}$, we note that  $\widehat{\CCC}_{ij}=\widehat{\CCC}_{ii}^{1/2}\widehat{V}_{ij}\widehat{\CCC}_{jj}^{1/2}$ for $i,j \in \{1,2\}$ for some $\widehat{V}_{ij}$ satisfying $\|\widehat{V}_{ij}\|_{\op} \leq 1$ regardless of $T$  a.s. Then the rest of the proof is nearly identical, so we omit the details. 
\end{proof}
\begin{lemma} \label{lem2}
$\|\widehat{C}_{\Upsilon\Upsilon,\KK}^{-1}\|_{\op} = O_p(\|\widehat{\SA}_{\KK}^{-1}\|_{\op})$ and $\|\widehat{C}_{\ZZ\XX,\KK}^{-1}\|_{\op} = O_p(\|\widehat{\SS}_{\KKK}^{-1}\|_{\op})$. 
\end{lemma}
\begin{proof}[Proof of Lemma \ref{lem2}]
For any $v\in \widetilde{\mathcal H}$, we can find $v_1$ and $v_2$ such that $v=\left[\begin{smallmatrix}v_1\\v_2\end{smallmatrix} \right]$, $v_1 \in \mathbb{R}^{m}$, and  $v_2 \in \mathcal H$. We similarly write $w  = \left[\begin{smallmatrix}w_1\\w_2\end{smallmatrix}\right]$. Then, to show the former result,  we first observe that  $\|\widehat{C}_{\Upsilon\Upsilon,\KK}^{-1}\|_{\op} = \sup_{\|v\|=\|w\| = 1} |\langle \widehat{C}_{\Upsilon\Upsilon,\KK}^{-1}v,w\rangle|$. 
Note that $\widehat{\CCC}_{11}$ (and its inverse), $\widehat{\CCC}_{12}$, and $\widehat{\CCC}_{21}$ are all convergent (see Assumption \ref{assum2}) in the operator norm, which implies that they are all $O_p(1)$ uniformly in $T$; see Lemma S5.1 of \cite{seo2024functional} based on some results given in \cite{skorohod2001}. On the other hand, $\|\SA_{\KK}^{-1}\|_{\op}$ diverges as $T$ gets larger, and thus it is obvious that, for large enough $T$, 
\begin{align}
	|	\langle \widehat{C}_{\Upsilon\Upsilon,\KK}^{-1}v,w \rangle| \leq&	|\langle \widehat{\CCC}_{11}^{-1}v_1 + \widehat{\CCC}_{11}^{-1}\widehat{\CCC}_{12} \widehat{\SA}^{-1}_{\KK} \nonumber
	\widehat{\CCC}_{21}\widehat{\CCC}_{11}^{-1}v_1 -\widehat{\CCC}_{11}^{-1} \widehat{\CCC}_{12} \widehat{\SA}^{-1}_{\KK} v_2  , w_1 \rangle|  \\ \nonumber
	& + |\langle -\widehat{\SA}^{-1}_{\KK} \widehat{\CCC}_{21}\widehat{\CCC}_{11}^{-1}  v_1 + \widehat{\SA}^{-1}_{\KK}v_2, w_2 \rangle |	\\  \nonumber
	&\leq O_p(\|\widehat{\CCC}_{11}^{-1}\|_{\op}) +  O_p(\|\widehat{\SA}^{-1}_{\KK}\|_{\op}).    \nonumber
\end{align} 
%	\begin{align}
	%		\langle \widehat{C}_{\Upsilon\Upsilon,\KK}^{-1}v,w \rangle|  &= 
	%			\begin{bmatrix}
		%			\langle \widehat{\CCC}_{11}^{-1}v_1 + \widehat{\CCC}_{11}^{-1}\widehat{\CCC}_{12} \widehat{S}^{-1}_{\KK} \widehat{\CCC}_{21}\widehat{\CCC}_{11}^{-1}v_1 , w_1 \rangle \\
		%			\langle -\widehat{S}^{-1}_{\KK} \widehat{\CCC}_{21}\widehat{\CCC}_{11}^{-1}  v_1 + \widehat{S}^{-1}_{\KK}v_2, w_2 \rangle
		%			\end{bmatrix}
	%	\langle \widehat{\CCC}_{11}^{-1}v_1 + \widehat{\CCC}_{11}^{-1}\widehat{\CCC}_{12} \widehat{S}^{-1}_{\KK} \widehat{\CCC}_{21}\widehat{\CCC}_{11}^{-1}v_1 -\widehat{\CCC}_{11}^{-1} \widehat{\CCC}_{12} \widehat{S}^{-1}_{\KK} v_2  -\widehat{S}^{-1}_{\KK} \widehat{\CCC}_{21}\widehat{\CCC}_{11}^{-1}  v_1 + \widehat{S}^{-1}_{\KK}v_2, w \rangle|
	%\\& \leq  O_p(1) + 3 \|\widehat{S}_{\KK}\|_{\op}^{-1}
	%	\end{align}
Since $\widehat{\CCC}_{11}^{-1}$ is also bounded uniformly in $T$ (Assumption \ref{assum2}\ref{assum2c}), the desired result immediately follows. The latter result can similarly be established, so the details are omitted. 
\end{proof}
Note that ${C}_{\Upsilon\Upsilon,\KK}^{-1}$ and ${C}_{\ZZ\XX,\KK}^{-1}$, defined in Theorems \ref{thm2} and \ref{thm4}, are given by  $${C}_{\Upsilon\Upsilon,\KK}^{-1}  = \left[\begin{matrix} 
{\CCC}_{11}^{-1} + {\CCC}_{11}^{-1}{\CCC}_{12} \SA^{-1}_{\KK} {\CCC}_{21}{\CCC}_{11}^{-1} & -{\CCC}_{11}^{-1} {\CCC}_{12} \SA^{-1}_{\KK} \\ -\SA^{-1}_{\KK} {\CCC}_{21}{\CCC}_{11}^{-1} &\SA^{-1}_{\KK}\end{matrix}\right]$$ and \begin{equation} \label{eqeqeq001}
{C}_{\ZZ\XX,\KK}^{-1}  = \left[\begin{matrix} {\DDD}_{11}^{-1} + {\DDD}_{11}^{-1}{\DDD}_{12} \SS^{-1}_{\KK} {\DDD}_{21}{\DDD}_{11}^{-1} & -{\DDD}_{11}^{-1} {\DDD}_{12} \SS^{-1}_{\KK} \\ -\SS^{-1}_{\KK} {\DDD}_{21}{\DDD}_{11}^{-1} &\mathcal{S}^{-1}_{\KK}	\end{matrix}\right].
\end{equation}
From the properties of the operator norm and similar arguments used in the proof of Lemma~\ref{lem2}, it is straightforward to show that 
\begin{equation} \label{addeqrem}
\|{C}_{\Upsilon\Upsilon,\KK}^{-1}\|_{\op} = O_p(\|\SA_{\KK}^{-1}\|_{\op}) = O_p(|\lambda_{\KK}|^{-1})  \quad\text{and}\quad  \|{C}_{\ZZ\XX,\KK}^{-1}\|_{\op} = O_p(\|\SS_{\KK}^{-1}\|_{\op})= O_p(|\llambda_{\KK}|^{-1}).  
\end{equation}
Furthermore, the following can be shown:
\begin{lemma} \label{lem3} As $\reg T \to \infty$, 
$\|\widehat{C}_{\Upsilon\Upsilon,\KK}^{-1}-{C}_{\Upsilon\Upsilon,\KK}^{-1}\|_{\op} = O_p(\|\widehat{\SA}_{\KK}^{-1}-\SA_{\KK}^{-1}\|_{\op}) + o_p(1)$ and $\|\widehat{C}_{\ZZ\XX,\KK}^{-1}-{C}_{\ZZ\XX,\KK}^{-1}\|_{\op} = O_p(\|\widehat{\SS}_{\KKK}^{-1}-{\SS}_{\KKK}^{-1}\|_{\op}) + o_p(1)$. 
\end{lemma}
\begin{proof}[Proof of Lemma \ref{lem3}]
To show the former result, we note that $\widehat{\CCC}_{ij} - \CCC_{ij} = O_p(T^{-1/2})$ and $\widehat{\CCC}_{11}^{-1} - {\CCC}_{11}^{-1} = O_p(T^{-1/2})$. Using these results and some algebra, we find that the stochastic orders of the elements of the operator matrix $\widehat{C}_{\Upsilon\Upsilon,\KK}^{-1} - {C}_{\Upsilon\Upsilon,\KK}^{-1}$ are given by
%	\begin{align}
	%		\begin{bmatrix} \CCC_{11}^{-1}\CCC_{12}(\widehat{\SA}_{\KK}^{-1} - {\SA}_{\KK}^{-1}) \CCC_{21}\CCC_{11}^{-1}  + O_p(T^{-1/2}) + O_p(T^{-1/2}\widehat{\SA}_{\KK}^{-1})  & O_p(T^{-1/2})\widehat{\SA}_{\KK}^{-1} \\
		%			O_p(T^{-1/2})\widehat{\SA}_{\KK}^{-1}  & \widehat{\SA}_{\KK}^{-1} -  {\SA}_{\KK}^{-1}
		%		\end{bmatrix}.
	%	\end{align}
\begin{align}
	\begin{bmatrix} O_p(\widehat{\SA}_{\KK}^{-1} - {\SA}_{\KK}^{-1})  + O_p(T^{-1/2}) + O_p(T^{-1/2}\widehat{\SA}_{\KK}^{-1})  & O_p(\widehat{\SA}_{\KK}^{-1} - {\SA}_{\KK}^{-1})+ O_p(T^{-1/2}\widehat{\SA}_{\KK}^{-1}) \\
		O_p(\widehat{\SA}_{\KK}^{-1} - {\SA}_{\KK}^{-1})+	O_p(T^{-1/2}\widehat{\SA}_{\KK}^{-1} ) & O_p(\widehat{\SA}_{\KK}^{-1} - {\SA}_{\KK}^{-1})
	\end{bmatrix}.
\end{align}
Since $\|\widehat{\SA}_{\KK}^{-1}\|_{\op} = \reg^{-1/2}$ by our choice of $\KK$ in \eqref{eqshur1},  $O_p(T^{-1/2})\widehat{\SA}_{\KK}^{-1}=O_p(\reg^{-1/2}T^{-1/2}) = o_p(1)$, from which the desired result follows.

The proof for the latter result is nearly identical, so omitted.
\end{proof}
\begin{remark} \label{remaddapp01} 	 {We consider the adjoint ${C}_{\XX\ZZ,\KK}^{-1}$ of ${C}_{\ZZ\XX,\KK}^{-1}$ in Theorem \ref{thm4}.  
	If we express ${C}_{\ZZ\XX,\KK}^{-1}$ in \eqref{eqeqeq001} as ${C}_{\ZZ\XX,\KK}^{-1} = \left[\begin{smallmatrix}
		a & b \\ c& d
	\end{smallmatrix}\right]$, then its adjoint ${C}_{\XX\ZZ,\KK}^{-1}$ is given by $	{C}_{\XX\ZZ,\KK}^{-1} = \left[\begin{smallmatrix}
		a^\ast & c^\ast \\ b^\ast& d^\ast
	\end{smallmatrix}\right]$. Moreover, since $\|A\|_{\op} = \|A^\ast\|_{\op}$ for any linear map $A$ between Hilbert spaces (see e.g., \citealp{Megginson2012}, Proposition 3.1.2), it follows that  $\|{C}_{\XX\ZZ,\KK}^{-1}\|_{\op} = O_p(\|\SS_{\KK}^{-1}\|_{\op})= O_p(|\llambda_{\KK}|^{-1})$ and also $\|\widehat{C}_{\XX\ZZ,\KK}^{-1}-{C}_{\XX\ZZ,\KK}^{-1}\|_{\op} = O_p(\|\widehat{\SS}_{\KKK}^{-1}-{\SS}_{\KKK}^{-1}\|_{\op}) + o_p(1)$. } 
	\end{remark}

	\subsection{Proofs of the main theoretical results} \label{app_proof}
	\paragraph{Proof of Proposition \ref{prop: svar: identification}} 
	Let $\mathfrak S$ be the Schur complement of $\mathcal B=\left[\begin{smallmatrix} I_1& \beta_{12} \\ \beta_{21}& I_2 \end{smallmatrix}\right]$, i.e., $\mathfrak S = I_1-\beta_{12}\beta_{21}$. Then the operator matrix $\mathcal B$ is invertible as long as $(I_1-\beta_{12}\beta_{21})x$ is nonzero for some nonzero $x \in \mathbb{R}$ and its inverse, viewed as a map depending on $\beta_{12}$ and $\beta_{21}$,  is given as follows (see \citealp{Bart2007}, p.\ 29):
	\begin{align} \label{eqinversefn}
\mathcal B^{-1}(\beta_{12},\beta_{21}) = \begin{bmatrix} \mathfrak S^{-1} & -\mathfrak S^{-1} \beta_{12} \\ - \beta_{21} \mathfrak S^{-1} & I_2 + \beta_{21} \mathfrak S^{-1} \beta_{12} \end{bmatrix}.
\end{align}
We first consider the case when $\beta_{12} = 0$. In this case, we know that $\Gamma$ in the RFVAR is given by $\Gamma = \left[\begin{smallmatrix} \gamma_{11} & \gamma_{12} \\ \gamma_{21} & \gamma_{22} \end{smallmatrix}\right]  =\left( \mathcal B^{-1}(0,\beta_{21}) \right)\left[\begin{smallmatrix} \alpha_{11} & \alpha_{12} \\ \alpha_{21}& \alpha_{22} \end{smallmatrix}\right] =  \left[\begin{smallmatrix} \alpha_{11} & \alpha_{12} \\ -\beta_{21}\alpha_{11} + \alpha_{21} & -\beta_{21}\alpha_{12}+ \alpha_{22} \end{smallmatrix}\right]$ and $	\left[\begin{smallmatrix} \varepsilon_{1t} 
\\ \mathcal E_{2t}\end{smallmatrix}\right] = \left(\mathcal B^{-1}(0,\beta_{21}) \right)\left[\begin{smallmatrix} u_{1t} 
\\ U_{2t}\end{smallmatrix}\right] = \left[\begin{smallmatrix} u_{1t} \\ -\beta_{21}u_{1t}+U_{2t}\end{smallmatrix}\right]$. The covariance operator   $\mathbf \Sigma_\varepsilon$ of $\left[\begin{smallmatrix}\varepsilon_{1t}\\ 
\mathcal E_{2t}\end{smallmatrix}\right]$ can be written as
\begin{equation*}
\mathbf \Sigma_\varepsilon  = \begin{bmatrix}
	\sigma_{\varepsilon,11}  &\sigma_{\varepsilon,12} \\\sigma_{\varepsilon,21} & {\Sigma}_{\varepsilon,22}
\end{bmatrix} = \begin{bmatrix}
	\sigma_{11} & -\sigma_{11} \beta_{21}^\ast  \\- \beta_{21} \sigma_{11}  & \beta_{21} \sigma_{11} \beta_{21}^\ast + \Sigma_{22} 
\end{bmatrix}.
\end{equation*} 
Observing that $\sigma_{11}^{-1}$ and $\sigma_{\varepsilon,11}^{-1}$ are well defined (and they even commute with any linear map in the considered case; see footnote \ref{foot1}), we find that  $\beta_{21} = - \sigma_{\varepsilon,21}\sigma_{\varepsilon,11}^{-1}$, $\alpha_{11}=\gamma_{11}$, $\alpha_{12}=\gamma_{12}$, $\alpha_{21}= \gamma_{21} - (\sigma_{\varepsilon,21}\sigma_{\varepsilon,11}^{-1})\gamma_{11}$, $\alpha_{22}=\gamma_{22}-(\sigma_{\varepsilon,21}\sigma_{\varepsilon,11}^{-1})\gamma_{12}$. Moreover, from that  $\sigma_{11}=\sigma_{\varepsilon,11}$ and $\beta_{21}=- \sigma_{\varepsilon,21}\sigma_{\varepsilon,11}^{-1}$, we find that $\Sigma_{22}=\Sigma_{\varepsilon,22} - (\sigma_{\varepsilon,21}\sigma_{\varepsilon,21}^\ast \sigma_{\varepsilon,11}^{-1})$. Thus for the given set of reduced-form parameters, we can find the unique structural parameters, which is the desired identification result.

We next consider the case when $\beta_{21} = 0$. Then \eqref{eqreduced1} holds with $\left[\begin{smallmatrix} \gamma_{11} & \gamma_{12} \\ \gamma_{21} & \gamma_{22} \end{smallmatrix}\right]=  \left(\mathcal B^{-1}(\beta_{12},0)\right) \left[\begin{smallmatrix} \alpha_{11} & \alpha_{12} \\ \alpha_{21}& \alpha_{22} \end{smallmatrix}\right] =  \left[\begin{smallmatrix} \alpha_{11} -\beta_{12}\alpha_{21}& \alpha_{12} -\beta_{12}\alpha_{22} \\ \alpha_{21} & \alpha_{22} \end{smallmatrix}\right]$ and in this case % $	\mathbf \Sigma_e = \left[\begin{smallmatrix}				\sigma_{11}  + \beta_{12}\Sigma_{22}\beta_{12}^\ast & -\beta_{12} \Sigma_{22}  \\-\Sigma_{22}\beta_{12}^\ast & \Sigma_{22} 			\end{smallmatrix}\right]$.
\begin{equation*}
\mathbf \Sigma_\varepsilon  = \begin{bmatrix}
	\sigma_{\varepsilon,11}  &\sigma_{\varepsilon,12} \\\sigma_{\varepsilon,21} & {\Sigma}_{\varepsilon,22}
\end{bmatrix} =  \begin{bmatrix}
	\sigma_{11}  + \beta_{12}\Sigma_{22}\beta_{12}^\ast & -\beta_{12} \Sigma_{22}  \\-\Sigma_{22}\beta_{12}^\ast & \Sigma_{22} 
\end{bmatrix}.
\end{equation*}
Observe first that $\Sigma_{\varepsilon, 22} = \Sigma_{22}$, but $\beta_{12}$ is not identified from the equation $\sigma_{\varepsilon,12} = \beta_{12}{\Sigma}_{\varepsilon, 22} = \beta_{12}\Sigma_{22}$ unless $\Sigma_{22}$ is injective (see Proposition 2.1 and the following discussion in \citealp{Mas2007}). Moreover, even if we know $\sigma_{11}$, we cannot determine $\beta_{12}$ from the equation $\sigma_{\varepsilon,11} = \sigma_{11} + \beta_{12}\Sigma_{22}\beta_{12}^\ast$ if $\Sigma_{22}$ is not injective. It thus may be easily deduced that $\beta_{12}$ is not uniquely identified from the reduced-form parameters. Of course, in this functional setting $\Sigma_{\varepsilon,22}$ is not invertible and thus the characterization via the inverse transform, $\beta_{12}=\sigma_{\varepsilon,12}\Sigma_{\varepsilon,22}^{-1}$, is invalid. Nevertheless, it turns out that if $\Sigma_{22}$ is injective (and thus $\Sigma_{\varepsilon,22}$ is injective), the equation $\sigma_{\varepsilon,12} = \beta_{12}\Sigma_{\varepsilon,22}$ defines a unique bounded linear map $\beta_{12}$ (see e.g., \citealp{carrasco2007linear}). Once $\beta_{12}$ is uniquely identified by $\sigma_{\varepsilon,12}$ and $\Sigma_{\varepsilon,22}$,  then, from similar algebra as in the above case with $\beta_{12}=0$, it is straightforward to show that the remaining structural parameters, $\alpha_{ij}$ for $i,j \in \{1,2\}$ and $\sigma_{11}$, can be written as functions of the reduced-form parameters and $\beta_{12}$. \qed

\paragraph{Proof of Proposition \ref{prop: svar: identification: a}}
Let $Y_t =\left[\begin{smallmatrix}y_t\\ X_t\end{smallmatrix}\right]$, $\mathcal E_t = \left[\begin{smallmatrix}\varepsilon_{1,t}\\\mathcal E_{2,t}\end{smallmatrix}\right]$ and  $\Gamma = \left(	\mathcal B^{-1}(\beta_{12},\beta_{21})\right) \mathcal A$ (see \eqref{eqinversefn}). From \eqref{eqreduced1}, we find that 
\begin{align}
Y_{t+h} &=   \Gamma Y_{t+h-1} + \mathcal E_{t+h} = \Gamma^2 Y_{t+h-2} +\Gamma \mathcal E_{t+h-1} + \mathcal E_{t+h} = \cdots = \Gamma^h Y_{t} + \tilde{V}_{t}, \label{eqpf01}
\end{align}
where $\tilde{V}_{t} = \sum_{j=0}^{h-1}\Gamma^j \mathcal E_{t+h-j}$, and we write this as $\tilde V_t = \left[\begin{smallmatrix}\tilde v_{1,t}\\ \tilde V_{2,t}\end{smallmatrix}\right]$. We may view $\Gamma^h$ as the following operator matrix $\Gamma^h = \left[\begin{smallmatrix} \gamma_{h,11} & \gamma_{h,12} \\ \gamma_{h,21} & \gamma_{h,22} \end{smallmatrix}\right]$, then \eqref{eqpf01} implies that 
\begin{align}
y_{t+h}&=  \gamma_{h,11} y_t + \gamma_{h,12} X_{t} + \tilde v_{1,t},  \label{eqpf02}\\
X_{t+h}&=  \gamma_{h,21} y_t + \gamma_{h,22} X_{t} + \tilde V_{2,t}.\label{eqpf03}
\end{align}
Since $\tilde V_t$ is $(\mathbb{R}\times \mathcal H)$-valued MA($h-1$) process, $\tilde v_{1,t}$ and $\tilde V_{2,t}$ are MA($h-1$) processes in the relevant spaces.   

If $\beta_{12}=0$, $\mathcal B^{-1} = \mathcal B^{-1}(0,\beta_{21}) = \left[\begin{smallmatrix} I_1& 0 \\ -\beta_{21}& I_2 \end{smallmatrix}\right]$ and thus  we have
\begin{align} \label{eqproofadd01}
\IRF_{h,12}=\pcr \Gamma^{h} \mathcal B^{-1}\pchh  = \pcr \begin{bmatrix} \gamma_{h,11} & \gamma_{h,12} \\ \gamma_{h,21} & \gamma_{h,22} \end{bmatrix}\begin{bmatrix} I_1& 0 \\ -\beta_{21}& I_2 \end{bmatrix} \pchh.
\end{align}
Noting that $\pcr \left[\begin{smallmatrix} a &b \\c&d \end{smallmatrix}\right] \pchh = b$ for all appropriately defined operator elements $a$, $b$, $c$ and $d$, we find from \eqref{eqproofadd01} that  $\gamma_{h,12}: \mathcal H\to \mathbb{R}$ is equivalent to $\IRF_{h,12}$. 	From this result, \eqref{eqpf02}, and the facts that (i) any linear map $g:\mathcal H \to \mathbb{R}$ can be understood as the map $\langle h_g, \cdot  \rangle:\mathcal H\to \mathbb{R}$ for some uniquely identified $h_g \in \mathcal H$ by the Riesz representation theorem (e.g. \citealp{Conway1994}, p.\ 13) and (ii) any linear map $g:\mathbb{R} \to \mathbb{R}$ can be understood as a scalar multiplication (see footnote \ref{foot1}), the desired result \eqref{prop2: eq1} is deduced. %  $\gamma_{h,12} X_{t} = \langle \beta_h^y,X_t\rangle$ for some $\beta_h^y \in \mathcal H
%Due to the Riesz representation theorem (e.g. \citealp{Conway1994}, p.\ 13), we may understand $\gamma_{h,12} X_{t} = \langle \beta_h^y,X_t\rangle$ for some $\beta_h^y \in \mathcal H$. 
%	From this result and \eqref{eqpf02}, the desired result  is established. 
On the other hand, 
\begin{align}\label{eqpf02a}
\IRF_{h,21}=\pch \Gamma^{h} \mathcal B^{-1}\pcrr  = \pch \begin{bmatrix} \gamma_{h,11} & \gamma_{h,12} \\ \gamma_{h,21} & \gamma_{h,22} \end{bmatrix}\begin{bmatrix} I_1& 0 \\ -\beta_{21}& I_2 \end{bmatrix} \pcrr = \gamma_{h,21}-\gamma_{h,22}\beta_{21}.
\end{align}
Observing that $X_t   = -\beta_{21}y_t+  \alpha_{21}y_{t-1}+\alpha_{22}X_{t-1} + U_{2t}$, \eqref{eqpf03} can be written as
\begin{align}
X_{t+h}&=  \gamma_{h,21} y_t + \gamma_{h,22} (-\beta_{21}y_t+  \alpha_{21}y_{t-1}+\alpha_{22}X_{t-1} + U_{2t} )   +  \tilde V_{2,t} \notag \\
&= (\gamma_{h,21} -  \gamma_{h,22} \beta_{21})y_t  +   \gamma_{h,22}\alpha_{21}y_{t-1}+ \gamma_{h,22}\alpha_{22}X_{t-1} +  V_{2,t},\label{eqpf03a}
\end{align}
where $V_{2,t} = \gamma_{h,22} U_{2t}  +  \tilde V_{2,t}$, which is obviously an MA($h-1$) process. The desired result \eqref{prop2: eq2} follows from 
\eqref{eqpf02a} and \eqref{eqpf03a}.

Now suppose that $\beta_{21}=0$ and $\mathbf\Sigma_{22}$ is injective. We then have $\mathcal B^{-1}  = \mathcal B^{-1}(\beta_{12},0)= \left[\begin{smallmatrix} I_1& -\beta_{12} \\ 0 & I_2 \end{smallmatrix}\right]$. Thus 
\begin{equation}
\IRF_{h,21}=\pch \Gamma^{h} \mathcal B^{-1}\pcrr  = \pch \begin{bmatrix} \gamma_{h,11} & \gamma_{h,12} \\ \gamma_{h,21} & \gamma_{h,22} \end{bmatrix}\begin{bmatrix} I_1& -\beta_{12} \\ 0 & I_2 \end{bmatrix} \pcrr = \gamma_{h,21}.
\end{equation}
Therefore, $\gamma_{h,21}: \mathcal H\to \mathbb{R}$ is equivalent to  $\IRF_{h,21}$. Combining this result with \eqref{eqpf03}, the desired result \eqref{prop2: eq3} is obtained.  On the other hand, 
\begin{align} \label{eqpf03a2}
\IRF_{h,12}=\pcr \Gamma^{h} \mathcal B^{-1}\pchh = \pcr \begin{bmatrix} \gamma_{h,11} & \gamma_{h,12} \\ \gamma_{h,21} & \gamma_{h,22} \end{bmatrix}\begin{bmatrix} I_1& -\beta_{12} \\ 0 & I_2 \end{bmatrix} \pchh = -\gamma_{h,11}\beta_{12} + \gamma_{h,12}.
\end{align}
%Observing that $y_t   = -\beta_{12}X_t+  \alpha_{21}y_{t-1}+\alpha_{22}X_{t-1} + U_{2t}$, \eqref{eqpf02} can be written as
Observing that $y_t   = -\beta_{12}X_t+  \alpha_{11}y_{t-1}+\alpha_{12}X_{t-1} + u_{1t}$, \eqref{eqpf02} can be written as
%	\begin{align}
%		y_{t+h}&=  \gamma_{h,11} y_t + \gamma_{h,12} X_{t} + \tilde{v}_{1,t} = \gamma_{h,11}( -\beta_{12}X_t+  \alpha_{21}y_{t-1}+\alpha_{22}X_{t-1} + U_{2t})+ \gamma_{h,12} X_{t} + \tilde{v}_{1,t} \\ &= (-\gamma_{h,11}\beta_{12} + \gamma_{h,12})X_t + \gamma_{h,11} \alpha_{21}y_{t-1} + \gamma_{h,11}\alpha_{22}X_{t-1} + \gamma_{h,11}U_{2t} + \tilde{e}_{1,t},	\label{eqpf03b} 
%	\end{align}  
\begin{align}
y_{t+h}&=  \gamma_{h,11} y_t + \gamma_{h,12} X_{t} + \tilde{v}_{1,t} = \gamma_{h,11}( -\beta_{12}X_t+  \alpha_{11}y_{t-1}+\alpha_{12}X_{t-1} + u_{1t})+ \gamma_{h,12} X_{t} + \tilde{v}_{1,t} \notag  \\ &= (-\gamma_{h,11}\beta_{12} + \gamma_{h,12})X_t + \gamma_{h,11} \alpha_{11}y_{t-1} + \gamma_{h,11}\alpha_{12}X_{t-1} + v_{1,t},	\label{eqpf03b} 
\end{align}  
where $v_{1,t}= \gamma_{h,11}u_{1t} + \tilde{v}_{1,t}$, which is an MA($h-1$) process. Then the desired result \eqref{prop2: eq4} follows from  \eqref{eqpf03a2} and \eqref{eqpf03b}.\qed

\paragraph{Proof of Theorem  \ref{thm1}}  
We note that ${\SA}$ in \eqref{eqshur0} does not change when ${v}_j$ is redefined as $\sgn(\langle \hat{v}_j,v_j \rangle)v_j$. Thus, with a slight abuse of notation, we hereafter let  ${v}_j$ denote $\sgn(\langle \hat{v}_j,v_j \rangle)v_j$ in this proof for notational simplicity. 
Observe that $\theta_h = \left[\begin{smallmatrix}\alpha_h\\\beta_h\end{smallmatrix}\right]$, $\hat\theta_h =\left[\begin{smallmatrix}\hat \alpha_h\\\hat \beta_h\end{smallmatrix}\right]$ and $\widehat{C}_{\Upsilon y} = \widehat{C}_{\Upsilon \Upsilon} \theta_h + \widehat{C}_{\Upsilon u}$, from which we find that
\begin{align} 
\hat{\theta}_h &= \widehat{C}_{\Upsilon\Upsilon,\KK}^{-1}  \widehat{C}_{\Upsilon y} =  \widehat{C}_{\Upsilon\Upsilon,\KK}^{-1}  \widehat{C}_{\Upsilon \Upsilon}  \theta_h   + \widehat{C}_{\Upsilon\Upsilon,\KK}^{-1}  \widehat{C}_{ \Upsilon u} \notag \\ &= 
\begin{bmatrix}
	\alpha_h + \widehat{\CCC}_{11}^{-1}\widehat{\CCC}_{12}(I_2-\widehat{\Pi}_{\KK}) \beta_h \\
	\widehat{\Pi}_{\KK} \beta_h
\end{bmatrix} + \widehat{C}_{\Upsilon\Upsilon,\KK}^{-1}  \widehat{C}_{ \Upsilon u}.
%	\theta_{1,h} + \left(\widehat{\CCC}_{11}^{-1}\widehat{\CCC}_{12}(I_2-\widehat{\Pi}_{\KK}) + \widehat{\Pi}_{\KK}\right)\theta_{2,h}  + \widehat{C}_{\Upsilon\Upsilon,\KK}^{-1}  \widehat{C}_{\Upsilon u}.\label{eqtheta}
\end{align} 
Note that for large enough $T$, $\|\widehat{C}_{\Upsilon\Upsilon,\KK}^{-1}\|_{\op}=O_p(\|\widehat{\SA}_{\KK}^{-1}\|_{\op})$ (see Lemma \ref{lem2}). 			
Since  $\widehat{\lambda}_j^{-1} \leq \reg^{-1/2}$ for $j=1,\ldots,\KK$, we find the following:
\begin{equation*} \label{eqpfadd01}
\|  \widehat{C}_{\Upsilon\Upsilon,\KK}^{-1}  \widehat{C}_{\Upsilon u }\|_{\op} \leq O(\reg^{-1/2}) \|\sum_{j=1}^{\KK} \hat{v}_j \otimes \hat{v}_j\|_{\op}  \|\widehat{C}_{ \Upsilon u}\|_{\op}  = O_p(\reg^{-1/2}T^{-1/2}) \to_p 0. 
\end{equation*}
From a little algebra, we find that 
\begin{equation*}
\hat{\theta}_{h} - \theta_{h} =  	\begin{bmatrix}
	\widehat{\CCC}_{11}^{-1}\widehat{\CCC}_{12}(I_2-{\Pi}_{\KK}) \beta_h + \widehat{\CCC}_{11}^{-1}\widehat{\CCC}_{12}({\Pi}_{\KK}-\widehat{\Pi}_{\KK})  \beta_h \\
	- (I_2-{\Pi}_{\KK})  \beta_h - ({\Pi}_{\KK}-\hat{\Pi}_{\KK}) \beta_h
\end{bmatrix}  +  O_p(\reg^{-1/2}T^{-1/2}).
\end{equation*}
Therefore, 
% 	and thus $\hat{\theta}_{\ha} = \theta_{1,h} +  \widehat{\CCC}_{11}^{-1}\widehat{\CCC}_{12}(I_2-\widehat{\Pi}_{\KK})\theta_{\hb} + O_p(\reg^{-1/2}T^{-1/2})$.  and $\hat{\theta}_{\hb} = \theta_{\hb} - (I-\widehat{\Pi}_{\KK})\theta_{2,h} + O_p(\reg^{-1/2}T^{-1/2})$.
% $\hat{\theta}_h - \theta_h = \widehat{\CCC}_{11}^{-1}\widehat{\CCC}_{12}({\Pi}_{\KK}-\widehat{\Pi}_{\KK})\theta_{2,h} + (\widehat{\Pi}_{\KK}-{\Pi}_{\KK})\theta_{2,h}  + \widehat{\CCC}_{11}^{-1}\widehat{\CCC}_{12} (I_2-{\Pi}_{\KK})\theta_{2,h} + ({\Pi}_{\KK}-I_2)\theta_{2,h} + O_p(\reg^{-1/2}T^{-1/2})$, from which we obtain the following:
\begin{equation}\label{eqpfaddarguments}
\hat{\theta}_h - \theta_h = O_p((\widehat{\Pi}_{\KK}-{\Pi}_{\KK}) \beta_h) + O_p((I_2-{\Pi}_{\KK}) \beta_h) + O_p(\reg^{-1/2}T^{-1/2}).
\end{equation} 
Observe that $\|\widehat{\SA}-\SA\|_{\op}=O_p(T^{-1/2})$ under Assumption \ref{assum2}. From Lemma 4.3 of \cite{Bosq2000} and Assumption \ref{assum3}, we find that $\|{v}_j-\hat{v}_j\| \leq O_p((\lambda_j-\lambda_{j+1})^{-1}T^{-1/2})=O_p((\lambda_j^{2}-\lambda_{j+1}^{2})^{-1} (\lambda_j+\lambda_{j+1})T^{-1/2})$. %\textcolor{violet}{We note that $\reg^{1/2} >\widehat\lambda_{\KK +1} - \lambda_{\KK+1} + \lambda_{\KK+1} - \lambda_{\KK}  + \lambda_{\KK} $, from which, Assumption~\ref{assum3} and the fact that $\lambda_{\KK} \geq \mathtt{C} ^{-1}\rho ^{-1}\KK^{-\rho/2}$ (see \citealp[p.\ 20]{imaizumi2018}) we have \begin{equation*}
	% \KK^{-\rho/2} \mathtt{C}^\ast \leq \lambda_{\KK+1} - \lambda_{\KK} + \lambda_{\KK} < \tau^{-1/2} +\lambda_{\KK+1} - \widehat{\lambda}_{\KK+1} \leq \tau^{1/2} (1+ O_p( (\tau T)^{-1/2}) )
	%\end{equation*} for $\mathtt{C}^\ast = ((\mathtt{C}\rho)^{-1} - \mathtt{C} \KK^{-1})  = O(1)$. Because $\tau T \to \infty$ under the employed assumption, we conclude that $\tau^{-1/2} \KK^{-\rho/2} \leq (\mathtt{C}^{\ast})^{-1} (1+o_p(1)) = O_p(1)$. Moreover, because $\reg^{1/2} \leq \widehat{\lambda}_{\KK} = O_p(T^{-1/2}) + \mathtt{C} \KK ^{-\rho/2}$ and $\tau T \to \infty$, we have $1 \leq O_p(\reg^{-1/2}T^{-1/2}) + \reg^{-1/2} \KK^{-\rho/2} = o_p(1) + \reg^{-1/2} \KK^{-\rho/2}$, which means that $\reg ^{-1/2} \KK^{-\rho/2}$ is bounded away from zero for large $T$. Combining these, we have $\reg \KK^{\rho} = O_p(1)$.  }
	{Also, note that $\reg^{1/2} \leq \hat{\lambda}_{\KK} = \hat{\lambda}_{\KK} - {\lambda}_{\KK} + {\lambda}_{\KK} \leq  \hat{\lambda}_{\KK} - {\lambda}_{\KK} +  \KK^{-\rho/2}$, and hence we obtain that $\tau^{1/2}\KK^{\rho/2} +  ( {\lambda}_{\KK} -\hat{\lambda}_{\KK})\KK^{\rho/2} \leq 1$ for every $T$. Since $\hat{\lambda}_{\KK} - {\lambda}_{\KK} = O_p(T^{-1/2})$, $\lambda_{\KK} \leq O_p(\KK^{-\rho/2})$ and $T^{-1/2}\tau^{-1/2} \to 0$ under the employed assumptions, we find that $O_p(1) = \tau^{1/2}\KK^{\rho/2} +  ( {\lambda}_{\KK} -\hat{\lambda}_{\KK})\KK^{\rho/2}  = \tau^{1/2}\KK^{\rho/2}  +  \tau^{1/2}\KK^{\rho/2} ( \tau^{-1/2} ({\lambda}_{\KK} -\hat{\lambda}_{\KK}))% =  \tau^{1/2}\KK^{\rho/2}  (1+( \tau^{-1/2} ({\lambda}_{\KK} -\hat{\lambda}_{\KK}))) 
		=\tau^{1/2}\KK^{\rho/2}(1+o_p(1)),$ from which we conclude that $\tau^{1/2}\KK^{\rho/2} = O_p(1)$.}
	%	Observe that  $\reg\KK^{\rho} = O_p(T^{-1/2}\KK^{\rho/2} + T^{-1}\KK^{\rho}) +O_p(1)$, from which we find that $\tau^{1/2} \KK^{\rho/2}$ must be stochastically bounded as $T$ gets larger, i.e., $\reg^{1/2}\KK^{\rho/2} = O_p(1)$.
	% Also, note that $\reg \leq \hat{\lambda}_{\KK}^2 = \hat{\lambda}_{\KK}^2 - {\lambda}_{\KK}^2 + {\lambda}_{\KK}^2 \leq O_p(T^{-1/2}) (\hat{\lambda}_{\KK} + {\lambda}_{\KK}) + \KK^{-\rho}=O_p(T^{-1/2}) (2{\lambda}_{\KK}+O_p(T^{-1/2})) + \KK^{-\rho} = O_p(T^{-1/2}\tau^{-1/2}\tau^{1/2}\KK^{-\rho/2})+\KK^{-\rho}$. Since $T^{-1/2}\tau^{-1/2} \to 0$, this implies that $\reg \KK^{\rho} = o_p(\reg^{1/2}\KK^{\rho/2}) + O_p(1)$, and hence $\reg^{1/2}\KK^{\rho/2} = o_p(1) + O_p(\reg^{-1/2}\KK^{-\rho/2})$, from which we find that $\tau^{1/2} \KK^{\rho/2}$ must be stochastically bounded as $T$ gets larger, i.e., $\reg^{1/2}\KK^{\rho/2} = O_p(1)$.
	Observe that 
	\begin{align}
		\|(\widehat{\Pi}_{\KK}- {\Pi}_{\KK}) \beta_h\| %&=  \|\sum_{j=1}^{\KK} \langle \hat{v}_j, \beta_h \rangle \hat{v}_j -  \sum_{j=1}^{\KK} \langle {v}_j, \beta_h \rangle {v}_j\|
		&= \|\sum_{j=1}^{\KK} \langle \hat{v}_j-v_j, \beta_h \rangle \hat{v}_j -  \sum_{j=1}^{\KK} \langle {v}_j, \beta_h \rangle ({v}_j-\hat{v}_j)\| \notag \\
		&\leq O_p(T^{-1/2}) \sum_{j=1}^{\KK} j^{\rho+1} (\lambda_j+\lambda_{j+1})=  O_p(T^{-1/2})\sum_{j=1}^{\KK} j^{\rho/2+1} \notag \\ &\leq  O_p(T^{-1/2}\reg^{-1/2-2/\rho}\reg^{1/2+2/\rho}\KK^{\rho/2+2})=O_p(T^{-1/2}\reg^{-1/2-2/\rho}). \label{eqappadd01}
	\end{align}
	%   We thus find that
	%  \begin{equation} \label{eqpfadd02}
		%       \widehat{\Pi}_{\KK}\theta_{2,h} \to_p  {\Pi}_{\KK}\theta_{2,h}.
		%   \end{equation}
	By allocating appropriate vectors $v_j$ to zero eigenvalues of $\SA$, we may assume that $\{v_j\}_{j \geq 1}$ forms an orthonormal basis of $\mathcal H$ (see the proof of Lemma A.1 of \citealp{Seo2024}). We then note that
	\begin{equation}\label{eqpf04}
		\|{\Pi}_{\KK} \beta_h -  \beta_h\|^2   = \|\sum_{j=\KK+1}^{\infty} \langle  \beta_h ,v_j \rangle v_j\|^2 = \sum_{j=\KK+1}^{\infty} \langle  \beta_h ,v_j\rangle^2,
	\end{equation}
	where the second equality holds due to the Parseval's identity. From \eqref{eqpf04}, Assumption \ref{assum3}\ref{assum3b}, the Euler-Maclaurin summation formula for the Riemann zeta-function (see (5.6) of \citealp{ibukiyama2014euler}), we find that 
	$ \|{\Pi}_{\KK} \beta_h -  \beta_h\|^2  \leq \KK^{1-2\varsigma}$. 
	% and the fact that $\langle \theta_{2,h} , v \rangle =0$ for any $v\in [\ran \CCC_{22}]^\perp$ (Assumption \ref{assum1}). 
	Moreover, from similar arguments used in the proof in (S2.13) of \cite{seong2021functional}, we deduce that $\KK^{1-2\varsigma} \leq [(1+d_T) \reg]^{(2\varsigma-1)/\rho}$, where $d_T$ is a sequence depending on $T$  and satisfies that $(1+d_T) \reg/\reg \to_p 1$. Combining this result with \eqref{eqappadd01}, the desired result is established.  \qed
	%Note that $\sum_{j=1}^{\infty} \langle \theta_{2,h} ,v_j\rangle^2$ is convergent, implying that \eqref{eqpf04} decays to zero as $\KK$ gets larger. 
	%% (UPTO HERE)

	\paragraph{Proof of Theorem \ref{thm2}}
	As in our proof of Theorem \ref{thm1}, we let  ${v}_j$ denote $\sgn(\langle \hat{v}_j,v_j \rangle)v_j$ for notational simplicity. \\
	
	\noindent {1. Proof of  \ref{thm2i0}:} Note that $\hat{\theta}_h = \widehat{P}_{\KK} \theta_h + \widehat{C}_{\Upsilon\Upsilon,\KK}^{-1}  \widehat{C}_{\Upsilon u}$ and thus
	\begin{align*}
		\sqrt{\frac{T}{\psi_{\KK}(\zeta)}} \widehat{\Theta}_1  
		&=  \left\langle \widehat{C}_{\Upsilon\Upsilon,\KK}^{-1} \frac{1}{\sqrt{T\psi_{\KK}(\zeta)}}\sum_{t=1}^T u_{h,t} \Upsilon_t, \zeta \right\rangle.   
	\end{align*}
	From similar arguments used in our proof of Lemma \ref{lem2}, $\|\widehat{C}_{\Upsilon\Upsilon,\KK}^{-1} - {C}_{\Upsilon\Upsilon,\KK}^{-1}\|_{\op} = O_p(\|\widehat{\SA}_{\KK}^{-1}-{\SA}_{\KK}^{-1}\|_{\op}) + o_p(1)$ can be deduced. Observe that
	
	\begin{align}
		&\widehat{\SA}_{\KK}^{-1}-{\SA}_{\KK}^{-1} =  \sum_{j=1}^{\KK} (\hat{\lambda}_j^{-1} - \lambda_j^{-1}){v}_j \otimes {v}_j + \sum_{j=1}^{\KK}  \hat{\lambda}_j^{-1}(\hat{v}_j \otimes \hat{v}_j-v_j\otimes v_j), %\\ &\leq O_p(\reg^{-1/2} \lambda_{\KK}^{-1}T^{-1/2}) + O_p(\reg^{-1/2}T^{-1/2}) \sum_{j=1}^{\KK} (\lambda_j-\lambda_{j+1})^{-1} \leq O_p(\reg^{-1/2}T^{-1/2})(\KK^{\rho/2} + \sum_{j=1}^{\KK} j^{\rho/2+1} )  \\& \leq  O_p(\reg^{-1/2}T^{-1/2})(\KK^{\rho/2+2} ) = \commWK{O_p(T^{-1/2}\reg^{-1-2/\rho})=o_p(1)}
	\end{align}
	where, using the fact that $\|\widehat{\SA}-\SA\|_{\op} = O_p(T^{-1/2})$, we find that the first term is bounded above by $O_p(\reg^{-1/2} T^{-1/2}\lambda_{\KK}^{-1})$ (see (S2.5) of \citealp{seong2021functional}) and the second is bounded above by $ O_p(\reg^{-1/2}T^{-1/2} \sum_{j=1}^{\KK} (\lambda_j-\lambda_{j+1})^{-1})$  under Assumption \ref{assum2} (see Lemma 4.3 of \citealp{Bosq2000}). These are all $o_p(1)$ under the conditions given in Assumption \ref{assum3},  and hence
	\begin{equation} 
		\sqrt{\frac{T}{\psi_{\KK}(\zeta)}} \widehat{\Theta}_1  = \left\langle ({C}_{\Upsilon\Upsilon,\KK}^{-1} + o_p(1)) \frac{1}{\sqrt{T\psi_{\KK}(\zeta)}} \sum_{t=1}^T u_{h,t}  \Upsilon_t, \zeta \right\rangle. \label{eqpfa1}
	\end{equation}
	From the functional central limit theorem for $L^2$-$m$-approximable sequences (see \citealp{berkes2013weak}, Theorem 1.1) and the Skorokhod representation theorem, we know that there exists a random element $V_T$ such that  $\|T^{-1/2} \sum_{t=1}^T u_{h,t} \Upsilon_t - V_T\| \to_p 0$ such that $V_T =_d N(0,\Lambda_{\UU})$ (Gaussian random element with mean zero and covariance $\Lambda_{\UU}$) for every $T$. Thus, neglecting asymptotically negligible terms, \eqref{eqpfa1} can be written as 
	\begin{equation}
		\mathcal {V}_T + \mathcal W_T, \label{eqpfa2}
	\end{equation}
	where $\mathcal{V}_T = \langle V_T, \frac{1}{\sqrt{\psi_{\KK}(\zeta)}}{C}_{\Upsilon\Upsilon,\KK}^{-1}\zeta \rangle$ and $\mathcal W_T = \langle T^{-1/2} \sum_{t=1}^T u_{h,t} \Upsilon_t-V_T, \frac{1}{\sqrt{\psi_{\KK}(\zeta)}}{C}_{\Upsilon\Upsilon,\KK}^{-1}\zeta \rangle.$ We let $\zeta_{\KK}=\frac{1}{\sqrt{\psi_{\KK}(\zeta)}}{C}_{\Upsilon\Upsilon,\KK}^{-1}\zeta$ and note that  
	\allowdisplaybreaks{\begin{align}
			\left|\frac{\langle T^{-1/2} \sum_{t=1}^T u_{h,t} \Upsilon_t - V_T, \zeta_{\KK} \rangle}{\langle V_T,\zeta_{\KK}\rangle}\right| 	&=  \frac{|\langle T^{-1/2} \sum_{t=1}^T u_{h,t} \Upsilon_t - V_T, \frac{\zeta_{\KK}}{\|\zeta_{\KK}\|} \rangle|}{|\langle V_T,\zeta_{\KK}\rangle|/\|\zeta_{\KK}\|} \notag \\ 
			%						&\leq \frac{\sup_{\|v\|\leq 1}|\langle T^{-1/2} \sum_{t=1}^T u_{h,t} \Upsilon_t - V_T, v \rangle|}{|\langle V_T,\frac{\zeta_{\KK}}{\|\zeta_{\KK}\|}\rangle|} 
			&\leq \frac{\|T^{-1/2} \sum_{t=1}^T u_{h,t} \Upsilon_t - V_T\|}{|\langle V_T,\frac{\zeta_{\KK}}{\|\zeta_{\KK}\|}\rangle|} \leq \frac{o_p(1)}{|\langle V_T,\frac{\zeta_{\KK}}{\|\zeta_{\KK}\|}\rangle|}. \label{eqpfadd03}
	\end{align}}We observe that $\zeta_{\KK}/\|\zeta_{\KK}\| = \sum_{j=1}^{\KK} c_{\KK,j}(\zeta) \vtw_j$ and $c_{\KK,j}(\zeta)  \neq 0$ for some $j$ under Assumption \ref{assum4}\ref{assum4aa}, and also $\sum_{j=1}^{\KK} c_{\KK,j}(\zeta)^2 = 1$ for every $\KK$. From the property of $V_T$, we find that $\langle V_T,\zeta_{\KK}/\|\zeta_{\KK}\|\rangle$ is a normal random variable with mean zero, and its variance is given by $\sum_{j=1}^{\KK} \sum_{\ell=1}^{\KK} c_{\KK,j}(\zeta) c_{\KK,\ell}(\zeta) \langle \Lambda_{\UU}\vtw_j,\vtw_{\ell} \rangle$,  
	%\begin{align}
	%\sum_{j=1}^{\KK} \sum_{\ell=1}^{\KK} c_j c_{\ell} \langle \Lambda_{u\Upsilon}\vtw_j,\vtw_{\ell} \rangle. %=   \sum_{j=1}^{\KK} \sum_{\ell=1}^{\KK} \frac{\lambda_j^{-1} \lambda_\ell^{-1}\langle \zeta,v_j \rangle\langle \zeta,v_\ell \rangle}{{\sum_{j=1}^{\KK} \lambda_j^{-2}\langle \zeta,v_j \rangle^2}} \langle \Lambda_{u\Upsilon}v_j,v_{\ell} \rangle > 0
	%\end{align}
	which converges to a positive constant under Assumption \ref{assum4}\ref{assum4bb}. This implies that the latter term in \eqref{eqpfa2} is  asymptotically negligible. 
	%As long as $\zeta_{\KK} \notin \Upsilon_{u\Upsilon}$, the denominator is a normal random variable with positive variance, while the numerator is $o_p(1)$. 
	%\commWK{Thus, if the latter term of \eqref{eqpfa2} is bounded as $T$ gets larger, the first term becomes asymptotically negligible.} 
	Note that by the properties of $V_T$, we know that $\langle V_T, \zeta_{\KK} \rangle$ is normally distributed with mean zero and its variance is given by % the latter term is equal in distribution to $N(0, \langle \Lambda_{u\Upsilon} \frac{1}{\sqrt{\psi_{\KK}(\zeta)}}{C}_{\Upsilon\Upsilon,\KK}^{-1}\zeta\rangle, \frac{1}{\sqrt{\psi_{\KK}(\zeta)}}{C}_{\Upsilon\Upsilon,\KK}^{-1}\zeta)$, 
	\begin{equation*}
		%	\langle \Lambda_{\UU} \frac{1}{\sqrt{\psi_{\KK}(\zeta)}}{C}_{\Upsilon\Upsilon,\KK}^{-1}\zeta,\frac{1}{\sqrt{\psi_{\KK}(\zeta)}}{C}_{\Upsilon\Upsilon,\KK}^{-1}\zeta \rangle =  \frac{1}{{\psi_{\KK}(\zeta)}}\langle \Lambda_{\UU} {C}_{\Upsilon\Upsilon,\KK}^{-1}\zeta,  {C}_{\Upsilon\Upsilon,\KK}^{-1}\zeta \rangle = 1,
		\langle \Lambda_{\UU} \zeta_{\KK},\zeta_{\KK} \rangle =  \frac{1}{{\psi_{\KK}(\zeta)}}\langle \Lambda_{\UU} {C}_{\Upsilon\Upsilon,\KK}^{-1}\zeta,  {C}_{\Upsilon\Upsilon,\KK}^{-1}\zeta \rangle = 1,
	\end{equation*}
	which establishes the result of \ref{thm2i0}.\\ % If $\zeta \in \mathbb{R}^{k+2}$, it can be easily shown that 
	%\begin{align}
	%\widehat{\Theta}_2 &= \langle (\widehat{\CCC}_{11}^{-1}\widehat{\CCC}_{12}-{\CCC}_{11}^{-1}{\CCC}_{12})(I_2-\widehat{\Pi}_{\KK})\theta_2 + {\CCC}_{11}^{-1}{\CCC}_{12}({\Pi}_{\KK}-\widehat{\Pi}_{\KK})\theta_2, \zeta\rangle \\ 
	%&=O_p(T^{-1/2}) \sum_{j=1}^{\KK} (\lambda_j-\lambda_{j+1})^{-1}
	%\end{align}
	%$\widehat{\Theta}_2 =  \widehat{\Theta}_3 = 0$ holds, which implies that  $\langle \hat{\theta}_h -\theta_h,\zeta \rangle = \widehat{\Theta}_h$ and thus the proof is complete.  \\

	%and the fact that,  $\psi_{\KK}(\zeta) = \langle {C}_{\Upsilon\Upsilon,\KK}^{-1} \Lambda_{u \Upsilon}{C}_{\Upsilon\Upsilon,\KK}^{-1}\zeta, \zeta\rangle$, %$\frac{1}{\sqrt{T\langle \Lambda_{u \Upsilon}v, v\rangle}} \sum_{t=1}^T \langle u_{t,h}  \Upsilon_t, v \rangle \to_d N(0,1).$$ 
	%we find that (CHECK)
	%\begin{align}
	% \frac{1}{\sqrt{T\psi_{\KK}(\zeta)}} \sum_{t=1}^T \langle u_{t,h}  \Upsilon_t, {C}_{\Upsilon\Upsilon,\KK}^{-1}\zeta \rangle \to_d N(0,1)\label{eqpf06}
	%\end{align}
	%as desired.
	%\eqref{eqpf05} and \eqref{eqpf06} imply that 
	%\begin{align}
	%\sqrt{\frac{T}{\psi_{\KK}(\zeta)}} \langle \hat{\theta}_h- \widehat{\Pi}_{\KK} \theta_h, \zeta\rangle \to_d N(0,1)
	%\end{align}
	\noindent{2. Proof of \ref{thm2i1}:} 
	Observe that $\widehat{\Theta}_{2A} = \langle\widehat{P}_{\KK} \theta_h - {P}_{\KK}\theta_h,\zeta_2\rangle =  O_p(\langle (\widehat{D}-{D})\beta_h,\zeta_1\rangle) + O_p(\langle (\widehat{\Pi}_{\KK}- {\Pi}_{\KK})\beta_h,\zeta_2\rangle) $, where $\widehat{D} = \widehat{\CCC}_{11}^{-1}\widehat{\CCC}_{12}(I_2-\widehat{\Pi}_{\KK})$ and ${D} = {\CCC}_{11}^{-1}{\CCC}_{12}(I_2-{\Pi}_{\KK})$. We find that 
	\begin{align}
		&\langle(\widehat{\Pi}_{\KK} - {\Pi}_{\KK})\beta_h,\zeta_2\rangle =  \sum_{j=1}^{\KK} \langle \hat{v}_j,\beta_h\rangle \langle \hat{v}_j,\zeta_2\rangle-  \sum_{j=1}^{\KK} \langle {v}_j,\beta_h\rangle  \langle {v}_j,\zeta_2\rangle  \notag \\ 
		%&=   \sum_{j=1}^{\KK} \langle \hat{v}_j-v_j,\theta_{2,h}\rangle \langle \hat{v}_j,\zeta\rangle +  \sum_{j=1}^{\KK} \langle {v}_j,\theta_{2,h}\rangle  \langle \hat{v}_j-{v}_j,\zeta\rangle \notag\\
		&=   \sum_{j=1}^{\KK} \langle \hat{v}_j-v_j,\beta_h\rangle \langle {v}_j,\zeta_2\rangle +  \sum_{j=1}^{\KK} \langle {v}_j,\beta_h\rangle  \langle \hat{v}_j-{v}_j,\zeta_2\rangle +   \sum_{j=1}^{\KK} \langle \hat{v}_j-v_j,\beta_h\rangle \langle \hat{v}_j-v_j,\zeta_2\rangle.  \label{eqpf07}
	\end{align}
	We write $\langle	(\hat{\SA} - \SA)v_j,v_{\ell}\rangle = F_1 + F_2$, where 
	\begin{align*}
		F_1 &= \frac{1}{T}\sum_{t=1}^T (\langle X_t,v_j \rangle  \langle X_t-\CCC_{21}\CCC_{11}^{-1}\mathbf{w}_t,v_{\ell} \rangle-\mathbb{E}[\langle X_t,v_j \rangle \langle X_t-\CCC_{21}\CCC_{11}^{-1}\mathbf{w}_t,v_{\ell} \rangle]),  \\
		F_2 &= \frac{1}{T}\sum_{t=1}^T \langle X_t,v_j \rangle \langle  ({\CCC}_{12}{\CCC}_{11}^{-1}-\hat{\CCC}_{12}\hat{\CCC}_{11}^{-1})\mathbf{w}_t, v_{\ell} \rangle . 
	\end{align*}
	%\commWK
	{Note that $F_2$ can be rewritten as $F_2 = \langle \hat{\CCC}_{12}v_j,  ({\CCC}_{11}^{-1}{\CCC}_{12}-\hat{\CCC}_{11}^{-1}\hat{\CCC}_{12})v_{\ell}\rangle.$  We observe that, under Assumption \ref{assum5}, $\mathbb{E}[\langle X_t,v_{\ell} \rangle \langle \mathbf{w}_t, v \rangle] = \mathbb{E}[\langle X_t,v_{\ell} \rangle^2]^{1/2} O(1) = O(\lambda_{\ell}^{1/2})$ for any arbitrary $v \in \mathbb{R}^m$, and this implies that $\CCC_{12}v_{\ell} = O(\lambda_{\ell}^{1/2})$.
		From Lemma 2.1 of \cite{hormann2010} (and the subsequent discussion on non-independent pair of $L^p$-$m$-approximable sequences for $p \geq 1$), we conclude that $\{\langle X_t, v_{\ell} \rangle \mathbf{w}_t\}$ is $L^2$-$m$-approximable. Furthermore, from the central limit theorem for $L^2$-$m$-approximable sequences (see e.g., \citealp{berkes2013weak}) and Assumption \ref{assum5}\ref{assum5a}, we obtain  $\sqrt{T}(\hat{\CCC}_{12}-\CCC_{12})v_\ell = O_p(\lambda_{\ell}^{1/2})$. Since $\hat{\CCC}_{11}^{-1}-{\CCC}_{11}^{-1} = O_p(T^{-1/2})$, we therefore find that
		$({\CCC}_{11}^{-1}{\CCC}_{12}-\hat{\CCC}_{11}^{-1}\hat{\CCC}_{12})v_{\ell}=  ({\CCC}_{11}^{-1}-\widehat{\CCC}_{11}^{-1}) \CCC_{21}v_{\ell} + \widehat{\CCC}_{11}^{-1}( {\CCC}_{12}-\hat{\CCC}_{12})v_{\ell}  = O_p(T^{-1/2} \lambda_{\ell}^{1/2})$ and also $\hat{\CCC}_{12}v_j = 
		(\hat{\CCC}_{12}-\CCC_{12})v_j +\CCC_{12}v_j = O(T^{-1/2}\lambda_j^{1/2}) + O(\lambda_j^{1/2}) = O({\lambda}^{1/2}_j)$. We thus conclude that $F_2 = O(\lambda_j^{1/2}\lambda_{\ell}^{1/2}T^{-1/2})$.  }
	
	%	From Assumption \ref{assum5}\ref{assum5a}, the $L^p$-$m$-approximability of $\{X_t\}$ and the fact that $\hat{\CCC}_{21}\hat{\CCC}_{11}^{-1}-{\CCC}_{21}{\CCC}_{11}^{-1} = O_p(T^{-1/2})$, we find that $F_2 = O_p(\lambda_j^{1/2} T^{-1})$. 
	Moreover, from the arguments used in \cite{seong2021functional} (for the unnumbered equation between (S2.28) and (S2.29)) and Assumption \ref{assum5}, the following can be shown: 
	\begin{equation*}
		T\mathbb{E}[F_1^2] \leq \sum_{s=0}^T\mathbb{E}[r_t(j,\ell)r_{t-s}(j,\ell)]\leq O(1) \mathbb{E}[\langle X_t,v_j \rangle^2\langle X_{w,t},v_{\ell} \rangle^2] = O(\lambda_j\lambda_{\ell}), 
	\end{equation*}
	where the last equality follows from the Cauchy-Schwarz inequality and Assumption \ref{assum5}\ref{assum5a}. Therefore, we deduce that $F_1 = O_p(T^{-1/2})\lambda_j^{1/2}\lambda_{\ell}^{1/2}$. From these results, it is concluded  that $\langle	(\hat{\SA} - \SA)v_j,v_{\ell}\rangle = O_p(T^{-1/2})\lambda_j^{1/2}\lambda_{\ell}^{1/2}$. 
	Combining this result with the employed assumptions and Lemma S1 of \cite{seong2021functional}, $\|\hat{v}_j-v_j\| = O_p(T^{-1/2}j^2)$ can be shown. In turn, the following can also be shown from nearly identical arguments used in the proofs of (S2.16) and (S2.33) of \cite{seong2021functional}:
	\begin{align}
		&\langle \hat{v}_j-v_j,\beta_h\rangle = O_p(T^{-1/2})j^{\rho/2+1-\varsigma}, \label{eqpf07a}\\
		&\langle \hat{v}_j-v_j,\zeta_2\rangle = O_p(T^{-1/2}) j^{ \rho/2 + 1-\delta}.\label{eqpf07b}
	\end{align}
	%\begin{align} \label{eqpf08}
	%\KK \leq (1+o_p(1)) a^{-1/\rho}.
	%\end{align}
	Note that $\rho/2 + 2 < \varsigma+ \delta$ and $\rho > 2$, and thus $\varsigma + \delta > 3$. From \eqref{eqpf07}, \eqref{eqpf07a}, \eqref{eqpf07b} and the conditions given in Assumption \ref{assum5}\ref{assum5b}, we find that 
	\begin{align}
		&\sqrt{\frac{T}{\psi_{\KK}(\zeta)}}\langle(\widehat{\Pi}_{\KK} - {\Pi}_{\KK})\beta_h,\zeta_2\rangle = O_p\left(\frac{1}{\sqrt{\psi_{\KK}(\zeta)}}\right) \left\{ \sum_{j=1}^{\KK} j^{\rho/2+1-\varsigma-\delta}+ \frac{1}{\sqrt{T}}\sum_{j=1}^{\KK} j^{\rho+2-\varsigma-\delta}\right\} \notag  \\ 
		&= O_p\left(\frac{1}{\sqrt{\psi_{\KK}(\zeta)}} \right) \left(1+ T^{-1/2}\KK^{\rho+3-\varsigma-\delta}\right) =  O_p\left(\frac{1}{\sqrt{\psi_{\KK}(\zeta)}} \right) \left(1+ T^{-1/2}\reg^{-1/2} \KK^{\rho/2+3-\varsigma-\delta}\right) \notag \\
		&=  O_p\left(\frac{1}{\sqrt{\psi_{\KK}(\zeta)}} \right) \left(1+\KK^{1-\varsigma-\delta}\right) =  O_p\left(\frac{1}{\sqrt{\psi_{\KK}(\zeta)}} \right).  \label{eqthmpf001}
	\end{align} 
	It only remains to show that $\sqrt{T}\langle (\widehat{D}-{D})\beta_h,\zeta_1\rangle=O_p({1}/{\sqrt{\psi_{\KK}(\zeta)}})$ for the desired result. We note that 	\begin{equation}
		(\langle (\widehat{D}-{D})\beta_h,\zeta_1\rangle)  = \langle (\widehat{\CCC}_{11}^{-1}\widehat{\CCC}_{12}-{\CCC}_{11}^{-1}{\CCC}_{12})(I_2-\widehat{\Pi}_{\KK})\beta_h ,\zeta_1\rangle -  \langle{\CCC}_{11}^{-1}{\CCC}_{12}(\widehat{\Pi}_{\KK} -\Pi_{\KK})\beta_h,\zeta_1 \rangle. \label{eqadd0070}
	\end{equation}
	The second term of \eqref{eqadd0070}, $\langle{\CCC}_{11}^{-1}{\CCC}_{12}(\widehat{\Pi}_{\KK} -\Pi_{\KK})\beta_h,\zeta_1 \rangle$, is equal to 
	\begin{align}
		%&=  \sum_{j=1}^{\KK} \langle \hat{v}_j,\beta_h\rangle \langle {\CCC}_{11}^{-1}{\CCC}_{12} \hat{v}_j,\zeta_1\rangle-  \sum_{j=1}^{\KK} \langle {v}_j,\beta_h\rangle  \langle {\CCC}_{11}^{-1}{\CCC}_{12}{v}_j,\zeta_1\rangle  \notag \\ 
		%&=   \sum_{j=1}^{\KK} \langle \hat{v}_j-v_j,\theta_{2,h}\rangle \langle \hat{v}_j,\zeta\rangle +  \sum_{j=1}^{\KK} \langle {v}_j,\theta_{2,h}\rangle  \langle \hat{v}_j-{v}_j,\zeta\rangle \notag\\
		& \sum_{j=1}^{\KK} \langle \hat{v}_j-v_j,\beta_h\rangle \langle \CCC_{11}^{-1}\CCC_{12}{v}_j,\zeta_2\rangle +  \sum_{j=1}^{\KK} \langle {v}_j,\beta_h\rangle  \langle  \CCC_{11}^{-1}\CCC_{12}(\hat{v}_j-{v}_j),\zeta_2\rangle \notag \\ & +   \sum_{j=1}^{\KK} \langle \hat{v}_j-v_j,\beta_h\rangle \langle  \CCC_{11}^{-1}\CCC_{12}(\hat{v}_j-v_j),\zeta_1\rangle. \label{eqpfadd07}
	\end{align}
	Under Assumptions \ref{assum1}-\ref{assum5}, the following holds:
	\begin{equation}
		%	&\langle \hat{v}_j-v_j,\theta_{2,h}\rangle = O_p(T^{-1/2})j^{\rho/2+1-\varsigma}, \\
		%	&\langle \hat{v}_j-v_j,\zeta\rangle = O_p(T^{-1/2}) j^{ \rho/2 + 1-\delta}\\
		\|\CCC_{11}^{-1}\CCC_{12} (\hat{v}_j-v_j)\| = O_p(T^{-1/2})j^{\rho/2 +1 - \delta}, \label{eqpfadd08}
	\end{equation} 
	which is deduced from nearly identical arguments used to show (S2.29) of \cite{seong2021functional}. From the expression of $\langle{\CCC}_{11}^{-1}{\CCC}_{12}(\widehat{\Pi}_{\KK} -\Pi_{\KK})\beta_h,\zeta_1 \rangle$  in \eqref{eqpfadd07}, Assumption \ref{assum5}\ref{assum5b}, and the results given by \eqref{eqpf07a} and \eqref{eqpfadd07}, the following is deduced exactly as in \eqref{eqthmpf001}:
	\begin{align*}
		&\sqrt{\frac{T}{\psi_{\KK}(\zeta)}}\langle{\CCC}_{11}^{-1}{\CCC}_{12}(\widehat{\Pi}_{\KK} -\Pi_{\KK})\beta_h,\zeta_1 \rangle 
		%= O_p\left(\frac{1}{\sqrt{\psi_{\KK}(\zeta)}}\right) \left\{ \sum_{j=1}^{\KK} j^{\rho/2+1-\varsigma-\delta}+ \frac{1}{\sqrt{T}}\sum_{j=1}^{\KK} j^{\rho+2-\varsigma-\delta}\right\} \notag  \\ 
		%	&= O_p\left(\frac{1}{\sqrt{\psi_{\KK}(\zeta)}} \right) \left(1+ T^{-1/2}\KK^{\rho+3-\varsigma-\delta}\right) =  O_p\left(\frac{1}{\sqrt{\psi_{\KK}(\zeta)}} \right) \left(1+ T^{-1/2}\reg^{-1/2} \KK^{\rho/2+3-\varsigma-\delta}\right) \notag \\
		=  O_p\left(\frac{1}{\sqrt{\psi_{\KK}(\zeta)}} \right). %\left(1+\KK^{1-\varsigma-\delta}\right) =  O_p\left(\frac{1}{\sqrt{\psi_{\KK}(\zeta)}} \right).  \label{eqthmpf001add}
	\end{align*} 
	Since $\widehat{\CCC}_{11}^{-1}\widehat{\CCC}_{12}-{\CCC}_{11}^{-1}{\CCC}_{12} = O_p(T^{-1/2})$ and $(I_2-\widehat{\Pi}_{\KK})\beta_h = O_p(1)$, the first term of \eqref{eqadd0070} obviously decays faster than the second term, thus  we conclude that $\sqrt{T}\langle (\widehat{D}-{D})\beta_h,\zeta_1\rangle=O_p({1}/{\sqrt{\psi_{\KK}(\zeta)}})$, which completes the proof. 				 \\
	
	\noindent{3. Proof of \ref{thm2i2}:} 
	We now consider $\widehat{\Theta}_{2B}$. We first observe that 
	$\widehat{\Theta}_{2B} = O_p(\langle {\CCC}_{11}^{-1}{\CCC}_{12}(I_2-\Pi_{\KK})\beta_h, \zeta_1 \rangle) + O_p(\langle (I_2-\Pi_{\KK})\beta_h, \zeta_2 \rangle)$. If we show that $\langle {\CCC}_{11}^{-1}{\CCC}_{12}(I_2-\Pi_{\KK})\beta_h, \zeta_1 \rangle$ and $\sqrt{T}\langle (I_2-\Pi_{\KK}) \beta_h, \zeta_2 \rangle$ are $O_p(1/\sqrt{\psi_{\KK}(\zeta)})$, then the desired result is established. Note that
	\begin{equation*}\label{eqpf08}
		%\sqrt{\frac{T}{\psi_{\KK}(\zeta)}}\sum_{j=\KK+1}^\infty \langle \theta, v_j \rangle\langle \zeta, v_j \rangle 
		\sqrt{\frac{T}{\psi_{\KK}(\zeta)}}\langle (I_2-\Pi_{\KK})  \beta_h, \zeta_2 \rangle=\sqrt{\frac{T}{\psi_{\KK}(\zeta)}}\sum_{j=\KK+1}^\infty \langle  \beta_h, v_j \rangle\langle \zeta_2, v_j \rangle 
		\leq O\left(\frac{1}{\sqrt{\psi_{\KK}(\zeta)}}\right) T^{1/2}\KK^{1-\delta - \varsigma},
	\end{equation*}
	where the last inequality follows from the Euler-Maclaurin summation formula for
	the Riemann zeta-function (e.g., (5.6) of \citealp{ibukiyama2014euler}). Since $\reg^{1/2}\KK^{\rho/2} \leq O_p(1)$,  
	\begin{equation} \label{eqpf09}
		\sqrt{\frac{T}{\psi_{\KK}(\zeta)}}\langle (I_2-\Pi_{\KK}) \beta_h,\zeta_2 \rangle \leq \frac{1}{\sqrt{\psi_{\KK}(\zeta)}} O_p(T^{1/2}\reg^{(\delta + \varsigma - 1)/\rho}).
	\end{equation}
	If $\delta$ and $\varsigma$ are large enough so that $T^{1/2}\reg^{(\delta + \varsigma - 1)/\rho} = O_p(1)$, the above is $O_p(1/\sqrt{\psi_{\KK}(\zeta)})$. 
	Similarly, we find that  
	\begin{align}\label{eqpf08add}
		\sqrt{\frac{T}{\psi_{\KK}(\zeta)}}\langle {\CCC}_{11}^{-1}{\CCC}_{12}(I_2-\Pi_{\KK})  \beta_h, \zeta_1 \rangle %&=\sqrt{\frac{T}{\psi_{\KK}(\zeta)}}\sum_{j=\KK+1}^\infty \langle  \beta_h, v_j \rangle\langle \zeta_1, {\CCC}_{11}^{-1}{\CCC}_{12} v_j \rangle 
		% \leq O\left(\frac{1}{\sqrt{\psi_{\KK}(\zeta)}}\right) T^{1/2}\KK^{1-\eta - \varsigma},
		\leq \frac{1}{\sqrt{\psi_{\KK}(\zeta)}}  O_p(T^{1/2}\reg^{(\delta + \varsigma - 1)/\rho}),
	\end{align}
	which is also $O_p(1/\sqrt{\psi_{\KK}(\zeta)})$.
	\\
	
	\noindent{4. Proof of the asymptotic equivalence between the choices of $\psi_{\KK}$ and $\widehat{\psi}_{\KK}$:} Under Assumptions \ref{assum1}-\ref{assum4}, we find that  $\hat{\theta}_h$ is consistent and    $\|\widehat{\SA}_{\KK}^{-1} - {\SA}_{\KK}^{-1}\|_{\op} \to_p 0$. Note that $\hat{u}_t = u_t + \hat{\delta}_t$, where $\hat{\delta}_t = \langle \theta_h-\hat{\theta}_h, \Upsilon_t \rangle$ and thus  $\hat{\delta}_t = o_p(1)$. We observe  that 
	\begin{equation*} \label{eqdelta}
		\sup_{1\leq t\leq T} \|\hat{\delta}_t \Upsilon_t\| %\leq \sup_{1\leq t\leq T}   |\langle \theta_h-\hat{\theta}_h, \Upsilon_t/\|\Upsilon_t\| \rangle| \|\Upsilon_t\|^2 
		\leq \| \theta_h-\hat{\theta}_h\| (\sup_{1\leq t\leq T} \|\Upsilon_t\|)^2 = O_p( \| \theta_h-\hat{\theta}_h\|). 
	\end{equation*}
	%where the  inequality follows from that $\sup_{\|x\|\leq 1} |\langle \theta_h-\hat{\theta}_h,x \rangle| = \| \theta_h-\hat{\theta}_h\|$. 
	Moreover, since $\{{u}_{h,t}\Upsilon_t\}$ is an $L^2$-$m$-approximable sequence,  %(see \eqref{lpmcon1}-\eqref{lpmcon2}  in Section \ref{Section_AFTS}), 
	we find that  $\|\widehat{\Lambda}_{\UU}^0-\Lambda_{\UU}\|_{\op} \to_p 0$ (\citealp{horvath2013estimation}, Theorem 1), where $\widehat{\Lambda}_{\UU}^0 = T^{-1} \sum_{s=-\bdw}^{\bdw}\mathrm{k}(s/\bdw) \sum_{t=|s|+1}^T  {u}_{h,t}\Upsilon_t \otimes {u}_{h,t-s}\Upsilon_{t-s}$. To obtain the desired result, it suffices to show that $\|\widehat{\Lambda}_{\UU}-\widehat{\Lambda}_{\UU}^0\| = o_p(1)$. To see this, observe that 
	\begin{align*}\label{eqdelta1}
		&\hat{u}_{h,t}\Upsilon_t \otimes \hat{u}_{h,t-s}\Upsilon_{t-s} - {u}_{h,t}\Upsilon_t \otimes {u}_{h,t-s}\Upsilon_{t-s}\\
		%\\&= (\hat{u}_{h,t}-u_{h,t})\Upsilon_t \otimes \hat{u}_{h,t-s}\Upsilon_{t-s} +  {u}_{h,t}\Upsilon_t \otimes (\hat{u}_{h,t-s}-{u}_{h,t-s})\Upsilon_{t-s}  \\ 
		%\\&= \hat{\delta}_t \Upsilon_t \otimes \hat{u}_{h,t-s}\Upsilon_{t-s} +  {u}_{h,t}\Upsilon_t \otimes \hat{\delta}_{t-s}\Upsilon_{t-s} \\ 
		&=\hat{\delta}_t \Upsilon_t \otimes {u}_{h,t-s}\Upsilon_{t-s} + \hat{\delta}_t \Upsilon_t \otimes \hat{\delta}_{t-s}\Upsilon_{t-s}  +  {u}_{h,t}\Upsilon_t \otimes \hat{\delta}_{t-s}\Upsilon_{t-s}. 
		%&\leq \|\hat{\delta}_t \Upsilon_t \otimes {u}_{h,t-s}\Upsilon_{t-s}\|_{\op} + \|\hat{\delta}_t \Upsilon_t \otimes \hat{\delta}_{t-s}\Upsilon_{t-s}\|_{\op}   +  \|{u}_{h,t}\Upsilon_t \otimes \hat{\delta}_{t-s}\Upsilon_{t-s}\|_{\op} \\
		% &\leq |\hat{\delta}_t| \|\Upsilon_t\| \|{u}_{h,t-s}\Upsilon_{t-s}\|+ |\hat{\delta}_t| \|\Upsilon_t\| |\hat{\delta}_{t-s}|\|\Upsilon_{t-s}\|   +  \|{u}_{h,t}\Upsilon_t\| | \hat{\delta}_{t-s}| \|\Upsilon_{t-s}\|
		%&= (\hat{u}_{h,t}-u_{h,t})\Upsilon_t \otimes \hat{u}_{h,t-s}\Upsilon_{t-s} +  {u}_{h,t}\Upsilon_t \otimes (\hat{u}_{h,t-s}-{u}_{h,t-s})\Upsilon_{t-s} 
	\end{align*} 
	Note that 
	\begin{align}\nonumber
		&\|T^{-1} \sum_{s=-\bdw}^\bdw \mathrm{k}\left(\frac{s}{\bdw}\right)  \sum_{t=|s|+1}^T \hat{\delta}_t \Upsilon_t \otimes {u}_{h,t-s}\Upsilon_{t-s}\|_{\op} \leq  T^{-1} \sum_{s=-\bdw}^\bdw \mathrm{k}\left(\frac{s}{\bdw}\right) \sum_{t=|s|+1}^T \| \hat{\delta}_t \Upsilon_t\| \|{u}_{h,t-s}\Upsilon_{t-s}\|\\ \nonumber
		&\leq  \sum_{s=-\bdw}^\bdw \mathrm{k}\left(\frac{s}{\bdw}\right) \sup_{1\leq t\leq T}\|\hat{\delta}_t \Upsilon_t\| \frac{1}{T} \sum_{t=1}^T \|u_t\Upsilon_t\| = \left(\sum_{s=-\bdw}^\bdw \mathrm{k}\left(\frac{s}{\bdw}\right)\right)  \|\hat{\theta}_h-\theta_h\| O_p(1)  = o_p(1).
		% \|T^{-1} \sum_{s=-h}^hw(s) \sum_{t=|s|+1}^T \langle \hat{\delta}_t \Upsilon_t / \|\hat{\delta}_t \Upsilon_t \|, \cdot \rangle {u}_{h,t-s}\Upsilon_{t-s}\|_{\op} \|\hat{\delta}_t \Upsilon_t\|\\
		%&= \sup_{\|h\|\leq 1}  \|T^{-1} \sum_{s=-h}^hw(s) \sum_{t=|s|+1}^T \langle \hat{\delta}_t \Upsilon_t / \|\hat{\delta}_t \Upsilon_t \|, h \rangle {u}_{h,t-s}\Upsilon_{t-s}\| \|\hat{\delta}_t \Upsilon_t\|  \\
		&% \leq \sup_{\|\tilde{h}_t\|\leq 1}\sup_{\|h\|\leq 1}  \|T^{-1} \sum_{s=-h}^hw(s) \sum_{t=|s|+1}^T \langle \tilde{h}, h \rangle {u}_{h,t-s}\Upsilon_{t-s}\| \|\hat{\delta}_t \Upsilon_t\| 
	\end{align}
	From similar arguments, we also find that $\|T^{-1} \sum_{s=-\bdw}^\bdw \mathrm{k}\left(\frac{s}{\bdw}\right)  \sum_{t=|s|+1}^T {u}_{h,t}\Upsilon_{t} \otimes\hat{\delta}_{t-s} \Upsilon_{t-s} \|_{\op} = o_p(1)$. Lastly, note that
	\begin{align*}
		&\|T^{-1} \sum_{s=-\bdw}^\bdw \mathrm{k}\left(\frac{s}{\bdw}\right)  \sum_{t=|s|+1}^T \hat{\delta}_t \Upsilon_t \otimes  \hat{\delta}_{t-s} \Upsilon_{t-s}\|_{\op} \leq  T^{-1} \sum_{s=-\bdw}^\bdw \mathrm{k}\left(\frac{s}{\bdw}\right) \sum_{t=|s|+1}^T \| \hat{\delta}_t \Upsilon_t\| \|\hat{\delta}_{t-s} \Upsilon_{t-s}\|\\
		&\leq  \left(\sum_{s=-\bdw}^\bdw \mathrm{k}\left(\frac{s}{\bdw}\right) \right)  O_p(\|\hat{\theta}_h-\theta_h\|^2) = o_p(1).
		% \|T^{-1} \sum_{s=-h}^hw(s) \sum_{t=|s|+1}^T \langle \hat{\delta}_t \Upsilon_t / \|\hat{\delta}_t \Upsilon_t \|, \cdot \rangle {u}_{h,t-s}\Upsilon_{t-s}\|_{\op} \|\hat{\delta}_t \Upsilon_t\|\\
		%&= \sup_{\|h\|\leq 1}  \|T^{-1} \sum_{s=-h}^hw(s) \sum_{t=|s|+1}^T \langle \hat{\delta}_t \Upsilon_t / \|\hat{\delta}_t \Upsilon_t \|, h \rangle {u}_{h,t-s}\Upsilon_{t-s}\| \|\hat{\delta}_t \Upsilon_t\|  \\
		&% \leq \sup_{\|\tilde{h}_t\|\leq 1}\sup_{\|h\|\leq 1}  \|T^{-1} \sum_{s=-h}^hw(s) \sum_{t=|s|+1}^T \langle \tilde{h}, h \rangle {u}_{h,t-s}\Upsilon_{t-s}\| \|\hat{\delta}_t \Upsilon_t\| 
	\end{align*}
	From these results, $\|\widehat{\Lambda}_{\UU}-\widehat{\Lambda}_{\UU}^0\|_{\op} = o_p(1)$ is established.  							\qed
	%From \eqref{eqdelta}, \eqref{eqdelta1}, and the fact that $\sup_{1\leq t\leq T}T^{-1/2}\sum_{t=1}^T u_t\Upsilon_t = O_p(1)$ (\citealp{berkes2013weak}, Theorem 1.1) and 
	%\begin{align}
	%\widehat{\Lambda}_{u\ZZ}-\widehat{\Lambda}_{u\Upsilon}^0 = T^{-1} \sum_{s=-h}^hw(s) \sum_{t=|s|+1}^T  ({u}_{h,t}\Upsilon_t \otimes {u}_{h,t-s}\Upsilon_{t-s}-\hat{u}_{h,t}\Upsilon_t \otimes \hat{u}_{h,t-s}\Upsilon_{t-s}) = O_p(T^{-1/2} h) = o_p(1).
	%\end{align}
	
	% Combining these results with the assumption that $T^{-1/2}h\to 0$, we find that % of along with Lemma A.1 of \cite{horvath2014test})
	%\begin{equation}
	%    \widehat{\Lambda}_{u\Upsilon} =  \frac{1}{T} \sum_{s=-h}^hw(s) \sum_{t=|s|+1}^T  {u}_{h,t}\Upsilon_t \otimes {u}_{h,t-s}\Upsilon_{t-s},
	%    \end{equation}
%$\|\widehat{\Lambda}_{u\Upsilon} -{\Lambda}_{u\Upsilon}\|_{\op}\to_p 0$.\qed %For the latter result, refer to Theorem 1 of \cite{horvath2013estimation} along with Lemma A.1 of \cite{horvath2014test}. From these results, the desired conclusion follows.\qed 
\paragraph{Proof of Corollary \ref{coradd}}
$\sqrt{{T}/{\psi_{\KK}(\zeta)}} \widehat{\Theta}_1 \to_d N(0,1)$ can be shown from the same arguments used in our proof of Theorem \ref{thm2}. 
If $\zeta_2 = 0$, we find that 
\begin{equation*}
	\widehat{\Theta}_{2A}+\widehat{\Theta}_{2B} = \langle    \hat P_{\KK} \theta_h - \theta_h , \zeta\rangle %= \langle \widehat\CCC_{11} ^{-1}\widehat\CCC_{12} (I_2 - \widehat\Pi_{\KK}) \beta_h, \zeta_1 \rangle 
	=					F_1+F_2,
\end{equation*}	
where $F_1 = \langle (\widehat\CCC_{11} ^{-1}\widehat\CCC_{12} - \CCC_{11} ^{-1}\CCC_{12} ) (I_2 - \widehat\Pi_{\KK}) \beta_h, \zeta_1 \rangle$ and $F_2=	\langle \CCC_{11} ^{-1}\CCC_{12} (I_2 - \widehat\Pi_{\KK}) \beta_h, \zeta_1 \rangle$.
We first note that  
\begin{equation*}
	|  F_1|  \leq \Vert \widehat\CCC_{11} ^{-1}\widehat\CCC_{12} - \CCC_{11} ^{-1}\CCC_{12}\Vert_{\op} \Vert \zeta_1 \Vert \Vert \sum_{j= \KK+1} ^\infty \langle  \beta_h , \widehat{v}_j \rangle \widehat v_j \Vert. %\nonumber\\							&\leq O_p(T^{-1/2}) (\sum_{j= K_{\alpha}+1} ^\infty \langle \theta_{2,h} , \widehat{v}_j \rangle^2 )^{1/2} = o_p(T^{-1/2}),
\end{equation*}
Since $\Vert \widehat\CCC_{11} ^{-1}\widehat\CCC_{12} - \CCC_{11} ^{-1}\CCC_{12}\Vert_{\op} = O_p(T^{-1/2})$ under Assumption \ref{assum2}\ref{assum2c} and $\Vert \sum_{j= \KK+1} ^\infty \langle  \beta_h , \widehat{v}_j \rangle \widehat v_j \Vert \to_p 0$ as $\KK \to \infty$ and $T \to \infty$ (because $\| \beta_h\| < \infty$), the above term is $o_p(T^{-1/2})$. We write
\begin{equation*}
	F_2 = F_3+F_4+F_5,%& = \sum_{j= \KK+1} ^\infty \langle \widehat v_j , \theta_{2,h}\rangle \langle \CCC_{11}  ^{-1}\CCC_{12} \widehat v_j, \zeta_1\rangle \nonumber\\
	%	& =   \sum_{j= \KK+1} ^\infty \langle \widehat v_j  -v_j , \theta_{2,h}\rangle \langle \CCC_{11} ^{-1}\CCC_{12} \widehat v_j, \zeta_1\rangle  \nonumber\\
	%	&\quad+   \sum_{j= \KK+1} ^\infty \langle v_j , \theta_{2,h}\rangle \langle \CCC_{11} ^{-1}\CCC_{12} ( \widehat v_j - v_j ), \zeta_1\rangle  \nonumber\\
	%	& \quad +  \sum_{j= \KK+1} ^\infty \langle v_j , \theta_{2,h}\rangle \langle \CCC_{11} ^{-1}\CCC_{12}   v_j  , \zeta_1 \rangle 
\end{equation*}
where $F_3=\sum_{j= \KK+1} ^\infty \langle \widehat v_j  -v_j , \beta_{h}\rangle \langle \CCC_{11} ^{-1}\CCC_{12} \widehat v_j, \zeta_1\rangle $, $F_4=\sum_{j= \KK+1} ^\infty \langle v_j , \beta_{h}\rangle \langle \CCC_{11} ^{-1}\CCC_{12} ( \widehat v_j - v_j ), \zeta_1\rangle$, and $F_5= \sum_{j= \KK+1} ^\infty \langle v_j , \beta_{h}\rangle \langle \CCC_{11} ^{-1}\CCC_{12}   v_j  , \zeta_1 \rangle $. 
Observe that
\begin{align*}
	F_3 &=   \sum_{j= \KK+1} ^\infty \langle \widehat v_j  - v_j , \beta_h\rangle \langle \CCC_{11} ^{-1}\CCC_{12} (\widehat v_j -v_j), \zeta_1\rangle  +   \sum_{j= \KK+1} ^\infty \langle \widehat v_j  -v_j , \beta_h\rangle \langle \CCC_{11} ^{-1}\CCC_{12} v_j, \zeta_1\rangle\notag  \\
	&= O_p(T^{-1} \KK^{\rho +3 -\varsigma - \delta} + T^{-1/2}\KK^{\rho/2 + 2  -\varsigma-\delta}),
\end{align*}
where the last equality follows from \eqref{eqpf07a}, \eqref{eqpfadd08} and the Euler-Maclaurin summation formula. 
Using Assumption \ref{assum5}\ref{assum5b}, we similarly find that  
\begin{equation*}
	F_4 = O_p(T^{-1/2}\KK^{\rho/2 + 2  -\varsigma-\delta}), \quad 
	F_5 =  O_p(\KK^{1-\varsigma - \delta}).
\end{equation*} 
Under the employed assumption, we have $\rho/2 + 2 < \varsigma + \delta$, and thus  $\KK^{\rho/2 + 2  -\varsigma-\delta} = o_p(1)$ and $T^{-1/2} \KK^{\rho +3 -\varsigma - \delta}  =o_p(1)$. Moreover, $T^{1/2}\KK^{1  -\varsigma-\delta} = O_p(T^{1/2} \reg^{(\varsigma + \delta - 1)/\rho}) = o_p(1)$ under the conditions on the decay rate of $\reg$. These results show that $F_2 = F_3+F_4+F_5 = o_p(T^{-1/2})$, and hence $	\sqrt{T}(\widehat{\Theta}_{2A}+\widehat{\Theta}_{2B}) \to_p 0$.
\qed 					

\subsection{Proofs of the results in Section \ref{sec:est2}}
\paragraph{Proof of Theorem \ref{thm3}} %\commWK{(NEED TO BE POLISHED)}
%\begin{equation} 
%\widehat{C}_{\XX\ZZ,\KKK}^{-1}\widehat{C}_{\XX\ZZ}  = \begin{bmatrix} 
	%I_1 & \widehat{\DDD}_{11}^{-1} \widehat{\DDD}_{12} %(I-\widehat{\PPI}_{\KKK})  \\ 0 &  \widehat{\PPI}_{\KKK}  
	%\end{bmatrix}
	%\end{equation}
	As in our proof of Theorem \ref{thm1}, we find that $\tilde{\theta}_h=\widehat{\PP}_{\KKK}\theta_h + \widehat{C}_{\ZZ\XX,\KKK}^{-1}  \widehat{C}_{ \ZZ u}$ and $\widehat{C}_{\ZZ\XX,\KKK}^{-1}  \widehat{C}_{ \ZZ u} = O_p(\reg^{-1/2}T^{-1/2})$, from which the following is deduced:  
	\begin{equation*}
		\tilde{\theta}_{h} - \theta_{h} =  	\begin{bmatrix}
			\widehat{\DDD}_{11}^{-1}\widehat{\DDD}_{12}(I_2-{\PPI}_{\KK}) \beta_h + \widehat{\DDD}_{11}^{-1}\widehat{\DDD}_{12}({\PPI}_{\KK}-\widehat{\PPI}_{\KK})  \beta_h \\
			- (I_2-{\PPI}_{\KK})  \beta_h - ({\PPI}_{\KK}-\hat{\PPI}_{\KK}) \beta_h
		\end{bmatrix}  +  O_p(\reg^{-1/2}T^{-1/2}),
	\end{equation*}
	where $\widehat{\PPI}_{\KKK} = \sum_{j=1}^{\KKK} \hat{\wv}_j \otimes \hat{\wv}_j$. 
	As in \eqref{eqpfaddarguments}, we find that $\tilde{\theta}_h - \theta_h = O_p((\widehat{\PPI}_{\KKK}-{\PPI}_{\KKK})\beta_h) + O_p((I_2-{\PPI}_{\KKK})\beta_h) + O_p(\reg^{-1/2}T^{-1/2})$.
	%From nearly identical arguments used in the proof of Theorem 3 of \cite{seong2021functional} (see %particularly (S2.28) and the following arguments), we find that 
	%\begin{align}
	%\|\hat{\vv}_j-\vv_j\|^2= O_p(1) \sum_{\ell \neq j} (\lambda_j^2-\lambda_{\ell})^{-2} %\lambda_{\ell}^3\lambda_{j} 
	%\end{align} 
	Observe that $\hat{\wv}_j$ is an eigenvector of $\hat\SS^\ast \hat\SS$ and $\|\hat\SS^\ast \hat\SS - \SS^\ast \SS\|_{\op} = O_p(T^{-1/2})$. From Lemma 4.3 of \cite{Bosq2000} and Assumption \ref{assumpIV2}, we find that $\|\hat{\wv}_j-{\wv}_j\| \leq O_p((\llambda_j^{2}-\llambda_{j+1}^{2})^{-1}T^{-1/2})$. From (S2.11) in the proof of Theorem 3 of \cite{seong2021functional} and the fact that $\reg \leq \hat{\llambda}_{\KKK}^2 = \hat{\llambda}_{\KKK}^2 - {\llambda}_{\KKK}^2 + {\llambda}_{\KKK}^2 = O_p(T^{-1/2}) + \KKK^{-\rho}$, % $\hat{\nu}_j$ is the eigen $T^{-1/2}\tau \to 0$, we find similar arguments used in our proof of Theorem \ref{thm1}  $\reg^{1/2}\KKK^{\rho/2}\leq O_p(1)$. 
	it can be shown that $\reg^{1/2}\KKK^{\rho/2} = O_p(1)$. We observe that
	\begin{align}
		&	\|(\widehat{\PPI}_{\KKK}- {\PPI}_{\KKK})\beta_{h}\| %&=  \|\sum_{j=1}^{\KKK} \langle \hat{w}_j,\theta_{2,h} \rangle \hat{w}_j -  \sum_{j=1}^{\KKK} \langle {w}_j,\theta_{2,h} \rangle {w}_j\|
		= \|\sum_{j=1}^{\KKK} \langle \hat{\wv}_j-\wv_j,\beta_{h}\rangle \hat{\wv}_j -  \sum_{j=1}^{\KKK} \langle {\wv}_j,\beta_{h} \rangle ({\wv}_j-\hat{\wv}_j)\|\notag  \\
		&\leq O_p(T^{-1/2}) \sum_{j=1}^{\KKK} j^{\rho+1}\leq  O_p(T^{-1/2}\reg^{-1-2/\rho}\reg^{1+2/\rho}\KKK^{\rho+2})=O_p(T^{-1/2}\reg^{-1-2/\rho}). \label{eqpf04add2}
	\end{align}
	Moreover, the following is deduced as in our proof of Theorem \ref{thm1}:
	\begin{align}\label{eqpf04add}
		\|({\PPI}_{\KKK} - I_2) \beta_{h}\|^2   = \|\sum_{j=\KK+1}^{\infty} \langle \beta_{h} ,\wv_j \rangle \wv_j\|^2 = \sum_{j=\KK+1}^{\infty} \langle \beta_{h} ,\wv_j\rangle^2\leq O_p(\reg^{(2\varsigma-1)/\rho}).
	\end{align}
	Combining this result with \eqref{eqpf04add2}, the desired result is obtained.  \qed
	%Note that $\sum_{j=1}^{\infty} \langle \theta_{2,h} ,v_j\rangle^2$ is convergent, implying that \eqref{eqpf04} decays to zero as $\KK$ gets larger. 
	%% (UPTO HERE)

	%\begin{align}\label{eqpf04}
	%  \|{\Pi}_{\KKK}\theta_h - \theta_h\|^2   = \|\sum_{j=\KKK+1}^{\infty} \langle \theta,\vv_j \rangle \vv_j\|^2 = \sum_{j=\KKK+1}^{\infty} \langle \theta,\vv_j\rangle^2,
	% \end{align}
%   where the second equality holds due to the Parseval's identity and the fact that $\langle \theta, v \rangle =0$ for any $v\in [\ran C_{\XX \ZZ}]^\perp$ (Assumption \ref{assumpIV1}). Note that $\sum_{j=1}^{\infty} \langle \theta,v_j\rangle^2$ is well defined (convergent), implying that \eqref{eqpf04} decays to zero as $\KKK$ gets larger. 

\paragraph{Proof of Theorem \ref{thm4}}				
\noindent {1. Proof of  \ref{thm4i0}:} Since $\tilde{\theta}_h = \widehat{\PP}_{\KKK} \theta_h+\widehat{C}_{\ZZ\XX,\KKK}^{-1}  \widehat{C}_{\ZZ u}$, we have 
\begin{align}
	%\sqrt{\frac{T}{\psi_{\KK}(\zeta)}} \langle \hat{\theta}_h- \widehat{\Pi}_{\KK} \theta_h, \zeta\rangle 
	\sqrt{\frac{T}{\omega_{\KKK}(\zeta)}} \widetilde{\Theta}_1  
	&=  \left\langle \widehat{C}_{\ZZ\XX,\KKK}^{-1} \frac{1}{\sqrt{T\omega_{\KKK}(\zeta)}}\sum_{t=1}^T u_{h,t} \ZZ_t, \zeta \right\rangle.  \label{eqpf05a}
\end{align}
We find from Lemma \ref{lem3} that $\| \widehat{C}_{\ZZ\XX,\KKK}^{-1} - {C}_{\ZZ\XX,\KKK}^{-1}\|_{\op} = O_p(\|\widehat{\SS}_{\KKK}^{-1}-\SS_{\KKK}^{-1}\|_{\op}) + o_p(1)$.  %Note that in the spectral decompositions given in \eqref{eqspectral1}, we may assume that $\nu_j \geq 0$ (resp.\ $\hat{\nu}_j \geq 0$) by replacing $w_j$ with $-w_j$ (resp.\ $\hat{w}_j$ with $-\hat{w}_j$) if necessary. 
Redefine ${\wv}_j$ (resp.\ $\ww_j$) as $\sgn(\langle \hat{\wv}_j,\wv_j \rangle)\wv_j$ (resp.\ $\sgn(\langle \hat{\ww}_j,\ww_j \rangle)\ww_j$), and accordingly, redefine $\llambda_j$ as $\sgn(\langle \hat{\wv}_j,\wv_j \rangle)\sgn(\langle \hat{\ww}_j,\ww_j \rangle) \llambda_j$ as in the proof of Theorem 2 of \cite{seong2021functional}. We then find that  $\|\widehat{\SS}_{\KKK}^{-1}-\SS_{\KKK}^{-1}\|_{\op}$ is bounded above by $\|\sum_{j=1}^{\KKK} (\nu_j^{-1} - \hat{\nu}_j^{-1}){\ww}_j \otimes {\wv}_j \|_{\op} + \|\sum_{j=1}^{\KKK} \hat{\nu}_j^{-1} (\hat{\ww}_j\otimes \hat{\wv}_j   -{\ww}_j \otimes  {\wv}_j  )\|_{\op}$.
From similar arguments used in (S2.4)-(S2.6) for the proof of Theorem 2 of \cite{seong2021functional} and the facts that $\|\hat{\wv}_j-{\wv}_j\| \leq O_p((\nu_j^2-\nu_{j+1}^2)^{-1}T^{-1/2})=O_p(j^{\rho+1} T^{-1/2})$ (see Lemma 4.3 of \citealp{Bosq2000}) and $\|\widehat{\SS} -\SS\|_{\op} = O_p(T^{-1/2})$, we find that 
\begin{equation*} 
	%\|\left(\widehat{C}_{\XX\ZZ}^\ast \widehat{C}_{\XX\ZZ}\right)_{\KK}^{-1}\widehat{C}_{\XX\ZZ}^\ast - \left({C}_{\XX\ZZ}^\ast {C}_{\XX\ZZ}\right)_{\KK}^{-1}{C}_{\XX\ZZ}^\ast\| &\leq
	\|\widehat{\SS}_{\KKK}^{-1}-\SS_{\KKK}^{-1}\|_{\op}  \leq O(\reg^{-1/2}j^{\rho+1} T^{-1/2}) +O(\reg^{-1/2}T^{-1/2})\sum_{j=1}^{\KKK} j^{\rho+1} 
	%& \leq O_p(T^{-1/2}\reg^{-1/2}\KKK^{\rho+2}) 
	= O_p(T^{-1/2}\reg^{-3/2-2/\rho}) =o_p(1),
\end{equation*}
implying that  $\| \widehat{C}_{\ZZ\XX,\KKK}^{-1} - {C}_{\ZZ\XX,\KKK}^{-1}\|_{\op} =o_p(1)$ (see Lemma \ref{lem3}).
Combining this result with \eqref{eqpf05a}, we have
\begin{equation}
	\sqrt{\frac{T}{\omega_{\KKK}(\zeta)}} \widetilde{\Theta}_1  = \left\langle  \left({C}_{\ZZ\XX,\KKK}^{-1} + o_p(1)\right) \frac{1}{\sqrt{T\omega_{\KKK}(\zeta)}} \sum_{t=1}^T u_{h,t}  \ZZ_t, \zeta \right\rangle.\label{eqpfb1}
\end{equation}
%Let ${C}_{\XX\ZZ,\KKK}^{-1}$ be the adjoint of ${C}_{\ZZ\XX,\KKK}^{-1}$, i.e., ${C}_{\XX\ZZ,\KKK}^{-1}=({C}_{\ZZ\XX,\KKK}^{-1})^\ast$.
From similar arguments used in our proof of Theorem \ref{thm2}, %the functional central limit theorem for $L^2$-$m$-approximable sequences (see \citealp{berkes2013weak}, Theorem 1.1) and the Skorokhod representation theorem, we know that there exists a random element $V_T$ such that  $\|T^{-1/2} \sum_{t=1}^T u_{h,t} \ZZ_t - V_T\| \to_p 0$ such that $V_T =_d N(0,\Lambda_{u\ZZ})$ for every $T$. 
%Neglecting asymptotically negligible terms, \eqref{eqpfb1} can be written as 
%\begin{equation}
%\langle T^{-1/2} \sum_{t=1}^T u_{h,t} \ZZ_t - V_T, \frac{1}{\sqrt{\psi_{\KKK}(\zeta)}} {C}_{\XX\ZZ,\KKK}^{-1}\zeta \rangle + \langle V_T, \frac{1}{\sqrt{\psi_{\KKK}(\zeta)}}{C}_{\XX\ZZ,\KKK}^{-1}\zeta \rangle. \label{eqpfb2}
%\end{equation}
we write  \eqref{eqpfb1} as $\mathcal {V}_T + \mathcal W_T$ with neglecting asymptotically negligible terms, where $\mathcal{V}_T = \langle V_T, \frac{1}{\sqrt{\omega_{\KKK}(\zeta)}}{C}_{\XX\ZZ,\KKK}^{-1}\zeta \rangle$ and $\mathcal W_T = \langle T^{-1/2} \sum_{t=1}^T u_{h,t} \ZZ_t-V_T, \frac{1}{\sqrt{\omega_{\KKK}(\zeta)}}{C}_{\XX\ZZ,\KKK}^{-1}\zeta \rangle.$ 
As in \eqref{eqpfadd03}, we let  $\zeta_{\KK}=\frac{1}{\sqrt{\omega_{\KK}(\zeta)}}{C}_{\XX\ZZ,\KKK}^{-1}\zeta$ and obtain the following:  
\begin{equation*}
	\left|\frac{\langle T^{-1/2} \sum_{t=1}^T u_{h,t} \ZZ_t - V_T, \zeta_{\KKK} \rangle}{\langle V_T,\zeta_{\KKK}\rangle}\right|% = \frac{|\langle T^{-1/2} \sum_{t=1}^T u_{h,t} \ZZ_t - V_T, \frac{\zeta_{\KKK}}{\|\zeta_{\KKK}\|} \rangle|}{\langle V_T,\zeta_{\KKK}\rangle/\|\zeta_{\KKK}\|} \\ 
	%&\leq \frac{\sup_{|v|\leq 1}|\langle T^{-1/2} \sum_{t=1}^T u_{h,t} \ZZ_t - V_T, v \rangle|}{\langle V_T,\frac{\zeta_{\KKK}}{\|\zeta_{\KKK}\|}\rangle} 
	\leq \frac{\|T^{-1/2} \sum_{t=1}^T u_{h,t} \ZZ_t - V_T\|}{|\langle V_T,\frac{\zeta_{\KKK}}{\|\zeta_{\KKK}\|}\rangle|} 
	\leq \frac{o_p(1)}{|\langle V_T,\frac{\zeta_{\KKK}}{\|\zeta_{\KKK}\|}\rangle|}.
\end{equation*}
Under Assumption \ref{assumpIV4}\ref{assumpIV4aa},  $\zeta_{\KKK}/\|\zeta_{\KKK}\| = \sum_{j=1}^{\KKK} d_{\KKK,j}(\zeta) \wwtw_j$, $d_{\KKK,j}(\zeta)  \neq 0$ for some $j$, and $\sum_{j=1}^{\KKK} d_{\KKK,j}^2(\zeta) = 1$ for every $\KKK$.  %aaaaaaaaaunder Assumption \ref{assumpIV4}\ref{assumpIV4aa}. 
By construction of $V_T$, we find that $\langle V_T,\zeta_{\KKK}/\|\zeta_{\KKK}\|\rangle$ is a normal random variable with mean zero, and its variance is 
%\begin{align}
$\sum_{j=1}^{\KKK} \sum_{\ell=1}^{\KKK}d_{\KKK,j}(\zeta) d_{\KKK,\ell}(\zeta) \langle \Lambda_{\UUU}\wwtw_j,\wwtw_{\ell} \rangle$, %=   \sum_{j=1}^{\KK} \sum_{\ell=1}^{\KK} \frac{\lambda_j^{-1} \lambda_\ell^{-1}\langle \zeta,v_j \rangle\langle \zeta,v_\ell \rangle}{{\sum_{j=1}^{\KK} \lambda_j^{-2}\langle \zeta,v_j \rangle^2}} \langle \Lambda_{u\Upsilon}v_j,v_{\ell} \rangle > 0
%\end{align}
which converges to a positive constant under Assumption \ref{assumpIV4}\ref{assumpIV4bb}. This implies that $\mathcal W_T$ is  asymptotically negligible. 
From similar arguments used in our proof of Theorem \ref{thm1}, it can be shown that  $\langle V_T, \zeta_{\KK} \rangle $ is a zero-mean normal random variable whose variance is equal to $\langle \Lambda_{\UUU} {C}_{\XX\ZZ,\KK}^{-1} \zeta, {C}_{\XX\ZZ,\KK}^{-1}\zeta \rangle/ \omega_{\KKK}(\zeta) =   1,$ which completes the proof.\\

% the latter term is equal in distribution to $N(0, \langle \Lambda_{u\Upsilon} \frac{1}{\sqrt{\psi_{\KK}(\zeta)}}{C}_{\Upsilon\Upsilon,\KK}^{-1}\zeta\rangle, \frac{1}{\sqrt{\psi_{\KK}(\zeta)}}{C}_{\Upsilon\Upsilon,\KK}^{-1}\zeta)$, 
%	\begin{align}
	%		\frac{1}{{\omega_{\KKK}(\zeta)}} \langle \Lambda_{\UUU} {C}_{\XX\ZZ,\KK}^{-1} \zeta, {C}_{\XX\ZZ,\KK}^{-1}\zeta \rangle  = 1,
	%	\end{align}
%	which completes the proof. \\%If $\langle \theta_{h},\zeta \rangle$= $\langle \theta_{1,h},\zeta \rangle$, then obviously $\widetilde{\Theta}_1=\widetilde{\Theta}_2=0$ and thus the desired result is obtained.

%and the fact that,  $\psi_{\KK}(\zeta) = \langle {C}_{\Upsilon\Upsilon,\KK}^{-1} \Lambda_{u \Upsilon}{C}_{\Upsilon\Upsilon,\KK}^{-1}\zeta, \zeta\rangle$, %$\frac{1}{\sqrt{T\langle \Lambda_{u \Upsilon}v, v\rangle}} \sum_{t=1}^T \langle u_{t,h}  \Upsilon_t, v \rangle \to_d N(0,1).$$ 
%we find that (CHECK)
%\begin{align}
% \frac{1}{\sqrt{T\psi_{\KK}(\zeta)}} \sum_{t=1}^T \langle u_{t,h}  \Upsilon_t, {C}_{\Upsilon\Upsilon,\KK}^{-1}\zeta \rangle \to_d N(0,1)\label{eqpf06}
%\end{align}
%as desired.
%\eqref{eqpf05} and \eqref{eqpf06} imply that 
%\begin{align}
%\sqrt{\frac{T}{\psi_{\KK}(\zeta)}} \langle \hat{\theta}_h- \widehat{\Pi}_{\KK} \theta_h, \zeta\rangle \to_d N(0,1)
%\end{align}
\noindent{2. Proof of \ref{thm4i1}:} 
Note that $\widetilde{\Theta}_{2A} = \langle\widehat{\PP}_{\KKK} \theta_h - {\PP}_{\KKK}\theta_h,\zeta\rangle =  O_p(\langle (\widehat{D}-{D})\beta_{h},\zeta_1\rangle) + O_p(\langle(\widehat{\PPI}_{\KKK} - {\PPI}_{\KKK})\beta_{h},\zeta_2\rangle)$, where $\widehat{D} = \widehat{\DDD}_{11}^{-1}\widehat{\DDD}_{12}(I_2-\widehat{\PPI}_{\KKK})$ and ${D} = {\DDD}_{11}^{-1}{\DDD}_{12}(I_2-{\PPI}_{\KKK})$.  As in \eqref{eqpf07}, we find that 
\begin{align}
	\langle(\widehat{\PPI}_{\KKK} - {\PPI}_{\KK})\beta_{h},\zeta_2\rangle =&% =  \sum_{j=1}^{\KKK} \langle \hat{w}_j,\theta_{2,h}\rangle \langle \hat{w}_j,\zeta\rangle-  \sum_{j=1}^{\KKK} \langle {w}_j,\theta_{2,h}\rangle  \langle {w}_j,\zeta\rangle  \notag \\ 
	%&=   \sum_{j=1}^{\KKK} \langle \hat{w}_j-w_j,\theta_{2,h}\rangle \langle \hat{w}_j,\zeta\rangle +  \sum_{j=1}^{\KKK} \langle {w}_j,\theta_{2,h}\rangle  \langle \hat{w}_j-{w}_j,\zeta\rangle \notag\\
	\sum_{j=1}^{\KKK} \langle \hat{\wv}_j-\wv_j,\beta_{h}\rangle \langle {\wv}_j,\zeta_2\rangle +  \sum_{j=1}^{\KKK} \langle {\wv}_j,\beta_{h}\rangle  \langle  \hat{\wv}_j-{\wv}_j,\zeta_2\rangle 
	\notag \\ & +   \sum_{j=1}^{\KKK} \langle \hat{\wv}_j-\wv_j,\beta_{h}\rangle \langle \hat{\wv}_j-\wv_j,\zeta_2\rangle. \label{eqpf07add}
\end{align}
Similarly as in our proof of Theorem \ref{thm2}, we find that $\langle (\widehat{\SS}-\SS) {\wv}_j,{\ww}_\ell \rangle$ can be written as $\langle (\widehat{\SS}-\SS) {\wv}_j,{\ww}_\ell \rangle = F_1+F_2$, where 
\begin{align*}
	F_1 &= \frac{1}{T}\sum_{t=1}^T (\langle X_t,\wv_j \rangle \langle Z_{2,t}-\DDD_{21}\DDD_{11}^{-1}\mathbf{z}_{1,t},\ww_{\ell} \rangle-\mathbb{E}[\langle X_t,\wv_j \rangle \langle Z_{2,t}-\DDD_{21}\DDD_{11}^{-1}\mathbf{z}_{1,t},\ww_{\ell} \rangle]), \\  
	F_2 &= \frac{1}{T}\sum_{t=1}^T \langle X_t,\wv_j \rangle \langle ({\DDD}_{21}{\DDD}_{11}^{-1}-\hat{\DDD}_{21}\hat{\DDD}_{11}^{-1}) \mathbf{z}_{1,t},\ww_{\ell} \rangle.
\end{align*}
We rewrite $F_2 = \langle \hat{\DDD}_{12}\wv_j,  (({\DDD}_{11}^{-1})^\ast{\DDD}_{21}^\ast-(\hat{\DDD}_{11}^{-1})^\ast \hat{\DDD}_{21}^\ast)\ww_{\ell}\rangle.$ Under Assumption \ref{assumpIV5}, $\DDD_{21}^\ast\ww_{\ell} = O(|\llambda_{\ell}|^{1/2})$ since $\mathbb{E}[\langle Z_{2,t},\ww_{\ell} \rangle \langle \mathbf{w}_{t}, v \rangle] = \mathbb{E}[\langle Z_{2,t},\ww_{\ell} \rangle^2]^{1/2} O(1) = O(|\llambda_{\ell}|^{1/2})$ for any arbitrary $v \in \mathbb{R}^m$. Similarly, it can also be shown that $\DDD_{12}\wv_{j} = O(|\llambda_{j}|^{1/2})$. Noting that, by Assumption \ref{assumpIV2} and Lemma 2.1 of \cite{hormann2010} (and their subsequent discussion on non-independent pairs of $L^p$-$m$-approximable sequences for $p \geq 1$), the sequences $\{\langle \XX_t,v \rangle \ZZ_t - \mathbb{E}[\langle \XX_t,v \rangle \ZZ_t] \}$  and $\{\langle \ZZ_t,v \rangle   - \mathbb{E}[\langle \ZZ_t,v \rangle \XX_t] \}$ are $L^2$-$m$-approximable for any $v \in \mathcal H$, and using the central limit theorem for $L^2$-$m$-approximable sequences (see e.g., \citealp{berkes2013weak}) along with Assumption \ref{assumpIV5}\ref{assumpIV5b}, we find that $\sqrt{T}(\hat{\DDD}_{12}-\DDD_{12})\wv_j = O_p(|\llambda_{j}|^{1/2})$ and  $\sqrt{T}(\hat{\DDD}_{21}^\ast-\DDD_{21}^\ast)\ww_\ell = O_p(|\llambda_{\ell}|^{1/2})$.
Combining these results with the employed assumptions, we obtain  	$ (({\DDD}_{11}^{-1})^\ast{\DDD}_{21}^\ast-(\hat{\DDD}_{11}^{-1})^\ast \hat{\DDD}_{21}^\ast)\ww_{\ell}= ({\DDD}_{11}^{-1}-\widehat{\DDD}_{11}^{-1})^\ast \DDD_{21}^\ast \ww_{\ell} + (\widehat{\DDD}_{11}^{-1})^\ast( {\DDD}_{21}^\ast-\hat{\DDD}_{21}^\ast)\ww_{\ell}  = O_p(T^{-1/2} |\llambda_{\ell}|^{1/2})$ and also $\hat{\DDD}_{12}\wv_j = 	(\hat{\DDD}_{12}-\DDD_{12})\wv_j +\DDD_{12}\wv_j = O_p(T^{-1/2}|\llambda_j|^{1/2}) + O(|\llambda_j|^{1/2}) = O_p(|{\llambda}_j|^{1/2})$. These results imply that $F_2 = O_p(T^{-1/2})|\llambda_j\llambda_{\ell}|^{1/2}$.
%	Combining these results with the employed assumptions, we find that  	$ ({\DDD}_{11}^{-1}{\DDD}_{12}-\hat{\DDD}_{11}^{-1}\hat{\DDD}_{12})\ww_{\ell}= ({\DDD}_{11}^{-1}-\widehat{\DDD}_{11}^{-1}) \DDD_{12}\ww_{\ell} + \widehat{\DDD}_{11}^{-1}( {\DDD}_{12}-\hat{\DDD}_{12})\ww_{\ell}  = O_p(T^{-1/2} \llambda_{\ell}^{1/2})$ and also $\hat{\DDD}_{12}\wv_j = 	(\hat{\DDD}_{12}-\DDD_{12})\wv_j +\DDD_{12}\wv_j = O_p(T^{-1/2}) + O(\llambda_j^{1/2})$, which is $O_p({\llambda}^{1/2}_j)$ since $T^{-1/2}$ decay faster than $\nu_j^{1/2} = O(j^{})$. These results imply that $F_2 = O_p(\llambda_j^{1/2}\llambda_{\ell}^{1/2}T^{-1/2})$.}

From similar arguments in our proof of Theorem \ref{thm2}, we find that 
\begin{equation*}
T\mathbb{E}[F_1^2] \leq O(1) \mathbb{E}[\langle X_t,\wv_j \rangle^2\langle Z_{w,t},\ww_{\ell} \rangle^2] = O(|\llambda_j\llambda_{\ell}|), 
\end{equation*}
from which we deduce that $F_1 = O_p(T^{-1/2})|\llambda_j\llambda_{\ell}|^{1/2}$ and conclude that $\langle	(\hat{\SS} - \SS)\wv_j,\ww_{\ell}\rangle = O_p( T^{-1/2})|\llambda_j\llambda_{\ell}|^{1/2}$. 
By combining this result with the given assumptions and Lemma S1 of \cite{seong2021functional}, it can be shown that $\|\hat{\wv}_j - \wv_j\| = O_p(T^{-1/2}j^2)$; in turn, the following can also be obtained using nearly identical arguments as those in the proofs of (S2.16) and (S2.33) of \cite{seong2021functional}:
%	Combining this result with the employed assumptions, we deduce the following from nearly identical arguments used in the proofs of (S2.16) and (S2.33) of \cite{seong2021functional}:
\begin{align}
&\langle \hat{\wv}_j-\wv_j,\beta_{h}\rangle = O_p(T^{-1/2})j^{\rho/2+1-\varsigma}, \label{eqpf07adda}\\
&\langle \hat{\wv}_j-\wv_j,\zeta_2\rangle = O_p(T^{-1/2}) j^{ \rho/2 + 1-\delta}.\label{eqpf07addb}
\end{align}
%\begin{align} 
%\KK \leq (1+o_p(1)) a^{-1/\rho}.
%\end{align}
When  $\rho/2 + 2 < \varsigma+ \delta$, the following result is deduced from \eqref{eqpf07add}, \eqref{eqpf07adda}, \eqref{eqpf07addb} and Assumption \ref{assumpIV5}\ref{assumpIV5c}:  
\begin{align}
&\sqrt{\frac{T}{\omega_{\KKK}(\zeta)}}\langle (\widehat{\PPI}_{\KKK} - {\PPI}_{\KKK})\beta_{h},\zeta_2\rangle = \frac{1}{\sqrt{\omega_{\KKK}(\zeta)}} \left\{ \sum_{j=1}^{\KKK} j^{\rho/2+1-\varsigma-\delta} + \frac{1}{\sqrt{T}}\sum_{j=1}^{\KKK} j^{\rho+2-\varsigma-\delta}\right\} \notag  \\ 
%	&= O_p\left(\frac{1}{\sqrt{\omega_{\KKK}(\zeta)}} \right) \left(1+ T^{-1/2}\KKK^{\rho+3-\varsigma-\delta}\right) =  O_p\left(\frac{1}{\sqrt{\omega_{\KKK}(\zeta)}} \right) \left(1+ T^{-1/2}\reg^{-1/2} \KKK^{\rho/2+3-\varsigma-\delta}\right) \notag \\
&=  O_p\left(\frac{1}{\sqrt{\omega_{\KKK}(\zeta)}} \right) \left(1+\KKK^{1-\varsigma-\delta}\right) =  O_p\left(\frac{1}{\sqrt{\omega_{\KKK}(\zeta)}} \right),  \label{eqthmpf001a}
\end{align} 
where nearly identical arguments to those used in \eqref{eqthmpf001} are applied.  
It only remains to verify that $\sqrt{T}\langle (\widehat{D}-{D})\beta_h,\zeta_1\rangle=O_p({1}/{\sqrt{\omega_{\KK}(\zeta)}})$ for the desired result. From some algebra, we find that $\langle(\widehat{D}-{D})\beta_h,\zeta_1\rangle = \langle (\widehat{\DDD}_{11}^{-1}\widehat{\DDD}_{12}-{\DDD}_{11}^{-1}{\DDD}_{12})(I_2-\widehat{\PPI}_{\KKK})\beta_h ,\zeta_1\rangle - \langle{\DDD}_{11}^{-1}{\DDD}_{12}(\widehat{\PPI}_{\KKK} -\PPI_{\KKK})\beta_h,\zeta_1 \rangle$, where the latter term, $ \langle{\DDD}_{11}^{-1}{\DDD}_{12}(\widehat{\PPI}_{\KKK} -\PPI_{\KKK})\beta_h,\zeta_1 \rangle$, is equal to 
\begin{align}
\sum_{j=1}^{\KKK} \langle \hat{\wv}_j-\wv_j,\beta_h\rangle \langle \DDD_{11}^{-1}\DDD_{12}\wv_j,\zeta_1\rangle +  \sum_{j=1}^{\KKK} \langle \wv_j,\beta_h\rangle  \langle  \DDD_{11}^{-1}\DDD_{12}(\hat{\wv}_j-{\wv}_j),\zeta_1\rangle \nonumber\\  +   \sum_{j=1}^{\KKK} \langle \hat{\wv}_j-\wv_j,\beta_h\rangle \langle  \DDD_{11}^{-1}\DDD_{12}(\hat{\wv}_j-\wv_j),\zeta_1\rangle; \label{eqthmpf001b}
\end{align}
see \eqref{eqpfadd07}.	Under Assumptions \ref{assumpIV2}-\ref{assumpIV5}, the following is deduced from nearly identical arguments used to show (S2.29) of \cite{seong2021functional}:
\begin{equation}
%	&\langle \hat{v}_j-v_j,\theta_{2,h}\rangle = O_p(T^{-1/2})j^{\rho/2+1-\varsigma}, \\
%	&\langle \hat{v}_j-v_j,\zeta\rangle = O_p(T^{-1/2}) j^{ \rho/2 + 1-\delta}\\
\|\DDD_{11}^{-1}\DDD_{12} (\hat{\wv}_j-\wv_j)\| = O_p(T^{-1/2})j^{\rho/2 +1 - \delta}. \label{eqpfadd08add}
\end{equation} 
From \eqref{eqpf07adda}, \eqref{eqthmpf001b}, \eqref{eqpfadd08add} and Assumption \ref{assumpIV5}\ref{assumpIV5c}, we find that $\sqrt{\frac{T}{\omega_{\KKK}(\zeta)}}\langle \langle{\DDD}_{11}^{-1}{\DDD}_{12}(\widehat{\PPI}_{\KKK} - {\PPI}_{\KKK})\beta_{h},\zeta_2\rangle = O_p\left({1}/{\sqrt{\omega_{\KK}(\zeta)}} \right)$. On the other hand,  because  $ \widehat{\DDD}_{11}^{-1}\widehat{\DDD}_{12}-{\DDD}_{11}^{-1}{\DDD}_{12} = O_p(T^{-1/2})$ and $(I_2-\widehat{\PPI}_{\KKK})\beta_{h} = O_p(1)$,  the first term \eqref{eqpf07adda} decays faster than the second term, and hence,  $\sqrt{T}\langle (\widehat{D}-{D})\beta_h,\zeta_1\rangle=O_p({1}/{\sqrt{\omega_{\KKK}(\zeta)}})$ as desired.
\\

\noindent{3. Proof of \ref{thm4i2}:} 
We observe that 
$\widetilde{\Theta}_{2B} = O_p(\langle {\DDD}_{11}^{-1}{\DDD}_{12}(I_2-\PPI_{\KKK})\beta_{h},\zeta_1\rangle) + O_p (\langle(I_2-\PPI_{\KKK}) \beta_h,\zeta_2 \rangle)$. As in our proof of Theorem \ref{thm2} (see particularly \eqref{eqpf09} and \eqref{eqpf08add}), we find that  						
%	If we show that $\langle (I_2-\PPI_{\KKK}) \theta_{2,h} \rangle$ is $O_p(1/\sqrt{\omega_{\KKK}(\zeta)})$, then the desired result is established. Note that
%		\begin{align}\label{eqpf08a}
%\sqrt{\frac{T}{\psi_{\KK}(\zeta)}}\sum_{j=\KK+1}^\infty \langle \theta, v_j \rangle\langle \zeta, v_j \rangle 
%				\sqrt{\frac{T}{\omega_{\KKK}(\zeta)}}\langle (I_2-\PPI_{\KKK}) \theta_{2,h} \rangle=\sqrt{\frac{T}{\omega_{\KKK}(\zeta)}}\sum_{j=\KKK+1}^\infty \langle \theta_{2,h}, w_j \rangle\langle \zeta, w_j \rangle 
%				\leq \frac{1}{\sqrt{\omega_{\KKK}(\zeta)}} T^{1/2}\KKK^{1-\delta - \varsigma},
%			\end{align}
%			where the last inequality follows from the Euler-Maclaurin summation formula as in \eqref{eqpf04} and \eqref{eqpf08}. Since $\reg^{1/2}\KKK^{\rho/2} \leq O_p(1)$,
\begin{equation} \label{eqpf09add1}
\sqrt{\frac{T}{\omega_{\KKK}(\zeta)}}\langle (I_2-\PPI_{\KKK}) \beta_h,\zeta_2 \rangle \leq \frac{1}{\sqrt{\omega_{\KKK}(\zeta)}} T^{1/2}\reg^{(\delta + \varsigma - 1)/\rho},
\end{equation}
and 	
\begin{align}\label{eqpf08adda}
\sqrt{\frac{T}{\omega_{\KKK}(\zeta)}}\langle {\DDD}_{11}^{-1}{\DDD}_{12}(I_2-\PPI_{\KKK})  \beta_h, \zeta_1 \rangle %&=\sqrt{\frac{T}{\psi_{\KK}(\zeta)}}\sum_{j=\KK+1}^\infty \langle  \beta_h, v_j \rangle\langle \zeta_1, {\CCC}_{11}^{-1}{\CCC}_{12} v_j \rangle 
% \leq O\left(\frac{1}{\sqrt{\psi_{\KK}(\zeta)}}\right) T^{1/2}\KK^{1-\eta - \varsigma},
\leq \frac{1}{\sqrt{\omega_{\KKK}(\zeta)}}  O_p(T^{1/2}\reg^{(\delta + \varsigma - 1)/\rho}).
\end{align}
If $\delta$ and $\varsigma$ are large enough so that $T^{1/2}\reg^{(\delta + \varsigma - 1)/\rho} = O_p(1)$, \eqref{eqpf09add1} and \eqref{eqpf08adda} are $O_p(1/\sqrt{\omega_{\KKK}(\zeta)})$, as desired. \\

\noindent{4. Proof of the asymptotic equivalence between the choices of $\omega_{\KKK}$ and $\widehat{\omega}_{\KKK}$:} Under the employed assumptions, the following can be shown using similar arguments used in our proof of Theorem \ref{thm2}:  $\|\widehat{C}_{\ZZ\XX,\KKK}^{-1}- {C}_{\ZZ\XX,\KKK}^{-1}\|_{\op} \to_p 0, \|\widehat{C}_{\XX\ZZ,\KKK}^{-1}- {C}_{\XX\ZZ,\KKK}^{-1}\|_{\op} \to_p 0$ (see Remark \ref{remaddapp01}) and $\|\widehat{\Lambda}_{\UUU} -{\Lambda}_{\UUU}\|_{\op}\to_p 0$.  %For the latter result, refer to Theorem 1 of \cite{horvath2013estimation} along with Lemma A.1 of \cite{horvath2014test}. 
From these, the desired result follows.\qed

\newpage 

\section*{Supplementary results}
\section{Direct estimation of $\IRF$s from the RFVAR}\label{sec: svar: est}
In this section, we discuss the estimation of $\IRF$s directly from the RFVAR model as mentioned earlier in Section \ref{sec_IRF}. As in that section, we consider the augmented Hilbert space $\widetilde{\mathcal H} = \mathbb{R} \times \mathcal H$, and define $\Upsilon_t = \left[\begin{smallmatrix}y_t \\X_t \end{smallmatrix}\right]$. Then, the reduced-form of the SVAR model \eqref{eq: model: svar} can be written as 
\begin{align}
	\Upsilon_t = \Gamma \Upsilon_{t-1} + \mathcal E_t,
\end{align}
where $\mathcal E_t =  \left[\begin{smallmatrix}\varepsilon_{1t}\\\mathcal E_{2t} \end{smallmatrix}\right]$ and $\Gamma = \mathcal B^{-1}\mathcal A$. Assuming stationarity of $\{\Upsilon_t\}$ and  $\mathbb{E}[\Upsilon_{t-1}\otimes \mathcal E_t] = 0$, we obtain the following operator equation: $D_{\Upsilon\Upsilon} = \Gamma C_{\Upsilon\Upsilon}$, where $D_{\Upsilon\Upsilon} =\mathbb{E}[\Upsilon_{t-1}\otimes \Upsilon_t]$.
This is a standard population moment equation for the stationary functional AR(1) model, and $\Gamma$ and $\Sigma_\varepsilon = \mathbb{E}[\mathcal E_t\otimes \mathcal E_t]$ are the parameters of interest. Under some standard regularity conditions, we can construct estimators $\widehat{\Gamma}$ and $\widehat{\Sigma}_{\varepsilon}$ (see e.g., \citealp{Bosq2000,Park2012}) satisfying  $\|\widehat{\Gamma}-\Gamma\|_{\op} \to_p 0$ and $\|\widehat{\Sigma}_\varepsilon - \Sigma_\varepsilon\|_{\op} \to_p 0$. As in \eqref{eq: model: svar}, we write $\widehat{\Gamma}$ and $\widehat{\Sigma}_{U}$ as follows:
\begin{align}
	\widehat{\Gamma} = \begin{bmatrix}
		\hat{\gamma}_{11} & \hat{\gamma}_{12} \\
		\hat{\gamma}_{21} & \hat{\gamma}_{22} \\
	\end{bmatrix}, \qquad 
	\widehat{\Sigma}_\varepsilon = \begin{bmatrix}
		\hat{\sigma}_{\varepsilon,11} & \hat{\sigma}_{\varepsilon,12} \\
		\hat{\sigma}_{\varepsilon,21} & \hat{\Sigma}_{\varepsilon,22} 
	\end{bmatrix}. \nonumber
\end{align}
Then under the identification scheme $\beta_{12}=0$, the structural parameters $\beta_{21}$, $\alpha_{11}$, $\alpha_{12}$, $\alpha_{21}$, $\alpha_{22}$, $\sigma_{11}$, and $\Sigma_{22}$ in \eqref{eq: model: svar} can be, respectively, consistently estimated as follows:
\begin{align}
	&\hat{\sigma}_{11}=\hat{\sigma}_{\varepsilon,11}, \quad \hat{\beta}_{21} = \hat{\sigma}_{U,21}\hat{\sigma}_{\varepsilon,11}^{-1}, \quad \widehat{\Sigma}_{22}=\widehat{\Sigma}_{\varepsilon,22} -\hat{\beta}_{21}\hat{\sigma}_{\varepsilon,11}^{-1}\hat{\beta}_{21}^\ast, \nonumber \\ &\hat{\alpha}_{11}=\hat{\gamma}_{11},\quad \hat{\alpha}_{12}=\hat{\gamma}_{12}, \quad
	\hat{\alpha}_{21}= \hat{\gamma}_{21}+\hat{\beta}_{21}\hat{a}_{11},   \quad \hat{\alpha}_{22}= \hat{\gamma}_{22}+\hat{\beta}_{21}\hat{a}_{12}, \nonumber
\end{align}
where $\hat{\sigma}_{\varepsilon,11}^{-1}$ is well defined (see our proof of Proposition \ref{prop: svar: identification}) and $\widehat \beta_{21} ^\ast$ is the adjoint of $\widehat\beta_{21}$. 
Then the impulse response of interest, $\IRF_{12,h}$, can also be consistently estimated. However, it is given by a complicated function of the above operator-valued estimators whose domains and codomains are different, and at present, to the best of the authors' knowledge, there is no requisite theory to obtain more detailed asymptotic properties beyond consistency for such an object from the estimation results of the standard functional autoregressive model. Thus, we leave this for future study.

\section{Extension: Inference on  $\IRF_{21,h}$} \label{sec_extension}

\subsection{Inference on $\IRF_{21,h}$ with exogenous predictors} \label{sec_appen_infer}
We now consider statistical inference on $\IRF_{21,h}$ for $h\geq 1$ under $\beta_{12}=0$ in \eqref{eq: model: svar} according to Proposition \ref{prop: svar: identification: a}. To accommodate a more general case with a set of control variables, we consider the following model: for bounded linear maps $A_{1,h}: \mathbb{R}^m \to \mathcal H$ and $A_{2,h}:\mathcal H\to \mathcal H$,
\begin{equation}
	X_{t+h+1} =   A_{1,h} \mathbf{w}_t  + A_{2,h} X_{t} +   \UUUU_{h,t}, \label{eq: model: benchmark: reduced4}
\end{equation}
where $\mathbf{w}_t = (w_{1,t}, \ldots, w_{m,t})'$ denotes the collection of  variables, including the $y_{t+1}$, $y_{t}$ and other control variables, and $\UUUU_{h,t}$ denotes the error term, which is a functional random variable (unlike $u_{h,t}$ in Section \ref{sec:est}). The model \eqref{prop2: eq2} given in Proposition \ref{prop: svar: identification: a} may be viewed as a special case of \eqref{eq: model: benchmark: reduced4}, with some notational differences which are made for notational simplicity in the subsequent development. In particular, note that we intentionally increase the time index of some variables; this is done to reuse the notation and some assumptions from the previous sections. %, and (ii) more generally allow for the possibility of multiple control variables. 
If we let $w_{1,t} = y_{t+1}$ and consider a unit shock applied to $y_{t+1}$, then $A_{1,h}\zeta_m$, where  $\zeta_m=(1,0,\ldots,0) \in \mathbb{R}^m$, is interpreted as the impulse response function of interest under appropriate conditions (see \eqref{eqpartialirf2} and Proposition  \ref{prop: svar: identification: a}).
%However, as will become evident, incorporating those variables is straightforward with only a minor modification of our approach (refer to Remark \ref{rem_controls}).
%	For convenience, we let $Y_{h,t} = (0, X_{t+h})'$, which is understood as $\elltwo$-valued random element. 

For convenience, we hereafter let 
\begin{equation} \label{eqdefY1}
	Y_{h,t} = X_{t+h+1} \quad \text{and} \quad	\Upsilon_t =   \begin{bmatrix}\mathbf{w}_t\\X_t\end{bmatrix}.
\end{equation}
As in Section \ref{sec:est}, we assume that \( \Upsilon_t \) and \( Y_{h,t} \) are zero-mean random elements for simplicity, as extending the subsequent results to the case with nonzero means is straightforward. 
Then the considered model can be rewritten as follows: for some $A_h: \elltwo \to \mathcal H$, 
\begin{equation}
	Y_{h,t} =  A_h \Upsilon_t +  \UUUU_{h,t}. \label{eq: model: benchmark: reduced5}
\end{equation}
Let $C_{Y\Upsilon} = \mathbb{E}[\Upsilon_t\otimes Y_{h,t}]$, $C_{\Upsilon \Upsilon} = \mathbb{E}[\Upsilon_t\otimes \Upsilon_t]$, and $C_{\UUUU \Upsilon } = \mathbb{E}[\Upsilon_t \otimes \UUUU_{h,t}]$ and let their sample counterparts be denoted by $\widehat{C}_{Y\Upsilon}$, $\widehat{C}_{\Upsilon \Upsilon}$, and $\widehat{C}_{\UUUU\Upsilon}$, respectively.
We also write ${C}_{\Upsilon\Upsilon}$ (resp.\ $\widehat{C}_{\Upsilon\Upsilon}$) as operator matrices on $\elltwo$ with elements $\CCC_{ij}$ (resp.\ $\hat \CCC_{ij}$)  as in \eqref{covblock}, and then define the operator Schur complements $\SA$ and $\widehat{\SA}$ as in \eqref{eqopschur}. Moreover,  we let $(\lambda_j, v_j)$ (resp.\ $(\hat{\lambda}_j, \hat{v}_j)$) be the eigenelements of $\SA$ (resp.\ $\widehat{\SA}$) as in \eqref{eqshur0}. We also define $\KK$ and $\widehat{\SA}_{\KK}^{-1}$ as in \eqref{eqshur1}.

When \eqref{eq: model: benchmark: reduced5} is given, $\Upsilon_t \otimes Y_{h,t} = A_h (\Upsilon_t \otimes \Upsilon_t) + \Upsilon_t \otimes \UUUU_{h,t}$. If the operator-valued sequences appearing in this equation are stationary and ${C}_{\UUUU\Upsilon} = \mathbb{E}[\Upsilon_t \otimes \UUUU_{h,t}] = 0$, the following population relationship holds: $C_{Y\Upsilon} = A_h C_{\Upsilon\Upsilon}$. Based on this, we define our estimator as follows:
\begin{equation}
	\widehat{A}_h = \widehat{C}_{Y\Upsilon} \widehat{C}_{\Upsilon\Upsilon,\KK}^{-1}.  
\end{equation}
For the unique identification of $A_h$ and our asymptotic analysis, we employ the following assumption: below, $\{\Upsilon_t\otimes \UUUU_{h,t} \}$ is a sequence of Hilbert-Schmidt operators and understood as Hilbert-valued random elements  (see e.g., \citealp{Bosq2000}, p.\ 34). Moreover, $A_{h}$ is equivalently understood as the operator matrix $A_{h} = [\begin{matrix} A_{1,h} & A_{2,h}	\end{matrix}]$ (see Section \ref{sec_prelim2}). %, where $A_{1,h}: \mathbb{R}^m \mapsto \mathcal H$ and $A_{2,h}: \mathcal H \mapsto \mathcal H$ (see Section \ref{sec_prelim2}).
%\begin{assumption} \label{assum1add} % and for any orthonormal basis $\{g_j\}_{j\geq 1}$ of $\mathcal {H}$, $\sum_{j=1}^\infty \langle \theta_h, g_j \rangle^2 < \infty$. %, i.e., $\langle \theta_h, v_j \rangle = 0$ for all $j$ such that $\lambda_j = 0$. 
%\end{assumption}
%We will investigate asymptotic properties of $\hat{A}_h$. For our asymptotic analysis, we employ the following assumptions: 
\begin{assumpA} \label{assum2add}
	\begin{enumerate*}[(i)]
		\item \label{assum2adda0}  $ A_h\ker C_{\Upsilon\Upsilon}=\{0\}$;
		\item \label{assum2adda} \eqref{eq: model: benchmark: reduced4} (or equivalently \eqref{eq: model: benchmark: reduced5}) holds with ${C}_{\UUUU\Upsilon}=0$;
		\item\label{assum2addb} $\{\Upsilon_t \}$,  $\{\UUUU_{h,t}\}$ and  $\{\Upsilon_t\otimes \UUUU_{h,t}\}$ are stationary and $L^4$-$m$-approximable;  %$\{\Upsilon_t-\mathbb{E}[\Upsilon_t]\}$, $\{u_{h,t}\}$ and $\{\Upsilon_t \otimes u_{h,t}\}$ is $L^2$-$m$-approximable, and
		\item \label{assumaddc} $\|\widehat{C}_{\UUUU\Upsilon}\|_{\op} = O_p(T^{-1/2})$, $\|\widehat{C}_{\Upsilon\Upsilon} -C_{\Upsilon\Upsilon}\|_{\op} = O_p(T^{-1/2})$,  and $\|\widehat{\CCC}_{11}^{-1} - \CCC_{11}^{-1}\|_{\op} = O_p(T^{-1/2})$; \item\label{assum3adda} Assumption \ref{assum3}\ref{assum3a} holds; %For $\rho>2$,     $\lambda_j^2 \leq \CC j^{-\rho}$, $\lambda_j^2-\lambda_{j+1}^2 \geq \CC j^{-\rho-1}$;
		\item\label{assum3addb} for $\varsigma > 1/2$ and $\gamma>1/2$, $|\langle A_{2,h}  v_j,v_{\ell} \rangle| \leq \CC j^{-\varsigma}\ell^{-\gamma}$.
	\end{enumerate*} 
\end{assumpA}
% \begin{assumption} \label{assum3add}
	% \end{assumption}
The consistency and rate of convergence of $\widehat{A}_h$ are given in the following theorem:
\begin{theoremA} \label{thm1addadd} Let Assumption \ref{assum2add} hold and $T\reg^{1+4/\rho} \to \infty$. Then, 
	\begin{equation}
		\|\widehat{A}_h -A_h\|_{\op} = O_p(T^{-1/2}\reg^{-1/2-2/\rho} + \reg^{(2\varsigma-1)/2\rho}).
	\end{equation}
\end{theoremA}
As in Section \ref{sec:est}, we write $\zeta \in \elltwo$ as $\zeta = \left[\begin{smallmatrix}\zeta_1 \\\zeta_2 \end{smallmatrix}\right]$, and in this section, we consider the case where $\zeta_2 = 0 \in \mathcal H$ in \eqref{eqzetadecom}. 
In this case,  $A_h\zeta = A_{1,h}\zeta_1$, and thus the impulse response of interest with respect to a unit shock on $y_{t+1}$ can be obtained by considering a specific choice of $\zeta_1$, as discussed earlier. Noting that $A_h\zeta= A_{1,h}\zeta_1$ itself is $\mathcal H$-valued, we develop inference on $\langle A_h \zeta, \varph \rangle$ for any arbitrary $\varph  \in \mathcal H$.  
%we may decompose $\elltwo$ into $\mathbb{R}^{m+3} \times \{0\}$ and $\{0\} \times \mathcal H$. The impulse response function with respect to a shock $\zeta$ on the variables $y_t$ can be written as $A_h \zeta \in \mathcal H$ for some $\zeta \in \mathbb{R}^{m+3} \times \{0\}$. We thus consider inference on $\langle A_h \zeta, \varph \rangle$ for any known element $\zeta \in \mathbb{R}^{m+3}\times \{0\}$ and $\varph \in  \{0\}\times \mathcal H$. This reduces to inference on the impulse response $\IRF_{21,h}$ under appropriate conditions (see Proposition \ref{prop: svar: identification: a}).  
Define the real quantity
\begin{equation}
	{\psi}_{\KK}^{(\varph)}(\zeta) = \langle  {\Lambda}_{\UU}^{(\varph)}{C}_{\Upsilon\Upsilon,\KK}^{-1}\zeta, {C}_{\Upsilon\Upsilon,\KK}^{-1}\zeta\rangle,
\end{equation}
where 
\begin{equation} \label{eqlrvadd}
	\Lambda_{\UU}^{(\varph)} =\sum_{s=-\infty}^\infty \mathbb{E}[\UUUU_{h,t}^{(\varph)}\Upsilon_t \otimes \UUUU_{h,t-s}^{(\varph)}\Upsilon_s] \quad \text{and} \quad \UUUU_t^{(\varph)} = \langle \UUUU_t,\varph \rangle.
\end{equation}
We let  $\widehat{\Lambda}_{\UU}^{(\varph)}$ be the sample counterpart of $\Lambda_{\UU}^{(\varph)}$ computed from the residual $\hat{\UUUU}_{h,t}^{(\varph)}= \langle Y_{h,t} - \widehat{A}_h\Upsilon_t, \varph \rangle$ as follows: 
\begin{equation} 
	\widehat{\Lambda}_{\UU}^{(\varph)} =  \frac{1}{T} \sum_{s=-h}^h\mathrm{k}\left(\frac{s}{\bdw}\right) \sum_{1 \leq t, t-s \leq T} \hat{\UUUU}_{h,t}^{(\varph)}\Upsilon_t \otimes \hat{\UUUU}_{h,t-s}^{(\varph)}\Upsilon_{t-s}. 
\end{equation}
Following the spectral decomposition of $C_{\Upsilon\Upsilon}$, we define $c_{m,j}(\zeta)$ as in \eqref{eqcm}. 
%\begin{equation}
%c_{m,j}(\zeta) =  \lambdatw_j^{-1}\langle \zeta,\vtw_j \rangle\bigg/\sqrt{\sum_{j=1}^{m} \lambdatw_j^{-2}\langle \zeta,\vtw_j \rangle^2}
%{equation}
%and note that $\sum_{j=1}^{m} c_{m,j}^2(\zeta) = 1$ for every $m$ unless $\langle \zeta,w_j \rangle \neq 0$ for at least one $j$. 
% $\Lambda_{u\Upsilon}^{(\varph)}$ is self-adjoint, nonnegative and compact for every $\varph \in \mathcal H$, so it allows the eigendecomposition. 
To establish asymptotic normality of our estimator, we employ the following assumptions: below, $\CC$ denotes a generic positive constant.
\begin{assumpA} \label{assum4add} 
	\begin{enumerate*}[(i)]
		\item \label{assum4adda}  Assumption \ref{assum4}\ref{assum4a} holds;
		\item \label{assum4addaa}  $\zeta \notin \ker C_{\Upsilon\Upsilon}$ and $\Lambda_{\UU}^{(\varph)} \vtw_j \neq 0$ for $\vtw_j$ corresponding to a nonzero eigenvalue of $C_{\Upsilon\Upsilon}$; 
		\item  \label{assum4addbb} $\sum_{j=1}^{m} \sum_{\ell=1}^{m} c_{m,j}(\zeta) c_{m,\ell}(\zeta) \langle \Lambda_{\UU}^{(\varph)}\vtw_j,\vtw_{\ell} \rangle \to \CC>0$ as $m\to \infty$; \item \label{assum4addcc} $\sup_{1\leq t\leq T} \|\Upsilon_t\| = O_p(1)$.
	\end{enumerate*}
\end{assumpA}
\begin{assumpA} \label{assum5add}
	\begin{enumerate*}[(i)]
		\item\label{assum5adda} Assumption \ref{assum5}\ref{assum5a}  holds; %$\mathbb{E}[\langle f_t,v_j \rangle f_t] \leq \CC \lambda_j^2$  and  $\mathbb{E}[r_t(j,\ell)r_{t-s}(j,\ell)]\leq \CC s^{-\tilde{\CC}}\mathbb{E}[r_t^2(j,\ell)]$ for some $\tilde{\CC} > 1$ and $s\geq 1$; 
		%T\item $\lambda_j^2 \leq C j^{-\rho}$  $\lambda_j^2 - \lambda_j^2 \geq C j^{-\rho-1}$. 
		\item\label{assum5addb} for some $\delta>1/2$ and $|\langle \CCC_{21}\CCC_{11}^{-1}\zeta_1,v_j \rangle| \leq \CC j^{-\delta}$. 
		%\item\label{assum4d} $h=o(T^{1/2})$.
	\end{enumerate*}
\end{assumpA}
We decompose  $\langle (\widehat{A}_h-A_h)\zeta,\varph\rangle$ into $\langle (\widehat{A}_h-A_h)\zeta,\varph\rangle=\widehat{\Theta}_1+\widehat{\Theta}_{2A} + \widehat{\Theta}_{2B}$, where
\begin{equation} \label{eqdecomthetaad}
	\widehat{\Theta}_1 =  \langle (\widehat{A}_h-A_h\widehat{P}_{\KK}^\ast)\zeta, \varph \rangle,\quad \widehat{\Theta}_{2A}=  \langle (A_h\widehat{P}_{\KK}^\ast -A_h{P}_{\KK}^\ast)  \zeta, \varph \rangle, \quad \widehat{\Theta}_{2B}=  \langle (A_h{P}_{\KK}^\ast -A_h)  \zeta, \varph \rangle
\end{equation}
and $\widehat{P}_{\KK}^\ast$ (resp.\ ${P}_{\KK}^\ast$) is the adjoint of $\widehat{P}_{\KK}$ (resp.\  ${P}_{\KK}^\ast$) defined as in \eqref{eqprojections}.
%We then have $\widehat{\Theta} = \widehat{\Theta}_1+\widehat{\Theta}_2+\widehat{\Theta}_3$.  
\begin{theoremA} \label{thm2add}  Suppose that Assumptions \ref{assum2add}-\ref{assum4add} hold and  $T\reg^{2+4/\rho} \to \infty$. Then the following holds:
	\begin{enumerate}[(i)]
		\item\label{thm2addi0} $\sqrt{{T}/{\psi_{\KK}^{(\varph)}(\zeta)}}\widehat{\Theta}_1 \to_d N(0,1)$.
		%\item\label{thm2i0} \begin{equation} 
			%\sqrt{{T}/{\psi_{\KK}(\zeta)}}\widehat{\Theta}_1 \to_d N(0,1). %\label{eqthm1}
			%\end{equation}
		\end{enumerate}
		%Suppose that $\langle \theta_{h}, \zeta\rangle \neq  \langle \theta_{1,h}, \zeta\rangle$, 
		If Assumption \ref{assum5add} is additionally satisfied with $\rho/2 + 2 < \varsigma+ \delta$, the following hold:
		\begin{enumerate}[(i)]\addtocounter{enumi}{1}
			\item\label{thm2addi1}  If $\psi_{\KK}^{(\varph)}(\zeta) \to_p \infty$, $
			\sqrt{{T}/{\psi_{\KK}^{(\varph)}(\zeta)}}\widehat{\Theta}_{2A} \to_p 0$.
			\item \label{thm2addi2} If $\psi_{\KK}^{(\varph)}(\zeta) \to_p \infty$ and  $T^{1/2}\reg^{(\delta+\varsigma-1)/\rho} \to 0$, 
			$\sqrt{{T}/{\psi_{\KK}^{(\varph)}(\zeta)}}\widehat{\Theta}_{2B} \to_p 0.$
			%and thus, if $\psi_{\KK}(\zeta) \to_p \infty$, 
			%$$\sqrt{{T}/{\psi_{\KK}(\zeta)}}\widehat{\Theta} \to_d N(0,1).$$
			
			%   \item If  HH holds and $\delta + \varsigma > \rho/2 + 2$
			%\item\label{thm2i3} 
		\end{enumerate}
		The results given in \ref{thm2i0}-\ref{thm2i2} hold when $\psi_{\KK}^{(\varph)}(\zeta)$ is replaced by $\widehat{\psi}_{\KK}^{(\varph)}(\zeta) = \langle  \widehat{\Lambda}_{\UU}^{(\varph)}\widehat{C}_{\Upsilon\Upsilon,\KK}^{-1}\zeta,\widehat{C}_{\Upsilon\Upsilon,\KK}^{-1} \zeta\rangle.$ % under the assumptions that each of  \ref{thm2i0}-\ref{thm2i2} requires. 
	\end{theoremA}
	
	\subsection{Inference on $\IRF_{21,h}$ with endogenous predictors}
	%As in Section \ref{sec:est2}, we now let $\XX_t$ be an observable predictor with an additive measurement error $e_t$, and let $\ZZ_t$ is another variable satisfying $\mathbb{E}[ \ZZ_t \otimes u_t]=0$ and $\mathbb{E}[ \ZZ_t \otimes \XX_t] \neq 0$. In this section, we use the notation introduced in Section \ref{sec:est2} and let $C_{Y\ZZ} = \mathbb{E}[\ZZ_t\otimes Y_t]$ and $\widehat C_{Y\ZZ } = T^{-1} \sum_{t=1}^T \ZZ_t\otimes Y_t $.  We employ the following assumptions:
	As in Appendix~\ref{sec:est2}, we now let $\XX_t$ in \eqref{eq: model: benchmark: reduced5} be an endogenous predictor and let $\ZZ_t$ be another variable satisfying $C_{\UUUU\ZZ} = \mathbb{E}[ \ZZ_t \otimes \UUUU_{h,t}]=0$ and $C_{\XX\ZZ}=\mathbb{E}[ \ZZ_t \otimes \XX_t] \neq 0$. In this section, we reuse the notation introduced in Appendix~\ref{sec:est2}, and let $C_{Y\ZZ} = \mathbb{E}[\ZZ_t \otimes Y_{h,t}]$, %$C_{\UUUU\ZZ} = \mathbb{E}[\ZZ_t \otimes \UUUU_{h,t}]$,
	$\widehat C_{Y\ZZ} = T^{-1} \sum_{t=1}^T \ZZ_t \otimes Y_{h,t}$, $\widehat C_{\XX\ZZ} = T^{-1} \sum_{t=1}^T \ZZ_t \otimes \XX_{t}$ and $\widehat C_{\UUUU\ZZ} = T^{-1} \sum_{t=1}^T \ZZ_t \otimes \UUUU_{h,t}$. We employ the following assumption:
	
	%\begin{assumption} \label{assum1addadd} . % and for any orthonormal basis $\{g_j\}_{j\geq 1}$ of $\mathcal {H}$, $\sum_{j=1}^\infty \langle \theta_h, g_j \rangle^2 < \infty$. %, i.e., $\langle \theta_h, v_j \rangle = 0$ for all $j$ such that $\lambda_j = 0$. 
	%\end{assumption}
	\begin{assumpB} \label{assumpIV2add}
		\begin{enumerate*}[(i)]
			\item \label{assumpIV2adda0} $\mathcal A_h\ker C_{\XX\ZZ}=\{0\}$ and $C_{\XX\ZZ}\neq 0$;
			\item \label{assumpIV2adda} $Y_{h,t} = A_h \XX_t + \UUUU_{h,t}$ holds with $\mathbb{E}[\XX_t \otimes \UUUU_{h,t}]\neq 0$ and $\mathbb{E}[\ZZ_t \otimes \UUUU_{h,t}]=0$; 
			\item\label{assumpIV2addb} $\{\XX_t\}$, $\{\ZZ_t\}$, $\{\UUUU_{h,t}\}$, % $\{\XX_t \otimes \ZZ_t - C_{\ZZ\XX}\}$ 
			and $\{\ZZ_t \otimes \UUUU_{h,t}\}$ are stationary and $L^4$-$m$-approximable; 
			% \item For some iid sequence $\{\varepsilon\}_{t\in \mathbb{Z}}$ satisfying $\mathbb{E}[\varepsilon_t]=0$ and $\mathbb{E}[\|\varepsilon_t\|^{2+\delta}] < \infty$, 
			% $$X_t = \sum_{j=0}^\infty \Psi_j \varepsilon_{t-j},$$
			% where $\{\Psi_j\}_{j\geq 0}$ is a sequence of linear operators satisfying $\sum_{j=1}^\infty j\|\Psi_j\| < \infty$.  
			\item \label{assumpIV2addc} $\|\widehat{C}_{\XX\ZZ} -C_{\XX\ZZ}\|_{\op} = O_p(T^{-1/2})$, $\|\widehat{C}_{\UUUU\ZZ}\|_{\op} = O_p(T^{-1/2})$  and $\|\widehat{\Delta}_{11}^{-1} - \Delta_{11}^{-1}\| = O_p(T^{-1/2})$; 
			\item\label{assumpIV3adda} Assumption \ref{assumpIV2}\ref{assumpIV3a} holds;
			\item\label{assumpIVaddb} for $\varsigma > 1/2$ and $\gamma > 1/2$, $|\langle A_{h}\wv_j,\ww_\ell\rangle| \leq \CC j^{-\varsigma}\ell^{-\gamma}$.
			%   \item\label{assumpIV2d} for some  $\rho>2$ and $\CC>0$, $\llambda_j^2 \leq \CC j^{-\rho}$ and $\llambda_j^2-\llambda_{j+1}^2 \geq \CC j^{-\rho-1}$. 
		\end{enumerate*} 
	\end{assumpB}
	Based on the population moment equation $C_{Y\ZZ} = A_h C_{\XX\ZZ}$, which holds under the above assumptions, we construct our proposed estimator as follows:
	\begin{equation}
		\widetilde{A}_h =  \widehat{C}_{Y\ZZ }  \widehat{C}_{\XX\ZZ,\KKK}^{-1}.
	\end{equation}
	%For our asymptotic analysis, we employ the following assumptions: 
	% \begin{assumption} \label{assumpIV3add} \begin{enumerate*}[(i)] 
			
			%\end{enumerate*}
			% \end{assumption}
		
		\begin{theoremA} \label{thm1add} Let Assumption \ref{assumpIV2add} hold and $T\reg^{2+4/\rho} \to \infty$. Then, 
			\begin{equation}
				\|\widetilde{A}_h -A_h\|_{\op} = O_p(T^{-1/2}\reg^{-1/2-2/\rho} + \reg^{(2\varsigma-1)/2\rho}).
			\end{equation}
		\end{theoremA}

		For statistical inference on $\langle A_h\zeta,\varph \rangle$ for $\zeta = \left[\begin{smallmatrix}\zeta_1 \\0 \end{smallmatrix} \right]\in \elltwo$ and $\varph \in \mathcal H$ as in Appendix~\ref{sec_appen_infer}, we employ the following assumptions: below, $\wwtw_j, d_{m,j}(\zeta)$,  and $\tilde{r}_t(j,\ell)$ are defined as in Appendix~\ref{sec:est2}, and $\Lambda_{\UUU}^{(\varph)}$ is defined by replacing $\Upsilon_t$ with $\ZZ_t$ in \eqref{eqlrvadd}.
		\begin{assumpB} \label{assumpIV4add} 
			\begin{enumerate*}[(i)] 
				\item \label{assumpIV4adda} Assumption \ref{assum4}\ref{assum4a} holds;
				\item \label{assumpIV4addaa}  $\zeta \notin \ker C_{\ZZ\XX}$ and $\Lambda_{\UUU}^{(\varph)} \wwtw_j \neq 0$ for all $\wwtw_j $ corresponding to nonzero eigenvalues of $C_{\XX\ZZ}$,
				\item  \label{assumpIV4addbb} $\sum_{j=1}^{m} \sum_{\ell=1}^{m} d_{m,j}(\zeta) d_{m,\ell}(\zeta) \langle \Lambda_{\UUU}^{(\varph)}\wwtw_j,\wwtw_{\ell} \rangle \to \CC > 0$ as $m\to \infty$; \item \label{assumpIV4addcc} $\sup_{1\leq t\leq T} \|\ZZ_t\| = O_p(1)$.
			\end{enumerate*}
		\end{assumpB}
		
		\begin{assumpB} \label{assumpIV5add} 
			\begin{enumerate*}[(i)] 
				\item \label{assumpIV5addb} Assumption \ref{assumpIV5}\ref{assumpIV5b} holds; %$\mathbb{E}[\langle \ZZ_t,w_j \rangle \XX_t] \leq \CC \llambda_j^2$, $\mathbb{E}[\langle \XX_t,w_j \rangle \ZZ_t] \leq \CC \llambda_j^2$  and  $\mathbb{E}[\tilde{r}_t(j,\ell)\tilde{r}_{t-s}(j,\ell)]\leq \CC s^{-\tilde{\CC}}\mathbb{E}[\tilde{r}_t^2(j,\ell)]$ for some $\tilde{\CC}>1$ and $s\geq 1$;
				\item\label{assumpIV5addc} for some $\delta>1/2$, $| \langle \zeta_1,  \DDD_{11}^{-1} \DDD_{12} \wv_j \rangle| \leq \CC j^{-\delta}$.   
			\end{enumerate*}
		\end{assumpB}

		We let $\widehat{\Lambda}_{\UUU}^{(\varph)}$ be the sample counterpart of ${\Lambda}_{\UUU}^{(\varph)}$ defined as follows:
		\begin{align}
			%  &\Lambda_{\UUU}^{(\varph)} =\sum_{s=-\infty}^\infty \mathbb{E}[u_{h,t}^{(\varph)}\ZZ_t \otimes u_{h,t-s}^{(\varph)}\ZZ_s], \quad  %= \sum_{j=1}^\infty \tilde{\mu}_j \tilde{\varpi}_j \otimes \tilde{\varpi}_j, \quad \tilde{\mu}_1 \geq \tilde{\mu}_2 \geq \ldots\geq 0 
			&\widehat{\Lambda}_{\UUU}^{(\varph)} =  \frac{1}{T} \sum_{s=-h}^h\mathrm{k}\left(\frac{s}{\bdw}\right) \sum_{1 \leq t, t-s \leq T}  \tilde{\UUUU}_{h,t}^{(\varph)}\ZZ_t \otimes \tilde{\UUUU}_{h,t-s}^{(\varph)}\ZZ_{t-s},
		\end{align}
		where $\tilde{\UUUU}_{h,t}^{(\varph)} = \langle  Y_{h,t} -\widetilde{A}_h \XX_t, \varph \rangle$.
		%where the second equality follows from the spectral decomposition of $ \Lambda_{\UUU}$ as in \eqref{eqspectrallambda}. 
		We then let $\omega_{\KKK}^{(\varph)}(\zeta)$ be defined by
		\begin{equation}
			\omega_{\KKK}^{(\varph)}(\zeta) %=  \langle  {\Lambda}_{\UUU}^{(\varph)} {C}_{\XX\ZZ,\KKK}^{-1}\zeta,  {C}_{\XX\ZZ,\KKK}^{-1} \zeta\rangle 
			=  \langle {C}_{\ZZ \XX,\KKK}^{-1} {\Lambda}_{\UUU}^{(\varph)} {C}_{\XX\ZZ,\KKK}^{-1}\zeta,  \zeta\rangle. 
		\end{equation}
		As in \eqref{eqdecomthetaad}, we write $\langle (\widetilde{A}_h-A_h)\zeta,\varph\rangle=\widetilde{\Theta}_1+\widetilde{\Theta}_{2A} + \widetilde{\Theta}_{2B}$, where
		\begin{equation} 
			\widetilde{\Theta}_1 =  \langle (\widetilde{A}_h-A_h\widehat{\PP}_{\KKK}^\ast)\zeta, \varph \rangle,\quad \widetilde{\Theta}_{2A}=  \langle (A_h\widehat{\PP}_{\KKK}^\ast -A_h{\PP}_{\KKK}^\ast)  \zeta, \varph \rangle, \quad \widetilde{\Theta}_{2B}=  \langle (A_h{\PP}_{\KKK}^\ast -A_h)  \zeta, \varph \rangle
		\end{equation}
		and $\widehat{\PP}_{\KK}^\ast$ (resp.\ ${\PP}_{\KK}^\ast$) is the adjoint of $\widehat{\PP}_{\KK}$ (resp.\  ${\PP}_{\KK}^\ast$) defined as in \eqref{eqwkadd}.
		\begin{theoremA} \label{thm2addadd}  Suppose that Assumptions \ref{assumpIV2add}-\ref{assumpIV4add} hold and $T\reg^{3+4/\rho} \to \infty$. Then the following holds:
			\begin{enumerate}[(i)]
				\item\label{thm2addaddi0} $\sqrt{{T}/{\omega_{\KK}^{(\varph)}(\zeta)}}\widetilde{\Theta}_1 \to_d N(0,1)$.
				%\item\label{thm2i0} \begin{equation} 
					%\sqrt{{T}/{\psi_{\KK}(\zeta)}}\widehat{\Theta}_1 \to_d N(0,1). %\label{eqthm1}
					%\end{equation}
				\end{enumerate}
				%Suppose that $\langle \theta_{h}, \zeta\rangle \neq  \langle \theta_{1,h}, \zeta\rangle$, 
				If Assumption \ref{assumpIV5add} is additionally satisfied with $\rho/2 + 2 < \varsigma+ \delta$, the following hold:
				\begin{enumerate}[(i)]\addtocounter{enumi}{1}
					\item\label{thm2addaddi1}  If $\omega_{\KK}^{(\varph)}(\zeta) \to_p \infty$, $
					\sqrt{{T}/{\omega_{\KK}^{(\varph)}(\zeta)}}\widetilde{\Theta}_{2A} \to_p 0$.
					\item \label{thm2addaddi2} If $\omega_{\KK}^{(\varph)}(\zeta) \to_p \infty$ and $T^{1/2}\reg^{(\delta+\varsigma-1)/\rho} \to 0$, 
					$\sqrt{{T}/{\omega_{\KK}^{(\varph)}(\zeta)}}\widetilde{\Theta}_{2B} \to_p 0.$
					%and thus, if $\psi_{\KK}(\zeta) \to_p \infty$, 
					%$$\sqrt{{T}/{\psi_{\KK}(\zeta)}}\widehat{\Theta} \to_d N(0,1).$$
					
					%   \item If  HH holds and $\delta + \varsigma > \rho/2 + 2$
					%\item\label{thm2i3} 
				\end{enumerate}
				The results given in \ref{thm2i0}-\ref{thm2i2} hold when $\omega_{\KK}^{(\varph)}(\zeta)$ is replaced by $\widehat{\omega}_{\KKK}^{(\varph)}(\zeta) = \langle \widehat{C}_{\ZZ\XX,\KK}^{-1} \widehat{\Lambda}_{\UUU}^{(\varph)}\widehat{C}_{\XX\ZZ,\KK}^{-1}\zeta,   \zeta\rangle.$ % under the assumptions that each of  \ref{thm2i0}-\ref{thm2i2} requires. 
			\end{theoremA}

			\section{Mathematical Appendix} \label{sec:app:pf}
			\subsection{Preliminaries} \label{sec_prelim}
			\subsubsection{Hilbert-valued random elements and linear operators}\label{sec_prelim1}
			Let $\mathcal H_1$ (resp.\ $\mathcal H_2$) be a separable Hilbert space with inner product $\langle \cdot,\cdot \rangle_1$ (resp.\ $\langle \cdot,\cdot \rangle_2$)  
			and norm $\|\cdot\|_1$  (resp.\ $\|\cdot\|_2$). We call $X$ a $\mathcal H_1$-valued random variable if it is a measurable map from the underlying probability space to $\mathcal H_1$. $X$ is said to be square-integrable if $\mathbb{E}[\|X\|_1^2]<\infty$. We let $\otimes$ denotes the tensor product associated with $\mathcal H_1$, $\mathcal H_2$, or both, i.e., for any $\zeta \in \mathcal H_k$ and $h \in \mathcal H_\ell$, 
			\begin{equation}
				\zeta\otimes h (\cdot) = \langle \zeta,\cdot \rangle_k h, 
			\end{equation}
			which is a map from $\mathcal H_k$ to $\mathcal H_\ell$ for $k \in \{1,2\}$ and $\ell \in \{1,2\}$.
			%	Note that $\zeta\otimes  h$ is a map from $\mathcal H_1$ to $\mathcal H_2$. 
			
			For any square-integrable random variables $X \in \mathcal H_k$ and $Y \in \mathcal H_\ell$, $\mathbb{E}[X \otimes Y]$ is called the covariance operator of $X$ and $Y$. We call $A:\mathcal H_k \to \mathcal H_\ell$ a bounded linear map if $A$ is linear  and its uniform operator norm, defined as $\|A\|_{\op}=\sup_{\| x\|\leq 1, x \in \mathcal H_k} \|Ax\|_\ell$, is bounded. For such an $A$, $A^\ast$ denotes the adjoint defined by the property $\langle A\zeta,h \rangle_\ell=\langle \zeta,A^\ast h \rangle_k$ for all $\zeta\in \mathcal H_k$ and $h \in \mathcal H_\ell$. $A$ is called a compact operator if $\mathcal H_k=\mathcal H_\ell$ and there exist two orthonormal bases $\{h_{1,j}\}_{j\geq 1}$ and $\{h_{2,j}\}_{j\geq 1}$ and a real sequence  $\{a_j\}_{j\geq 1}$ tending to zero such that $A = \sum_{j=1}^\infty a_j h_{1,j} \otimes h_{2,j}$. In this expression, we may assume that $a_j\geq 0$ and $h_{1,j}=h_{2,j}$ if $A$ is self-adjoint (i.e., $A=A^\ast$) and nonnegative (i.e., $\langle A\zeta,\zeta \rangle_1\geq 0$ for any $\zeta$), and in this case, $a_j$ (resp.\ $\zeta_{1,j}$) becomes an eigenvalue (resp.\ eigenvector) of $A$.   A compact operator $A$ is Hilbert-Schmidt if $\sum_{j=1}^\infty \|A \zeta_j\|_k^2 < \infty$ for some orthonormal basis $\{\zeta_j\}_{j\geq1}$, and in this case, the Hilbert-Schmidt norm $\|A\|_{\HS}$ is defined as  $\|A\|_{\HS}= \sqrt{\sum_{j=1}^\infty \|A \zeta_j\|_k^2}$. It is well known that $\|A\|_{\op} \leq \|A\|_{\HS}$.
			
			\subsubsection{Approximable functional time series} \label{Section_AFTS}
			We briefly introduce the property of $L^p$-$m$-approximability, which we use in our theoretical development. A more detailed discussion on the property can be found in \cite{hormann2010}, \cite{horvath2013estimation} and \cite{berkes2013weak}.  
			
			Let $\overline{\mathcal H}$ be any arbitrary separable Hilbert space with inner product $\langle \cdot,\cdot \rangle_{\overline{\mathcal H}}$ and norm $\|\cdot\|_{\overline{\mathcal H}}$. %; particularly, in the present paper $\overline{\mathcal H}$ is typically a Cartesian product of Hilbert spaces where a tuple of random elements take values,  such as $\mathcal H \times \mathbb{R}^{m}$ (in Assumption \ref{assum2}) or $\mathcal H \times \mathcal H \times \mathbb{R}^m$ (in Assumption \ref{assumpIV2}) for some $m \geq 1$.   
			For some measurable function %$f:(\overline{\mathcal H})^\infty \to \overline{\mathcal H}$, 
			$f$, $p\in \mathbb{N}$, $\delta \in (0,1)$, and iid elements $\{e_t\}$, we consider a $\overline{\mathcal H}$-valued sequence $h_t$ given by 
			\begin{equation} \label{lpmcon1}
				h_t = f(e_t,e_{t-1},\ldots)  
			\end{equation}
			satisfying
			\begin{equation} \label{lpmcon2}
				\mathbb{E}[h_t]=0, \quad \mathbb{E}[\|h_t\|^{p+\delta}] <\infty, \quad 
				\sum_{m\geq 1} \left(\mathbb{E}\left[\|h_t-h_{t,m}\|^{p+\delta}\right]\right)^{1/\kappa} < \infty\quad \text{for some $\kappa > p +\delta$}, 
			\end{equation}
			where 
			\begin{equation} %\label{lpmcon2a}
				h_{t,m}=f(e_t,\ldots,e_{t-m+1},e_{t,t-m}^{(m)},e_{t,t-m-1}^{(m)},\ldots)
			\end{equation}
			and the sequences $\{e^{(m)}_{t,k}\}$ are defined as independent copies of the sequence $\{e_{t}\}$. The sequence $\{h_t\}$ satisfying \eqref{lpmcon1}-\eqref{lpmcon2} is said to be $L^p$-$m$-approximable, which is necessarily stationary. The $L^p$-$m$-approximablility (particularly, the third condition in \eqref{lpmcon2}) has been widely used as a measure of weak dependence of functional time series; in addition to the articles mentioned at the beginning of this section, see also  \cite{horvath2014test} and \cite{HORVATH2016676}.
			
			%	In the present paper, for some stationary time series  $\tilde{h}_t$ with $\mathbb{E}[\tilde{h}_t] \neq 0$, we often require $L^p$-$m$-approximability of its centered version ($\tilde{h}_t-\mathbb{E}[\tilde{h}_t])$ (see e.g. Assumption \ref{assum2}). 
			%More generally, in this paper, for any stationary time series $\tilde{\eta}_t$ with $\mathbb{E}[\tilde{\eta}_t] \neq 0$, we say that $\{\tilde{\eta}_t\}$ is $L^p$-$m$-approximable if $\{\tilde{\eta}_t -\mathbb{E}[\tilde{\eta}_t]\}$ satisfies the requirements \eqref{lpmcon1}-\eqref{lpmcon2a}. 
			
			%Let $\eta_t = (\eta_{1,t}, \eta_{2,t},\ldots,\eta_{m,t})$ be any tuple of stationary Hilbert-valued random elements with a possibly nonzero mean. Noting that $\{\eta_{j,t} \otimes \eta_{k,t}\}$ is a sequence in the Hilbert space of Hilbert-Schmidt operators (see e.g., \citealp{Bosq2000}, p.\ 34), we say that $\eta_t$ is second-order $L^p$-$m$-approximable if $\{\eta_{j,t}-\mathbb{E}[\eta_{j,t}]\}$ and $\{\eta_{j,t}\otimes \eta_{k,t}-\mathbb{E}[\eta_{j,t}\otimes \eta_{k,t}]\}$ are $L^p$-$m$-approximable for any $j$ and $k$.

			%Note that $\sum_{j=1}^{\infty} \langle \theta_{2,h} ,v_j\rangle^2$ is convergent, implying that \eqref{eqpf04} decays to zero as $\KK$ gets larger. 
			
			%Note that $\sum_{j=1}^{\infty} \langle \theta_{2,h} ,v_j\rangle^2$ is convergent, implying that \eqref{eqpf04} decays to zero as $\KK$ gets larger. 
			
			\subsubsection{Operator matrices on a Hilbert space}\label{sec_prelim2}
			Now assume that $\mathcal H_3 = \mathcal H_1\times \mathcal H_2$ of which any element $x \in \mathcal H_3$ may be understood as a tuple $\left[\begin{smallmatrix} x_1 \\ x_2\end{smallmatrix}\right]$, where $x_1\in \mathcal H_1$,  $x_2 \in \mathcal H_2$ and the inner product on $\mathcal H_3$ is defined as $\langle \cdot,\cdot \rangle_3= \langle \cdot,\cdot \rangle_1 + \langle \cdot,\cdot \rangle_2$. We consider the coordinate projection map $P_1:\mathcal H_3 \to \mathcal H_1$ given by $P_1\left[\begin{smallmatrix} x_1 \\ x_2\end{smallmatrix}\right] = x_1$ for any $x=\left[\begin{smallmatrix} x_1 \\ x_2\end{smallmatrix}\right] \in \mathcal H_3$ and its adjoint $P_1^\ast$ defined by $P_1^\ast x_1 = \left[\begin{smallmatrix} x_1 \\ 0\end{smallmatrix}\right]$ for any  $x_1\in \mathcal H_1$. $P_2$ and $P_2^\ast$ can similarly be defined associated with $\mathcal H_2$. Then it is obvious that 
			\begin{equation} \label{eqopmatrix1}
				Ax = A \left[\begin{smallmatrix} x_1 \\ x_2\end{smallmatrix}\right]  = A \left[\begin{smallmatrix} x_1 \\ 0\end{smallmatrix}\right] + A \left[\begin{smallmatrix} 0 \\ x_2\end{smallmatrix}\right] =  A P_1^\ast x_1 + A P_2^\ast x_2 = \left[\begin{smallmatrix} P_1 A P_1^\ast x_1 \\  P_2A P_1^\ast x_1 \end{smallmatrix}\right] + \left[\begin{smallmatrix} P_1 A P_2^\ast x_2 \\  P_2A P_2^\ast x_2 \end{smallmatrix}\right],
			\end{equation}
			where the last equality follows from that $A P_1^\ast x_1 $ (resp.\ $A P_2^\ast x_2 $ ) can be written as the tuple $\left[\begin{smallmatrix} P_1A P_1^\ast x_1 \\  P_2A P_1^\ast x_1 \end{smallmatrix}\right]$ (resp.\ $\left[\begin{smallmatrix} P_1 A P_2^\ast x_2 \\  P_2A P_2^\ast x_2 \end{smallmatrix}\right]$). The result given in \eqref{eqopmatrix1} can be understood similarly as a matrix transformation of $\left[\begin{smallmatrix} x_1 \\ x_2 \end{smallmatrix}\right]$, which is regarded as a (2$\times$1) vector, where the transformation matrix is given by 
			\begin{equation} 
				\label{eqopmatrix}	\mathcal A  = \begin{bmatrix} P_1A P_1^\ast & P_1A P_2^\ast \\ P_2A P_1^\ast &  P_2A P_2^\ast \end{bmatrix}.
			\end{equation}
			Thus the linear map $A$ on $\mathcal H_3$ may be understood as the operator matrix $\mathcal A$ given above, and the operator matrices that appear in Sections \ref{sec: svar}-\ref{sec:est} are defined in this way. From this result, we find that the cross-covariance operator of $\mathcal H_3$-valued random element $Z$ and $\widetilde{Z}$, $\mathbb{E}[Z\otimes \widetilde{Z}]$, can also be written as the following operator matrix:
			\begin{equation} \label{eqcovop}
				\mathbb{E}[Z\otimes \widetilde{Z}]  = \begin{bmatrix} P_1\mathbb{E}[Z\otimes \widetilde{Z}] P_1^\ast & P_1\mathbb{E}[Z\otimes \widetilde{Z}] P_2^\ast \\ P_2\mathbb{E}[Z\otimes\widetilde{Z}] P_1^\ast &  P_2\mathbb{E}[Z\otimes \widetilde{Z}] P_2^\ast \end{bmatrix} = \begin{bmatrix} \mathbb{E}[ P_1Z \otimes P_1 \widetilde{Z}] & \mathbb{E}[P_2 Z\otimes P_1 \widetilde{Z}]  \\ \mathbb{E}[P_1 Z\otimes P_2 \widetilde{Z}] &  \mathbb{E}[P_2 Z\otimes P_2 \widetilde{Z}] \end{bmatrix},
			\end{equation}
			where we used the equality that $B\mathbb{E}[Z\otimes \widetilde{Z}]C = \mathbb{E}[C^\ast Z\otimes B Z]$ for any bounded linear maps $B$ and $C$.  
			
			If $A$ is a map from $\mathcal H_3$ to $\mathcal H_j$ for $j=1$ or $2$, then it is obvious from the above discussion that this operator $A$ can be represented as the operator matrix 
			\begin{equation}
				\mathcal A = \begin{bmatrix}	P_jAP_1^\ast  & P_jA P_2^\ast \end{bmatrix}. \label{eqopmatrix2}
			\end{equation}

				\subsection{Proofs of the results in Section \ref{sec_extension}}\label{sec:ext_pf}
				We first discuss some useful lemmas, which will be repeatedly used in the subsequent sections, and then provide proofs of the theoretical results. 	With a slight abuse of notation, we will use $\langle \cdot, \cdot \rangle$ (resp. $\|\cdot\|$) to denote the inner product (resp. norm) regardless of the underlying Hilbert space. This may cause little confusion while significantly reducing notational burden. As another way to simplify notation, we let $I_1$ denote the identity map on $\mathbb{R}^m$ for various $m$ depending on the context. This is used together with $I_2$, denoting the identity map on $\mathcal H$, and $I$, denoting the identity map on $\mathbb{R}^m \times \mathcal H$.

				\paragraph{Proof of Theorem \ref{thm1addadd}}
				As in our proof of Theorem \ref{thm1}, let  ${v}_j$ denote $\sgn(\langle \hat{v}_j,v_j \rangle)v_j$ for notational simplicity.
				We note that $\|\widehat{A}_h -A_h\|_{\op} = \|\widehat{A}_h^\ast -A_h^\ast\|_{\op}$ and thus analyze the latter to obtain the desired result. Observe that 
				\begin{align} \label{eqthetadd}
					\widehat{A}_h^\ast &=  \widehat{P}_{\KK}  A_h^\ast   + \widehat{C}_{\Upsilon\Upsilon,\KK}^{-1}  \widehat{C}_{\Upsilon \UUUU} = 
					\begin{bmatrix}
						A_{1,h}^\ast + \widehat{\CCC}_{11}^{-1}\widehat{\CCC}_{12}(I_2-\widehat{\Pi}_{\KK})  A_{2,h}^\ast \\
						\widehat{\Pi}_{\KK} A_{2,h}^\ast
					\end{bmatrix} + \widehat{C}_{\Upsilon\Upsilon,\KK}^{-1}  \widehat{C}_{\Upsilon \UUUU}.
					% \\ &= \mathcal A_{1,h} +  \left(\widehat{\Gamma}_{11}^{-1}\widehat{\Gamma}_{12}(I_2-\widehat{\Pi}_{\KK}) + \widehat{\Pi}_{\KK}\right)\mathcal A_{2,h}  + \widehat{C}_{\Upsilon\Upsilon,\KK}^{-1}  \widehat{C}_{u \Upsilon}.
				\end{align}
				%\begin{align}
				%\widehat{P}_{\KK}\mathcal A^\ast &=\widehat{C}_{\Upsilon\Upsilon,\KK}^{-1} \widehat{C}_{\Upsilon\Upsilon} 
				% = \begin{bmatrix}
					%I_1 & \widehat{\CCC}_{11}^{-1}\widehat{\CCC}_{12}(I_2-\widehat{\Pi}_{\KK}) \\ 0 & \widehat{\Pi}_{\KK}
					%\end{bmatrix}\begin{bmatrix}
					% A_{11}^\ast & A_{12}^\ast \\  A_{21}^\ast & A_{22}^\ast
					%\end{bmatrix}  \\ &= \begin{bmatrix}
					%\mathcal A_{11}^\ast + \widehat{\CCC}_{11}^{-1}\widehat{\CCC}_{12}(I_2-\widehat{\Pi}_{\KK})  A_{12}^\ast  &\mathcal A_{12}^\ast + \widehat{\CCC}_{11}^{-1}\widehat{\CCC}_{12}(I_2-\widehat{\Pi}_{\KK})  A_{22}^\ast \\ \widehat{\Pi}_{\KK} A_{21}^\ast &\widehat{\Pi}_{\KK} A_{22}^\ast
					%\end{bmatrix}
					%\end{align} 
					%\begin{align}
					%\widehat{P}_{\KK} A_h^\ast &=\widehat{C}_{\Upsilon\Upsilon,\KK}^{-1} \widehat{C}_{\Upsilon\Upsilon} 
					% = \begin{bmatrix}%I_1 & \widehat{\CCC}_{11}^{-1}\widehat{\CCC}_{12}(I_2-\widehat{\Pi}_{\KK}) \\ 0 & \widehat{\Pi}_{\KK}
						%\end{bmatrix}\begin{bmatrix}
						% A_{1}^\ast \\ A_{2}^\ast 
						%\end{bmatrix}  \\ &= \begin{bmatrix}
						% A_{1,h}^\ast + \widehat{\CCC}_{11}^{-1}\widehat{\CCC}_{12}(I_2-\widehat{\Pi}_{\KK})  A_{2,h}^\ast   \\\widehat{\Pi}_{\KK} A_{2,h}^\ast
						%\end{bmatrix}
						%\end{align} 
						From nearly identical arguments used in our proof of Theorem \ref{thm1} (and the fact that $\|\widehat{C}_{\UUUU\Upsilon }\|_{\op} = \|\widehat{C}_{\Upsilon \UUUU}\|_{\op}$), the following can be shown: (a) $\widehat{C}_{\Upsilon\Upsilon,\KK}^{-1}  \widehat{C}_{\Upsilon \UUUU} = O_p(\reg^{-1/2}T^{-1/2})$, (b) $\widehat{A}_h^\ast - A_h^\ast=O_p((\widehat{\Pi}_{\KK}-{\Pi}_{\KK})A_{2,h}^\ast) + O_p((I_2-{\Pi}_{\KK})A_{2,h}^\ast) + O_p(\reg^{-1/2}T^{-1/2}),$ (c) $\|v_j-\hat{v}_j\| = O_p((\lambda_j^{-2}-\lambda_{j+1}^{-2})(\lambda_j+\lambda_{j+1}) T^{-1/2})$ and (d) $\reg^{1/2}\KK^{\rho/2} = O_p(1)$. From the latter two results and similar arguments used to show \eqref{eqappadd01}, we find that 
						\begin{equation}
							\|(\widehat{\Pi}_{\KK}- {\Pi}_{\KK})A_{2,h}^\ast\|_{\op} % &=  \|\sum_{j=1}^{\KK} \langle  A_{2,h}\hat{v}_j, \cdot\rangle \hat{v}_j -  \sum_{j=1}^{\KK} \langle   A_{2,h}{v}_j, \cdot \rangle {v}_j\|_{\op} \\ &
							= \|\sum_{j=1}^{\KK} \langle A_{2,h}(\hat{v}_j-v_j),\cdot \rangle \hat{v}_j -  \sum_{j=1}^{\KK}\langle  A_{2,h} {v}_j,\cdot \rangle ({v}_j-\hat{v}_j)\|_{\op} %\\ &
							%\\ &\leq O(1) \sum_{j=1}^{\KK} j^{\rho+1} (\lambda_j+\lambda_{j+1})\|\widehat{C}_{\Upsilon\Upsilon}-{C}_{\Upsilon\Upsilon}\|_{\op}
							%=  O_p(T^{-1/2})\sum_{j=1}^{\KK} j^{\rho/2+1} 
							%\leq O_p(T^{-1/2}\reg^{-1/2-2/\rho}\reg^{1/2+2/\rho}\KK^{\rho/2+2})
							={O_p(T^{-1/2}\reg^{-1/2-2/\rho})}. \nonumber
						\end{equation}
						We also note that, under Assumption \ref{assum2add}, the following holds:
						\begin{align}\label{eqpf04adda}
							\|(I_2-{\Pi}_{\KK})A_{2,h}^\ast\|_{\op}^2  = \sum_{j=\KK+1}^{\infty}  \|A_{2,h}v_j\|^2 \leq O(1)\sum_{j=\KK+1}^{\infty} j^{-2\varsigma} \leq O( \reg^{(2\varsigma-1)/\rho}). 
						\end{align}
						From these results, the desired result is established.  \qed
						%Note that $\sum_{j=1}^{\infty} \langle \theta_{2,h} ,v_j\rangle^2$ is convergent, implying that \eqref{eqpf04} decays to zero as $\KK$ gets larger. 
						%% (UPTO HERE)
						
						\paragraph{Proof of Theorem \ref{thm2add}}
						\noindent {1. Proof of  \ref{thm2addi0}:} Note that $\langle \hat{A}_h\zeta ,\varph\rangle=\langle  (A_h\widehat{C}_{\Upsilon\Upsilon} + \widehat{C}_{\UUUU \Upsilon}) \widehat{C}_{\Upsilon\Upsilon,\KK}^{-1} \zeta ,\varph\rangle=  \langle A_h\widehat{P}_{\KK}^\ast\zeta ,\varph\rangle   + \langle \widehat{C}_{\UUUU \Upsilon}  \widehat{C}_{\Upsilon\Upsilon,\KK}^{-1}\zeta ,\varph\rangle$. %Since $\varph \in \mathcal H$ and $\zeta \in \mathbb{R}^{k+3}$, we find that $\langle \widehat{C}_{\Upsilon u}  \widehat{C}_{\Upsilon\Upsilon,\KK}^{-1}\zeta ,\varph\rangle = \langle P_2\widehat{C}_{\Upsilon u}  \widehat{C}_{\Upsilon\Upsilon,\KK}^{-1}\zeta ,\varph\rangle$, where $P_2$ is the projection onto $\mathcal H$ along $\mathbb{R}^{k+3}$. 
						We thus find that 
						\begin{align*}
							%\sqrt{\frac{T}{\psi_{\KK}(\zeta)}} \langle \hat{\theta}_h- \widehat{\Pi}_{\KK} \theta_h, \zeta\rangle 
							\sqrt{\frac{T}{\psi_{\KK} ^{(\varph)}(\zeta)}} \widehat{\Theta}_1  
							&=  \frac{1}{\sqrt{T\psi_{\KK}^{(\varph)}(\zeta)}}\sum_{t=1}^T \langle \Upsilon_t,  \widehat{C}_{\Upsilon\Upsilon,\KK}^{-1}  \zeta \rangle  \UUUU_{h,t}^{(\varph)} =  \frac{1}{\sqrt{T\psi_{\KK}^{(\varph)}(\zeta)}}\sum_{t=1}^T \langle \Upsilon_t, \widehat{C}_{\Upsilon\Upsilon,\KK}^{-1}  \zeta \rangle  \UUUU_{h,t}^{(\varph)}. %\label{eqpf05add}
						\end{align*}
						From similar arguments used in our proof of Theorem \ref{thm2}, we find that 
						\begin{equation} 
							\sqrt{\frac{T}{\psi_{\KK}^{(\varph)}(\zeta)}} \widehat{\Theta}_1  = \frac{1}{\sqrt{T\psi_{\KK}^{(\varph)}(\zeta)}}\sum_{t=1}^T \langle \Upsilon_t, ({C}_{\Upsilon\Upsilon,\KK}^{-1} + o_p(1))  \zeta \rangle  \UUUU_{h,t}^{(\varph)} \label{eqpfa1add}
						\end{equation}
						and $\|T^{-1/2} \sum_{t=1}^T  \Upsilon_t \UUUU_{h,t}^{(\varph)}  - V_T\| \to_p 0$ for $V_T =_d N(0,\Lambda_{\UU}^{(\varph)})$ for every $T$. Thus, neglecting asymptotically negligible terms, \eqref{eqpfa1add} can be written as $\mathcal {V}_T + \mathcal W_T$, where $\mathcal{V}_T = \langle V_T, \zeta_{\KK} \rangle$, $\mathcal W_T = \langle T^{-1/2} \sum_{t=1}^T \UUUU_{h,t}^{(\varph)} \Upsilon_t-V_T, \zeta_{\KK} \rangle$ and $\zeta_{\KK}= {C}_{\Upsilon\Upsilon,\KK}^{-1}\zeta /\sqrt{\psi_{\KK}^{(\varph)}(\zeta)}$.  Note also that  
						\begin{align*}
							&\left| \frac{\langle T^{-1/2} \sum_{t=1}^T \UUUU_{h,t}^{(\varph)} \Upsilon_t - V_T, \zeta_{\KK} \rangle}{\langle V_T,\zeta_{\KK}\rangle}\right| %= \frac{|\langle T^{-1/2} \sum_{t=1}^T  \langle u_{h,t}, \varph \rangle \Upsilon_t - V_T, \frac{\zeta_{\KK}}{\|\zeta_{\KK}\|} \rangle|}{\langle V_T,\zeta_{\KK}\rangle/\|\zeta_{\KK}\|} \notag \\ 
							&\leq \frac{\sup_{\|v\|\leq 1}|\langle T^{-1/2} \sum_{t=1}^T  \UUUU_{h,t}^{(\varph)} \Upsilon_t - V_T, v \rangle|}{|\langle V_T,\frac{\zeta_{\KK}}{\|\zeta_{\KK}\|}\rangle|} 
							%=\frac{\|T^{-1/2} \sum_{t=1}^T  \langle u_{h,t}, \varph \rangle \Upsilon_t - V_T\|}{\langle V_T,\frac{\zeta_{\KK}}{\|\zeta_{\KK}\|}\rangle} 
							\leq \frac{o_p(1)}{|\langle V_T,\frac{\zeta_{\KK}}{\|\zeta_{\KK}\|}\rangle|}.
						\end{align*}
						We observe that $\zeta_{\KK}/\|\zeta_{\KK}\| = \sum_{j=1}^{\KK} c_{\KK,j}(\zeta) \vtw_j$ and $c_{\KK,j}(\zeta)  \neq 0$ for some $j$ under Assumption \ref{assum4add}\ref{assum4addaa}, and also $\sum_{j=1}^{\KK} c_{\KK,j}(\zeta)^2 = 1$ for every $\KK$. From the property of $V_T$, we find that $\langle V_T,\zeta_{\KK}/\|\zeta_{\KK}\|\rangle$ is a normal random variable with mean zero and variance $\sum_{j=1}^{\KK} \sum_{\ell=1}^{\KK} c_{\KK,j}(\zeta) c_{\KK,\ell}(\zeta) \langle \Lambda_{\UU}^{(\varph)}\vtw_j,\vtw_{\ell} \rangle$ 
						which converges to a positive constant under Assumption \ref{assum4add}\ref{assum4addbb}. This implies that $\mathcal W_T$ is  asymptotically negligible. 
						%As long as $\zeta_{\KK} \notin \Upsilon_{u\Upsilon}$, the denominator is a normal random variable with positive variance, while the numerator is $o_p(1)$. 
						%\commWK{Thus, if the latter term of \eqref{eqpfa2} is bounded as $T$ gets larger, the first term becomes asymptotically negligible.} 
						Note that by construction of $V_T$, we know that $\langle V_T, \zeta_{\KK} \rangle$ is normally distributed with mean zero; its variance is given by % the latter term is equal in distribution to $N(0, \langle \Lambda_{u\Upsilon} \frac{1}{\sqrt{\psi_{\KK}(\zeta)}}{C}_{\Upsilon\Upsilon,\KK}^{-1}\zeta\rangle, \frac{1}{\sqrt{\psi_{\KK}(\zeta)}}{C}_{\Upsilon\Upsilon,\KK}^{-1}\zeta)$, 
						$\langle \Lambda_{\UU}^{(\varph)}  \zeta_{\KK},  \zeta_{\KK} \rangle =   \langle  \Lambda_{\UU}^{(\varph)} {C}_{\Upsilon\Upsilon,\KK}^{-1}\zeta,{C}_{\Upsilon\Upsilon,\KK}^{-1} \zeta \rangle/{\psi_{\KK}^{(\varph)}(\zeta)} = 1$, which establishes Theorem \ref{thm2add}\ref{thm2i0}. \\
						
						%and the fact that,  $\psi_{\KK}(\zeta) = \langle {C}_{\Upsilon\Upsilon,\KK}^{-1} \Lambda_{u \Upsilon}{C}_{\Upsilon\Upsilon,\KK}^{-1}\zeta, \zeta\rangle$, %$\frac{1}{\sqrt{T\langle \Lambda_{u \Upsilon}v, v\rangle}} \sum_{t=1}^T \langle u_{t,h}  \Upsilon_t, v \rangle \to_d N(0,1).$$ 
						%we find that (CHECK)
						%\begin{align}
						% \frac{1}{\sqrt{T\psi_{\KK}(\zeta)}} \sum_{t=1}^T \langle u_{t,h}  \Upsilon_t, {C}_{\Upsilon\Upsilon,\KK}^{-1}\zeta \rangle \to_d N(0,1)\label{eqpf06}
						%\end{align}
						%as desired.
						%\eqref{eqpf05} and \eqref{eqpf06} imply that 
						%\begin{align}
						%\sqrt{\frac{T}{\psi_{\KK}(\zeta)}} \langle \hat{\theta}_h- \widehat{\Pi}_{\KK} \theta_h, \zeta\rangle \to_d N(0,1)
						%\end{align}
						
						\noindent{2. Proof of \ref{thm2addi1}:} 
						From a little algebra, we find that %$\widehat{\Theta}_{2A} = \langle (A_{h}\widehat{P}_{\KK}^\ast - A_{h}{P}_{\KK}^\ast)\zeta, \varph\rangle = O_p( \langle A_{h}(I_2-\widehat{\Pi}_{\KK})\widehat{\CCC}_{21}\widehat{\CCC}_{11}^{-1}\zeta-A_{h}(I_2-{\Pi}_{\KK}){\CCC}_{21}{\CCC}_{11}^{-1}\zeta),\varph \rangle$. %, where $\hat{\zeta}=\widehat{\CCC}_{21}\widehat{\CCC}_{11}^{-1}\zeta$ and $\tilde{\zeta}={\CCC}_{21}{\CCC}_{11}^{-1}\zeta$. %(I_2-\widehat{\Pi}_{\KK})\widehat{\CCC}_{21}\widehat{\CCC}_{11}^{-1}$ and ${D} =(I_2-{\Pi}_{\KK}){\CCC}_{21}{\CCC}_{11}^{-1}$.  Observe that 
						%		\begin{align}
							%			&\widehat{\Theta}_{2A} = O_p(\langle A_{2,h}(I_2-\widehat{\Pi}_{\KK})(\widehat{\CCC}_{21}\widehat{\CCC}_{11}^{-1}\zeta_1-{\CCC}_{21}{\CCC}_{11}^{-1}\zeta_1),\varph \rangle) + O_p(\langle A_{2,h}({\Pi}_{\KK}-\widehat{\Pi}_{\KK}){\CCC}_{21}{\CCC}_{11}^{-1}\zeta_1,\varph \rangle),  \notag
							%	\\ & = O_p(T^{-1/2}) + \sum_{j=1}^{\KK} (\langle {v}_j, \tilde{\zeta} \rangle \langle A_{22,h} {v}_j,\varph\rangle-\langle \hat  v_j, \tilde{\zeta} \rangle \langle A_{22,h} \hat v_j,\varph\rangle) \notag \\
							%	&= O_p(T^{-1/2}) + \sum_{j=1}^{\KK} \langle {v}_j-\hat{v}_j, \tilde{\zeta} \rangle \langle A_{22,h} {v}_j,\varph\rangle + \sum_{j=1}^{\KK} \langle \hat{v}_j, \tilde{\zeta} \rangle \langle A_{22,h} ({v}_j-\hat{v}_j),\varph\rangle, \label{eqappadd001} % + \sum_{j=1}^{\KK} \langle \hat{v}_j-v_j, \tilde{\zeta} \rangle \langle A_{h} ({v}_j-\hat{v}_j),\varph\rangle. \notag
							%		\end{align}
						\begin{equation}
							\widehat{\Theta}_{2A} =  O_p(\langle A_{2,h}(I_2-\widehat{\Pi}_{\KK})(\widehat{\CCC}_{21}\widehat{\CCC}_{11}^{-1}-{\CCC}_{21}{\CCC}_{11}^{-1})\zeta_1,\varph \rangle)   + O_p(\langle A_{2,h}({\Pi}_{\KK}-\widehat{\Pi}_{\KK}){\CCC}_{21}{\CCC}_{11}^{-1}\zeta_1,\varph \rangle),  \label{eqasdf01}
							%	\\ & = O_p(T^{-1/2}) + \sum_{j=1}^{\KK} (\langle {v}_j, \tilde{\zeta} \rangle \langle A_{22,h} {v}_j,\varph\rangle-\langle \hat  v_j, \tilde{\zeta} \rangle \langle A_{22,h} \hat v_j,\varph\rangle) \notag \\
							%	&= O_p(T^{-1/2}) + \sum_{j=1}^{\KK} \langle {v}_j-\hat{v}_j, \tilde{\zeta} \rangle \langle A_{22,h} {v}_j,\varph\rangle + \sum_{j=1}^{\KK} \langle \hat{v}_j, \tilde{\zeta} \rangle \langle A_{22,h} ({v}_j-\hat{v}_j),\varph\rangle, \label{eqappadd001} % + \sum_{j=1}^{\KK} \langle \hat{v}_j-v_j, \tilde{\zeta} \rangle \langle A_{h} ({v}_j-\hat{v}_j),\varph\rangle. \notag
						\end{equation}
						where the first term is $O_p(T^{-1/2})$ since $A_{2,h}(I_2-\widehat{\Pi}_{\KK}) = O_p(1)$ and $\widehat{\CCC}_{21}\widehat{\CCC}_{11}^{-1}-{\CCC}_{21}{\CCC}_{11}^{-1} = O_p(T^{-1/2})$. Therefore, \eqref{eqasdf01} can be written as
						\begin{equation*}
							\widehat{\Theta}_{2A}  %&=  O_p(1) \sum_{j=1}^{\KK} (\langle {v}_j, {\CCC}_{21}{\CCC}_{11}^{-1}\zeta_1 \rangle \langle A_{2,h} {v}_j,\varph\rangle-\langle \hat  v_j, {\CCC}_{21}{\CCC}_{11}^{-1}\zeta_1 \rangle \langle A_{2,h} \hat v_j,\varph\rangle) + O_p(T^{-1/2}), \notag \\
							=  O_p(T^{-1/2}) +  O_p(1) (F_1+F_2+F_2), %\label{eqappadd001} %\sum_{j=1}^{\KK}  \langle v_j-\hat{v}_j, {\CCC}_{21}{\CCC}_{11}^{-1}\zeta_1\rangle \langle A_{2,h} v_j,\varph\rangle\notag  \\ &+   O_p(1) \sum_{j=1}^{\KK} \langle\hat{v}_j-v_j, {\CCC}_{21}{\CCC}_{11}^{-1}\zeta_1 \rangle \langle A_{2,h} ({v}_j-\hat{v}_j),\varph\rangle  \notag  \\ &+  O_p(1) \sum_{j=1}^{\KK}  \langle v_j, {\CCC}_{21}{\CCC}_{11}^{-1}\zeta_1 \rangle \langle A_{2,h} ({v}_j-\hat{v}_j),\varph\rangle  
							%	&= O_p(T^{-1/2}) + \sum_{j=1}^{\KK} \langle {v}_j-\hat{v}_j, \tilde{\zeta} \rangle \langle A_{22,h} {v}_j,\varph\rangle + \sum_{j=1}^{\KK} \langle \hat{v}_j, \tilde{\zeta} \rangle \langle A_{22,h} ({v}_j-\hat{v}_j),\varph\rangle, \label{eqappadd001} % + \sum_{j=1}^{\KK} \langle \hat{v}_j-v_j, \tilde{\zeta} \rangle \langle A_{h} ({v}_j-\hat{v}_j),\varph\rangle. \notag
						\end{equation*}
						where $F_1=\sum_{j=1}^{\KK}  \langle v_j-\hat{v}_j, {\CCC}_{21}{\CCC}_{11}^{-1}\zeta_1\rangle \langle A_{2,h} v_j,\varph\rangle$, $F_2= \sum_{j=1}^{\KK} \langle\hat{v}_j-v_j, {\CCC}_{21}{\CCC}_{11}^{-1}\zeta_1 \rangle \langle A_{2,h} ({v}_j-\hat{v}_j),\varph\rangle$ and $F_3=\sum_{j=1}^{\KK}  \langle v_j, {\CCC}_{21}{\CCC}_{11}^{-1}\zeta_1 \rangle \langle A_{2,h} (\hat{v}_j-{v}_j),\varph\rangle$.
						As in our proof of Theorem \ref{thm2}, it can be shown that $\langle	(\hat{\SA} - \SA)v_j,v_{\ell}\rangle = O_p(T^{-1/2})\lambda_j^{1/2}\lambda_{\ell}^{1/2}$, and combining this result with the arguments used in the proofs of (S2.16) and (S2.33) of \cite{seong2021functional}, we obtain the following: (a)~$\langle A_{2,h}(\hat{v}_j-v_j),\varph\rangle = O_p(T^{-1/2})j^{\rho/2+1-\varsigma}$ and (b) $\langle \hat{v}_j-v_j,{\CCC}_{21}{\CCC}_{11}^{-1}\zeta_1\rangle = O_p(T^{-1/2}) j^{ \rho/2 + 1-\delta}$. Using these results with similar arguments used for \eqref{eqthmpf001}, we find that  $\sqrt{{T}/{\psi_{\KK}^{(\varph)}(\zeta)}}\widehat{\Theta}_{2A}= O_p(\sqrt{1/{\psi_{\KK}^{(\varph)}(\zeta)}})$. \\
						
						\noindent{3. Proof of \ref{thm2addi2}:} 
						Observe that, for $\zeta \in \left[\begin{smallmatrix}\zeta_1\\0\end{smallmatrix}\right]$ and $\varph \in \mathcal H$, 
						\begin{align*}
							\langle A_{h}({P}_{\KK}^\ast-I)\zeta, \varph\rangle &=  \langle A_{2,h}(I-{\Pi}_{\KK})\CCC_{21}\CCC_{11}^{-1}\zeta_1, \varph\rangle =\sum_{j=\KK+1}^\infty \langle v_j, \CCC_{21}\CCC_{11}^{-1}\zeta_1\rangle \langle A_{2,h}v_j, \varph \rangle.
						\end{align*}
						From this result, it can be shown that  $\sqrt{{T}/{\psi_{\KK}^{(\varph)}(\zeta)}}\langle A_{h}({P}_{\KK}^\ast-I)\zeta, \varph\rangle = O_p(1/\sqrt{\psi_{\KK}^{(\varph)}(\zeta)})$  
						as in our proof of Theorem \ref{thm2}\ref{thm2i2} (see \eqref{eqpf09}).\\
						% which can be shown $O_p(\KK^{1-\delta-\varsigma})$  We thus find that $\sqrt{{T}/{\psi_{\KK}^{(\varph)}(\zeta)}}\langle A_{h}({P}_{\KK}^\ast-I)\zeta, \varph\rangle$ 
						%We find that the above is $O_p(1/\sqrt{\psi_{\KK}(\zeta)})$  \\
						
						\noindent{4. Proof of the asymptotic equivalence between the choices of ${\psi}_{\KK}^{(\varph)}$ and $\widehat{\psi}_{\KK}^{(\varph)}$:} Under the employed assumptions, we found that $\|\widehat{\SA}_{\KK}^{-1} - {\SA}_{\KK}^{-1}\|_{\op} \to_p 0$. Let $\hat{\UUUU}_t = \UUUU_t + \hat{\delta}_t$, where $\hat{\delta}_t = (A_h-\widehat{A}_h) \Upsilon_t$ and this is $o_p(1)$ since $\|\widehat{A}_h-A_h\|_{\op} = o_p(1)$. Noting that $\{\UUUU_{h,t}^{(\varph)}\Upsilon_t\}$ is an $L^2$-$m$-approximable sequence, we apply similar arguments used in our proof of Theorem \ref{thm2}, and find that % of along with Lemma A.1 of \cite{horvath2014test})
						%\begin{equation}
						%    \widehat{\Lambda}_{u\Upsilon} =  \frac{1}{T} \sum_{s=-h}^h\mathrm{k}(s) \sum_{t=|s|+1}^T  {u}_{h,t}\Upsilon_t \otimes {u}_{h,t-s}\Upsilon_{t-s},
						%    \end{equation}
					$\|\widehat{\Lambda}_{\UU}^{(\varph)} -{\Lambda}_{\UU}^{(\varph)} \|_{\op}\to_p 0$, from which the desired result follows. \qed  %For the latter result, refer to Theorem 1 of \cite{horvath2013estimation} along with Lemma A.1 of \cite{horvath2014test}. From these results, the desired conclusion follows.\qed 
					
					\paragraph{Proof of Theorem \ref{thm1add}}
					Observe that $\widetilde{A}_h^\ast =  \widehat{\PP}_{\KK}  A_h^\ast   + (\widehat{C}_{U\ZZ} \widehat{C}_{\XX\ZZ,\KK}^{-1})^\ast$ and $\|(\widehat{C}_{U\ZZ} \widehat{C}_{\XX\ZZ,\KK}^{-1})^\ast\|_{\op}=\|\widehat{C}_{U\ZZ} \widehat{C}_{\XX\ZZ,\KK}^{-1}\|_{\op}$ by the properties of the operator norm. It is thus deduced that 
					\begin{align} \label{eqthetaddadd}
						\widetilde{A}_h^\ast & %=  \widehat{\PP}_{\KK}  A_h^\ast   + O_p(\widehat{C}_{U\ZZ} \widehat{C}_{\ZZ\XX,\KK}^{-1})
						= \begin{bmatrix}
							A_{1,h}^\ast + \widehat{\DDD}_{11}^{-1}\widehat{\DDD}_{12}(I_2-\widehat{\PPI}_{\KK})  A_{2,h}^\ast  \\  \widehat{\PPI}_{\KK} A_{2,h}^\ast
						\end{bmatrix}+  O_p(\widehat{C}_{U\ZZ} \widehat{C}_{\ZZ\XX,\KK}^{-1}).
						% \\ &= \mathcal A_{1,h} +  \left(\widehat{\Gamma}_{11}^{-1}\widehat{\Gamma}_{12}(I_2-\widehat{\Pi}_{\KK}) + \widehat{\Pi}_{\KK}\right)\mathcal A_{2,h}  + \widehat{C}_{\Upsilon\Upsilon,\KK}^{-1}  \widehat{C}_{u \Upsilon}.
					\end{align}
					We find the following as in our proof of Theorem \ref{thm3}:
					(a) $\widehat{C}_{U\ZZ} \widehat{C}_{\ZZ\XX,\KK}^{-1} = O_p(\reg^{-1/2}T^{-1/2})$ and hence $\widetilde{A}_h^\ast - A_h^\ast=O_p((\widehat{\PPI}_{\KK}-{\PPI}_{\KK})A_{2,h}^\ast) + O_p((I_2-{\PPI}_{\KK})A_{2,h}^\ast) + O_p(\reg^{-1/2}T^{-1/2}),$ (b) $\|\hat{\wv}_j-{\wv}_j\| \leq O_p((\llambda_j^2-\llambda_{j+1}^2)^{-1}T^{-1/2})$ and (c) $\reg^{1/2}\KK^{\rho/2} = O_p(1)$. 
					From the latter two results, we find that
					\begin{align}
						\|(\widehat{\PPI}_{\KK}- {\PPI}_{\KK})A_{2,h}^\ast\|_{\op} % &=  \|\sum_{j=1}^{\KK} \langle  A_{2,h}\hat{v}_j, \cdot\rangle \hat{v}_j -  \sum_{j=1}^{\KK} \langle   A_{2,h}{v}_j, \cdot \rangle {v}_j\|_{\op} \\ &
						&= \|\sum_{j=1}^{\KK} \langle A_{2,h}(\hat{\wv}_j-\wv_j),\cdot \rangle \hat{\wv}_j -  \sum_{j=1}^{\KK}\langle  A_{2,h} {\wv}_j,\cdot \rangle ({\wv}_j-\hat{\wv}_j)\|_{\op} \notag \\ &
						%\\ &\leq O(1) \sum_{j=1}^{\KK} j^{\rho+1} (\lambda_j+\lambda_{j+1})\|\widehat{C}_{\Upsilon\Upsilon}-{C}_{\Upsilon\Upsilon}\|_{\op}
						=  O_p(T^{-1/2})\sum_{j=1}^{\KK} j^{\rho+1} 
						%\leq O_p(T^{-1/2}\reg^{-1/2-2/\rho}\reg^{1/2+2/\rho}\KK^{\rho/2+2})
						={O_p(T^{-1/2}\reg^{-1-2/\rho})}. \label{eqappadd1}
					\end{align}
					Moreover, under the employed assumptions, we deduce the following as in \eqref{eqpf04add}:
					\begin{equation}\label{eqpf04addadd}
						\|({\PPI}_{\KKK}-I_2)A_{2,h}^\ast\|_{\op}^2   \leq \sum_{j=\KK+1}^{\infty} \| A_{2,h} \wv_j\|^2 = O( \reg^{(2\varsigma-1)/\rho}).
					\end{equation}
					From \eqref{eqappadd1}-\eqref{eqpf04addadd}, the results (a)-(c) given above, and the fact that $\|\widetilde{A}_h^\ast-{A}_h^\ast\|_{\op}=\|\widetilde{A}_h-{A}_h\|_{\op}$, the desired result immediately follows. \qed 
					
					\paragraph{Proof of Theorem \ref{thm2addadd}}
					\noindent {1. Proof of  \ref{thm2addaddi0}:}  Note that $\langle \widetilde{A}_h\zeta ,\varph\rangle=\langle  (A_h\widehat{C}_{\XX\ZZ} + \widehat{C}_{\UUUU \ZZ}) \widehat{C}_{\XX\ZZ,\KK}^{-1} \zeta ,\varph\rangle=  \langle A_h\widehat{\PP}_{\KK}^\ast\zeta ,\varph\rangle   + \langle \widehat{C}_{\UUUU \ZZ}  \widehat{C}_{\XX\ZZ,\KK}^{-1}\zeta ,\varph\rangle$. %Since $\varph \in \mathcal H$ and $\zeta \in \mathbb{R}^{k+3}$, we find that $\langle \widehat{C}_{\Upsilon u}  \widehat{C}_{\Upsilon\Upsilon,\KK}^{-1}\zeta ,\varph\rangle = \langle P_2\widehat{C}_{\Upsilon u}  \widehat{C}_{\Upsilon\Upsilon,\KK}^{-1}\zeta ,\varph\rangle$, where $P_2$ is the projection onto $\mathcal H$ along $\mathbb{R}^{k+3}$. 
					From similar arguments used in our proof of Theorem \ref{thm2add}, we find that 
					\begin{equation} 
						\sqrt{\frac{T}{\omega^{(\varph)}_{\KK}(\zeta)}} \widetilde{\Theta}_1  = \frac{1}{\sqrt{T\omega^{(\varph)}_{\KK}(\zeta)}}\sum_{t=1}^T \langle \ZZ_t, ({C}_{\XX\ZZ,\KK}^{-1} + o_p(1))  \zeta \rangle  \UUUU_{h,t}^{(\varph)} \label{eqpfa1addadd}
					\end{equation}
					%From the functional central limit theorem for $L^2$-$m$-approximable sequences (see \citealp{berkes2013weak}, Theorem 1.1) and the Skorokhod representation theorem, we know that 
					and there exists a random element $V_T$ such that  $\|T^{-1/2} \sum_{t=1}^T  \ZZ_t \UUUU_{h,t}^{(\varph)}  - V_T\| \to_p 0$ where $V_T =_d N(0,\Lambda_{\UUU}^{(\varph)})$ for every $T$. Neglecting asymptotically negligible terms, \eqref{eqpfa1addadd} can be written as $
					\mathcal {V}_T + \mathcal W_T$, where $\mathcal{V}_T = \langle V_T, {C}_{\XX\ZZ,\KK}^{-1}\zeta \rangle/{\sqrt{\omega_{\KK}^{(\varph)}(\zeta)}}$ and $\mathcal W_T = \langle T^{-1/2} \sum_{t=1}^T  \UUUU_{h,t}^{(\varph)} \ZZ_t-V_T, {C}_{\XX\ZZ,\KK}^{-1}\zeta\rangle/{\sqrt{\omega_{\KK}^{(\varph)}(\zeta)}}.$ We let $\zeta_{\KK}={C}_{\XX\ZZ,\KK}^{-1}\zeta/{\sqrt{\omega_{\KK}^{(\varph)}(\zeta)}}$ and note that  
					\begin{align}
						&\left|\frac{\langle T^{-1/2} \sum_{t=1}^T  \UUUU_{h,t}^{(\varph)} \ZZ_t - V_T, \zeta_{\KKK} \rangle}{\langle V_T,\zeta_{\KKK}\rangle}\right|% = \frac{|\langle T^{-1/2} \sum_{t=1}^T u_{h,t} \ZZ_t - V_T, \frac{\zeta_{\KKK}}{\|\zeta_{\KKK}\|} \rangle|}{\langle V_T,\zeta_{\KKK}\rangle/\|\zeta_{\KKK}\|} \\ 
						%&\leq \frac{\sup_{|v|\leq 1}|\langle T^{-1/2} \sum_{t=1}^T u_{h,t} \ZZ_t - V_T, v \rangle|}{\langle V_T,\frac{\zeta_{\KKK}}{\|\zeta_{\KKK}\|}\rangle} 
						\leq \frac{\|T^{-1/2} \sum_{t=1}^T \UUUU_{h,t}^{(\varph)} \ZZ_t - V_T\|}{|\langle V_T,\frac{\zeta_{\KKK}}{\|\zeta_{\KKK}\|}\rangle|} 
						\leq \frac{o_p(1)}{|\langle V_T,\frac{\zeta_{\KKK}}{\|\zeta_{\KKK}\|}\rangle|}.
					\end{align}
					Under Assumption \ref{assumpIV4add}\ref{assumpIV4addaa}.
					$\zeta_{\KKK}/\|\zeta_{\KKK}\| = \sum_{j=1}^{\KKK} d_{\KKK,j}(\zeta) \wwtw_j$ and $d_{\KKK,j}(\zeta)  \neq 0$ for some $j$ and also $\sum_{j=1}^{\KKK} d_{\KKK,j}(\zeta)^2 = 1$ for every $\KKK$. As in our proof of Theorem \ref{thm2add}, we deduce that $\langle V_T,\zeta_{\KKK}/\|\zeta_{\KKK}\|\rangle$ is a mean-zero normal random variable and its variance is given by $
					\sum_{j=1}^{\KKK} \sum_{\ell=1}^{\KKK} d_{\KKK,j}(\zeta) d_{\KKK,\ell}(\zeta) \langle \Lambda_{\UUU}^{(\varph)} \wwtw_j, \wwtw_{\ell} \rangle$, %=   \sum_{j=1}^{\KK} \sum_{\ell=1}^{\KK} \frac{\lambda_j^{-1} \lambda_\ell^{-1}\langle \zeta,v_j \rangle\langle \zeta,v_\ell \rangle}{{\sum_{j=1}^{\KK} \lambda_j^{-2}\langle \zeta,v_j \rangle^2}} \langle \Lambda_{u\Upsilon}v_j,v_{\ell} \rangle > 0
					which converges to a positive constant by Assumption \ref{assumpIV4add}\ref{assumpIV4addbb}. This implies that $\mathcal W_T$ is asymptotically negligible.
					Moreover, from construction of $V_T$, we find that $\langle V_T, \zeta_{\KK} \rangle$ is normally distributed with mean zero and variance $ \langle \Lambda_{\UUU}^{(\varph)} {C}_{\XX\ZZ,\KK}^{-1} \zeta, {C}_{\XX\ZZ,\KK}^{-1}\zeta \rangle /  {{\omega_{\KKK}^{(\varph)}(\zeta)}}  = \langle {C}_{\ZZ \XX,\KK}^{-1} \Lambda_{\UUU}^{(\varph)} {C}_{\XX\ZZ,\KK}^{-1} \zeta, \zeta \rangle /  {{\omega_{\KKK}^{(\varph)}(\zeta)}} = 1$, from which the desired result follows. \\

					\noindent{2. Proof of \ref{thm2addaddi1}:} 
					We first note that $\widetilde{\Theta}_{2A} = O_p(\langle A_{2,h} (\widehat{D}^\ast - D^\ast) \zeta_1,\varph \rangle)$, where $\widehat{D} = \hat{\DDD}_{11}^{-1} \hat{\DDD}_{12} (I_2 - \widehat{\PPI}_{\KK})$ and $D = {\DDD}_{11}^{-1} {\DDD}_{12} (I_2 - {\PPI}_{\KK})$. From similar algebra used in \eqref{eqasdf01}, this can further be simplified as follows: 	$\widetilde{\Theta}_{2A} = O_p(T^{-1/2})+ O_p(\langle A_{2,h}({\PPI}_{\KK}-\widehat{\PPI}_{\KK}) D_0^\ast \zeta_1,\varph \rangle)$, where $D_0  =  {\DDD}_{11}^{-1} {\DDD}_{12}$. By expanding $\langle A_{2,h}({\PPI}_{\KK}-\widehat{\PPI}_{\KK}) D_0^\ast \zeta_1,\varph \rangle$, we rewrite this as $\widetilde{\Theta}_{2A} = O_p(T^{-1/2}) + O_p(1) (F_1+F_2+F_3)$, where $F_1 = \sum_{j=1}^{\KK} 	\langle \wv_j-\hat{\wv}_j, D_0^\ast {\zeta}_1 \rangle \langle A_{2,h} \wv_j,\varph\rangle$, $F_2= \sum_{j=1}^{\KK} \langle \wv_j-\hat{\wv}_j, D_0^\ast {\zeta_1} \rangle \langle A_{2,h} (\hat\wv_j-{\wv}_j),\varph \rangle$ and $F_3=\sum_{j=1}^{\KK} \langle \wv_j, D_0^\ast{\zeta_1} \rangle \langle  A_{2,h} (\wv_j-\hat{\wv}_j),\varph\rangle$.
					Observe that $\langle \phi_j, D_0^\ast \zeta_1\rangle = \langle \DDD_{11}^{-1}\DDD_{12} \phi_j , \zeta_1\rangle \leq \CC j^{-\delta}$. Moreover, under Assumptions \ref{assumpIV2add}-\ref{assumpIV5add},  the following can be shown from nearly identical arguments used in the proofs of (S2.16) and (S2.33) of \cite{seong2021functional}: (a) $\langle A_{2,h}(\hat{\wv}_j-\wv_j),\varph\rangle = O_p(T^{-1/2})j^{\rho/2+1-\varsigma}$ and (b) $\langle \hat{\wv}_j-\wv_j, D_0^\ast{\zeta_1}\rangle = O_p(T^{-1/2}) j^{ \rho/2 + 1-\delta}$. Combining all these results, we deduce that 
					$\sqrt{{T}/{\omega_{\KK}^{(\varph)}(\zeta)}}\widetilde{\Theta}_{2A}= O_p\left(1/{\sqrt{\omega_{\KK}^{(\varph)}(\zeta)}} \right)$, where similar arguments to those used for  \eqref{eqthmpf001} are applied. \\ 
					%\begin{align}
					%  &\sqrt{\frac{T}{\omega_{\KK}(\zeta)}}\widehat{\Theta}_{2A}= O_p\left(\frac{1}{\sqrt{\omega_{\KK}(\zeta)}} \right).  \label{eqthmpf001add}
					%\end{align} 
					
					\noindent {3. Proof of  \ref{thm2addaddi2}:}  Observe that $	\langle A_{2,h}({\PP}_{\KK}^\ast-I)\zeta, \varph\rangle =  \langle A_{2,h}(I_2-{\PPI}_{\KK})D_0^\ast \zeta_1, \varph\rangle $, where $D_0 = \DDD_{11}^{-1}\DDD_{12}$. Therefore, 
					\begin{equation*}
						\langle A_{2,h}({\PP}_{\KK}^\ast-I)\zeta, \varph\rangle  =\sum_{j=\KK+1}^\infty \langle  \wv_j, D_0^\ast \zeta_1\rangle \langle A_{2,h}\wv_j, \varph \rangle=\sum_{j=\KK+1}^\infty \langle  \DDD_{11}^{-1}\DDD_{12} \wv_j,  \zeta_1\rangle \langle A_{2,h}\wv_j, \varph \rangle.
					\end{equation*}
					From this result and Assumption \ref{assumpIV5add}, it can be shown that   $\sqrt{T/{\omega_{\KK}^{(\varph)}(\zeta)})}\langle A_{2,h}({\PP}_{\KK}^\ast-I)\zeta, \varph\rangle=O_p(1/\sqrt{\omega_{\KK}^{(\varph)}(\zeta)})$  
					as in our proof of Theorem \ref{thm2}\ref{thm2i2} (see \eqref{eqpf09}).
					\\ 
					
					%	\begin{equation} \label{eqpf09add}
						%		\sqrt{\frac{T}{\omega_{\KK}(\zeta)}}\langle A_{2,h}({\PP}_{\KK}^\ast-I)\zeta, \varph\rangle  \leq \frac{1}{\sqrt{\omega_{\KK}(\zeta)}} T^{1/2}\reg^{(\delta + \varsigma - 1)/\rho}.
						%	\end{equation}
					%	If $\delta$ and $\varsigma$ are large enough so that $T^{1/2}\reg^{(\delta + \varsigma - 1)/\rho} = O_p(1)$, the above is $O_p(1/\sqrt{\omega_{\KK}(\zeta)})$. \\
					
					\noindent{4. Proof of the asymptotic equivalence between the choices of $\omega_{\KK}^{(\varph)}$ and $\widehat{\omega}_{\KK}^{(\varph)}$:} The proof is similar to that given to show that $\psi_{\KK}^{(\varph)}$ can be replaced by $\widehat{\psi}_{\KK}^{(\varph)}$ in our proof of Theorem \ref{thm2add} and is thus omitted. \qed

\bibliographystyle{apalike}

\end{document}